\documentclass[11pt]{article}
\usepackage{amsfonts,amsmath,amssymb,amsthm,boxedminipage,color,url,fullpage}
\usepackage{enumerate,paralist}
\usepackage[numbers]{natbib} 
\usepackage[pdfstartview=FitH,colorlinks,linkcolor=blue,filecolor=blue,citecolor=blue,urlcolor=blue]{hyperref} \usepackage[labelfont=bf]{caption}
\usepackage{aliascnt}
\usepackage{cleveref}
\usepackage{xspace}

\usepackage{relsize}
\usepackage{graphics}

\providecommand{\remove}[1]{}
\newcommand{\Draft}[1]{\ifdefined\IsDraft\texttt{ #1} \fi}

\ifdefined\IsDraft
    \newcommand{\authnote}[2]{{\bf [{\color{red} #1's Note:} {\color{blue} #2}]}}
\else
    \newcommand{\authnote}[2]{}
\fi

\newcommand{\mparagraph}[1]{\paragraph{#1.}}

\newcommand{\sdotfill}{\textcolor[rgb]{0.8,0.8,0.8}{\dotfill}} 

\newenvironment{protocol}{\begin{proto}}{\end{proto}}

\newenvironment{algorithm}{\begin{algo}}{\vspace{-\topsep}\end{algo}}
\newenvironment{experiment}{\begin{expr}}{\vspace{-\topsep}\sdotfill\end{expr}}

\newcommand{\aka} {also known as\ }
\newcommand{\resp}{resp.,\ }
\newcommand{\ie}  {i.e.,\ }
\newcommand{\eg}  {e.g.,\ }
\newcommand{\whp}  {with high probability\ }
\newcommand{\wrt} {with respect to\ }
\newcommand{\wlg} {without loss of generality\ }
\newcommand{\cf}{{cf.,\ }}

\newcommand{\abs}[1]{\left\lvert #1 \right\rvert}
\newcommand{\ceil}[1]{\left\lceil #1 \right\rceil}

\newcommand{\seq}[1]{\langle #1 \rangle}
\newcommand{\set}[1]{\ens{#1}}

\newcommand{\floor}[1]{\left \lfloor#1 \right \rfloor}

\newcommand{\eqdef}{:=}

\newcommand{\R}{{\mathbb R}}
\newcommand{\N}{{\mathbb{N}}}
\newcommand{\Z}{{\mathbb Z}}

\newcommand{\zo}{\set{0,1}}
\newcommand{\oo}{\set{-1,1}}

\newcommand{\xor}{\oplus}

\newcommand{\eps}{\varepsilon}

\newcommand{\la}{\gets}

\newcommand{\poly}{\operatorname{poly}}
\newcommand{\polylog}{\operatorname{polylog}}
\newcommand{\loglog}{\operatorname{loglog}}

\newcommand{\Exp}{\operatorname*{E}}
\newcommand{\Var}{\operatorname{Var}}

\newcommand{\negl}{{\operatorname{neg}}}

\newcommand{\Supp}{\operatorname{Supp}}

\newcommand{\MathAlg}[1]{\mathsf{#1}}

\newcommand{\sign}{\MathAlg{sign}}

\renewcommand{\cref}{\Cref}

\newtheorem{theorem}{Theorem}[section]

\newaliascnt{lemma}{theorem}
\newtheorem{lemma}[lemma]{Lemma}
\aliascntresetthe{lemma}
\crefname{lemma}{Lemma}{Lemmas}

\newaliascnt{claim}{theorem}
\newtheorem{claim}[claim]{Claim}
\aliascntresetthe{claim}
\crefname{claim}{Claim}{Claims}

\newaliascnt{corollary}{theorem}

\aliascntresetthe{corollary}
\crefname{corollary}{Corollary}{Corollaries}

\newaliascnt{construction}{theorem}

\aliascntresetthe{construction}
\crefname{construction}{Construction}{Constructions}

\newaliascnt{fact}{theorem}
\newtheorem{fact}[fact]{Fact}
\aliascntresetthe{fact}
\crefname{fact}{Fact}{Facts}

\newaliascnt{proposition}{theorem}
\newtheorem{proposition}[proposition]{Proposition}
\aliascntresetthe{proposition}
\crefname{proposition}{Proposition}{Propositions}

\newaliascnt{conjecture}{theorem}

\aliascntresetthe{conjecture}
\crefname{conjecture}{Conjecture}{Conjectures}

\newaliascnt{definition}{theorem}
\newtheorem{definition}[definition]{Definition}
\aliascntresetthe{definition}
\crefname{definition}{Definition}{Definitions}

\newaliascnt{remark}{theorem}
\newtheorem{remark}[remark]{Remark}
\aliascntresetthe{remark}
\crefname{remark}{Remark}{Remarks}

\newaliascnt{notation}{theorem}

\aliascntresetthe{notation}
\crefname{notation}{Notation}{Notation}

\newaliascnt{proto}{theorem}

\newtheorem{proto}[proto]{Protocol}

\aliascntresetthe{proto}
\crefname{proto}{protocol}{protocols}

\newaliascnt{algo}{theorem}
\newtheorem{algo}[algo]{Algorithm}
\aliascntresetthe{algo}
\crefname{algo}{algorithm}{algorithms}

\newaliascnt{expr}{theorem}
\newtheorem{expr}[expr]{Experiment}
\aliascntresetthe{expr}
\crefname{experiment}{experiment}{experiments}

\newcommand{\Stepref}[1]{Step~\ref{#1}}

\def\FullBox{$\Box$}
\def\qed{\ifmmode\qquad\FullBox\else{\unskip\nobreak\hfil
\penalty50\hskip1em\null\nobreak\hfil\FullBox
\parfillskip=0pt\finalhyphendemerits=0\endgraf}\fi}

\def\qedsketch{\ifmmode\Box\else{\unskip\nobreak\hfil
\penalty50\hskip1em\null\nobreak\hfil$\Box$
\parfillskip=0pt\finalhyphendemerits=0\endgraf}\fi}

\newcommand{\ex}[2]{\Exp_{#1}\left[#2\right]}
\newcommand{\eex}[1]{\ex{}{#1}}

\renewcommand{\Pr}{{\mathrm {Pr}}}

\newcommand{\ppr}[2]{\Pr_{#1}\left[#2\right]}
\newcommand{\pr}[1]{\ppr{}{#1}}

\newcommand{\Sa}{\mathsf{S}}

\newcommand{\Sc}{\Sa}

\newcommand{\Ac}{\mathsf{A}}
\newcommand{\Bc}{\mathsf{B}}
\newcommand{\Cc}{\mathsf{C}}
\newcommand{\Pc}{{\mathsf{P}}}
\newcommand{\Af}{{\mathbb{A}}}

\newcommand{\cC}{{\mathcal{C}}}

\newcommand{\ens}[1]{\{#1\}}
\newcommand{\size}[1]{\left|#1\right|}

\newcommand{\out}{{\operatorname{out}}}

\newcommand{\Uni}{{\mathord{\mathcal{U}}}}

\newcommand{\prob}[1]{\mathsf{\textsc{#1}}}

\newcommand{\SD}{\prob{SD}}

\newcommand{\cW}{{\cal{W}}}

\newcommand{\pptm}{{\sc pptm}\xspace}
\newcommand{\ppt}{{\sc ppt}\xspace}

\newcommand{\cH}{{\cal{H}}}

\newcommand{\cB}{\mathcal{B}}
\newcommand{\cS}{\mathcal{S}}

\newcommand{\cs}{{\cal{S}}}

\newcommand{\cX}{{\cal{X}}}
\newcommand{\cZ}{{\cal{Z}}}
\newcommand{\cI}{{\cal{I}}}

\newcommand{\Tableofcontents}{
\thispagestyle{empty}
\pagenumbering{gobble}
\clearpage
\tableofcontents
\thispagestyle{empty}
\clearpage
\pagenumbering{arabic}
}

\newcommand{\vect}[1]{{ \bf #1}}

\newcommand{\OT}{\ensuremath{\operatorname{OT}}\xspace}

\newcommand{\CoinToss}{\ensuremath{\MathAlg{CoinFlip}}\xspace}

\newcommand{\Aadv}{{\Ac}}
\newcommand{\iAadv}{{\Af}}

\newcommand{\Abort}{\ensuremath{\MathAlg{Abort}}\xspace}
\newcommand{\Ideal}{\operatorname{IDEAL}}
\newcommand{\Real}{\operatorname{REAL}}

\newcommand{\Inote}[1]{\authnote{Iftach}{#1}}
\newcommand{\Enote}[1]{\authnote{Eliad}{#1}}
\newcommand{\Nnote}[1]{\authnote{Nissan}{#1}}

\newcommand{\fvec}[1]{f^{\mathsf{vec}}_{#1}}
\newcommand{\fhyp}[1]{f^{\mathsf{hyp}}_{#1}}

\newcommand{\BerooSign}{{\mathcal{C}}}
\newcommand{\Beroo}[1]{{\BerooSign_{#1}}}
\newcommand{\vBeroo}[1]{{\widehat{\BerooSign}_{#1}}}

\newcommand{\sBias}[2]{\widehat{\BerooSign}^{-1}_{#1}(#2)}

\newcommand{\BerzoSign}{{\mathcal{B}er}}
\newcommand{\Berzo}[1]{\BerzoSign(#1)}

\newcommand{\HT}{\ensuremath{\mathsf{HT}\xspace}}
\newcommand{\st}{\mbox{\rm s.t. }}

\newcommand{\minus}{\scalebox{0.5}[1.0]{$-$}}

\newcommand{\rpm}{\raisebox{.2ex}{$\scriptstyle\pm$}}
\newcommand{\ra}{\rightarrow}

\newcommand{\sss}{{\mathsf{sum}}}
\newcommand{\ells}{{\mathsf{\ell}}}
\newcommand{\xs}[2]{{\sss_{#1}(#2)}}
\newcommand{\xl}[2]{{\ells_{#1}(#2)}}

\newcommand{\ms}[1]{\xs{\rnd}{#1}}
\newcommand{\ml}[1]{\xl{\rnd}{#1}}
\newcommand{\NS}[1]{{\xs{n}{#1}}}
\newcommand{\NL}[1]{\xl{n}{#1}}

\newcommand{\elli}[1]{\rnd +1 -i}

\newcommand{\HypSign}{\mathcal{HG}}
\newcommand{\Hyp}[1]{{\HypSign_{#1}}}
\newcommand{\vHyp}[1]{{\widehat{\HypSign}_{#1}}}

\newcommand{\Dst}{\mathcal{D}}
\newcommand{\vDst}{\widehat{\Dst}}
\newcommand{\vNegDst}{\widehat{-\Dst}}

\newcommand{\w}{w}

\newcommand{\bias}{\MathAlg{Bias}}
\newcommand{\val}{\MathAlg{val}}

\newcommand{\isShare}[1]{{\# {#1}}}

\newcommand{\ShVecO}[2]{\vect{#1^{\isShare{#2}}}}

\newcommand{\Eps}{\mathcal{E}}

\newcommand{\cVS}[1]{\ShVecO{c}{#1}}

\newcommand{\bankV}[1]{\vect{b^{#1}}}
\newcommand{\tp}{h}

\newcommand{\vr}{v}
\newcommand{\rng}{r}

\newcommand{\partNum}{t}
\newcommand{\partNumInn}{r}

\newcommand{\secParam}{\kappa}

\newcommand{\rnd}{m}

\newcommand{\game}{\mathsf{G}}

\newcommand{\ratioo}{\MathAlg{ratio}}
\newcommand{\error}{\MathAlg{error}}

\newcommand{\trnd}{{\widetilde{\rnd}}}

\newcommand{\z}{z}

\newcommand{\Ph}{{\widehat{\Pc}}}
\newcommand{\Pih}{{\widehat{\Pi}}}

\newcommand{\Party}[1]{\Pc}

\newcommand{\PartyF}[1]{\Ph}

\newcommand{\Protocol}[1]{\Pi^{#1}}
\newcommand{\ProtocolF}[1]{\Pih^{#1}}

\newcommand{\vct}{v}
\newcommand{\const}{\lambda}
\newcommand{\aset}{\cH}

\newcommand{\Defense}{{\mathsf{Defense}}}
\newcommand{\DefenseTilde}{{\mathsf{\widetilde{Defense}}}}

\newcommand{\Coin}{{\mathsf{Coin}}}
\newcommand{\AddNoise}{{\mathsf{Noise}}}
\newcommand{\noise}{\AddNoise}

\newcommand{\share}{{\mathsf{share}}}
\newcommand{\Noise}{{\mathsf{Noise}}}

\newcommand{\alp}[1]{\alpha_{#1}}

\newcommand{\Value}{B}
\newcommand{\ElementSet}{\mathcal{A}}
\newcommand{\ElementVar}{A}
\newcommand{\element}{a}
\newcommand{\PredictAdv}{\Gamma}

\newcommand{\GoodElements}{\ElementSet^*}
\newcommand{\GoodVectors}{\mathcal{V}^*}

\newcommand{\weight}{w}

\newcommand{\HintSet}{\cH}
\newcommand{\HintFunc}{f}
\newcommand{\SecHintFunc}{g}
\newcommand{\CompHintFunc}{\SecHintFunc \circ \HintFunc}
\newcommand{\HintVar}{\HintFunc(\ElementVar)}
\newcommand{\GoodHints}{\HintSet^*}
\newcommand{\GoodHintsFinal}{\HintSet}
\newcommand{\hint}{h}

\newcommand{\tail}{\MathAlg{tail}}

\newcommand{\vctBaseLen}{s}
\newcommand{\vctLenFact}{\alpha}

\newcommand{\vctTotalLen}{\vctLenFact \cdot \vctBaseLen}
\newcommand{\hypVctLenFact}{\beta}
\newcommand{\hypTotalVctLen}{\hypVctLenFact \cdot \vctBaseLen}

\newcommand{\curRound}{i}
\newcommand{\curCoins}{c}
\newcommand{\CurCoinsVarSign}{C}
\newcommand{\CurCoinsVar}{\CurCoinsVarSign_{\curRound}}
\newcommand{\GoodCoins}{\cC_{\curRound}^*}
\newcommand{\prevCoins}{b}

\newcommand{\HintVarBin}{\HintFunc(\CurCoinsVar)}

\newcommand{\NextCoinsVar}{\sum_{j=\curRound}^{\rnd}\CurCoinsVarSign_j}
\newcommand{\NextNextCoinsVar}{\sum_{j=\curRound+1}^{\rnd}\CurCoinsVarSign_j}

\newcommand{\hypSample}{t}
\newcommand{\hypBankWeight}{p}

\newcommand{\ratioBoundFactor}{\gamma}

\newcommand{\coinsFuncSym}{\ell}
\newcommand{\coinsSumSym}{\hat{\coinsFuncSym}}
\newcommand{\coinsFunc}[1]{\coinsFuncSym(#1)}
\newcommand{\coinsSum}[1]{\coinsSumSym(#1)}
\newcommand{\coinsSumIndex}{t}

\newcommand{\biasTerm}{{\frac{t \cdot 2^t \cdot \sqrt{\log m}}{m^{1/2+1/\left(2^{t-1}-2  \right)}}}}
\newcommand{\biasTermSingle}{{\frac{2^t \cdot \sqrt{\log m}} {m^{1/2+1/\left(2^{t-1}-2  \right)}}}}
\newcommand{\biasTermSingleParam}[2]{{\frac{2^{#1} \cdot \sqrt{\log m}}{m^{1/2+1/\left(2^{#2-1}-2  \right)}}}}
\newcommand{\outcome}{{\mathsf{outcome}}}

\newcommand{\HTProtocol}{\Pi^{\HT}}
\newcommand{\HTDefenseProtocol}{{\mathsf{\widetilde{Defense^{\HT}}}}}
\newcommand{\HTDefenseRound}{{\mathsf{RoundDefense^{\HT}}}}

\newcommand{\hintSet}{{\mathcal{H}}}
\newcommand{\Strategy}{{\Bc}}

\makeatletter
\let\xx@thm\@thm
\AtBeginDocument{\let\@thm\xx@thm}
\makeatother

\title{Fair Coin Flipping:\\ Tighter Analysis  and the Many-Party Case
\Draft{\\{\small \sc Working Draft: Please Do Not Distribute}}}
\author{Niv Buchbinder\thanks{Statistics and Operations Research, Tel Aviv university. E-mail:\texttt{niv.buchbinder@gmail.com}.} \and Iftach Haitner\thanks{School of Computer Science, Tel Aviv University. E-mail:\{\texttt{iftachh@cs.tau.ac.il}, \texttt{nisnis.levi@gmail.com}, \texttt{eliadtsfadia@gmail.com}\}. Research supported by ERC starting grant 638121.}~\thanks{Member of the  
		Check Point Institute for Information Security.}  \and Nissan Levi$^{\dagger}$ \and Eliad Tsfadia$^{\dagger}$}

\begin{document}
\sloppy
\maketitle
\begin{abstract}
	
In a multi-party  \emph{fair} coin-flipping protocol, the parties output a common (close to) unbiased  bit, even when some adversarial parties try to bias the output. In this work we focus on the case of an arbitrary number of corrupted parties.  \citet{Cleve86} [STOC 1986] has shown that in  \emph{any} such $m$-round coin-flipping protocol, the corrupted parties can bias the honest parties' common output bit by $\Theta(1/m)$. For more than two decades, however,  the best known coin-flipping protocol was the one of \citet*{AwerbuchBCGM1985}  [Manuscript 1985], who presented a $t$-party, $m$-round  protocol with bias  $\Theta(t/\sqrt{m})$. This was changed by the breakthrough result of \citet*{MoranNS09} [Journal of Cryptology 2016], who constructed an $m$-round, \emph{two}-party coin-flipping protocol with optimal bias  $\Theta(1/m)$. More recently, \citet{HaitnerT17} [SIAM Journal on Computing 2017]  constructed an $m$-round, \emph{three}-party coin-flipping protocol with  bias  $O(\log^3m / m)$.  Still for the case of more than three parties, the best known protocol remained the $\Theta(t/\sqrt{m})$-bias protocol of \cite{AwerbuchBCGM1985}.


We make a step towards eliminating the above gap, presenting a $t$-party,  $m$-round coin-flipping protocol, with bias $O(\frac{t^4 \cdot 2^t \cdot \sqrt{\log m}}{m^{1/2+1/\left(2^{t-1}-2\right)}})$ for any $t\le \tfrac12 \cdot \loglog m$. This improves upon the  $\Theta(t/\sqrt{m})$-bias protocol of \cite{AwerbuchBCGM1985}, and in particular, for $t\in O(1)$ it is an $1/m^{\frac12 + \Theta(1)}$-bias protocol. For the three-party case, it is an $O(\sqrt{\log m}/m)$-bias protocol, improving over the  $O(\log^3m / m)$-bias protocol  of \cite{HaitnerT17}.

Our protocol generalizes that of  \cite{HaitnerT17},  by presenting an appropriate ``recovery protocol'' for the remaining parties to interact in, in the case that some parties abort or are caught cheating (\cite{HaitnerT17} only presented a two-party  recovery protocol, which limits their final protocol to handle three parties). We prove the fairness  of the  new protocol by  presenting  a new paradigm for analyzing fairness of coin-flipping protocols;  the claimed fairness is proved  by  mapping  the set of adversarial strategies that try to bias the honest parties' outcome in the protocol to the set of the feasible solutions of a linear program. The gain each strategy achieves is the value of the corresponding solution. We then bound the optimal value of the linear program by constructing a feasible solution to its dual.

\end{abstract}

\noindent\textbf{Keywords:} coin-flipping; fair computation; stopping time problems

\ifdefined\SBM
\thispagestyle{empty}
\pagenumbering{gobble}
\clearpage
\pagenumbering{arabic}

\else
\Tableofcontents
\fi

\section{Introduction}\label{sec:Introduction}
In a multi-party  \emph{fair} coin-flipping protocol, the parties wish to output a common (close to) unbiased bit, even though some of the parties may be adversarial and try to bias the output. More formally, such protocols should satisfy the following two properties: first, when all parties are honest (\ie follow the prescribed protocol), they all output the \emph{same} bit, and this bit is unbiased (\ie uniform over $\zo$). Second, even when some parties are corrupted (\ie collude and arbitrarily deviate from the protocol), the remaining parties should still output the \emph{same} bit, and this bit should not be too biased (\ie its distribution should be close to uniform over $\zo$). We emphasize that unlike weaker variants of coin-flipping protocol known in the literature, the honest parties should \textbf{always} output a common bit, regardless of what the corrupted parties do, and in particular  they are not allowed to abort if a cheat was detected.

 When a majority of the parties are honest, efficient and \emph{completely} fair coin-flipping protocols are known as a special case of secure multi-party computation with an honest majority \cite{BenGolWig88}.\footnote{Throughout, we assume a broadcast channel is available to the parties. By \cite{CohenHOR2016}, broadcast channel is necessary for fair coin-flipping protocol secure against one third or more, corruptions.} However, when there is no honest majority, the situation is more complex.

 \mparagraph{Negative results}
 \citet{Cleve86} showed that for \emph{any} efficient two-party $m$-round coin-flipping protocol, there exists an efficient adversarial strategy to  bias the output of the honest party by $\Theta(1/m)$. This lower bound extends to the multi-party case, with no honest majority,  via a simple reduction.

 \mparagraph{Positive results}
 \citet*{AwerbuchBCGM1985} showed that if one-way functions exist, a simple $m$-round majority protocol can be used to derive a $t$-party coin-flipping protocol with bias $\Theta(t/\sqrt{m})$.\footnote{The result of \citet{AwerbuchBCGM1985} was never published, and it is contributed to them by \citet{Cleve86} who analyzed the two-party case. \citet{Cleve86}'s analysis  extends to the many-party case in  a straightforward manner. The protocol of \cite{Cleve86} is using  family of trapdoor permutations, but  the latter were merely used to construct commitment schemes, which we currently  know how to construct from any one-way function \cite{HaitnerReVaWe09,HastadImLeLu99,Naor91}. Roughly, the $t$-party protocol of \cite{AwerbuchBCGM1985} is the following: in each round $i \in [m]$, each party $j \in [t]$ commits on a uniformly random coin $c_{i,j} \in \oo$. After the commitments phase, each party then decommits on its coin, and the parties agree on the value $c_i = \prod_{j \in [t]} c_{i,j}$. The final outcome is set to $\sign(\sum_{i=1}^m c_i)$.}

 For more than two decades, \citeauthor{AwerbuchBCGM1985}'s  protocol was the best known fair coin-flipping protocol (without honest majority), under \emph{any} hardness assumption  and for \emph{any} number of parties. In their breakthrough result, Moran, Naor, and Segev \cite{MoranNS09} constructed an $m$-round, \emph{two}-party coin-flipping protocol with optimal bias of $\Theta( 1/ m)$.  In a subsequent work, \citet*{BeimelOO10} extended the result of \cite{MoranNS09} for the multi-party case in which \emph{less than} $\frac23$ of the parties can be corrupted. More specifically, for any $\ell < \frac23 \cdot t$, they presented an $m$-round, $t$-party protocol with bias $\frac{2^{2^{2\ell -t}}}{m}$ against (up to) $\ell$ corrupted parties.  Recently, \citet{HaitnerT17}   constructed an $m$-round, \emph{three}-party coin-flipping protocol with  bias  $O(\log^3m / m)$ against two corruptions. In a subsequent work, \citet{AlonOmri16Full} extended the result of \citet{HaitnerT17} for the multi-party case in which \emph{less than} $\frac34$ of the parties can be corrupted.  More specifically, for any $t\in O(1)$ and $\ell < \frac34 \cdot t$, they presented an $m$-round, $t$-party protocol with bias $O(\log^3m / m)$ against (up to) $\ell$ corrupted parties. All the above results hold  under the assumption that oblivious transfer protocols exist.   Yet, for the case of more than three parties  (and unbounded number of corruptions),  the best known protocol  was the $\Theta(t/\sqrt{m})$-bias protocol of \cite{AwerbuchBCGM1985}.

 \subsection{Our Result}
 Our main result is a new multi-party coin flipping protocol.
 \begin{theorem}[main theorem, informal]\label{thm:mainInf}
 	Assuming the existence of oblivious transfer protocols, for any $\rnd = \rnd(n) \le \poly(n)$ and $t = t(n)\le \frac12
 	\cdot \loglog m$, there exists an $\rnd$-round, $t$-party coin-flipping protocol with bias $O(\frac{t^4 \cdot 2^t \cdot \sqrt{\log m}}{m^{1/2+1/\left(2^{t-1}-2\right)}})$ (against up to $t-1$ corrupted parties).
 \end{theorem}
 The above protocol improves upon the  $\Theta(t/\sqrt{m})$-bias protocol of \citet{AwerbuchBCGM1985} for any $t\le \tfrac12 \cdot \loglog m$. For  $t\in O(1)$, this   yields an $1/m^{\frac12 + \Theta(1)}$-bias protocol. For the three-party case, the above  yields an $O(\sqrt{\log m}/m)$-bias protocol, improving over the $O(\log^3m / m)$-bias protocol  of \citet{HaitnerT17}.

 We analyze  the new protocol by  presenting  a new paradigm for analyzing fairness of coin-flipping protocols. We upper bound the bias of the protocol by upper-bounding the value of  a linear program that \emph{characterizes it}:  there exists an onto mapping from the set of adversarial  strategies that try to bias the honest parties' outcome in the protocol, to the program's feasible solutions, such that the gain a strategy achieves is, essentially, the  value of the solution of the program it is mapped to. See \cref{intro:sec:LP} for more details.

 \subsection{The New Multi-Party Fair Coin-Flipping Protocol}\label{intro:sec:newprotocol}
 Our coin-flipping protocol follows the paradigm of  \citet{HaitnerT17}. In the following we focus on efficient \emph{fail-stop} adversaries: ones that follow the protocol description correctly and their only adversarial action is to abort prematurely (forcing the remaining  parties to decide on their common output without them). Compiling  a protocol that is secure against such fail-stop adversaries into a protocol of the same bias that is secure against \emph{any} efficient  adversary, can be done using standard cryptographic tools.
 
 In addition, we assume  the parties can \textit{securely compute with abort}  any efficient functionality, where according to this security definition, if a cheat is detected or if one of the parties aborts, the remaining parties are not required to output anything.
 The only information a party obtains from such a computation is its  local output (might be a different output per party). The order of which the outputs are given by such functionality, however,  is  arbitrary. In particular, a ``rushing'' party that aborts after obtaining its own output, prevents  the remaining  parties from getting their outputs.  For every efficient functionality,  a constant-round protocol that securely compute it with abort can be constructed using \textit{oblivious transfer} protocol. As explained in \cite{BeimelOO10}, this can be done using (a variation on) the protocol of \cite{Pass:2004:BSM:1007352.1007393}.

 The protocol of \citet{HaitnerT17} enhances  the basic majority coin-flipping protocol of \citet{AwerbuchBCGM1985} with \emph{recovery protocols} for the remaining parties to interact in, if  some of the parties abort. We consider the following generalization of the  $t$-party  $m$-round  protocol of \cite{HaitnerT17}, for arbitrary value of $t$  and odd value of $m$. The  functionality $\Defense$ and the sub-protocols $\set{\Protocol{t'}}_{t' < t}$ used in the protocol are specified later.

\begin{samepage}

\begin{protocol}[$\ProtocolF{t}  = (\Party{}_1,\Party{}_2,\ldots,\Party{}_{t})$]\label{intro:ProtocolOuter}
	
	\item For $i=1$ to $\rnd$:
	\begin{enumerate}
		
		
		\item Every (proper) subset of parties $\cZ \subsetneq \set{\Party{}_1,\Party{}_2,\ldots,\Party{}_{t}}$ securely compute $\Defense^{\size{\cZ}}()$. Let $\share^{\#\z,\cZ}$ be the output party $\Party{}_z$ received  from this call.\label{intro:step:defense}
		
		\item  The parties  securely  compute  $\Coin()$ that returns  a common uniform $\oo$ coin $c_i$.\label{intro:step:coin}\footnote{In the formal description of the protocol, see \cref{sec:protocol}, in round $i$ $\Coin$ returns $(m+1-i)^2$ coins. As shown in \cite{HaitnerT17}, given more weight to earlier rounds is necessary:  otherwise, an adversary can easily gain a bias of $1/\sqrt{m}$ by aborting in one of the last rounds.}
		
		
	\end{enumerate}
	\item[Output:] All parties output $\sign(\sum_{i=1}^m c_i) \:$   (i.e., $1$ if $\sum_{i=1}^m c_i > 0$ and $0$ otherwise).
	
	\item [Abort:] ~
	Let $\cZ \subsetneq \set{\Party{}_1,\ldots,\Party{}_t}$ be the remaining (non-aborting) parties. To decide on a common output, the parties in $\cZ$ interact in the ``recovery" protocol  $\Protocol{\size{\cZ}}$, where party $\Party{}_\z$'s private input is $\share^{\#\z,\cZ}$. If  $\cZ = \set{\Party{}_\z}$ (\ie  only a single non-aborting party remained), the  party $\Party{}_\z$ outputs $\share^{\#\z,\cZ}$.
\end{protocol}
	\end{samepage}
Note that (since $m$ is odd) the common output in  an all-honest execution  is a uniform bit. 
To instantiate the above protocol, one needs to define the functionality $\Defense^{t'}$ and the protocol $\Protocol{t'}$, for all ${t'< t}$. But first let's discuss whether  we need these functionalities and protocols at all? That is, why not simply instruct the remaining parties to re-toss the coin $c_i$ if some parties abort in \Stepref{intro:step:coin} of the $i$'th round. This simple variant is  essentially the vanilla  protocol of  \citet{AwerbuchBCGM1985}, and it is not hard to get convinced that a malicious (fail-stop) party can bias the output of the protocol by $\Theta(1/\sqrt{{m}})$. To see that, note that the sum of $m$ unbiased $\oo$ coins is roughly uniform over $[-\sqrt{m}, \sqrt{m}]$. In particular, the  probability that the sum is in $\oo$ is $\Theta(1/\sqrt{m})$. It follows that if a party aborts after seeing in \Stepref{intro:step:coin} of the first round that $c_1 =1$, and by that causes the remaining parties to re-toss $c_1$, it biases the final outcome of the protocol towards $0$ by $\Theta(1/\sqrt{m})$.

To improve upon this   $1/\sqrt{m}$ barrier, \cite{HaitnerT17}  have defined  $\Defense$ such that the expected outcome of  $\Protocol{t'}(\Defense^{t'}())$  equals $\delta_i = \pr{\sign(\sum_j c_j) =1 | c_1,\ldots,c_i}$ for every $t'$. Namely, the expected outcome of the remaining parties has \emph{not changed}, if some parties abort in \Stepref{intro:step:coin} of the protocol.\footnote{The above  definition  requires $\Defense$  to share a (secret) state with  $\Coin$, since both functionalities are defined \wrt the same coin $c_i$. This non-standard requirement is only  for the sake of presentation, and in the actual protocol we replace it with stateless functionalities that share, and maintain, their ``state'' by secret sharing it between the parties. See \cref{sec:protocol}.} The above correlation of the defense values returned by $\Defense$ in \Stepref{intro:step:defense} and the value of $c_i$ returned by $\Coin$ in \Stepref{intro:step:coin}, however,  yields that they  give some information about $c_i$, and thus the (only) weak point of the protocol has shifted to \Stepref{intro:step:defense}.  Specifically, the bias achieved by aborting  in \Stepref{intro:step:defense} of round $i$  is the difference between  $\delta_{i-1}$, the expected value of the protocol  given the coins $c_1,\ldots,c_{i-1}$ flipped in the previous rounds, and the expected outcome of the protocol given these coins and the defense values given to the corrupted parties in \Stepref{intro:step:defense}.  If done properly, only limited information about  $c_i$ is revealed in \Stepref{intro:step:defense}, and thus attacking there is not as effective as attacking in \Stepref{intro:step:coin} of the vanilla (no defense)  protocol.

For $t=2$, \cite{HaitnerT17} have set the defense for the remaining party to be a bit $b_i$ that is set  to $1$ with probability  $\delta_i$. Namely, if a party aborts in \Stepref{intro:step:coin} of the $i$'th round, the other party outputs $1$ with probability $\delta_i$, and $0$ otherwise. Since $\eex{b_i} = \delta_i$, attacking in \Stepref{intro:step:coin}  (in any round) of the protocol is useless. Moreover, since $b_i$ only leaks ``limited information'' about $\delta_i$ (and thus about $c_i$), it is  possible to show that attacking in \Stepref{intro:step:defense} (in any round) biases the protocol by (roughly) $(\delta_i-\delta_{i-1})^2$ (to compare to  the   $(\delta_i-\delta_{i-1})$ bias achieved in \Stepref{intro:step:coin} of the vanilla protocol).\footnote{To see where the square power is coming from, consider for simplicity the first round in which $\delta_0 = 1/2$. Let $\Delta \eqdef |\delta_1-1/2| \in \Theta(1/\sqrt{m})$ (\ie $\delta_1 = 1/2 + c_1 \Delta$). It follows that\\$\pr{\out=1 \mid b_1=1} = \sum_{x \in \oo}\pr{\out=1\mid c_1=x} \cdot \pr{c_1=x\mid b_1=1} = (1/2-\Delta)^2 + (1/2+\Delta)^2 = 1/2 + 2\Delta^2$. Namely, revealing $b_1$ only causes a bias of $\Theta(\Delta^2)$ (and not  $\Delta$).} These observations yield  (see \cref{intro:sec:LP}) that the protocol's bias is $\polylog (m)/m$.\footnote{More precisely, this bound was proved for the weighted variant of the above protocol, where in round $i$ the functionality $\Coin$ returns the sum of $m-i+1$ independent coins. See \cref{sec:protocol}.} Generalizing the above for even  $t=3$  is non-trivial. The defense values of the remaining parties should allow them to interact in a \emph{fair} protocol $\Protocol{2}$ of expected outcome $\delta_i$. Being fair,  protocol $\Protocol{2}$ should contain a defense mechanism of its own to avoid one of the remaining parties to bias its outcome by too much (this was not an issue in the case  $t=1$, in which  there is only one remaining party).  
Yet, \cite{HaitnerT17} managed to find such an implementation of the $\Defense$ functionality and $\Protocol{2}$ that yield a $\polylog(m)/m$-bias protocol.\footnote{The  analysis we employ in this paper, see \cref{intro:sec:LP}, shows that the bias of (a simple variant of) the \cite{HaitnerT17} protocols is actually $\sqrt{\log m}/m$.} The rather complicated approach used by \cite{HaitnerT17}  was  tailored for the case that the recovery sub-protocol $\Protocol{2}$ is a two-party protocol. In particular, it critically relies on the fact that in a two-party protocol, there is no recovery sub-protocol (rather, the remaining party decides on its output by its own). We take a  different approach to implement the $\Defense$ functionality and its accompanied recovery protocol $\Protocol{t'}$.\footnote{Actually, for subsets of size two, we are still using the mechanism of \cite{HaitnerT17} that handles  such subsets better. We ignore this subtlety for the sake of the introduction.} 

\begin{algorithm}[The  $\Defense$ functionality]
\item [Input:] $1^{t'}$

\noindent
//Recall that $c_1,\ldots,c_{i-1}$ are the coins flipped in the previous rounds, $c_i$ is the coin to be output in this round call to $\Coin$, and $\delta_i = \pr{\sign(\sum_j c_j) =1 | c_1,\ldots,c_i}$.
\begin{enumerate}
	
	\item Let $\delta_i' = \delta_i + \noise$, where $\noise$ is  random variable of expectation $0$.
	
	 (Additional requirements regarding the distribution of the noise will be given below.)
	
	\item  Let  $\share^{\#1},\ldots,\share^{\#t'}$ be  $t'$-out-of-$t'$  secret sharing of $\delta'_i$.\footnote{I.e., $\set{\share^{\#i}}$ are uniform   strings conditioned on  $\bigoplus_{i=1}^{t'}\share^{\#i} = \delta'_i$. (We assume for simplicity that $\delta_i'$ has a short binary representation.)} Return  $\share^{\#i}$ to the $i$'th party.
\end{enumerate}
\end{algorithm}
Namely, $\Defense$  computes a noisy version of $\delta_i$ and secret-shares the result between the calling parties.
\begin{protocol} [$\Protocol{t'} = (\Party{}_1,\Party{}_2,\ldots,\Party{}_{t'})$]

	 \item[Input:] Party $\Party{}_z$'s input is  $\share^{\#\z}$.
	
	\begin{enumerate}

   \item Each party $\Party{}_z$ sends its input $\share^{\#\z}$ to the other parties, and all parties set $\delta' = \bigoplus_{z=1}^{t'} \share^{\#z}$.\label{step:rec:newdelta}
	
	\item Continue as the $\delta'$-biased  version of Protocol $\ProtocolF{t'}$:
	\begin{itemize}
		\item $\Coin$ sets the coin  $c_i$ to be $1$ with probability $1/2 + \eps$ (rather than $1/2$), for $\eps\in [-1/2,1/2]$  being the value such that $\delta'$ is the probability that the sum of $m$ independent  $(1/2 + \eps)$-biased $\oo$ coins  is positive.
		\item The definition  of $\Defense$ is changed accordingly to reflect this change in the bias of the coins.
	\end{itemize}
\end{enumerate}
\end{protocol}
Since $\eex{\noise}=0$, the expected outcome of $\Protocol{t'}(\Defense(1^{t'}))$ is indeed $\delta_i$. Note that since the corrupted parties can use their shares to reconstruct the value of $\delta_i'$ sampled in the all-corrupted  calls to $\Defense$ (those calls made by subsets in which all parties are corrupted), the values returned by $\Defense$ do leak some information about $\delta_i$, and thus about the coin $c_i$. But if $\noise$ is ``noisy enough" (\ie high enough variance), then $\delta_i'$ does not leak too much information about $\delta_i$.  Hence, by taking noisy enough $\noise$, we make $\ProtocolF{t}$  robust against a single abort (this is similar to the two-party protocol). On its second abort, however,  an attacker is actually attacking the  above sub-protocol $\Protocol{t'}$, which  provides the attacker a very effective attack opportunity:  the attacker who is first to reconstruct $\delta' = \delta_i'$, can choose to abort and by that make the remaining parties continue with an execution  whose expected outcome is $\delta_i$. Hence, it can bias the protocol's outcome by $\delta_i - \delta_i'$.   If $\delta_i'$ is \whp  far from $\delta_i$, this makes the resulting protocol  unfair.  A partial  solution for this problem is to defend the reconstruction of $\delta'$ in a similar way to how we  defend the reconstruction of the coin $c_i$; before reconstruction the value of $\delta'$ (Step~\ref{step:rec:newdelta} of Protocol $\Protocol{t'}$), call (a variant of) $\Defense$ to defend the parties in the case an abort happens in the reconstruction step. Namely, each subset of parties will get new defense values for executing a recovery protocol with expected output $\delta'$. As in the two-party protocol mentioned before, the use of defense values reduces the bias from  $\delta_i - \delta_i'$ to (roughly) $(\delta_i - \delta_i')^2$. This limitation dictates $\noise$ of bounded variance, but when using such a $\noise$ function we are no longer in the situation where $\delta_i'$ does not leak significant information about $\delta_i$, making the  protocol $\ProtocolF{t}$ vulnerable to aborting attacks. 
The solution is to choose a  variance of  $\noise$ that compromises between these two contradicting requirements.  For not too large $t$, the right choice of parameters yields a protocol of the claimed bias,  significantly improving over the $(1/\sqrt{m})$-bias  vanilla protocol. More details below.

\paragraph{The $\Noise$ function.}\label{sec:intro:noise}
Our $\Noise$ function, parameterized by $\alpha>1$, as follows.

\begin{algorithm}[$\Noise_\alpha$]
	\item [Parameter:] $\alpha > 1$.
	\item [Input:] $\delta \in [0,1]$.
	
	\noindent
	\begin{enumerate}
		
		\item Let $\eps$ be the value such that $\delta$ is the probability that the sum of $m$ independent  $(1/2 + \eps)$-biased $\oo$ coins is positive.
		
		\item Sample an $\alpha\cdot  m$-size  set $\cS$  of independent values in $\oo$, each  taking the value  $1$ with probability $1/2 + \eps$. 
		
		\item Let $\delta'$ be the probability that an $m$-size random subset of $\cS$ has a majority of ones. 
		
		\item Return $\delta'$.
		
	\end{enumerate}
\end{algorithm}

By definition, it is clear that for every $\alpha$, $\eex{\Noise_{\alpha}(\delta)} = \delta$, and note that the variance of $\Noise_{\alpha}(\delta)$ increases with $\alpha$. It can be shown that for large enough $\alpha$, the information that $\delta'$ leaks about $\delta$ is essentially like revealing $\alpha$ independent samples, each taking one w.p.\ $\delta$ and zero otherwise. Using similar arguments to the single sample case, it can be shown that revealing such $\alpha$ samples results with a bias of (roughly) $\sqrt{\alpha}/m$.
We choose $\alpha$ as a function of $k$ --- the number of active parties (hereafter, denote it by $\alpha_k$).  As explained in the previous section, the reconstruction of $\delta'$ should also be protected using a similar defense scheme. This means that now we need to protect the value that is induces by those $\alpha_k$ coins (rather than $m$ coins) using a similar process that now ``reveals'' $\alpha_{k-1}$ samples (rather than $\alpha_k$) for handling a single abort. Using similar arguments, this yields a bias of (roughly) $\sqrt{\alpha_{k-1}}/\alpha_k$ (the formal statement is given by applying \cref{lemma:HyperlProcessVectorHint:pre} with $\alpha=\alpha_{k-1}$ and $\beta=\alpha_k$). In order to minimize the bias, we want to minimize the maximum of $\set{\sqrt{\alpha_{t-1}}/m, \sqrt{\alpha_{t-2}}/\alpha_{t-1},\ldots,\sqrt{\alpha_3}/\alpha_4,1/\alpha_3}$, which holds whenever $\sqrt{\alpha_{t-1}}/m = \sqrt{\alpha_{t-1}}/\alpha_t = \ldots = \sqrt{\alpha_3}/\alpha_4 = 1/\alpha_3$. The solution is obtained by setting $\alpha_k = m^{\frac{2^{t-3}}{2^{t-2}-1}\cdot \frac{2^{k-2}-1}{2^{k-3}}}$, yielding a bias (per-round) of $\frac1{\alpha_3} = \frac1{m^{\frac12 + \frac1{2^{t-1}-2}}}$.

\subsection{Proving Fairness via Linear Program}\label{intro:sec:LP}
In the previous sections we explained how to bound the bias of aborting in a given round of \cref{intro:ProtocolOuter}. The actual situation, however,  is more complex since an adversary might use an \emph{adaptive} strategy for deciding on which round to abort.
As considered by \cite{HaitnerT17}, the security  of  \cref{intro:ProtocolOuter} can be reduced to the value of the  appropriate \emph{online binomial game}.

\mparagraph{Online binomial games}
 An \emph{$m$-round online-binomial game} is a  game between the (honest, randomized) challenger and an all-powerful player. The game is played for $m$ rounds. At round $i$, the  challenger
 tosses an independent $\oo$ coin $c_i$. The final outcome of the game is set to one if the overall sum of the coins is positive, otherwise it is set to $0$. Following each round, the challenger sends  the player some information (\ie \emph{hint})  about  the outcome of the coins tossed until this round. After getting the hint, the player decides whether  to \emph{abort}, or to continue playing. If it aborts, the game stops and the player is rewarded with $\delta_{i-1}$ --- the probability that the output of the game is one given coins $c_1,\ldots,c_{i-1}$ (not including $c_i$, this round coin). If it never aborts, the player is rewarded with the (final) outcome of the game.\footnote{An alternative (yet equivalent) definition of this  game is: in each round, after getting the hint, the player can instruct the challenger to  \emph{re-toss} the current round coin, but it can do that at most once during  the duration of the game. After the game ends, the player is rewarded with its final outcome.} The bias of an  $m$-round game $\game_m$, denoted $\bias(\game_m)$, is the  advantage  the \emph{best} all-powerful player achieves over the passive (non-aborting) player, namely its expected reward  minus $1/2$.

The connection between such online Binomial games and the coin-flipping protocols $\ProtocolF{t}$ described in the previous section is rather straightforward. Recall that an  adversary controlling some of the parties in an execution of Protocol $\ProtocolF{t}$ gains nothing by aborting in  \Stepref{intro:step:coin}, and thus we can assume \wlg that it only aborts, if ever,  at \Stepref{intro:step:defense} of some rounds.\footnote{Actually, in an inner sub-protocols $\Protocol{t'}$ the attacker can also aborts in  the steps where $\delta'$ is reconstructed. But bounding the effect of such aborts is rather simple, comparing to those done is \Stepref{intro:step:defense}, and we ignore  such aborts from the current discussion.} Recall that the gain achieved from aborting in \Stepref{intro:step:defense} of round $i$ is the difference between  $\delta_{i-1}$, here the expected outcome of the protocol  given the coins $c_1,\ldots,c_{i-1}$ flipped in the previous rounds, and the expected outcome of the protocol given these coins and the defense values given to the corrupted parties in \Stepref{intro:step:defense}.   It follows that the maximal bias obtained by a \emph{single} abort, is exactly  the bias of the online binomial game, in which the hints are set to the defense values of the corrupted  parties.  The bias achieved by $t$ aborts in the protocol is at most $t$ times the bias of the corresponding game.

\mparagraph{Bounding online binomial games via a linear program}
Upper-bounding the bias of even a rather simple binomial game is not easy.\footnote{Our problem fits in the well-studied area of {\em stopping-time problems}, \cf \citet{Ferguson06}, where the goal is to upper-bound the value of the optimal stopping (\ie aborting) strategy.}
Specifically, it is non-trivial  to take advantage of the fact that the player does not know beforehand which round will yield the largest gain. A pessimistic approach, taken in \cite{HaitnerT17}, is to consider non-adaptive players that can only abort in a predetermined round, and then upper-bound general players using a union bound.  This approach effectively assumes the player is told the round it is best to abort, and as we prove here misses the right bound by a $\polylog$ factor.

We take a very different approach by  showing how to map  the set of all possible strategies of the game into feasible solutions to a \emph{linear program (LP)}. The bias each strategy achieves is equal to the objective value of the corresponding solution. We then use LP weak duality to bound the maximal value of the LP.\footnote{Interestingly, we also prove the other direction: each feasible solution to the LP corresponds to a possible strategy of the game. This shows that bounding the value of the linear formulation is actually equivalent to bounding the value of the best strategy in the game.} This modular proof approach also yields tighter analysis than the one taken in \cite{HaitnerT17}.  The intuition of the linear program is simple. For a given binomial game we consider all possible states. Specifically, each state $u$ is characterized by the current round, $i$, the sum of coins tossed so far (in the first $i-1$ rounds), $b$, and the hint $h$ given to the strategy.
We use the notation $u=\seq{i,b,h}$. For  state $u$, let $p_u$ be the probability that the game visits state $u$. For two state $u$  and $v$, let $p_{v \mid u}$ be the probability that the game  visits $v$ given that it visits state $u$. (Note that $p_u$ and $p_{v \mid u}$ are determined by the game itself, and are not functions of the adversary.) We  write $u<v$, to indicate that the round of $v$ is strictly larger than the round of $u$. For a state $v$, let $c_v$ be the expected outcome of the game given that the strategy aborts at  $v$. Given a strategy $\Sc$, let $a_v^{\Sc}$ be the marginal probability that the strategy aborts at state $v$. It is easy to see that the bias achieved by strategy $\Sc$ can be written as:
\begin{align*}
	\bias_m(\Sc) =  \sum_{v \in \Uni} a_v^\Sc\cdot c_v  - \frac12.
\end{align*}

Next, we build a linear formulation whose variables are the marginal probabilities $a_v^\Sc$, capturing the probability that a strategy aborts at state $v$. One clear constraint on the variables is that the variables $a_v^\Sc$ are non-negative. Another obvious constraint is that $a_v^\Sc\le p_v$, \ie the  probability of aborting in a state  is at most the probability  the  game visits the state. A more refined constraint is that $a_v^\Sc + \sum_{u | u<v} a_u^\Sc \cdot p_{v|u} \leq p_v$.
Intuitively, this constraint stipulates that the marginal probability of aborting at state $v$ plus the probability that the game visits $v$ and the strategy aborted in a state $u<v$, cannot exceed the probability that the game visits $v$. We prove that this is indeed a valid constraint for any strategy and also that any solution that satisfies this constraint can be mapped to a valid strategy. As all our constraints and the objective function are linear, this gives us a linear program that characterizes all strategies.

Formulating the linear program is just the first step. Although there are many methods for solving (exactly) a specific linear program, we are interested in (bounding) the \emph{asymptotic} behavior of the optimal solution of the program as a function of $m$. To bound the solution we construct an asymptotic feasible solution to the dual program. This gives (by weak duality) an upper bound on the optimal bias obtained by any strategy.

 \subsection{Additional Related Work}\label{sec:relatedWork}
 \citet{CleveI93} showed that in the \textit{fail-stop model}, any two-party $\rnd$-round coin-flipping protocol has bias $\Omega(1/\sqrt{\rnd})$; adversaries in this model are
 computationally unbounded, but they must follow the instructions of the protocol, except for being
 allowed to abort prematurely. \citet{DachmanLMM11} showed that the same holds for $o(n/\log n)$-round protocols in the random-oracle model ---  the parties have oracle access to a uniformly  chosen  function over $n$ bit strings.  
Very recenetly, \citet{MajiW20} showed that any black-box construction of $\rnd$-round two-party coin-flipping protocol from one-way functions, has bias $\Omega(1/\sqrt{\rnd})$.

  Recently,  \citet*{BHMO17}  have shown that \emph{any} $m$-round $t$-party coin-flipping with $t^k > m$ for some $k\in \N$, can be biased by $1/(\sqrt{m} \cdot  (\log m)^k)$. Ignoring logarithmic factors,  this means that if the number of parties is  $m^{\Omega(1)}$, the majority  protocol of  \cite{AwerbuchBCGM1985} is optimal. 
  Where  \citet*{HaitnerMO18} proved that for any \emph{fixed} $m$, key-agreement is a necessary assumption for \emph{two-party} $m$-round coin-flipping  protocol of bias smaller than $1/\sqrt{m}$.

 There is a vast literature concerning coin-flipping protocols with weaker security guarantees. Most notable among these are protocols that are \textit{secure with abort}. According to this security definition, if a cheat is detected or if one of the parties aborts, the remaining parties are not required to output anything. This form of security is meaningful in many settings, and it is typically much easier to achieve; assuming one-way functions exist, secure-with-abort protocols of negligible bias are known to exist against any number of corrupted parties \cite{Blum83,HNORV08,Naor91}. To a large extent, one-way functions are also necessary for such coin-flipping protocols \cite{BermanHT14,HaitnerOmri11,ImpagliazzoLu89,Maji10}.

 Coin-flipping protocols were also studied in a variety of other models. Among these are collective
 coin-flipping in the \textit{perfect information model}: parties are computationally unbounded
 and all communication is public \cite{AlonNaor93,BenLin89,Feige99a,RussellZuc98,Saks89,GoldwasserKP15,kalai2018a,HaitnerK20}, and protocols are based on physical assumptions, such
 as quantum computation \cite{AharonovEtAl00,Ambainis04,AmbainisBDR04} or tamper-evident seals \cite{MoranNaor05}.

 Perfectly fair coin-flipping protocols (\ie having zero bias) are a special case of protocols for \emph{fair} secure function evaluation (SFE). Intuitively, the security of such protocols guarantees that when the protocol terminates, either everyone receives the (correct) output of the functionality, or no one does. While \citet{Cleve86}'s result yields that some functions do not have fair SFE, it was recently shown that many interesting function families do have (perfectly) fair SFE \cite{gordonHKL11,Ash14,ABMO15}.

 \subsection{Open Problems}
 Finding the optimal bias  $t$-party coin-flipping protocol for $t>2$ remained the main open question in this area. While the gap  between the upper and lower bound for the three-party case is now quite small (\ie an $O(\sqrt{\log m})$ factor), the gap for $t>3$ is still rather large, and for  $t>\frac12 \loglog m$  the best protocol remains the  $t/\sqrt{m}$-bias protocol of \cite{AwerbuchBCGM1985}. For the three parties case, while we improved the upper bound of \cite{HaitnerT17}  by a $\polylog m$ factor, it is still open whether the remaining $O(\sqrt{\log m})$ factor is necessary for this case.
 

\subsection*{Acknowledgment}
We thank Eran Omri for very useful discussions. We also thank the anonymous referees for very  useful  comments regarding the readability of this text.

\ifdefined\SBM
\subsection*{Paper Organization}
The definition of the new many-party coin-flipping protocol and the formal statement of the main theorem are given in \cref{S:sec:protocol}.  Standard  notations and definitions, full details of the definition of the  protocol  along with its security proof, and the tools developed for the latter security proof, are given in the appendix.
\else
\subsection*{Paper Organization}
Notations and the definitions used throughout the paper are given in \cref{section:Preliminaries}.
Our coin-flipping protocol  along with its security proof are given in \cref{sec:protocol}. The proofs given in \cref{sec:protocol} use the bounds given in \cref{sec:ExpChangeDueLeakage} on the change  knowing the defense values has on the expected output of the protocol, and the new bounds on the bias of online binomial games given in \cref{sec:BinomialViaLP}. Missing proofs can be found in \cref{sec:missinProofs}.
\fi


\ifdefined\SBM

\else
\section{Preliminaries}\label{section:Preliminaries}

\subsection{Notation}\label{sec:prelim:notation}
We use calligraphic letters to denote sets, uppercase for random variables and functions,  lowercase for values, boldface for vectors and capital boldface for matrices. All logarithms considered here are in base two. For a vector $\vct$, we denote its $i$-th entry by $\vct_i$ or $\vct[i]$. For $a\in \R$ and $b\geq 0$,  let $a\pm b$ stand for the interval $[a-b,a+b]$. Given sets $\cs_1,\ldots,\cs_k$ and $k$-input function $f$, let  $f(\cs_1,\ldots,\cs_k) \eqdef \set{f(x_1,\ldots,x_j) \colon x_i\in \cs_i}$, \eg $f(1\pm 0.1) = \set{f(x) \colon x\in [.9,1.1]}$. For $n\in \N$, let $[n] \eqdef \set{1,\ldots,n}$ and $(n) \eqdef \set{0,\ldots,n}$. Given a vector $\vct \in \oo^\ast$, let $\w(\vct)\eqdef \sum_{i\in [\size{\vct}]} \vct_i$. Given a vector $\vct \in \oo^\ast$ and a set of indexes $\cI \subseteq [\size{\vct}]$, let $\vct_{\cI} = (\vct_{i_1},\ldots,\vct_{i_{\size{\cI}}})$ where $i_1,\ldots,i_{\size{\cI}}$ are the ordered elements of $\cI$. We let the XOR of two integers, stands for the \emph{bitwise} XOR of their bit representations, and we let $\sign \colon \R \mapsto \zo$ be the function that outputs one on non-negative input and zero otherwise.

Let $\poly$ denote the set all polynomials,  \ppt denote  for probabilistic  polynomial time, and  \pptm denote a \ppt algorithm (Turing machine).  A function $\nu \colon \N \to [0,1]$ is \textit{negligible}, denoted $\nu(n) = \negl(n)$, if $\nu(n)<1/p(n)$ for every $p\in\poly$ and large enough $n$.

\paragraph{Distributions.}
Given a distribution $D$, we write $x\gets D$ to indicate that $x$ is selected according to $D$. Similarly, given a random variable $X$, we write $x\gets X$ to indicate that $x$ is selected according to $X$. Given a finite set $\cs$, we let $s\la \cs$ denote that $s$ is selected according to the uniform distribution on $\cs$. The support of a distribution $D$ over a finite set $\Uni$, denoted $\Supp(D)$, is defined as $\set{u\in\Uni: D(u)>0}$. The \emph{statistical distance} of two distributions $P$ and $Q$ over a finite set $\Uni$, denoted as $\SD(P,Q)$, is defined as $\max_{\cs\subseteq \Uni} \size{P(\cs)-Q(\cs)} = \frac{1}{2} \sum_{u\in \Uni}\size{P(u)-Q(u)}$.

For $\delta\in [0,1]$, let $\Berzo{\delta}$ be the Bernoulli probability distribution over $\zo$, taking the value $1$ with probability $\delta$ and $0$ otherwise. For $\eps \in [-1,1]$, let  $\Beroo{\eps}$ be the Bernoulli probability distribution over $\oo$, taking the value $1$ with probability   $\frac{1}{2}(1+\eps)$ and $-1$ otherwise. For $n\in \N$ and $\eps \in [-1,1]$, let $\Beroo{n,\eps}$ be the binomial distribution induced by the sum of $n$ independent random variables, each distributed according to $\Beroo{\eps}$. For $n \in \N$, $\eps \in [-1,1]$ and $k \in \Z$, let $\vBeroo{n, \eps}(k) \eqdef \ppr{x\la \Beroo{n,\eps}}{x \geq  k} = \sum_{t=k}^{n}\Beroo{n,\eps}(t)$. For $n \in \N$ and $\delta \in [0,1]$, let $\sBias{n}{\delta}$ be the value $\eps \in [-1,1]$ with $\vBeroo{n, \eps}(0) = \delta$. For $n \in \N$, $\ell \in [n]$ and $p\in \set{-n,\dots,n}$, define the hyper-geometric probability distribution $\Hyp{n,p,\ell}$ by $\Hyp{n,p,\ell}(k) \eqdef \ppr{\cI}{\w(\vct_\cI) = k}$, where $\cI$ is an $\ell$-size set uniformly chosen from $[n]$ and $\vct \in \oo^n$ with $w(\vct)= p$. Let $\vHyp{n,p,\ell}(k) \eqdef \ppr{x\la \Hyp{n,p,\ell}}{x \geq  k} = \sum_{t=k}^{\ell}\Hyp{n,p,\ell}(t)$. Let $\Phi\colon \R \mapsto (0,1)$ be the cumulative distribution function of the standard normal distribution, defined by $\Phi(x) \eqdef \frac{1}{\sqrt{2\pi}}\int_{x}^{\infty}e^{-\frac{t^2}{2}}dt$. Finally, for $n\in \N$ and $i \in [n]$, let $\NL{i} \eqdef (n+1-i)^2$ and $\NS{i} \eqdef \sum_{j=i}^n \NL{j}$. We summarize the different notations used throughout the paper in the following tables.

\begin{table}[ht]
	\begin{center}
		\caption{\label{fig:summeryOfBasicFunctions} Basic functions.}
		\begin{tabular}{|c|c|c|}
			\hline
			\textit{Definition} & \textit{Input Range} & \textit{Output value} \\
			\hline
			$[n]$ & $n \in \N$ & $\set{1,\ldots,n}$\\
			\hline
			$(n)$ & $n \in \N$ & $\set{0,\ldots,n}$\\
			\hline	
			$\NL{i}$ & $n\in \N$, $i \in [n]$ & $(n+1-i)^2$\\
			\hline
			$\NS{i}$ & $n\in \N$, $i \in [n]$ & $\sum_{j=i}^n \NL{j}$\\
			\hline
			$\w(\vct)$ & $\vct \in \oo^\ast$ & $\sum_{i\in \cI} \vct_i$\\
			\hline
			$\vct_{\cI}$ & $\vct \in \oo^\ast$, $\cI \subseteq [\size{\vct}]$ and $i_1,\ldots,i_{\size{\cI}}$ are the ordered elements of $\cI$ & $(\vct_{i_1},\ldots,\vct_{i_{\size{\cI}}})$\\
			\hline
			$a \pm b$ & $a \in \R$, $b \geq 0$ & $[a-b, a+b]$\\
			\hline
			$\sign(x)$ & $x \in \R$ & $1$ for $x \geq 0$\\
			& & and $0$ otherwise. \\
			\hline
		\end{tabular}
	\end{center}
\end{table}


\begin{table}[ht]
	\begin{center}
		\caption{\label{fig:summeryOfDistributions} Distributions.}
		\begin{tabular}{|c|c|c|}
			\hline
			\textit{Distribution} & \textit{Input Range} & \textit{Description} \\
			\hline
			$\Berzo{\delta}$ & $\delta\in [0,1]$ & $1$ with probability $\delta$ and $0$ otherwise.\\
			\hline
			$\Beroo{\eps}$ & $\eps \in [-1,1]$ & $1$ with probability $\frac{1}{2}(1+\eps)$ and $-1$ otherwise\\
			\hline		
			$\Beroo{n,\eps}$ & $n\in \N$, $\eps \in [-1,1]$ & sum of $n$ independent $\Beroo{\eps}$ random variables\\
			\hline
			$\Hyp{n,p,\ell}$ & $n \in \N$, $p\in \set{-n,\dots,n}$, $\ell \in [n]$ & The value of $\w(\vct_{\cI})$, where:\\
			& & (1) $\cI$ is an $\ell$-size set uniformly chosen from $[n]$, and\\
			& & (2) $\vct \in \oo^n$ is an (arbitrary) vector with $\w(\vct) = p$\\
			\hline
			
		\end{tabular}
	\end{center}
\end{table}

\newpage

\begin{table}[ht]
	\begin{center}
		\caption{\label{fig:summeryOfOtherFunctions} Distributions related functions.}
		\begin{tabular}{|c|c|c|}
			\hline
			\textit{Definition} & \textit{Input Range} & \textit{Output value} \\
				\hline
				$\Phi(x)$ & $x \in \R$ & $\frac{1}{\sqrt{2\pi}}\int_{x}^{\infty}e^{-\frac{t^2}{2}}dt$\\
			\hline		
			$\Beroo{n,\eps}(k)$ & $n\in \N$, $\eps \in [-1,1]$, $k \in \Z$ & $\ppr{x \la \Beroo{n,\eps}}{x = k}$\\
			\hline
			$\vBeroo{n,\eps}(k)$ & $n\in \N$, $\eps \in [-1,1]$, $k \in \Z$ & $\ppr{x \la \Beroo{n,\eps}}{x \geq k}$\\
			\hline
			$\sBias{n}{\delta}$ & $n\in \N$, $\delta\in [0,1]$ & The value $\eps \in [-1,1]$ with $\vBeroo{n, \eps}(0) = \delta$\\
			\hline
			$\Hyp{n,p,\ell}(k)$ & $n \in \N$, $p\in \set{-n,\dots,n}$, $\ell \in [n]$, $k \in \Z$ & $\ppr{x \la \Hyp{n,p,\ell}}{x = k}$\\
			\hline
			$\vHyp{n,p,\ell}(k)$ & $n \in \N$, $p\in \set{-n,\dots,n}$, $\ell \in [n]$, $k \in \Z$ & $\ppr{x \la \Hyp{n,p,\ell}}{x \geq k}$\\
			\hline
		\end{tabular}
	\end{center}
\end{table}

\subsection{Facts About the Binomial Distribution}\label{sec:prelim:Binomial}

\begin{fact}[Hoeffding's inequality for $\oo$]\label{claim:Hoeffding}
	Let $n,t \in \N$ and $\eps \in [-1,1]$. Then
	\begin{align*}
	\ppr{x\la \Beroo{n,\eps}}{\abs{x-\eps n} \geq t} \leq 2e^{-\frac{t^2}{2n}}.
	\end{align*}
\end{fact}
\begin{proof}
	Immediately follows by  \cite{Hoeffding1963}.
\end{proof}

The following proposition is proven in \cite{HaitnerT17}.
\def\propBinomProbEstimation{Let $n \in \N$, $t \in \Z$ and $\eps \in [-1,1]$ be such that $t \in \Supp(\Beroo{n, \eps})$, $\size{t} \leq n^{\frac{3}{5}}$ and $\size{\eps} \leq  n^{-\frac{2}{5}}$. Then
	\begin{align*}
	\Beroo{n, \eps}(t) \in  (1 \pm \error) \cdot \sqrt{\frac{2}{\pi}} \cdot \frac{1}{\sqrt{n}}\cdot  e^{-\frac{(t-\eps n)^2}{2n}},
	\end{align*}
	for $\error  = \xi \cdot (\eps^2 \abs{t} + \frac{1}{n} + \frac{\abs{t}^3}{n^2} + \eps^4 n)$ and a universal constant $\xi$.}

\begin{proposition}\label{prop:binomProbEstimation}
	\propBinomProbEstimation
\end{proposition}

The following propositions are proven in \cref{app:missinProofs:Binomial}.

\def\propBinomTailExpectation{
	Let $n \in \N$, $\eps \in [-1,1]$ and let $\mu \eqdef \ex{x \la \Beroo{n,\eps}}{x} = \eps\cdot n$. Then for every $k > 0$ it holds that
	\begin{enumerate}
		\item $\ex{x \la \Beroo{n,\eps} \mid \size{x-\mu} \leq k}{(x-\mu)^2} \leq \ex{x \la \Beroo{n,\eps}}{(x-\mu)^2} \leq n$.\label{prop:binomTailExpectation:Item1}
		
		\item $\ex{x \la \Beroo{n,\eps} \mid \size{x-\mu} \leq k}{\size{x-\mu}} \leq \ex{x \la \Beroo{n,\eps}}{\size{x-\mu}} \leq \sqrt{n}$.\label{prop:binomTailExpectation:Item2}
	\end{enumerate}
}

\begin{proposition}\label{prop:binomTailExpectation}
	\propBinomTailExpectation
\end{proposition}

\def\propBinomEps{Let $n, n' \in \N$, $k\in \Z$, $\eps \in [-1,1]$ and $\const > 0$ be such that $n \leq n'$, $\abs{k} \leq \const \cdot \sqrt{n \log n}$, $\size{\eps}  \leq \const \cdot \sqrt{\frac{\log n}{n}}$, and let $\delta = \vBeroo{n,\eps}(k)$. Then
	\begin{align*}
	\sBias{n'}{\delta} \in \frac{\eps n - k}{\sqrt{n\cdot n'}} \pm \error,
	\end{align*}
	for $\error = \varphi(\const) \cdot \frac{\log^{1.5} n}{\sqrt{n\cdot n'}}$ and a universal function $\varphi$.}

\begin{proposition}\label{prop:epsDiff}
	\propBinomEps
\end{proposition}

\subsection{Facts About the Hypergeometric Distribution}\label{sec:prelim:HypGeo}

\begin{fact}[Hoeffding's inequality for hypergeometric distribution]\label{fact:hyperHoeffding}
	Let $\ell \leq n \in \N$, and $p \in \Z$ with $\size{p}\ \leq n$. Then
	$$\ppr{x\la \Hyp{n,p,\ell}}{{\abs{x-\mu}} \geq t} \leq e^{-\frac{t^2}{2\ell}},$$
	for $\mu = \ex{x\la \Hyp{n,p,\ell}}{x} = \frac{\ell  \cdot p}{n}$.
\end{fact}
\begin{proof}
	Immediately follows by  \cite[Equations (10),(14)]{scala2009hypergeometric}.
\end{proof}

The following propositions are proven in \cref{app:missinProofs:HypGeo}.

\def\propHyperProbTightEstimation{Let $n \in \N$, $\ell\in[\floor{\frac{n}2}]$, $p, t\in \Z$ and $\const>0$ be such that $\size{p}\leq \const\cdot\sqrt{n\log n}$, $\size{t} \leq \const\cdot\sqrt{\ell\log \ell}$ and $t \in \Supp(\Hyp{n, p, \ell})$. Then
	\begin{align*}
	\Hyp{n, p, \ell}(t) = (1 \pm \error) \cdot \sqrt{\frac{2}{\pi}} \cdot \frac{1}{\sqrt{\ell(1-\frac{\ell}{n})}}\cdot e^{-\frac{(t-\frac{p\ell}{n})^2}{2\ell(1-\frac{\ell}{n})}},
	\end{align*}
	for $\error= \varphi(\const) \cdot \frac{\log^{1.5} \ell}{\sqrt{\ell}}$ and a universal function $\varphi$.}

\begin{proposition}\label{prop:hyperProbTightEstimation}
	\propHyperProbTightEstimation
\end{proposition}

\def\propHyperToNormal{Let $n \in \N$, $\ell \in [\floor{\frac{n}2}]$, $p, k\in [n]$ and $\const > 0$ be such that $\size{p} \leq \const\cdot \sqrt{n\log{n}}$ and $\size{k} \leq \const\cdot \sqrt{\ell\log{\ell}}$. Then
	\begin{align*}
		\vHyp{n, p, \ell}(k) \in \Phi\left(\frac{k-\frac{p\cdot \ell}{n}}{\sqrt{\ell(1-\frac{\ell}{n})}}\right) \pm \error,
	\end{align*}
	where $\error = \varphi(\const) \cdot \frac{\log^{1.5} \ell}{\sqrt{\ell}}$ for some universal function $\varphi$.}

\begin{proposition}\label{prop:hyperToNormal}
	\propHyperToNormal
\end{proposition}

\def\propHyperEpsEstimation{Let $n \in \N$, $\ell \in [\floor{\frac{n}2}]$, $p, k\in [n]$ and $\const > 0$ be such that $\size{p} \leq \const\cdot \sqrt{n\log{n}}$ and $\size{k} \leq \const\cdot \sqrt{\ell\log{\ell}}$ and let $\delta = \vHyp{n, p, \ell}(k)$. Then for every $m \geq \ell$ it holds that
	\begin{align*}
	\sBias{m}{\delta} \in \frac{\frac{p\cdot\ell}{n} - k}{\sqrt{m\cdot\ell(1-\frac{\ell}{n})}} \pm \error,
	\end{align*}
	where $\error = \varphi(\const) \cdot \frac{\log^{1.5} \ell}{\sqrt{m\cdot \ell}}$ for some universal function $\varphi$.}

\begin{proposition}\label{prop:hyperEpsEstimation}
	\propHyperEpsEstimation
\end{proposition}

\subsection{Multi-Party Computation}\label{sec:SFE}

\subsubsection{Protocols}\label{sec:Protocols}
To keep the discussion simple, in the following we focus  on no private input protocols. A  $\partNum$-party  protocol is defined using $\partNum$ Turing Machines (TMs) $\Pc_1,\ldots,\Pc_\partNum$,
having the security parameter $1^\secParam$ as their common input. In each round, the parties broadcast and receive messages on  a broadcast channel. At the end of protocol, each  party outputs some binary string. The parties communicate in a synchronous network, using only a broadcast channel: when a party broadcasts a message, all other parties see \emph{the same} message. This ensures some consistency between the information the parties have. There are no private channels and all the parties see all the messages, and can identify their sender. We do not assume simultaneous broadcast. It follows that  in each round, some parties might hear the messages sent by the other parties before broadcasting their messages. We assume that if a party aborts, it first broadcasts the message \Abort to the other parties, and \wlg only does so at the end of a round in which it is supposed to send a message.
A protocol is \emph{efficient}, if its parties are \pptm, and the protocol's number of rounds is a computable function of the security parameter.

This work focuses on  efficient protocols, and on malicious, static (\ie non-adaptive) \ppt adversaries for such protocols. An adversary is allowed to corrupt some subset of the parties; before the beginning of the protocol, the adversary corrupts a subset of the parties that from now on may arbitrarily deviate from the protocol. Thereafter, the adversary sees the messages sent to the corrupted parties and controls their messages. We also consider the so called \textit{fail-stop} adversaries. Such adversaries follow the prescribed protocol, but might abort prematurely. Finally, the honest parties follow the instructions of the protocol to its completion.

\subsubsection{The Real vs. Ideal Paradigm}\label{sec:realIdeal}
The security of multi-party computation protocols is defined using the \emph{real} vs. \emph{ideal} paradigm \cite{Canetti00,Goldreich04}. In this paradigm, the \emph{real-world model}, in which
protocols is  executed is compared to an \emph{ideal model} for executing the task at hand.
The latter model involves a trusted party whose functionality captures the security requirements
of the task. The security of the real-world protocol is argued by showing that it ``emulates''  the ideal-world protocol, in the following sense: for any real-life adversary $\Aadv$, there exists an ideal-model adversary (\aka simulator) $\iAadv$  such that the global output of an execution of the protocol with $\Aadv$ in the real-world model is distributed similarly to the global output of running
$\iAadv$ in the ideal model. The following discussion is restricted  to random,  no-input functionalities. In addition, to keep the presentation simple, we limit our attention to uniform adversaries.\footnote{All results stated in this paper, straightforwardly extend to the  non-uniform settings.}

\mparagraph{The Real Model}
Let $\pi$ be an $\partNum$-party protocol and let $\Aadv$ be an adversary controlling a subset $\cC \subseteq [\partNum]$ of the parties. Let $\Real_{\pi,\Aadv,\cC}(\secParam)$ denote the output of $\Aadv$  (\ie \wlg its view: its random input and the messages it received) and the outputs of the honest parties, in a random  execution of $\pi$ on common input $1^\secParam$. Recall that an adversary is \emph{fail stop}, if until they abort, the parties in its control follow the prescribed protocol (in particular, they property toss their private random coins). We call  an execution of $\pi$ with such a fail-stop adversary, a fail-stop execution. 

\mparagraph{The Ideal Model}
 Let  $f$ be a $\partNum$-output  functionality. If $f$ gets a security parameter (given in unary), as its first input, let $f_\secParam(\cdot) = f(1^\secParam,\cdot)$.  Otherwise, let $f_\secParam = f$. An ideal execution of $f$ \wrt an adversary $\iAadv$  controlling a subset $\cC \subseteq [\partNum]$ of the ``parties" and a security parameter  $1^\secParam$, denoted  $\Ideal_{f,\iAadv,\cC}(\secParam)$, is the output of the adversary $\iAadv$ and that of the trusted party, in the following experiment.
\begin{experiment}~
\begin{enumerate}
  \item The trusted party sets $(y_1,\dots,y_\partNum) = f_\secParam(X)$, where $X$ is a uniform element in the domain of $f_\secParam$,  and sends $\set{y_i}_{i\in \cC}$ to $\iAadv(1^\secParam)$.

  \item $\iAadv(1^\secParam)$ sends  the message $\mathsf{Continue}$/ $\mathsf{Abort}$ to the trusted party, and locally outputs some value.

  \item The trusted party outputs $\set{o_i}_{i\in [\partNum] \setminus \cC}$, for  $o_i$ being $y_i$ if $\iAadv$ instructs  $\mathsf{Continue}$, and $\perp$ otherwise.
\end{enumerate}
\end{experiment}
An adversary $\iAadv$ is non-aborting, if it never sends the $\mathsf{Abort}$ message.

\paragraph{$\alpha$-secure computation.}
The following definitions adopts the notion of $\alpha$-secure computation \cite{BeimelLOO11,GordonKatz10,Katz07} for our restricted settings. 
\begin{definition}[$\alpha$-secure computation]\label{def:deltaSecure}
An efficient  $\partNum$-party protocol $\pi$ computes  a $\partNum$-output functionality $f$ in a {\sf  $\alpha$-secure} manner [\resp against fail-stop adversaries], if for every $\cC \subsetneq [\partNum]$ and every [\resp fail-stop] \ppt
adversary $\Aadv$ controlling the parties indexed by $\cC$,\footnote{The requirement that $\cC$ is a \emph{strict} subset of $[\partNum]$, is merely for notational convinced.} there exists a \ppt $\iAadv$ controlling the same parties, such that
$$\SD\left(\Real_{\pi,\Aadv,\cC}(\secParam),\Ideal_{f,\iAadv,\cC}(\secParam)\right) \leq \alpha(\secParam),$$
for large enough $\secParam$. A protocol securely compute a functionality $f$, if it computes $f$ in a $\negl(\secParam)$-secure manner. The protocol $\pi$ computes $f$ in a {\sf simultaneous $\alpha$-secure} manner, if the above is achieved by a {\sf non-aborting} $\iAadv$.
\end{definition}
Note that being simultaneous $\alpha$-secure is a very strong requirement, as it dictates that the cheating  real adversary has no way to prevent the honest parties from getting their part of the output, and this should be achieved with no simultaneous broadcast mechanism (i.e., in each round, some parties might see the messages sent by the other parties before broadcasting their messages).

\subsubsection{Fair Coin-Flipping Protocols}\label{sec:CFprotocols}
\begin{definition}[$\alpha$-fair coin-flipping]\label{def:fairCTSim}
For $\partNum\in \N$ let $\CoinToss_\partNum$ be the $\partNum$-output functionality from $\zo$ to $\zo^\partNum$, defined by $\CoinToss_{\partNum}(b)= b \ldots b$ ($\partNum$ times).  A $\partNum$-party protocol $\pi$ is {\sf $\alpha$-fair coin-flipping protocol}, if it computes $\CoinToss_\partNum$ in  a {\sf simultaneous $\alpha$-secure} manner. 
\end{definition}

\paragraph{Proving fairness.}
\citet{HaitnerT17} gave an alternative characterization of fair coin-flipping protocols against fail-stop adversaries.  Specifically, \cref{prop:FairCTGameAlt} below reduces the task of proving fairness of a coin-flipping protocol,  against fail-stop adversaries, to proving the protocol is correct: the honest parties always output the same bit, and this bit is uniform in an all honest execution, and to proving the protocol is unbiased: a fail-stop adversary cannot bias the output of the honest parties by too much.

\begin{definition}[correct coin-flipping protocols]\label{def:CorrectCT}
A protocol is a {\sf correct coin flipping}, if
 \begin{itemize}
   \item When interacting with an fails-stop adversary controlling a subset of the parties, the honest parties {\sf always} output the same bit, and
   \item The common  output in  a random {\sf honest} execution of $\pi$, is uniform over $\zo$.
 \end{itemize}
\end{definition}

Given a partial view of a fail-stop adversary, we are interesting in the  expected outcome of the parties, conditioned on this and the adversary making no further aborts.
\begin{definition}[view value]\label{def:ViewVal}
Let $\pi$ be a protocol in which the honest parties always output the same bit value. For a partial  view $v$ of the parties in a fail-stop  execution of $\pi$, let $\Cc_\pi(v)$ denote the parties' full view in an {\sf honest} execution of $\pi$ conditioned on $v$ (\ie all parties that do not abort in $v$ act honestly in $\Cc_\pi(v)$). Let $\val_\pi(v) = \ex{v' \la \Cc_\pi(v)}{\out(v')}$,  where $\out(v')$ is the common output of the non-aborting  parties in $v'$.
\end{definition}

Finally, a protocol is unbiased, if no fail-stop adversary can bias the common output of the honest parties by too much.
\begin{definition}[$\alpha$-unbiased coin-flipping protocols, \cite{HaitnerT17}.]\label{def:FairCTGameAlt}
A $\partNum$-party, $\rnd$-round  protocol $\pi$ is {\sf $\alpha$-unbiased}, if the following holds for every fail-stop adversary $\Aadv$ controlling the parties indexed by a subset $\cC\subset [\partNum]$ (the corrupted parties). Let $V$ be the corrupted parties' view in a random execution of $\pi$ in which $\Aadv$ controls those parties, and let $I_j$ be the index of the $j$'th round in which $\Aadv$ sent an abort message (set to $\rnd+1$, if no such round). Let $V_{i}$ be the prefix of $V$ at the end of the $i$'th round, letting $V_0$  being the empty view, and let $V_i^{-}$ be the prefix of $V_i$ with the $i$'th round abort messages (if any) removed. Then
  $$\size{\ex{V}{\sum_{j\in \size{\cC}} \val(V_{I_j}) -  \val(V_{I_j}^-)}} \leq  \alpha,$$
  where $\val = \val_\pi$ is according to \cref{def:ViewVal}.
\end{definition}

\begin{lemma}[\cite{HaitnerT17}, Lemma 2.18]\label{prop:FairCTGameAlt}
Let $\pi$ be a correct,  $\alpha$-unbiased coin-flipping protocol with $\alpha(\secParam) \leq \frac12 - \frac1{p(\secParam)}$, for some $p\in \poly$, then $\pi$ is a $(\alpha(\secParam) + \negl(\secParam))$-secure coin-flipping protocol against fail-stop adversaries. 
\end{lemma}

\subsubsection{Oblivious Transfer}\label{sec:OT}
\begin{definition}\label{def:OTfunct}
The $\binom{1}{2}$  oblivious transfer ($\OT$ for short) functionality, is the two-output functionality $f$ over $\zo^3$, defined by $f(\sigma_0,\sigma_1,i) = ((\sigma_0,\sigma_1), (\sigma_i,i))$.
\end{definition}
Protocols the securely compute $\OT$, are known under several hardness assumptions (\cf \cite{AieIshRei01,EvenGL85,GentryPV2008,Haitner04,Kalai05,NaorPinkas01}).

\subsubsection{$f$-Hybrid Model}\label{sec:OTHybrid}
Let $f$ be a $\partNum$-output functionality. The {$f$-hybrid} model is identical to the real model of computation discussed above, but in addition, each $\partNum$-size subset of  the parties involved, has access to a trusted party realizing  $f$. It is important to emphasize  that the trusted party realizes $f$ in a \emph{non}-simultaneous manner: it sends a random output of $f$ to the parties in an arbitrary order. When a  party  gets its part of the output, it  instructs the trusted party to either continue sending the output to the other parties, or to send them the abort symbol (\ie the  trusted party  ``implements" $f$  in a perfect non-simultaneous manner). All notions given in \cref{sec:realIdeal,sec:CFprotocols} naturally extend to the $f$-hybrid model, for any functionality $f$. In addition,  the proof of \cref{prop:FairCTGameAlt} straightforwardly extends to this model. We also make use of the following known fact.

\begin{fact}\label{fact:fHybridMToReal}
	Let $f$ be a polynomial-time computable functionality, and assume there exists a $t$-party, $\rnd$-round, $\alpha$-fair coin-flipping  protocol in the $f$-hybrid model, making at most $k$ calls to $f$, were $t$, $\rnd$, $\alpha$ and $k$, are function of the security parameter $\kappa$. Assuming there exist a constant-round protocol for  securely computing \OT, then there exists a $t$-party, $(O(k\cdot t^2)+ \rnd)$-round, $(\alpha +\negl(\kappa))$-fair coin-flipping  protocol (in the real world).
\end{fact}


\begin{proof}
	Since $f$ is a polynomial-time computable and since we assume the existence of  a protocol for  securely computing \OT, there exists a constant-round protocol $\pi_f$ for  securely computing $f$:  a constant-round protocol for $f$ that is secure against semi-honest adversaries follows by \citet{BMR90} (assuming \OT), and the latter protocol can be compiled into a $O(t^2)$-round protocol that securely computes $f$, against arbitrary malicious adversaries,  using the techniques of \citet{GoldreichMiWi87} (assuming one-way functions, that follows by the existence of \OT). Let $\pi$ be a $t$-party, $\rnd$-round, $\alpha$-fair coin-flipping  protocol  in the $f$-hybrid model. \citet{Canetti00} yields that by replacing the trusted party for computing $f$ used in $\pi$ with the protocol $\pi_f$, we get an $(O(k\cdot t^2)+ \rnd)$-round, $(\alpha +\negl)$-fair coin-flipping  protocol.
\end{proof}


\newcommand{\defense}{{\sf def}}
\newcommand{\deltaDef}{{\delta_\defense}}
\newcommand{\DeltaDef}{{\Delta_\defense}}
\newcommand{\ProFH}{\widetilde{\Pi}} 
\newcommand{\ProFR}{\breve{\Pi}} 
\newcommand{\shareName}{\delta}

\section{The Many-Party Coin-Flipping Protocol}\label{sec:protocol} 
In \cref{sec:MultiPartyProtocolFinal}, the  many-party coin-flipping protocol  is defined in an hybrid model. The    security of the latter protocol is analyzed in \cref{sec:AnalysisOfProtocol}. The (real model) many-party coin-flipping protocol is defined and analyzed in \cref{sec:ProvingMainThm}.

\subsection{The Hybrid-Model Protocol}\label{sec:MultiPartyProtocolFinal} 
The coin-flipping protocol described below follows the high-level description given in the introduction. The main difference is that the number of coins flipped is every round is not one, but a decreasing function of the round index. This asymmetry, also done in \cite{HaitnerT17}, prevents the last rounds from having too high influence on the final outcome.

The protocols below are defined in an hybrid model in which the parties get joint oracle access to several ideal functionalities. We assume the following conventions about the model: all functionalities are guaranteed to function correctly, but do not guarantee fairness: an adversary can abort, and thus preventing the honest parties from getting their output, \emph{after} seeing the outputs of the corrupted parties in its control. We assume identified  abort: when a party aborts, its identity is revealed to  all other parties.  We also assume  that when the parties make  \emph{parallel} oracle calls, a party that aborts in one of these calls is forced to abort in all of them.

The protocols defined below will not be efficient, even in the hybrid model, since the parties are required to hold real numbers (which apparently  have infinite presentation), we handle this inefficiency when  defining the (efficient) real world protocol in \cref{sec:ProvingMainThm}.

Protocol $\ProtocolF{}$ defined next is our (hybrid model) coin-flipping protocol to be called.

This protocol  is merely a wrapper for protocol $\Protocol{}$:  the parties first correlate their private inputs using an oracle to the   $\Defense$ functionality, and then interact in $\Protocol{}$ with these inputs (protocol $\Protocol{}$ and the functionality $\Defense$ are defined below). For $\rnd,\partNum \in \N$, the $\partNum$-party, $O(\rnd\cdot \partNum)$-round protocol $\ProtocolF{\partNum}_{\rnd}$ is defined as follows.

\begin{protocol}[$\ProtocolF{\partNum}_{\rnd} = (\PartyF{\partNum}_1,\ldots,\PartyF{\partNum}_{\partNum})$]\label{protocol:outer}
\item[Oracle:] $\Defense$. 
\item [Protocol's description:]~
  \begin{enumerate}
  \item Let $\delta^{\isShare{1}}, \ldots,\delta^{\isShare{t}}$ be $t$-out-of-$t$ shares of $\frac12$.

  \item Let $\ell=\partNum$ be the defense-quality parameter.

  \item For every $\emptyset \neq \cZ \subseteq [\partNum]$ (in parallel), the parties jointly call $\Defense(1^\rnd, 1^\partNum, 1^\ell, \cZ, \delta^{\isShare{1}}, \ldots,\delta^{\isShare{t}})$,  where  $1^\rnd$, $1^\partNum$, $1^{\ell}$,  and $\cZ$ are common inputs, and input $\delta^{\isShare{k}}$ is provided  by party $\Party{\partNumInn}_k$.  Let $\shareName^{\isShare{z},\cZ}$ be the output of party $\PartyF{\partNum}_z$ returned by this call.\label{protocol:outer:defense}
    
  \item The parties interact in $\Protocol{\partNum}_\rnd = (\Party{\partNum}_1,\Party{\partNum}_2,\ldots,\Party{\partNum}_{\partNum})$ with common input $1^\ell$. Party $\PartyF{\partNum}_\z$ plays the role of $\Party{\partNum}_\z$  with private input $\set{\shareName^{\isShare{z},\cZ}}_{\emptyset \neq \cZ \subseteq [\partNum]}$.
  
  \item [Abort (during step $2$):] If there is a single remaining party, it outputs an unbiased bit. Otherwise, the remaining parties interact in $\ProtocolF{t'}_{\rnd}(1^\ell)$ for $t'<t$ being the number of the remaining parties.
  \end{enumerate}
\end{protocol}

\subsubsection{\texorpdfstring{Protocol $\Protocol{\partNumInn}_\rnd$}{The Inner Protocol}} \label{subsec:protocol:inner}

When defining $\Protocol{\partNumInn}_\rnd$, we make a distinction whether the number of parties is two or larger. We let $\Protocol{2}_\rnd$ be the two-party protocol $\HTProtocol_\rnd$, which is a variant of the of two-party protocol of \cite{HaitnerT17} defined in \cref{sec:htprotocol}. For the many-party case (three parties or more), we use the newly defined protocol given below. This distinction between the two-party and many-party cases is  made for improving the bias of the final protocol,  and all is well-defined if we would have used the protocol below also for the two-party case (on the first read, we encourage the reader to assume that this is indeed the case). See \cref{rem:whyUsingHT} for the benefit of using the \cite{HaitnerT17} protocol for the two-party case.

For $\rnd,\partNumInn \le \partNum \in \N$, the $\partNumInn$-party, $O(\rnd \cdot \partNumInn)$-round protocol $\Protocol{\partNumInn}_{\rnd}$ is defined as follows (the functionalities $\Defense$ and $\Coin$ the protocol uses are defined in \cref{subsubsec:nextCoin,subsubsec:Defense:t}, respectively).

\begin{protocol}[$\Protocol{\partNumInn}_\rnd = (\Party{\partNumInn}_1,\Party{\partNumInn}_2,\ldots,\Party{\partNumInn}_{\partNumInn})$ (for $\partNumInn > 2$)]\label{protocol:inner:t}
\item[Oracles:] $\Defense$, and $\Coin$.
\item[Common input:] defense-quality parameter $1^\ell$.
\item[$\Party{\partNumInn}_\z$'s inputs:] $\set{\shareName^{\isShare{z},\cZ}}_{\emptyset \neq \cZ \subseteq [\partNumInn]}$.\footnote{The type of $\shareName^{\isShare{z},\cZ}$ varies according to $\size{\cZ}$. For $\size{\cZ}=1$, $\shareName^{\isShare{z},\cZ}$ is simply a $\zo$ bit,  for $\size{\cZ} > 2$ it is a share of  $\size{\cZ}$-out-of$\size{\cZ}$ secret share of a number in $[0,1]$, and for  $\size{\cZ} = 2$ it has a more complex structure. See \cref{subsubsec:Defense:t} for details.}

\item [Protocol's description:]~
  \begin{enumerate}
  \item For every $\emptyset \neq \cZ \subsetneq [\partNumInn]$ (in parallel), the parties jointly call $\Defense(1^\rnd, 1^\partNumInn, 1^\ell, \cZ, \shareName^{\isShare{1},[\partNumInn]},\ldots,\shareName^{\isShare{\partNumInn},[\partNumInn]})$, where  $1^\rnd, 1^\partNumInn, 1^{\ell} , \cZ$ are common inputs, and input $\shareName^{\isShare{k},[\partNumInn]}$ is provided  by party $\Party{\partNumInn}_k$.
    \label{protocol:inner:t:defense:1}
    
  \item[\qquad  $\bullet$] For all $z \in \cZ$, party $\Party{\partNumInn}_z$ updates   $\shareName^{\isShare{z},\cZ}$ to the  value it received from this call.
    
  \item Each party $\Party{\partNumInn}_\z$ sends $\shareName^{\isShare{z},[\partNumInn]}$ to the other parties. \label{protocol:inner:t:reconstruct:delta}
  \item [\qquad $\bullet$] All parties set $\delta = \bigoplus_{\z=1}^\partNumInn \shareName^{\isShare{z},[\partNumInn]}$.

  \item For $i=1$ to $\rnd$:
    \begin{enumerate}
    \item The parties jointly  call $\Coin(1^\rnd, 1^\partNumInn, \delta, c_1,\ldots,c_{i-1})$.
    \item[\qquad $\bullet$] For $z \in \cZ$, let $(c_i^{\isShare{z}},\delta_i^{\isShare{z}})$ be the output of party $\Party{\partNumInn}_z$ returned by $\Coin$.\label{protocol:inner:t:coin}
      
    \item For every $\emptyset \neq \cZ \subsetneq [\partNumInn]$ (in parallel), the parties jointly call $\Defense(1^\rnd, 1^\partNumInn, 1^{\ell},\cZ, \delta_i^{\isShare{1}},\ldots,\delta_i^{\isShare{\partNumInn}})$, where $1^\rnd, 1^\partNumInn, 1^{\ell}, \cZ$ are common inputs, and the input $\delta_i^{\isShare{k}}$ is provided by party $\Party{\partNumInn}_k$. \label{protocol:inner:t:defense:2}
    
    \item[\qquad $\bullet$] For $z \in \cZ$, party $\Party{\partNumInn}_z$ updates   $\shareName^{\isShare{z},\cZ}$ to the  value it received  from this call.
      
    \item Each party $\Party{\partNumInn}_\z$ sends $c_i^{\isShare{z}}$ to the other parties. \label{protocol:inner:t:reconstruct:coin}
    \item[\qquad$\bullet$]   All parties set $c_i = \bigoplus_{\z=1}^ \partNumInn c_i^{\isShare{z}}$.

    \end{enumerate}
   \end{enumerate}
 
 \item[Output:] All parties output $\sign(\sum_{i=1}^\rnd c_i)$.
     
\item [Abort:]  Let $\emptyset \neq \cZ \subsetneq [\partNumInn]$ be the indices of the remaining parties. If $\cZ = \set{z_k}$, then the  party $\Party{\partNumInn}_k$ outputs $\shareName^{\isShare{k},\set{k}}$. Otherwise ($\size{\cZ} \geq 2$), assume for ease of notation that $\cZ = [h]$ for some $h\in [\partNumInn-1]$. To decide on a common output, the parties interact in $\Protocol{h}_{\rnd} = (\Party{h}_1',\ldots,\Party{h}_h')$ with common input $1^\ell$, where party $\Party{h,\partNumInn}_\z$ plays the role of $\Party{h}_\z'$ with private input $\set{\shareName^{\isShare{z},\cZ'}}_{\emptyset \ne \cZ' \subseteq \cZ}$.
\end{protocol}

That is at \Stepref{protocol:inner:t:defense:1}, the parties use $\Defense$ to be instructed what to do if some parties abort in the reconstruction of the value of $\delta$ that happens at \Stepref{protocol:inner:t:reconstruct:delta}. If \Stepref{protocol:inner:t:defense:1}  ends successfully (no aborts), then the expected outcome of the protocol is guaranteed to be $\delta$, even if some parties abort in he reconstruction of $\delta$ done in \Stepref{protocol:inner:t:reconstruct:delta} (but no further aborts). If an abort occurs in this \Stepref{protocol:inner:t:defense:1}, then the remaining parties use their inputs to  interact in a protocol whose expected outcome is $\delta'$, for  $\delta'$ being the input in the call to $\Defense$ that generated the parties' input.  The key point is that even though  $\delta$ might be rather far from $\delta'$, the corrupted parities who only holds parties information about $\delta$ (\ie the output of $\Defense$), cannot exploit this gap too effectively.

A similar thing happens when flipping each of the coins $c_i$. The parties  first use $\Coin$ and $\Defense$ to get shares of the new coin $c_i$ and to get instructed what to do if some parties abort in the reconstruction of $c_i$. If \Stepref{protocol:inner:t:defense:2} ends successfully, then the expected outcome of the protocol is $\delta_i = \pr{\sign(\sum_{i=1}^\rnd c_i) =1 \mid c_1,\ldots,c_i}$, even if some parties abort in the reconstruction of $c_i$ (but no further aborts).   If an abort occurs in \Stepref{protocol:inner:t:defense:2}, then the remaining parties use their inputs to  interact in a protocol whose expected outcome is $\delta_{i-1}$. Also in this case, the corrupted parities cannot exploit the gap between $\delta_i$ and  $\delta_{i-1}$ too effectively.

We note that in the  recursive invocations done in the protocol when abort happens, the number of interacting parties in the new protocol is  smaller. We also note that since all calls to the $\Defense$ functionality  taken is \Stepref{protocol:inner:t:defense:1} / \Stepref{protocol:inner:t:defense:2} are done in \emph{parallel}, the resulting protocol has indeed $O(\partNumInn \cdot m)$ rounds.

Finally, the role of the input parameter $\ell$ is to optimize the information the calls to $\Defense$ leak  through the execution of the protocol (including its sub-protocols executions that take place when aborts happen). Recall (see discussion in the introduction) that on one hand, we would like $\Defense$ to leak as little information as possible, to prevent an effective attack of the current execution of the protocol. For instance, the value return by $\Defense$ in \Stepref{protocol:inner:t:defense:1}, should not give too much information about the value of $\delta$. On  the other hand, a too hiding $\Defense$ will make an interaction done in a sub-protocol,  happens if an abort happens, less secure. Parameter  $\ell$  is set to $t$ in the parent call to the protocol done from  the $t$-party protocol $\ProtocolF{t}$ and is kept to this value throughout the different sub-protocol executions,  enables us to find the optimal balance between these contradicting requirements. See \cref{subsubsec:Defense:t} for details.


\subsubsection{The $\Coin$ Functionality} \label{subsubsec:nextCoin}

Functionality $\Coin$ performs the (non fair) coin-flipping operation done inside the main loop of $\Protocol{}$. It outputs shares of the $i$-th round's coin $c_i$, and also shares for the value of expected outcome of the protocol given $c_i$.\footnote{This redundancy in the  functionality description, \ie the shares of  coins can be used to compute the the second part of the output, simplifies the presentation of the protocol.}

Recall that $\Berzo{\delta}$ is the Bernoulli probability distribution over $\zo$ that  assigns    probability $\delta$ to $1$, that $\Beroo{\eps}$ is the Bernoulli probability distribution over $\oo$ that assigns  probability $\frac{1}{2}(1+\eps)$ to $1$, that $\Beroo{n, \eps}(k) = \pr{\sum_{i=1}^{n} x_i = k}$ for $x_i$'s that are i.i.d according to $\Beroo{\eps}$, and $\vBeroo{n, \eps}(k) = \ppr{x\la \Beroo{n,\eps}}{x \geq  k}$. Also recall that $\sBias{n}{\delta}$ is the value $\eps \in [-1,1]$ with $\vBeroo{n,\eps}(0) = \delta$, that $\ml{i} = (\rnd+1-i)^2$ (\ie  the number of coins tossed at round $i$), and that $\ms{i} = \sum_{j=i}^m \ml{j}$ (\ie the number of coins tossed after round $i$).

\begin{samepage}
\begin{algorithm}[$\Coin$]
\item[Input:] Parameters $1^\rnd$ and  $1^{\partNumInn}$,  $\delta \in [0,1]$,  and coins $c_1,\ldots,c_{i-1}$.
\item[Operation:]~
  \begin{enumerate}
  \item Let $\eps = \sBias{\ms{1}}{\delta}$. \label{coin:epsCalculation}
  \item Sample $c_i \la \Beroo{\ml{i}, \eps}$.
  \item Let $\delta_i = \vBeroo{\ms{i+1},\eps}(-\sum_{j=1}^{i} c_j)$

  \item Sample $\partNumInn$ uniform strings  $\share^{\isShare{1}},\ldots,\share^{\isShare{\partNumInn}}$ conditioned on $(c_i,\delta_i) = \bigoplus_{i=1}^\partNumInn \share^{\isShare{i}}$, and return  party $\Party{\partNumInn}_{i}$ the share $\share^{\isShare{i}}$.
  \end{enumerate}
\end{algorithm}
 \end{samepage}


\subsubsection{The $\Defense$ Functionality} \label{subsubsec:Defense:t}

The $\Defense$ functionality is used by  protocol $\Protocol{}$ to ``defend'' the remaining parties when some corrupted parties abort. When invoked  with a subset $\cZ \subsetneq [\partNumInn]$ and  $\delta \in [0,1]$, it produces the inputs the parties in $\cZ$ need in order to collaborate and produce a $\delta$-biased bit --- expected  value is $\delta$.

 As with protocols $\Protocol{\partNumInn}_\rnd$,  we make a distinction whether $\partNumInn = 2$ ($\partNumInn$ is the number of parties that call $\Defense$) or $\partNumInn>2$. In the former case,   we  use a simple variant of the \cite{HaitnerT17} defense functionality defined in \cref{sec:htprotocol}.  For all other values, we use the functionality defined below. (Also in this case, we encourage the first-time reader to ignore this subtlety, and assume we use the new definition for all cases.)


\begin{algorithm}[$\Defense$ functionality for $\partNumInn > 2$]
\item[Input:] Parameters $1^\rnd$,  $1^{\partNumInn}$, $1^{\ell}$, set $\cZ \subseteq [\partNumInn]$ and shares $\set{\delta^{\isShare{z}}}_{\z \in [\partNumInn]}$. 

\item[Operation:] Return  $\DefenseTilde(1^\rnd,1^\partNumInn,1^\ell,\cZ,\bigoplus_{\z \in [\partNumInn]}\delta^{\isShare{z}})$
\end{algorithm}
 
Namely, $\Defense$ just reconstructs $\delta$ and calls  $\DefenseTilde$ defined below.

\begin{algorithm}[$\DefenseTilde$]
\item[Input:] Parameter $1^\rnd$, $1^{\partNumInn}$, $1^{\ell}$, set $\cZ = \set {z_1,\ldots,z_k} \subsetneq [\partNumInn]$, and $\delta \in [0,1]$.
\item[Operation:]~
  \begin{enumerate}
  \item If $\size{\cZ}=1$, let  $o_1 \la \Beroo{\delta}$.
  \item If $\size{\cZ}=2$, let $(o_1,o_2) = \HTDefenseProtocol (1^m,\delta)$.
  \item If $\size{\cZ}>2$,
    \begin{enumerate}
    \item Let $\delta' = \AddNoise(1^\rnd, 1^\ell, \size{\cZ}, \delta)$.
    \item Sample $\size{\cZ}$ uniform shares $o_1,\ldots,o_k$ such that $\delta' = \bigoplus_{i=1}^k o_i$. 
    
    \end{enumerate}
  \item Return $o_i$ to party $\Party{\partNumInn}_{z_i}$, and $\perp$ to the other parties.
  \end{enumerate} 	
\end{algorithm}
It is clear that for the case $\size{\cZ}=1$, the expected value of the output bit of the party in $\cZ$ is indeed $\delta$. Since the expected value of $\delta'$ output by $\AddNoise(\cdot, \delta)$ (see below) is $\delta$, it is not hard to see that the same holds also for the case  $\size{\cZ}>2$. Finally, though somewhat more difficult to verify, the above also holds for the case $\size{\cZ}=2$ (see \cref{sec:htprotocol}).

\subsubsection{The $\AddNoise$ Functionality}
The $\AddNoise$ functionality, invoked by $\DefenseTilde$,  takes as input $\delta\in [0,1]$ and returns a ``noisy version'' of it $\delta'$ (\ie expected value is $\delta$). The amount of noise used is determined by the defense-quality parameter $\ell$ that reflects the number of players that interact in the parent protocol $\ProtocolF{t}$, the number of parties that will use the returned value in their sub-protocol  $\size{\cZ}$, and the round complexity of the protocol $\rnd$.


\begin{definition}[$\alpha$-factors]\label{def:alphaFactors}
  For $\rnd \geq 1$, $\ell \geq 2$ and $2 \leq k \leq \ell$, let $ \alpha(\rnd, \ell, k) =  m^{\frac{2^{\ell-3}}{2^{\ell-2}-1} \cdot \frac{2^{k-2}-1}{2^{k-3}}}$.
\end{definition}

\begin{algorithm}[$\AddNoise$]
\item[Input:] Parameter $1^\rnd$, $1^{\ell}$ and  $1^k$, and $\delta \in [0,1]$.
\item[Operation:]~
  \begin{enumerate}
  \item Let $\alpha = \alpha(\rnd, \ell, k)$ and $\eps = \sBias{\ms{1}}{\delta}$.\label{addnoise:epsCalculation}
  \item Sample $\vect{\bar{b}} \la\ (\Beroo{\eps})^{\alpha \cdot \ms{1}}$.\label{addnoise:sampleVec}
  \item Let $\delta' = \ppr{\cX \subseteq [\alpha \cdot \ms{1}], \size{\cX}=\ms{1}}{\sum_{x \in X} \vect{\bar{b}}[x] > 0}$.\footnote{I.e., $\delta'$ is the probability that when sampling $\ms{1}$ coins from $\vect{\bar{b}}$, their sum is positive.}
  \item Output $\delta'$.
  \end{enumerate}
\end{algorithm}

Namely, $\AddNoise$ sample a vector $\vect{\bar{b}}$ of $\alpha \cdot \ms{1}$ $\eps$-biased coins. The value of $\delta'$ is then determined as the probability to get a positive sum, when sampling $\ms{1}$-size subset of coins from $\vect{\bar{b}}$.

\subsubsection{The Protocol of \citeauthor{HaitnerT17}}\label{sec:htprotocol}

In this section we define the two-party protocol $\Protocol{2}$ and the functionality $\HTDefenseProtocol$.  For clarity,  in this subsection we name protocol $\Protocol{2}$  by  $\HTProtocol$. 

Protocol $\HTProtocol$ and functionality $\HTDefenseProtocol$ defined below are  close variants for those used by \citet{HaitnerT17} for construction their three-party coin-flipping protocol. For an elaborated  discussion  of the ratio underlying the following  definitions, see \cite{HaitnerT17}.

\paragraph{Protocol $\HTProtocol$.}
The two-party $m$-round protocol $\HTProtocol_\rnd$ is defined as follows  (the functionality $\HTDefenseRound$ used by the protocol is defined below). Recall that for $\ell \in \N$, $\tp(\ell) = \ceil{\log \ell} +1$  is the number of bits it takes to encode an integer in $[-\ell,\ell]$.

\begin{samepage} 
\begin{protocol}[$\HTProtocol_\rnd = (\Party{2}_1,\Party{2}_2)$] \label{protocol:inner:2} 
\item[Common input:] round parameter $1^\rnd$.
\item[Oracles:] $\HTDefenseRound$.
\item[$\Party{2}_\z$'s input:] $\cVS{\z} \in \zo^{\rnd\times \tp(\rnd)}$, $d^z \in \zo$, and $\bankV{\isShare{z},1},\bankV{\isShare{z},2} \in \zo^{2 \cdot \ms{1}}$.
\item [Protocol's description:]~
  \begin{enumerate}
  \item For $i=1$ to $\rnd$:
    \begin{enumerate}
    \item The parties  jointly call 
        $\HTDefenseRound(1^\rnd,c_1,\ldots,c_{i-1},\cVS{1}[i],\cVS{2}[i],
                         \bankV{\isShare{1},1},\bankV{\isShare{1},2},
                         \bankV{\isShare{2},1},\bankV{\isShare{2},2})$,
        where  $(1^\rnd, c_1,\ldots,c_{i-1})$ is the common input, and $(\cVS{\z}[i],\bankV{\isShare{z},1},\bankV{\isShare{z},2})$ is provided by the party $\Party{2}_z$. \label{protocol:inner:2:defense}
			
    \item[$\bullet$] For all $\z\in \set{1,2}$, party $\Party{2}_z$ updates  $d^{z}$ to the value it received from this call.
			
    \item $\Party{2}_1$ sends $\cVS{1}[i]$ to $\Party{2}_2$, and $\Party{2}_2$ sends $\cVS{2}[i]$ to $\Party{2}_1$.\label{protocol:inner:2:reconstructCoin}
    \item[$\bullet$] Both parties set $c_i = \cVS{1}[i] \xor \cVS{2}[i]$.
    \end{enumerate}
  \item Both parties output $\sign(\sum_{i=1}^\rnd c_i)$.
  \end{enumerate}
\item [Abort:] The remaining party  $\Party{2}_\z$ outputs $d^{\z}$.
\end{protocol}

That is, the parties get correlated shares for the rounds' coins, and they reveal them in the main loop at \Stepref{protocol:inner:2:reconstructCoin}. Prior to revealing them, the parties call the $\HTDefenseRound$ functionality to get a defense value in case the other party aborts during the coin reconstruction. 

\end{samepage} 

\begin{algorithm}[$\HTDefenseRound$]
\item[Input:] Parameter $1^\rnd$, coins $c_1,\ldots,c_{i-1}$, and 
  shares $\cVS{1}[i],\cVS{2}[i] \in \zo^{\tp(\rnd)}$ and $\bankV{\isShare{1},1},\bankV{\isShare{2},1},\bankV{\isShare{1},2},\bankV{\isShare{2},2} \in \oo^{2 \cdot \ms{1}}$.

\item[Operation:]~
\begin{enumerate}
\item Let $\bankV{1} = \bankV{\isShare{1},1} \xor \bankV{\isShare{2},1}$,  $\bankV{2} = \bankV{\isShare{1},2} \xor \bankV{\isShare{2},2}$ and  $c_i=\cVS{1}[i] \xor \cVS{2}[i]$.
  \item For both $\z\in \set{1,2}$: sample a random $(\ms{i+1})$-size subset $\cW^\z \subset [2 \cdot \ms{1}]$, and set $d^{\z}$ to one if $\sum_{j=1}^i c_j + \sum_{w \in \cW^\z} \bankV{z}[w] \geq 0$, and to zero otherwise.
  
  \item Return $d^z$ to party $\Party{2}_z$. 
   \end{enumerate}
\end{algorithm}

Namely, to generate a defense value $d_z$ for $\Party{2}_z$, $\HTDefenseRound$ samples  $(\ms{i+1})$-coins from the vector $\bankV{z}$, adds them to the coin $c_1,\cdots,c_i$ and  set $d_z$ to the  sign of this sum.

\paragraph{The $\HTDefenseProtocol$ functionality.} This functionality  prepares the  inputs for the  parties that interact in $\HTProtocol$.  

Recall that for $n\in \N$ and $\eps \in [-1,1]$, $\Beroo{n,\eps}$ is the binomial distribution induced by the sum of $n$ independent random $\rpm 1$ coins, taking the value $1$ with probability $\frac{1}{2}(1+\eps)$, and $-1$ otherwise.

\begin{samepage}
\begin{algorithm}[$\HTDefenseProtocol$]
\item[Input:] Parameter $1^\rnd$ and $\delta \in [0,1]$.

\item[Operation:]~
  \begin{enumerate}
  \item Let $\eps = \sBias{\ms{1}}{\delta}$. \label{htdefenseprotocol:epsCalculation}
  \item For $\z\in \set {1,2}$: sample $\bankV{\z} \la (\Beroo{\eps})^{2\cdot \ms{1}}$. \label{htdefenseprotocol:vecSample} 

  \item For $\z\in \set{1,2}$: sample a random $(\ms{1})$-size subset $\cI^\z \subset [2\cdot \ms{1}]$, and set $d^z$ to one if $\w(\bankV{\z}_{\cI^\z}) \geq 0$, and to zero otherwise.

  \item Let $\vect{c} = (c_1,\ldots,c_m)$ where for $i \in [\rnd]$, $c_i \la \Beroo{\ml{i} ,\eps}$.

  \item Sample two uniform shares $\cVS{1},\cVS{2}$ with  $\cVS{1} \xor \cVS{2} = \vect{c}$.  For both $z \in \set{1,2}$, sample  two uniform shares $\bankV{\isShare{1},z},\bankV{\isShare{2},z}$ with $\bankV{\isShare{1},z} \xor \bankV{\isShare{2},z} = \bankV{\z}$.

  \item Return: $((\cVS{z},\bankV{\isShare{z},1},\bankV{\isShare{z},2},d^z))_{z \in \set{1,2}}$.
  \end{enumerate}
\end{algorithm}
\end{samepage}

Namely, at \Stepref{htdefenseprotocol:epsCalculation}, $\HTDefenseProtocol(\delta)$ calculates  $\eps \in [-1,1]$ for which the probability that the  sum of $\ms{1}$ independent $\eps$-bias coins is positive, is  $\delta$. Then, $\HTDefenseProtocol$ uses this $\eps$ to sample the rounds' coins $c_i$, to be used in the two-party protocol $\HTProtocol_\rnd$, and the vectors that are used by $\HTDefenseRound$ to give defense values in every round of the loop of $\HTProtocol_\rnd$.

\subsection{Security Analysis of the Hybrid-Model Protocol}\label{sec:AnalysisOfProtocol}

In this section we prove the following theorem,  stating that \cref{protocol:outer} cannot be biased  much by a fail-stop adversary.

\begin{theorem}\label{thm:MainFailStop}
  Fix an integer function $t'$ with $t'(m) \leq \frac12\loglog m$. For integers $\rnd \equiv 1 \bmod 12$ and $\partNum = \partNum'(m)$, protocol $\ProtocolF{\partNum}_\rnd$ is a $(\partNum \cdot \rnd)$-round, $\partNum$-party, $O\left(\frac{t \cdot 2^t \cdot \sqrt{\log m}}{m^{1/2+1/\left(2^{t-1}-2\right)}}\right)$-fair, coin-flipping protocol, against unbounded fail-stop adversaries, in the $(\Defense,\Coin)$-hybrid model.\footnote{The hidden constant in the $O$ notation is independent of $t'$ and $m$.}
\end{theorem}

We prove \cref{thm:MainFailStop} in \cref{subsec:ProvingMainThmFailStop}, but first introduce the main tools and concepts used for this proof. \textit{Leakage from two-step boolean process} used to bound attack in \Stepref{protocol:inner:t:defense:1}, is presented in \cref{subsec:TwoStepProccess}. \textit{Binomial games} used to bound an attack inside the loop of \cref{protocol:inner:t}, are introduced in  \cref{subsec:formal_binomial}. Finally, in \cref{subsec:BasicObservations} we note several simple facts about the protocol.

\subsubsection{Leakage from Two-Step Boolean Processes}\label{subsec:TwoStepProccess}
Our main tool for analyzing  the effect of an  abort in \Stepref{protocol:inner:t:defense:1} of protocol $\Protocol{\partNumInn}_\rnd$, for $\partNumInn>2$, is bounding the leakage from the relevant ``two-step boolean process''. 
A two-step boolean process is a pair of jointly-distributed random variables $(\ElementVar,\Value)$, where $\Value$ is over $\zo$ and $\ElementVar$ is over an arbitrary domain $\ElementSet$. It is instructive to think that  the process' first step is  choosing $\ElementVar$, and its second step is to choose $\Value$ as a random function of $\ElementVar$. A leakage function $\HintFunc$ for a two-step process $(\ElementVar,\Value)$ is simply a randomized function over the support of $\ElementVar$. We will be interested in bounding by how much  the expected outcome of $\Value$ changes when $\HintFunc(\ElementVar)$ is leaked. This change is captured via the notion of prediction advantage.

\def\PredictionAdvantage{
        For a two-step process $\Pc = (\ElementVar,\Value)$ and a leakage function $\HintFunc$ for $\Pc$, define the \emph{prediction advantage} $\PredictAdv_{\Pc,\HintFunc}$ by $\PredictAdv_{\Pc,\HintFunc}(\hint) = \size{\pr{\Value = 1} - \pr{\Value = 1 \mid \HintVar = \hint}}$.
}
\begin{definition}[prediction advantage]\label{def:PredictionAdvantage:pre}
    \PredictionAdvantage
\end{definition}

We now define the notions of an hypergeometric process, and of vector leakage function. As we shall see later on, the boolean process that induced in \cref{protocol:inner:t:defense:1} of protocol $\Protocol{\partNumInn}_\rnd$, can be viewed as such a hypergeometric process, coupled to a vector leakage function.

\def\VectorLeakageFunction{
     Let $\vctBaseLen, \vctLenFact$  be integers. A randomized function $\HintFunc$ is a {\sf $(\vctBaseLen, \vctLenFact)$-vector leakage function}  for the two-step Boolean process $(\ElementVar,\Value)$, if on input  $\element \in \Supp(\ElementVar)$,  it outputs a vector in $\oo^{\vctTotalLen}$ according to $(\Beroo{\eps})^{\vctTotalLen}$, for $\eps = \sBias{\vctBaseLen}{\eex{\Value \mid \ElementVar = \element}}$.   
}

\begin{definition}[vector leakage function]\label{def:VectorLeakageFunction:pre}
    \VectorLeakageFunction
\end{definition}
 
\def\HypProcess{
    Let $\vctBaseLen, \hypVctLenFact \in \N$ and $\delta \in [0,1]$. 
    An {\sf $\bigl(\vctBaseLen, \hypVctLenFact, \delta\bigr)$-hypergeometric process} is the two-step Boolean process $(\ElementVar,\Value)$ defined by 
    \begin{enumerate}
        \item $\ElementVar = \vHyp{\hypVctLenFact \cdot \vctBaseLen, \w(v), \vctBaseLen}(0)$, for $v\la (\Beroo{\eps})^{\hypVctLenFact \cdot \vctBaseLen}$ and $\eps = \sBias{\vctBaseLen}{\delta}$.
        
        \item $\Value \la \Berzo{\ElementVar}$, 
    \end{enumerate}
}
\begin{definition}[Hypergeometric process]\label{def:HypProcess:pre}
    \HypProcess
\end{definition}

In \cref{subsec:ProvingMainThmFailStop} we use the following lemma to bound the gain an adversary can achieve by aborting at  \Stepref{protocol:inner:t:defense:1} of $\Protocol{\partNumInn}_\rnd$. The proof is given in \cref{sec:ExpChangeDueLeakage}.

\def\HyperProcessVectorHint{
    Assume $\vctBaseLen, \vctLenFact, \hypVctLenFact \in \N$ and $\delta \in [0,1]$,
    satisfy
    \begin{enumerate}
        
        \item $2 \leq \vctLenFact < \hypVctLenFact \leq \vctBaseLen$,\label{lemma:HyperlProcessVectorHint:asmp2}
        
        \item $\frac{\vctLenFact + \sqrt{\vctBaseLen}}{\vctBaseLen} \cdot \log^2 \vctBaseLen  \leq 10^{-5} \cdot \sqrt{\frac{\vctLenFact}{\hypVctLenFact}}$, and \label{lemma:HyperlProcessVectorHint:asmp3}
        
        \item $\sqrt{\frac{\vctLenFact}{\hypVctLenFact}}\cdot \log \vctBaseLen \leq \frac1{100}$.\label{lemma:HyperlProcessVectorHint:asmp4}
    \end{enumerate}
    Let $\Pc = (\ElementVar, \Value)$ be a $\bigl(\vctBaseLen, \hypVctLenFact, \delta\bigr)$-hypergeometric process according to \cref{\HyperProcessLabel} 
    , let $\HintFunc$ be an $\bigl(\vctBaseLen, \vctLenFact\bigr)$-vector leakage function for $\Pc$ according to \cref{\VectorLeakageFunctionLabel}
    , and let $\PredictAdv_{\Pc,\HintFunc}$ be according to \cref{\PredictionAdvantageLabel}
    .
    Then, there exists a universal constant $\const > 0$ such that
    \begin{align*}
    \ppr{\hint \la \HintVar}{\PredictAdv_{\Pc,\HintFunc}(\hint) > \const\cdot  \sqrt{\log \vctBaseLen} \cdot \frac{\sqrt{\vctLenFact}}{\hypVctLenFact}} \leq \frac1{\vctBaseLen^2}.
    \end{align*}
}

\def\HyperProcessLabel{def:HypProcess:pre}
\def\VectorLeakageFunctionLabel{def:VectorLeakageFunction:pre}
\def\PredictionAdvantageLabel{def:PredictionAdvantage:pre}
\begin{lemma}\label{lemma:HyperlProcessVectorHint:pre}
    \HyperProcessVectorHint
\end{lemma}

\paragraph{Proof sketch.}
\cref{lemma:HyperlProcessVectorHint:pre} is proved in \cref{subsubsec:HyperOneRoundGameVectorHint},  yet to make the current section more self contained we give a high-level proof sketch (under some simplifying assumptions).

Assume for simplicity that $\delta=1/2$, and therefore in the Hypergeometric process  we sample $v \la \oo^{\beta s}$ (\ie  each $v_i$ is an unbiased coin). Therefore, $\size{w(v)}$ is expected to be $\approx \sqrt{\beta s}$, yielding that $\size{A-1/2} \approx 1/\sqrt{\beta}$ (follows by \cref{prop:hyperToNormal}). For simplicity, assume that either $A=1/2+1/\sqrt{\beta}$ or $A=1/2-1/\sqrt{\beta}$, and each happens with probability $1/2$. In the $(\alpha,s)$-vector leakge, we essentially reveal $\alpha$ i.i.d.\ samples, each takes $1$ w.p.\ $A$ and $-1$ otherwise (this is because we reveal $\alpha s$ i.i.d.\ samples where the weight of each $s$ samples is positive with probability $A$). For simplicity, assume further that we only reveal whether the sum of those $\alpha$ samples (denote them by $Z_1,\ldots,Z_{\alpha}$) is positive.
Now, we need to analysis how much information the event $\set{\sum_{i=1}^\alpha Z_i \geq 0}$ leaks on the (independent) sample $B \la \Berzo{A}$. Compute

\begin{align*}
	\lefteqn{\pr{B=1 \mid \sum_{i=1}^\alpha Z_i \geq 0}}\\
	&= \pr{B=1 \mid A = 1/2+1/\sqrt{\beta}} \cdot \pr{A = 1/2+1/\sqrt{\beta} \mid \sum_{i=1}^\alpha Z_i \geq 0}\\
	&+ \pr{B=1 \mid A = 1/2-1/\sqrt{\beta}} \cdot \pr{A = 1/2-1/\sqrt{\beta} \mid \sum_{i=1}^\alpha Z_i \geq 0}\\
	&= (1/2+1/\sqrt{\beta}) \cdot \pr{\sum_{i=1}^\alpha Z_i \geq 0 \mid A = 1/2+1/\sqrt{\beta}}\cdot \frac{\pr{A = 1/2+1/\sqrt{\beta}}}{\pr{\sum_{i=1}^\alpha Z_i \geq 0}}\\
	&+ (1/2-1/\sqrt{\beta}) \cdot \pr{\sum_{i=1}^\alpha Z_i \geq 0 \mid A = 1/2-1/\sqrt{\beta}} \cdot \frac{\pr{A = 1/2-1/\sqrt{\beta}}}{\pr{\sum_{i=1}^\alpha Z_i \geq 0}}\\
	&\approx (1/2+1/\sqrt{\beta})(1/2+\sqrt{\alpha/\beta}) + (1/2-1/\sqrt{\beta})(1/2-\sqrt{\alpha/\beta})\\
	&= 1/2 + \sqrt{\alpha}/\beta.
\end{align*}
The ``$\approx$'' transition holds since: (1) $\pr{A = 1/2+1/\sqrt{\beta}} = \pr{A = 1/2-1/\sqrt{\beta}} = \pr{\sum_{i=1}^\alpha Z_i \geq 0} = 1/2$,  (2) The sum of $\alpha$ independent samples from $\Beroo{1/\sqrt{\beta}}$ is positive with probability $\approx 1/2+\sqrt{\alpha/\beta}$, and (3) The sum of $\alpha$ independent samples from $\Beroo{-1/\sqrt{\beta}}$ is positive with probability $\approx 1/2-\sqrt{\alpha/\beta}$.

Since $\pr{B=1} = 1/2$, we conclude that $\size{\pr{B=1} - \pr{B=1 \mid \sum_{i=1}^\alpha Z_i \geq 0}} \leq  \sqrt{\alpha}/\beta$. 

\subsubsection{Online-Binomial Games}\label{subsec:formal_binomial}

Our main tool for analyzing the effect an abort in the main loop of the protocol has, is bounding the bias of the relevant ``online-binomial games''. Following the informal discussion given in \cref{sec:Introduction}, we give here a formal definition of such games. While in the introduction we referred to a very narrow notion of binomial game, here we cover a wider class of games, letting the challenger to toss many, possibly biased, coins in each round. 

\def\onlineBinomialGame{
	Let $m \in \N$, $\eps \in [-1,1]$, and   $f$ be a randomized function over $[m] \times \Z \times \Z$. The {\sf $m$-round online binomial game $\game_{m,\eps,f}$} is the random variable  $\game_{m,\eps,f} = \set{C_1,\ldots,C_m,f}$, where for every $i \in [m]$, $C_i \la \Beroo{(m-i+1)^2,\eps}$.  
	We refer to each $C_i$ as the {\sf $i$'th round coins}, and to $f$ as the {\sf hint} function.
}

\begin{definition}[online-binomial game]\label{def:game}
	\onlineBinomialGame
\end{definition} 

We will be interested in bounding by how much the outcome of such a game can be biased.

\def\biasOnlineBinomialGame {
	Let $\game = \game_{m,\eps,f} = \set{C_1,\ldots,C_m,f}$ be an $m$-round online binomial game. 
	For $i \in \set{1,\ldots,m}$, let $S_i = \sum_{j=1}^i C_j$, letting $S_0=0$. 
	For $i \in \set{1,\ldots,m}$, let $H_i = f(i,S_{i-1},C_i)$, let $\delta_i(b) = \pr{S_m \geq 0 \mid  S_{i-1}=b}$, let $\delta_i(b,h) = \pr{S_m \geq 0 \mid  S_{i-1}=b,\ H_i=h}$, let $O_i=\delta_i(S_{i-1},H_i)$, and let $O_i^-=\delta_i(S_{i-1})$. Let also $O_{m+1}=O_{m+1}^-=1$ if $S_m \geq 0$, and let  $O_{m+1}=O_{m+1}^-=0$ if $S_m < 0$.

	For an algorithm $\Strategy$, let $I$ be the first round in which $\Bc$ outputs $1$ in the following $m$-round process:
	In round $i$, algorithm $\Bc$ is getting input $(S_{i-1},H_i)$ and outputs a $\zo$-value. Let $I=m+1$ if $\Bc$ never outputs a one. The {\sf bias $\Bc$ gains in $\game$} is defined by
	$$\bias_\Bc(\game) = \size{\eex{O_I - O_I^- }}$$
	The {\sf bias of $\game$} is defined by $\bias_{m,\eps,f} = \bias(\game) = \max_\Bc \set{ \bias_\Bc(\game)}$, where the maximum is over {\sf all} possible algorithms $\Bc$.	
}

\begin{definition}[The bias of $G_{m,\eps,f}$]\label{def:gameBias}
    \biasOnlineBinomialGame
\end{definition}

Namely, in the $i$'th round the algorithm $\Strategy$ is getting the sum of the coins flipped up to previous round - $S_{i-1}$, and a ``hint'' $H_i= f(i,S_{i-1}, C_i)$. If the $\Strategy$ decides to abort, it get rewarded by $\size{\delta_i(S_{i-1},H_i) - \delta_i(S_{i-1})}$. Hence, $\Strategy$'s ``goal'' is to find the round in which the above gain is maximized.

In the proof of \cref{thm:MainFailStop}, we use the following two lemmas (proven in \cref{sec:BinomialViaLP}). 

\def\VectorHint {
    For $\rnd,\ell \in \N$ and  $\eps \in [-1,1]$, define the random function $\fvec{\rnd,\eps,\ell} \colon [m] \times \Z \times \Z \mapsto \set{-1,1}^{\ell}$ as follows: on input $(i,b,c)$, it calculates  $\delta = \vBeroo{\ms{i+1},\eps}(-b-c)$, and  $\eps \eqdef \sBias{\ms{1}}{\delta}$, and returns a random sample from $(\Beroo{\eps})^\ell$.
}

\begin{definition}[Vector hint]\label{def:vectorHint} 
    \VectorHint
\end{definition}

\def\BinomialVectorHintInterfaceLemma{
	For $\rnd \in \N$, $k \in [\rnd]$, $\eps \in [-1,1]$, and $f=\fvec{\rnd,\eps,k \cdot \ms{1}}$, let $\game$ be the binomial game $\game_{m,\eps,f}$ according to \cref{\GameLabel}. 
	Assuming that $k \leq \frac{m}{\log^6 m}$, it holds that $\bias_{\game} \in O(\frac{\sqrt{k}}{m} \cdot \sqrt{\log \rnd})$.	
}

\def\GameLabel{def:game}
\begin{lemma}\label{binomial:game:vec}
	\BinomialVectorHintInterfaceLemma
\end{lemma}

\def\HyperHint {
	For $\rnd \in \N$, and an integer $p \in [-2 \cdot \ms{1},2 \cdot \ms{1}]$, define the random function $\fhyp{\rnd,p} \colon [m] \times \Z \times \Z \mapsto \set{-1,1}$ as follow: on input $(i,b,c) $ outputs $1$ with probability $\vHyp{2\cdot \ms{1},p,\ms{i+1}}(-b-c)$ and $-1$ otherwise. 
}

\begin{definition}[hypergeometric hint]\label{def:hypHint}
    \HyperHint
\end{definition}

\def\BinomialHyperHintInterfaceLemma{
	Let $\rnd \in \N$,  $\eps \in [-1,1]$, and let $p$ be integer in $[-2 \cdot \ms{1},2 \cdot \ms{1}]$. Assume  that $\size{p} \leq \const \cdot \sqrt{\log \rnd \cdot \ms{1}}$ for some constant $\const$, and let $f=\fhyp{\rnd,p}$. Let $\game$ be the binomial game  $\game_{m,\eps,f}$ according to \cref{\GameLabel}, then $\bias_{\game} \in O(\frac{\sqrt{\log \rnd}}{\rnd})$.	
}

\def\GameLabel{def:game}
\begin{lemma}\label{binomial:game:hyp}
	\BinomialHyperHintInterfaceLemma
\end{lemma}

\subsubsection{Basic Observations  about \cref{protocol:outer}}\label{subsec:BasicObservations}

The following  simple facts  are used within the proof of \cref{thm:MainFailStop}. We start with  a simple observation regarding the outcome of the $\Defense$ functionality. 


\begin{fact}\label{fact:defenseOutcome}\label{fact:simple:1}  
	Let $\rnd \geq 1$,  $\partNumInn > 2$,  $\ell \in \N$, $\delta \in [0,1]$ and $\cZ = (z_{1},\ldots,z_{\size{\cZ}})  \subseteq [\partNumInn]$, and let $S= (S_1,\ldots,S_\partNumInn)= \DefenseTilde(1^\rnd,1^\partNumInn,1^\ell,\cZ,\delta)$.  Let  $\outcome(S)$ be the outcome of a non-aborting  execution of protocol $\Protocol{\size{\cZ}}_\rnd$ on common input $1^\ell$, and the $j$'th  party private input is set to $S_{z_j}$. Then for every $\cB \subset [\partNumInn]$ with $\cZ \not \subset \cB$ and for every  $\vect{s} \in \Supp(S^\cB = \set{S_z}_{z\in \cB})$, it holds that $\eex{\outcome(S) \mid S^\cB  =\vect{s}} = \delta$.
\end{fact}
Namely, in an honest interaction that follows an  abort, the expected outcome of the interaction is $\delta$, for $\delta$ being the input in the last call to $\DefenseTilde$ that happened before the abort. The latter holds, even conditioned on the partial information held by the corrupted parties.
\begin{proof}
    Assume \wlg that $\cB = \set{2,\ldots,r}$, and that $\cZ = \set{1,\ldots,\size{\cZ}}$. Consider an honest execution of protocol $\Protocol{\size{\cZ}}_\rnd$, in which  party $\Party{rrr}_z$ for $z \in \cZ$ start with private inputs $S_j$ for $1 \leq j \leq \size{\cZ}$.\footnote{The parties that participate in this execution have more inputs for the case that some of them will abort later on. Since, however, we are interested in an honest execution, those additional inputs can be ignored.} Let $\vect{s}=(s_2,\ldots,s_r) \in \Supp(S^\cB)$.

    By construction of $\Defense$ functionality, and specifically since it breaks the output into random shares, it holds that $\eex{\bigoplus_{i=1}^{\size{\cZ}} S_i \mid S^\cB  =s} = \delta$. Writing it a bit differently:
    \begin{align}
	    \ex{S_1}{S_1 \xor s_2 \xor \ldots \xor s_{\size{\cZ}}} = \delta \label{subsec:BasicObservations:eq1}
    \end{align}
    By construction of \cref{protocol:inner:t}, it holds that 
    \begin{align}
       \eex{\outcome(S) \mid S_1,\ldots,S_{\size{\cZ}}} = S_1 \xor \ldots \xor S_{\size{\cZ}}  \label{subsec:BasicObservations:eq2}
    \end{align}
Putting it together we get:
    \begin{align*}
	    \eex{\outcome(S) \mid S^\cB = s} &= \eex{\outcome(S) \mid S_2=s_2,\ldots,S_r=s_r}\\
	    &= \ex{S_1}{\eex{\outcome(S) \mid S_1,S_2=s_2,\ldots,S_r=s_r} }  \nonumber   \\
	    &= \ex{S_1}{S_1 \xor s_2 \xor \ldots \xor s_{\size{\cZ}} }  \nonumber     \\
	    &= \delta,  \nonumber   
    \end{align*}
    as required.
\end{proof}

We remind the reader that the $\alpha$-factors are: $\alpha(\rnd, \ell, k) =  m^{\frac{2^{\ell-3}}{2^{\ell-2}-1} \cdot \frac{2^{k-2}-1}{2^{k-3}}}$ (see \cref{def:alphaFactors}). The following fact states some basic properties of the $\alpha$-factors.

\begin{fact}\label{fact:alphaFacts}
	Let $\rnd \geq 1$ and $\ell \geq 2$ be two integers, and denote for simplicity $\alp{k}=\alpha(\rnd, \ell, k)$. It holds that 
	\begin{enumerate}
		\item $\alp{\ell-1} = m^{1 - \frac1{2^{\ell-2}-1}}$.
		\item $ \alp{2} = 1 $.
		\item $\frac{\sqrt{\alp{2}}}{\alp{3}} = \ldots = \frac{\sqrt{\alp{\ell-3}}}{\alp{\ell-2}} = \frac{\sqrt{\alp{\ell-2}}}{\alp{\ell-1}} = \frac{\sqrt{\alp{\ell-1}}}{\rnd}
		= \frac{1}{\rnd^{\frac12+\frac1{2^{\ell-1}-2}  }   } $.
	\end{enumerate}
\end{fact}
\begin{proof}
	Immediate by definition.
\end{proof}

\subsubsection{Proving  \cref{thm:MainFailStop}} \label{subsec:ProvingMainThmFailStop}

\begin{proof}[Proof of \cref{thm:MainFailStop}]
	
By construction, in an all-honest execution the parties output a uniform bit, so it left to prove that the protocol cannot be biased by too much by fail-stop adversaries.  

Let $\Aadv$ be a fail-stop adversary controlling the parties $\set{\PartyF{\partNum}_\z}_{\z\in \cC}$ for some $\cC \subsetneq [\partNum]$.  Let $V$ be the (joint) view of the parties controlled by $\Aadv$,  let $V_i$ be the prefix of $V$ at the end of round $i$, and  $V_i^-$ be the prefix of $V_{i}$ with the $i$'th round abort messages (if any) removed.  Let $\val(v)$ be the expected outcome of an honest (non-aborting) execution of the protocol by the parties that do not abort in $v$, conditioned on $v$ (see \cref{def:ViewVal}).  We assume \wlg that if $\Ac$ instructs a corrupted party to abort at a given round, it does so \emph{after} seeing the honest parties' messages of that round.
	
For $k \in [\partNum-1]$, let $I_k$ be the $k$-th aborting communication round (that is, the $k$'th round in which at least one party aborts). Letting $I_k=\perp$ if less than $k$ aborting rounds happen, and let $V_{\perp} = V_{\perp}^-$. By \cref{prop:FairCTGameAlt}, to prove the theorem it is sufficient to show that:
\begin{align}\label{theorem_fs:main_eq}
  \size{\ex{V}{\sum_{k=1}^{\partNum-1} \val(V_{I_k}) - \val(V_{I_k}^-)}} \leq O \left(\biasTerm \right).
\end{align}
Since
\begin{align*}
  \size{\ex{V}{\sum_{k=1}^{\partNum-1} \val(V_{I_k}) - \val(V_{I_k}^-)}}  =  \size{\sum_{k=1}^{\partNum-1} \ex{V}{\val(V_{I_k}) - \val(V_{I_k}^-)}}
  \leq \sum_{k=1}^{\partNum-1} \size{\ex{V}{\val(V_{I_k}) - \val(V_{I_k}^-)}},
\end{align*}
it suffices to  show that 
\begin{align}
  \size{\ex{V}{\val(V_{I_k}) - \val(V_{I_k}^-)}} \leq  O \left(\biasTermSingle \right) \label{theorem_fs:single_abort}
\end{align}
for every $1 \leq k \leq t-1$.

Fix $1 \leq k \leq t-1$. The $k$'th abort can occur in one of the following places:
\begin{itemize}
	\item In \Stepref{protocol:outer:defense} of the parent protocol $\ProtocolF{\partNum}_\rnd$ (can only happen for $k=1$).
	\item During the execution of protocol $\Protocol{\partNumInn}_\rnd$, for some $\partNumInn \leq \partNum$.
\end{itemize}
Since by construction aborting in \Stepref{protocol:outer:defense} gives nothing to the adversary,  it is left to prove that \cref{theorem_fs:single_abort} holds for aborting done during an execution of $\Protocol{\partNumInn}_\rnd$.

Let $R= R(k)$ be the number of active parties when the $k$'th abort occur (that means that it occurs during the execution of $\Protocol{R}_\rnd$).  We show that for any value of $r \in \set{2,\ldots,t}$, it holds that
\begin{align}
\size{\ex{V| R=r}{\val(V_{I_k}) - \val(V_{I_k}^-)}} \leq  O \left(\biasTermSingle \right) \label{theorem_fs:single_abortR}
\end{align}
and \cref{theorem_fs:single_abort} will follow.

Let $I = I_k$. For $r \in \set{2,\ldots,t}$, we condition till the end of the proof on $R=r$. We distinguish between the case  $r>2$ case, and $r=2$.

\paragraph{The case $r>2$.}
Recall that protocol $\Protocol{\partNumInn}_\rnd$ has five step: \Stepref{protocol:inner:t:defense:1}, \Stepref{protocol:inner:t:reconstruct:delta},  \Stepref{protocol:inner:t:coin}, \Stepref{protocol:inner:t:defense:2}, and \Stepref{protocol:inner:t:reconstruct:coin}. We let $\cs = \set{1,2,3a,3b,3c}$, and let $T \in \cs$ to be the step executed in round $I$. Applying complete expectation on the left side of \cref{theorem_fs:single_abortR}, we get that (we remind the reader that with let $I_k=I$, and we fixed some value of $r$):
\begin{align}
\size{\ex{V}{\val(V_{I}) - \val(V_{I}^-)}} = \sum_{j \in \mathcal{S}} \size{\ex{V \mid T=j}{\val(V_{I} ) - \val(V_{I}^-)}} \cdot \pr{T=j} \label{theorem_fs:single_abort:specific_round}
\end{align}
In the following we prove that for $t \in \set{2,3a,3c}$ it holds that 
\begin{align}
  \val(V_{I}|_{T=t}) = \val(V_{I}^-|_{T =t}), \label{theorem_fs:single_abort:round:no_gain}
\end{align}
that for $t=1$ it holds that
\begin{align}
  \size{\eex{\val(V_{I}) - \val(V_{I}^-) \mid T=1}} \cdot \pr{T=1} \leq O \left(\biasTermSingle \right),  \label{theorem_fs:single_abort:round:s1}
\end{align}
and that for $t=3b$ it holds that
\begin{align}
  \size{\eex{\val(V_{I}) - \val(V_{I}^-)\mid T=3b}} \cdot \pr{T=3b} \leq O \left(\biasTermSingle \right). \label{theorem_fs:single_abort:round:s3}
\end{align}

Putting \cref{theorem_fs:single_abort:round:no_gain}, \cref{theorem_fs:single_abort:round:s1}, and \cref{theorem_fs:single_abort:round:s3}, in  \cref{theorem_fs:single_abort:specific_round}, yields that $\size{\ex{V}{\val(V_{I}) - \val(V_{I}^-)}} \leq O \left(\biasTermSingle \right)$, proving \cref{theorem_fs:single_abortR}.

The following random variables are define \wrt this interaction of $\Protocol{r}_\rnd$. Let  $\Delta$ be the value of $\delta$ calculated in \Stepref{protocol:inner:t:reconstruct:delta}, set to $\perp$ if an abort occurred before this round (\ie $T<2$), and let $\DeltaDef$ be the value of the parameter $\delta$  in the last call to $\DefenseTilde$ before $I$ (by definition, such a call is guaranteed to exist).

By \cref{fact:simple:1}, it holds that
\begin{align}\label{eq:mainLem:Delta}
\val(V_I) = \DeltaDef
\end{align}

\paragraph{Proving \cref{theorem_fs:single_abort:round:no_gain}.}  We prove separately for every $t\in \set{2,3a,3c}$.
\begin{description}
	\item[$t=2$:] By construction, in case of no abort, the expected outcome of the protocol at the end of Step $2$ is $\Delta$, namely   $\val(V_I^{-}|_{T=2}) = \Delta$. Since by construction $\Delta = \DeltaDef$, \cref{eq:mainLem:Delta} yields that $ \val(V_{I}|_{T=2}) = \val(V_{I}^-|_{T =2})$.

	 \item[$t=3a$:] Since $c_i$ and $\delta_i$ are shared using an $\partNumInn$-out-of-$\partNumInn$ secret sharing schemes, $V_I$ contains \emph{no information} about $c_i$ and $\delta_i$. Thus, $\val(V_I^{-}) = \val(V_{I-1})$. If $I$ is the very first round to reach \Stepref{protocol:inner:t:coin} (\ie we are in the first round of the loop), then by construction $\val(V_{I-1}) = \Delta= \DeltaDef$. Otherwise (not the first round in the loop), by definition $\val(V_{I-1}) = \pr{\sign(\sum_{i=1}^m c_i) =1 \mid V_{I-1}}$, which by construction is also equal to $\DeltaDef$.  Hence, by \cref{eq:mainLem:Delta},  $\val(V_{I}) = \val(V_{I}^-)$.
  
	  \item[$t=3c$:] Follows by an analogues argument to that used for proving the case $t=2$.
\end{description}

\paragraph{Proving \cref{theorem_fs:single_abort:round:s1}.}
In the following we condition on $V_{I-1} = v'$ for some  $v' \in \Supp(V_{I-1})|_{T=1}$.

Let $M$ be the messages that the corrupted parties receive during the execution of \cref{protocol:inner:t:defense:1}. We assume \wlg that $M$ contains also the vectors of coins sampled in \Stepref{addnoise:sampleVec} of functionality $\AddNoise$ (happened by the joint call to  $\Defense$ done in this round) that was used to generate the defense values (\ie the messages) of the corrupted parties.\footnote{An adversary that can bias the protocol without this additional information, can be emulated by an adversary that get this additional information.}
We remind the reader that $\DeltaDef$ is the value of the $\delta$ parameter passed to the last call of $\DefenseTilde$. By construction, $\DeltaDef$ is a deterministic function of $V_{I-1}$. Since we conditioned on $V_{I-1} = v'$, we conclude that $\DeltaDef$ has a fixed value, denote this value by $\deltaDef$.

 The proof follow by the next claim (proven below).
\begin{claim}\label{protocol:lemma:inner:many:defense1:interface} It holds that
  \begin{align*}
    \ppr{n \la M}{\size{\deltaDef - \eex{\Delta\mid M=n}} > \lambda \cdot \frac{\sqrt{2^{\partNumInn} \cdot \alp{\partNumInn-1}}}{\alp{\partNumInn}} \cdot \sqrt{\log \rnd}} \leq \frac1{\rnd^2}.
  \end{align*}
\end{claim}

Namely, with high probability, after the adversary sees the messages of \Stepref{protocol:inner:t:defense:1}, the value of $\Delta$ is not far from $\deltaDef$.  It follows that
\begin{align}
\ex{V_I \mid T=1}{\val(V_I)-\val(V_I^{-})} &= \ex{V_I \mid T=1}{\val(V_I)}-\ex{V_I \mid T=1}{\val(V_I^{-})} \label{main:theorem:mv0} \\
                                           &= \deltaDef-\ex{V_I \mid T=1}{\val(V_I^{-})} \label{main:theorem:mv0.1} \\
                                           &= \deltaDef-\ex{V_I \mid T=1}{\eex{\Delta \mid V_I }} \label{main:theorem:mv1}\\
                                           &= \ex{V_I \mid T=1}{\deltaDef-\eex{\Delta \mid V_I}}  \nonumber  \\
                                           &= \ex{M \mid T=1}{\deltaDef-\eex{\Delta \mid M}}  \nonumber
\end{align}

\cref{main:theorem:mv0.1} holds by \cref{fact:simple:1}. \cref{main:theorem:mv1} holds since conditioned on  $T=1$, $\val(V_I^{-}) = \eex{\Delta \mid V_I }$.
Applying triangle inequality to \cref{main:theorem:mv0}, and multipling it by $\pr{T=1}$, it holds that 
\begin{align*}
\size{\ex{V_I \mid T=1}{\val(V_I)-\val(V_I^{-})}} \cdot \pr{T=1} \leq \ex{M \mid T=1}{\size{\deltaDef - \eex{\Delta \mid M}}} \cdot \pr{T=1} 
\end{align*}
Continuing the evaluation, it holds that
\begin{align*}
\ex{N \mid T=1}{\size{\deltaDef - \eex{\Delta \mid N}}} \cdot \pr{T=1} &\leq  \ex{N}{\size{\deltaDef - \eex{\Delta \mid N}}}\\
                                                                         &= \sum_{n \in \Supp(N)} \size{\deltaDef - \eex{\Delta \mid N=n}} \cdot \pr{N=n}\nonumber,
\end{align*}
Applying \cref{protocol:lemma:inner:many:defense1:interface} to previous inequality yields that
\begin{align}
\ex{M \mid T=1}{\size{\deltaDef-\eex{\Delta \mid M}}} \cdot \pr{T=1} &\leq 1 \cdot \frac1{\rnd^2} + \lambda \cdot \frac{\sqrt{2^{\partNumInn} \cdot \alp{\partNumInn-1}}}{\alp{\partNumInn}} \cdot \sqrt{\log \rnd} 
\end{align}
for some universal constant $\lambda$. Finally, since	$\frac{\sqrt{2^{\partNumInn} \cdot \alp{\partNumInn-1}}}{\alp{\partNumInn}} \leq \frac{\sqrt{2^{\partNumInn} \cdot \alp{\ell-1}}}{\rnd} = \frac{\sqrt{2^{\partNumInn}}}{\rnd^{\frac12+\frac1{2^{\ell-1}-2}  } }$ (\cref{fact:alphaFacts}), we conclude that 
\begin{align*}
\size{\ex{V_I \mid T=1} {\val(V_{I}) - \val(V_{I}^-)}} \cdot \pr{T=1}  \leq  O \left(\biasTermSingleParam{\partNumInn}{\ell} \right) = 
                                                                             O \left(\biasTermSingleParam{\partNum}{\ell} \right),
\end{align*}
which is the same as \cref{theorem_fs:single_abort:round:s1}.

\paragraph{Proving \cref{theorem_fs:single_abort:round:s3}.} We prove \cref{theorem_fs:single_abort:round:s3} in the following claim.
\begin{claim}\label{claim:connectionToVectorGame}
	$\size{\eex{\val(V_I) - \val(V_I^-)\ \mid\  T=3b}} \cdot \pr{T=3b} \leq O \left(\biasTermSingleParam{\partNum}{\ell} \right)$.
\end{claim}

\paragraph{The case of $\partNumInn=2$.} \label{sec:ProofMainLemmaTwo}

The proof of this case  follows similar lines to that of \cite[Thm 3.10]{HaitnerT17}, using the new bound for binomial game given in   \cref{binomial:game:hyp}, instead of the bound used in \cite{HaitnerT17}. Details below.

We prove $\Protocol{2}_\rnd = \HTProtocol_\rnd$ is secured against an abort action. By construction, $\Aadv$ can abort either in \Stepref{protocol:inner:2:defense} (\ie during the call to $\HTDefenseRound$), or in \Stepref{protocol:inner:2:reconstructCoin} (\ie during the reconstruction of the coin). We let $T \in \set{1a,1b}$ to be the step executed in round $I$. Applying complete expectation on the left side of \cref{theorem_fs:single_abortR}, we get that :
\begin{align}
\size{\ex{V}{\val(V_{I}) - \val(V_{I}^-)}} = \sum_{j \in \set{1a,1b}} \size{\ex{V \mid T=j}{\val(V_{I} ) - \val(V_{I}^-)}} \cdot \pr{T=j} \label{theorem_fs:single_abort:specific_round:r2}
\end{align}
Similar lines to that used to analyze abort in \Stepref{protocol:inner:t:reconstruct:coin}  of protocol $\Protocol{\partNumInn}_\rnd$ with $\partNumInn > 2$, yield that conditioned on $T=1b$, it holds that $\eex{\val(V_{I})} = \eex{\val(V_{I}^-)}$. Putting it in \cref{theorem_fs:single_abort:specific_round:r2}, we get:
 
\begin{align*}
   \size{\ex{V}{\val(V_{I}) - \val(V_{I}^-)}} = \size{\ex{V \mid T=1a}{\val(V_{I} ) - \val(V_{I}^-)}} \cdot \pr{T=1a}
\end{align*}
Hence, the following finishes the proof of the theorem
\begin{align}
\size{\ex{V \mid T=1a}{\val(V_{I} ) - \val(V_{I}^-)}} \cdot \pr{T=1a} \leq O \left(\biasTermSingle \right) \label{theorem_fs:single_abort:specific_round:r2:1a}
\end{align}

let $P$ be the sum of all entries of the vector $\bankV{1}$ (sampled at \Stepref{htdefenseprotocol:vecSample} of $\HTDefenseProtocol$) during the last execution of $\HTDefenseProtocol$. \footnote{Using the notations from \cref{sec:prelim:notation}, we can define $P$ to be: $P=w(\bankV{1})$.} Let $\tau =  12 \cdot \sqrt{\log \rnd \cdot \ms{1}}$. It holds that,
\begin{align}
&\size{\eex{\val(V_I) - \val(V_I^-)\ \mid\ T=1a}} \cdot \pr{T=1a} = \nonumber \\
&\size{\eex{\val(V_I) - \val(V_I^-)\ \mid \size{P}>\tau, T=1a}}  \cdot \pr{\size{P}>\tau\ \mid\ T=1a}  \cdot \pr{T=1a}\ +  \label{case:r2:term1}\\
&\size{\eex{\val(V_I) - \val(V_I^-)\ \mid \size{P} \leq \tau, T=1a}} \cdot \pr{\size{P} \leq \tau\ \mid\ T=1a} \cdot \pr{T=1a} \label{case:r2:term2}
\end{align}
The term from Line \ref{case:r2:term1} contains in it $\pr{\size{P}>\tau\ \mid\ T=1a}  \cdot \pr{T=1a}$ which is bounded by $\pr{\size{P} > \tau}$.
By Hoeffding's inequality,
\begin{align}
\pr{\size{P} > \tau} &\leq \pr{\size{P -2 \eps \cdot \ms{1}} > 4 \cdot \sqrt{\log \rnd \cdot \ms{1}}} \leq \frac1\rnd   \label{proof_r_equal_2:eq1}.
\end{align}
The term from Line \ref{case:r2:term2} satisfies:
\begin{align*}
&\size{ \eex{\val(V_I) - \val(V_I^-)\ \mid \size{P} \leq \tau, T=1a} } \cdot \pr{\size{P} \leq \tau\ \mid\ T=1a} \cdot \pr{T=1a} \leq \\
&\size{ \eex{\val(V_I) - \val(V_I^-)\ \mid \size{P} \leq \tau, T=1a} } \cdot \pr{T=1a\ \mid\ \size{P} \leq \tau} 
\end{align*}

Hence, in order to prove \cref{theorem_fs:single_abort:specific_round:r2:1a} (and finish the proof), we prove the following:
\begin{claim}\label{claim:connectionToHypGame} 
 \begin{align*}
    \size{ \eex{\val(V_I) - \val(V_I^-)\ \mid \size{P} \leq \tau, T=1a} } \cdot \pr{T=1a \mid \size{P} \leq \tau} \leq O \left(\biasTermSingle \right)
 \end{align*}
\end{claim}
\remove {	
	\begin{claim}\label{claim:connectionToHypGame:old} 
		Let $\eps \in [-1,1]$, and let $p \in \Supp(P)$. Let $f=\fhyp{\rnd,p}$ be according to \cref{def:hypHint}, and let $\game=\game_{m,\eps,f}$ be a binomial game with hyper-geometric hint according to $\cref{def:game}$. Conditioning on $\Eps=\eps$, and on $P=p$, it holds that 
		$\size{\eex{\val(V_I) - \val(V_I^-)\ \mid\ T=1b}} \cdot \pr{T=1b} \leq \bias(\game)$, where $\bias$ is according to \cref{def:gameBias}. 
	\end{claim}
}
\end{proof}

\begin{remark}[On the setting of  the  $\alpha$-factors and  using the protocol of \cite{HaitnerT17} for the two-party sub-protocol.]\label{rem:whyUsingHT}

	 Following the notations from the proof of \cref{thm:MainFailStop}, let $\Delta$ be the value of $\delta$ calculated in \Stepref{protocol:inner:t:reconstruct:delta} of protocol $\Protocol{\partNumInn}_\rnd$. Let $\alp{k}=\alpha(\rnd, t, k)$ (\ie as in \cref{fact:alphaFacts}). By construction,   $\alpha_{\partNumInn} \cdot \ms{1}$  is the number (independent, possibly biases)  coins used in  \Stepref{addnoise:sampleVec} of the $\AddNoise$ functionality to determined the value of $\Delta$, and   (roughly) $\alpha_{\partNumInn-1} \cdot \ms{1}$  coins are used in by the $\Defense$ functionality at \Stepref{protocol:inner:t:defense:1} of protocol $\Protocol{\partNumInn}_\rnd$. It can be shown that (roughly):
    \begin{enumerate}
        \item Aborting  at \Stepref{protocol:inner:t:defense:2} of $\Protocol{\partNum}_\rnd$, gains  bias $\frac{\sqrt{\alpha_{\partNum-1}}}{\rnd}$. \label{remark:whyUsingHT:1}
        
        \item Aborting  at  \Stepref{protocol:inner:t:defense:1} of protocol $\Protocol{\partNumInn}_\rnd$ for $2 \leq \partNumInn < \partNum$, gains  bias $\frac{\sqrt{\alpha_{\partNumInn-1}}}{\alpha_{\partNumInn}}$. \label{remark:whyUsingHT:2}
    \end{enumerate}
and there are no other attacking opportunities.

   Since protocol $\Protocol{2}_\rnd$ (\ie protocol $\HTProtocol_\rnd$) uses  $\Theta(\ms{1})$ coins, by the above  observation about the $\alpha$'s it holds that  $\alpha_2=1$. Optimizing the choice of $\alpha$'s to minimize the bias they yield according to \cref{remark:whyUsingHT:1} and \cref{remark:whyUsingHT:2},  yields the following equation:
	\begin{align}
       \left(\frac{\sqrt{\alp{2}}}{\alp{3}} =\right) \frac{1}{\alp{3}} = \frac{\sqrt{\alp{3}}}{\alp{4}} = \ldots = \frac{\sqrt{\alp{\partNum-2}}}{\alp{\partNum-1}} = \frac{\sqrt{\alp{\partNum-1}}}{\rnd}  \label{remark:whyUsingHT:eq1}
    \end{align}
    Assume   that instead of using protocol  $\HTProtocol$ in the two players case, we would have used protocol $\Protocol{\partNumInn}_\rnd$ (\cref{protocol:inner:t}) with $\partNumInn=2$.  Now an adversary has an additional attacking  opportunity (at \Stepref{protocol:inner:t:defense:1} of protocol $\Protocol{\partNumInn}_{2}$), which gains bias $\frac{\sqrt{\alpha_1}}{\alpha_2}= \frac{1}{\alpha_2}$.
    
    As a result, when optimizing the parameters of the new protocol, \cref{remark:whyUsingHT:eq1} changes to 
	\begin{align}
    \frac{1}{\alp{2}} = \frac{\sqrt{\alp{2}}}{\alp{3}} = \ldots = \frac{\sqrt{\alp{\partNum-2}}}{\alp{\partNum-1}} = \frac{\sqrt{\alp{\partNum-1}}}{\rnd}  \label{remark:whyUsingHT:eq2}
    \end{align}

    
    Consider for instance  the case of four players (\ie $\partNum=4$). When of using $\HTProtocol$ (as we actually do), \cref{remark:whyUsingHT:eq1} becomes $\frac{1}{\alp{3}} = \frac{\sqrt{\alp{3}}}{\rnd}$,  implying that  $\alp{3}=m^{2/3}$. This yields roughly an overall bias of  $\frac1{\rnd^{2/3}}$. When using \cref{protocol:inner:t} also for the case $r=2$, \cref{remark:whyUsingHT:eq2} becomes 
    $\frac{1}{\alp{2}} = \frac{\sqrt{\alp{2}}}{\alp{3}} = \frac{\sqrt{\alp{3}}}{\rnd}$, implying that  $\alp{2} = \rnd^{4/7}$, yielding  roughly an overall bias of $\frac1{\rnd^{4/7}}$.
\end{remark}

\paragraph{Proving  \cref{protocol:lemma:inner:many:defense1:interface}}

\begin{proof}[Proof of \cref{protocol:lemma:inner:many:defense1:interface}]
Define the two-step process (see \cref{subsec:TwoStepProccess} for an introduction about leakage from two-step boolean processes) $P=(A,B)$, for  $A=\Delta$, and $B=\Berzo{A}$, and define a leakage function $f$ for $P$ by $f(a) = M|_{A=a}$ (\ie the messages received by the corrupted parties at round $s^1$). By definition,  $\size{\deltaDef - \pr{\Delta=1 \mid M=n}} = \PredictAdv_{P,f}(m)$ for every $n \in \Supp(M)$. Hence, it is left to prove that  
\begin{align}
  \ppr{n \la M}{\PredictAdv_{P,f}(n) >  \lambda \cdot \frac{\sqrt{2^{\partNumInn} \cdot \alp{\partNumInn-1}}}{\alp{\partNumInn}} \cdot \sqrt{\log \rnd}} \leq \frac1{\rnd^2}. \label{protocol:lemma:inner:many:defense1:interface:eq1}
\end{align}

Let $P'=(A',B')$ be a $\seq{\ms{1},\alp{\partNumInn},\deltaDef}$-hypergeometric process (see \cref{def:HypProcess:pre}), and let $f'$ be a $(\ms{1} , 2^{\partNumInn} \cdot \alp{\partNumInn-1})$-vector leakage function (see \cref{def:VectorLeakageFunction:pre}) for $P'$. By construction, it holds that $P \equiv P'$ (\ie the two random variables are distributed  the same). We remind the reader that we assume that $M$ contains also the vectors of coins sampled at \Stepref{addnoise:sampleVec} in the $\Noise$ algorithm, and that the messages that the corrupted parties get are a random function of those vectors. Hence, for every $a \in \Supp(A)$, $f(a)$ is a concatenation  of  $f'(a)$ and some random function of $f(a)$. Thus (see \cref{prop:RedundentInfoInTheLeakage}), for proving \cref{protocol:lemma:inner:many:defense1:interface:eq1} it suffices to show that  
\begin{align}
    \ppr{h \la f'(A')}{\PredictAdv_{P',f'}(h) > \lambda \cdot \frac{\sqrt{2^{\partNumInn} \cdot \alp{\partNumInn-1}}}{\alp{\partNumInn}} \cdot \sqrt{\log \rnd}} \leq \frac1{\rnd^2}
\end{align}

We prove the above equation by applying \cref{lemma:HyperlProcessVectorHint:pre} for the hypergeometric process $(A',B')$ with the vector leakage function $f'$, and parameters $\vctBaseLen = \ms{1}$, $\hypVctLenFact = \alp{\partNumInn}$ and $\vctLenFact = 2^{\partNumInn}\cdot \alp{\partNumInn-1}$. Note that the first and third conditions of \cref{lemma:HyperlProcessVectorHint:pre} trivially holds for this choice of parameters, whereas the second condition holds since $\frac{\log^2 \vctBaseLen}{\sqrt{\vctBaseLen}} = o(\sqrt{\frac{\vctLenFact}{\hypVctLenFact}})$ and since $\frac{\vctLenFact}{\vctBaseLen}\cdot \log^2 \vctBaseLen = o(\sqrt{\frac{\vctLenFact}{\hypVctLenFact}})$ for $\partNumInn \leq \partNum = o(\log \rnd)$. Therefore, \cref{lemma:HyperlProcessVectorHint:pre} yields that
\begin{align*}
  \ppr{\hint \la f'(A')}{\PredictAdv_{P',f'}(h) > \const' \sqrt{\log \vctBaseLen} \cdot \frac{\sqrt{\vctLenFact}}{\hypVctLenFact}} \leq \frac1{\vctBaseLen^2},
\end{align*}
for some universal constant $\const' > 0$. We conclude that
\begin{align*}
  \ppr{\hint \la f'(A')}{\PredictAdv_{P',f'}(h) > 2\const'\sqrt{\log \rnd} \cdot \frac{\sqrt{2^{\partNumInn} \cdot \alp{\partNumInn-1}}}{\alp{\partNumInn}}} \leq \frac1{\vctBaseLen^2} \leq \frac1{\rnd^2},
\end{align*}
 and the proof of the claim follows.
\end{proof}

\paragraph{Proving \cref{claim:connectionToVectorGame}.}
\begin{proof}[Proof of \cref{claim:connectionToVectorGame}]
Assume towards a contradiction that: 
\begin{align}
  \size{\eex{\val(V_I) - \val(V_I^-)\ \mid\  T=3b}} \cdot \pr{T=3b} = \omega \left(\biasTermSingleParam{\partNum}{\ell} \right)  \label{proof:connectionToVectorGame:eq1}
\end{align}

Let $\Delta$ be the value of $\delta$ calculated in \Stepref{protocol:inner:t:reconstruct:delta} of $\Protocol{\partNumInn}_\rnd$, and let  $\Eps=\sBias{\ms{1}}{\Delta}$. Note that $\Eps$ is the bias of the coins tossed in the main loop of $\Protocol{\partNumInn}_\rnd$.

Since
\begin{align*}
   &\size{\eex{\val(V_I) - \val(V_I^-)\ \mid\  T=3b}} \cdot \pr{T=3b} = \\
   &\sum_{\eps \in \Supp(\Eps)} \size{\eex{\val(V_I) - \val(V_I^-)\ \mid\ \Eps=\eps , T=3b}} \cdot \pr{\Eps=\eps \mid T=3b} \cdot \pr{T=3b} =\\
   &\sum_{\eps \in \Supp(\Eps)} \size{\eex{\val(V_I) - \val(V_I^-)\ \mid\ \Eps=\eps , T=3b}} \cdot \pr{T=3b \mid \Eps=\eps} \cdot \pr{\Eps=\eps},
\end{align*}
\cref{proof:connectionToVectorGame:eq1} yields that 
\begin{align}
  \size{\eex{\val(V_I) - \val(V_I^-)\ \mid\  \Eps=\eps',T=3b}} \cdot \pr{T=3b \mid \Eps=\eps'} = 
  \omega \left(\biasTermSingleParam{\partNum}{\ell} \right)  \label{proof:connectionToVectorGame:eq2}
\end{align}
for some  $\eps' \in \Supp(\Eps)$. 

Let $\tilde{r}_\Aadv$ and $\tilde{r}_h$ be a fixing of $\Aadv$ and the honest party respectively, that  cause  the protocol to reach the main loop of $\Protocol{\partNumInn}_\rnd$ with $\Eps=\eps'$. Let $\beta = \size{\eex{\val(V_I) - \val(V_I^-)\ \mid\  \Eps=\eps',T=3b}} \cdot \pr{T=3b \mid \Eps=\eps'}$ --- the gain of the adversary $\Aadv$, conditioned on $\Eps=\eps'$.

Let  $f=\fvec{\rnd,\eps',2^\partNumInn \cdot \alp{\partNumInn-1} \cdot \ms{1}}$ be according to \cref{def:vectorHint}, let $\game=\game_{m,\eps,f}$ be a binomial game with vector hint according to \cref{def:game} and let $\bias$ be according to \cref{def:gameBias}. We next show that $\beta \leq \bias(\game)$. Observe that the assumption $\partNumInn < \frac12\loglog m$ implies that $2^\partNumInn \cdot \alp{\partNumInn-1} < \frac{m}{\log^6 m}$. Hence, we can apply \cref{binomial:game:vec} and together with \cref{fact:alphaFacts} it holds that
\begin{align*}
 \beta \leq \bias(\game) \leq O(\frac{\sqrt{2^\partNumInn \cdot \alp{\partNumInn-1}}}{m} \cdot \sqrt{\log \rnd}) = O\left( \frac{\sqrt{2^{\partNumInn}} \cdot \sqrt{\log \rnd}}{\rnd^{\frac12+\frac1{2^{\ell-1}-2}  } } \right) = O\left(\biasTermSingleParam{\partNum}{\ell} \right), 
\end{align*}
contradicting \cref{proof:connectionToVectorGame:eq2}.

To show that $\beta \leq \bias(\game)$, we define a player $\Strategy$ for the game $\game$, that achieves bias $\beta$.

\begin{algorithm}[player $\Strategy$]\label{claim:2partyConnectionToSimpleGame:attacker}
\item[Operation:]~
  \begin{enumerate}
    
  \item Start emulating an execution of protocol $\ProtocolF{\partNum}_\rnd(1^\ell)$, with $\Aadv$ controlling parties $\Party{}_1,\ldots,\Party{}_{\partNum-1}$, where $\Aadv$ uses randomness $\tilde{r}_\Aadv$ , and the honest party $\Party{}_{\partNum}$ uses randomness $\tilde{r}_h$ until the main loop of $\Protocol{\partNumInn}_\rnd$ is reached. (If not reached,  $\Strategy$ never aborts.) From that point, continue the execution randomly using fresh new randomness.

  \item For $i=1$ to $\rnd$:
  
    \begin{enumerate}

     \item Let $(s_{i-1},h_i)$ be the $i$'th message sent by the challenger.

    \item If $i>1$, emulate \Stepref{protocol:inner:t:reconstruct:coin}: let $c_{i-1} = s_{i}-s_{i-1}$ and set  $c_{i-1}^{\isShare{\partNumInn}}$  such that 
    $c_{i-1} = \bigoplus_{\z \in [\partNumInn]}c_{i-1}^{\isShare{z}}$. Emulate the reconstruction of $c_{i-1}$, letting $c_{i-1}^{\isShare{\partNumInn}}$ be the message of the honest party.
        
    \item Emulate \Stepref{protocol:inner:t:coin}:  send the corrupted parties  $2 \cdot (\partNumInn-1)$ random strings $(c_i^{\isShare{1}},\delta_i^{\isShare{1}},\ldots,c_i^{\isShare{\partNumInn}-1},\delta_i^{\isShare{\partNumInn}-1})$ as the  answers of  $\Coin$.
        
    \item Emulate \Stepref{protocol:inner:t:defense:2}: emulate the parallel calls to $\Defense$  using the hint $h_i$. 
    
    Recall that the $\Defense$ functionality is merely a deterministic wrapper for the $\DefenseTilde$ functionality, and the latter, in turn,  is a wrapper to the $\AddNoise$ functionality.   Hence,  it suffices to shows how to use $h_i$ for emulating these calls to  $\AddNoise$.  The  $\AddNoise$ functionality uses $\alpha_{r-1} \cdot \ms{1}$  independent $\Beroo{\eps'}$-biased  coins per call, and there are at most $2^r$ such calls.   Also note that hint $h_i$ is a vector of $2^r \cdot \alpha_{r-1} \cdot \ms{1}$ entries of independent from $\Beroo{\eps'}$. 
    
   Thus to emulate this step, the samples in $h_i$ for these samples needed by $\AddNoise$.

    \item[\quad $\bullet$] If $\Aadv$ aborts at this step, output $1$ (\ie abort at round $i$). Otherwise, output $0$ (\ie continue to next round).
    
    \end{enumerate}
  \end{enumerate}
\end{algorithm}
By construction, $\Aadv$'s view in the above emulation has the same distribution as in the execution of \cref{protocol:outer}, condition on $\Eps=\eps'$. Recall that the bias of $\Strategy$ for a binomial game $\game=\game_{m,\eps',f}$ is defined by $\bias_{\Strategy}(\game) = \eex{ \abs{O_I - O_I^-}}$, where $I$ is the aborting round of $\Strategy$ (m+1 if no abort occurred), $O_i=\delta_i(S_{i-1},H_i)$, and $O_i^-=\delta_i(S_{i-1})$ for $i \in [m]$, and for $i=m+1$ it holds that $O_{m+1}=O_{m+1}^-$. Also recall that $S_{j}$ is the sum of coins tossed up to round $j$, $\delta_i(s_{i-1})$ is the expected outcome of the binomial game given $S_{i-1}=s_{i-1}$, and $\delta_i(s_{i-1},h_i)$ is the expected outcome of the binomial game given $S_{i-1}=s_{i-1}$ and the hint in round $i$ is $h_i$. By above notations, and since $O_{m+1}=O_{m+1}^-$, we can write:
$\bias_{\Strategy}(\game) = \abs{ \eex{ \delta_I(S_{I-1})-\delta_I(S_{I-1},H_I)\ \mid\ I\neq m+1} } \cdot \pr{I \neq m+1}$. 
By construction, $\val(V_I)=\delta_I(S_{I-1})$, $\val(V_I^-)=\delta_I(S_{I-1},H_I)$, and $T=3b$ if and only if $I \neq m+1$. 
It follows that $\bias_{m,\eps',f}(\Strategy) = \size{\eex{\val(V_I) - \val(V_I^-)\ \mid\ \Eps=\eps', T=3b}} \cdot \pr{T=3b \mid \Eps=\eps'} = \beta$. Since $\bias(\game)  = \max_\Strategy \set{ \bias_{\Strategy}(\game)}$, we conclude that $\beta \leq \bias(\game)$.
\end{proof}

\paragraph{Proving \cref{claim:connectionToHypGame}.}
\begin{proof}[Proof of \cref{claim:connectionToHypGame}]
	This proof follows the same line as the proof of \cref{claim:connectionToVectorGame}, so we omit several details. Starting as in the proof of \cref{claim:connectionToVectorGame}, we assume toward contradiction that:
\begin{align}
  \size{ \eex{\val(V_I) - \val(V_I^-)\ \mid T=1a, \size{P} \leq \tau} } \cdot \pr{T=1a \mid \size{P} \leq \tau} = \Omega \left(\biasTermSingle \right) \label{proof:connectionToHypGame:eq1}
\end{align}

Let $\DeltaDef$ be the $\delta$ parameter passed to the last call to $\HTDefenseProtocol$, and let $\Eps = \sBias{\ms{1}}{\DeltaDef}$ (\ie $\Eps$ is the last $\eps$ calculated in  \Stepref{htdefenseprotocol:epsCalculation} of $\HTDefenseProtocol$). Note that the $\HTProtocol$ can be thought as a majority protocol of $\Eps$-biased coins. As in the proof of \cref{claim:connectionToVectorGame}, it is guaranteed that there exists $\eps' \in \Supp(\Eps)$, and $p'\in \Supp(P)$, $-\tau \leq p' \leq \tau$, for which:

\begin{align}
\size{\eex{\val(V_I) - \val(V_I^-)\ \mid\  T=1a, \Eps=\eps',P=p'}} \cdot \pr{T=1a \mid \Eps=\eps',P=p'}&  = \label{proof:connectionToHypGame:eq2}\\
\Omega &\left(\biasTermSingleParam{\partNum}{\ell} \right)  \nonumber
\end{align}

Let $\tilde{r}_\Aadv$, and $\tilde{r}_h$ be a possible randomness' values such that when adversary $\Aadv$ uses $\tilde{r}_\Aadv$, and the honest party uses $\tilde{r}_h$, the protocol reaches $\HTProtocol_\rnd$ with $\Eps=\eps'$, and $P=p'$. Let $\beta = \size{\eex{\val(V_I) - \val(V_I^-)\ \mid\  T=1a, \Eps=\eps',P=p'}} \cdot \pr{T=1a \mid \Eps=\eps',P=p'}$, the gain of the adversary $\Aadv$, conditioned on $\Eps=\eps'$, and $P=p'$.

Let $f=\fhyp{\rnd,p'}$ be according to \cref{def:hypHint}, let $\game=\game_{m,\eps',f}$ be a binomial game with hyper-geometric hint according to $\cref{def:game}$, and let $\bias$ be according to \cref{def:gameBias}. In the following, we show that $\beta \leq \bias(\game)$. Assuming that, by \cref{binomial:game:hyp}, we get that
\begin{align*}
\beta \leq 
\bias(\game) \leq 
O(\frac{\sqrt{\log \rnd}}{\rnd}) = 
O\left(\biasTermSingleParam{\partNum}{\ell} \right)
\end{align*}
Contradicting \cref{proof:connectionToHypGame:eq2}.

As in the proof for \cref{claim:connectionToVectorGame}, to show that $\beta \leq \bias(\game)$, we define a player $\Strategy$ for the game $\game$, that achieves bias $\beta$.

\begin{algorithm}[Player $\Strategy$]
\item[Operation:]~
  \begin{enumerate}
   
  \item Start emulating an execution of protocol $\ProtocolF{\partNum}_\rnd(1^\ell)$, with $\Aadv$ controlling parties $\Party{}_1,\ldots,\Party{}_{\partNum-1}$, where $\Aadv$ uses randomness $\tilde{r}_\Aadv$ , and the honest party $\Party{}_{\partNum}$ uses randomness $\tilde{r}_h$ until the main loop of $\HTProtocol_\rnd$ is reached. (If not reached,  $\Strategy$ never aborts.) From that point, continue the execution randomly using fresh new randomness.
      
  \item For $i=1$ to $\rnd$:
	\begin{enumerate}
       	\item Let $(s_{i-1},h_i)$ be the $i$'th message sent by the challenger.         	
       	\item If $i>1$, emulate \Stepref{protocol:inner:2:reconstructCoin}: let $c_{i-1} = s_{i}-s_{i-1}$ and set  $c_{i-1}^{\isShare{1}}$  such that 
        $c_{i-1} = c_{i-1}^{\isShare{1}} \oplus c_{i-1}^{\isShare{2}}$. Emulate the reconstruction of $c_{i-1}$, letting $c_{i-1}^{\isShare{2}}$ be the message of the honest party.
        	
        \item Emulate \Stepref{protocol:inner:2:defense}: emulate the call to $\HTDefenseRound$ by sending $h_i$ to party $\Party{2}_1$.
        	
        \item[\quad $\bullet$] If $\Aadv$ aborts at this step, output $1$ (\ie abort at round $i$). Otherwise, output $0$ (\ie continue to next round).
    \end{enumerate}

  \end{enumerate}
\end{algorithm}

By construction of strategy $\Strategy$, $\Aadv$'s view in the above emulation has the same distribution as in his view in the execution of \cref{protocol:inner:2}, condition on $\Eps=\eps'$, and on $P=p'$. Using the very same argument that was use at the end of the proof of \cref{claim:connectionToVectorGame} we conclude that  $\bias_{m,\eps',f}(\Strategy) = \beta$. Since $\bias(\game)  = \max_\Strategy \set{ \bias_{\Strategy}(\game)}$, we conclude that $\beta \leq \bias(\game)$.
\end{proof}

\subsection{Proof of Main Theorem}\label{sec:ProvingMainThm}

In this section we prove our main result: the existence of   an $O(m)$-round, $\partNum$-party coin-flipping protocol, in the real (non-hybrid) model, that is $O(\frac{t^4 \cdot 2^t \cdot \sqrt{\log m}}{m^{1/2+1/\left(2^{t-1}-2\right)}})$-fair.

\begin{theorem}[Main theorem --- many-party, fair coin flipping]\label{thm:MainFullFledged}
  Assuming protocols for securely computing \OT exist, then for any polynomially bounded, polynomial-time computable, integer functions $\rnd= \rnd(\kappa)$ and $\partNum = \partNum(\kappa) \leq \frac12\loglog m$, there exists  a $\partNum$-party, $\rnd$-round, $O(\frac{t^4 \cdot 2^t \cdot \sqrt{\log m}}{m^{1/2+1/\left(2^{t-1}-2\right)}})$-fair,  coin-flipping protocol.
\end{theorem}

\begin{proof}[Proof of \cref{thm:MainFullFledged}]
	We compile  our hybrid protocol  defined in \cref{sec:MultiPartyProtocolFinal} into the desired real-world protocol. The main  part of the proof is showing how to modify  the $O(m't')$-round, $t'$-party hybrid protocol $\ProtocolF{t'}_{m'}$   (see \cref{protocol:outer}), for arbitrary  integers $m'$ and $t'$,  into a  form that allows this compilation. This modification involves several steps, all using standard techniques. In the following we fix $m' $ and $t'$, and let $\ProtocolF{}  = \ProtocolF{t'}_{m'}$.

	 First modification is that $\ProtocolF{}$ (through protocol $\Protocol{}$) uses real numbers. Specifically, the parties keep the value of  $\delta$, which is a real number in $[0,1]$, and also keep shares for such  values. We note that the value of  $\delta$ is always set to the probability that when sampling some $k$ $\eps$-biased $\oo$-coins, the bias is at least $b\in \Z$. Where in turn, $\eps$ is the value such that the sum of $n$ $\eps$-biased coins, is positive with probability $\delta'$, for some $\delta'$ whose value is already held by the parties. It follows that $\delta$ has short description given the value of $\delta'$ (\ie the values of $k$ and $b$), and thus all $\delta$ have short descriptions.

	Second modification is  to modify the functionalities used by the protocol as oracles into ones that  are polynomial-time computable in $\rnd'$ and $2^{\partNum'}$, without hurting the security of the protocol. By inspection,  the only calculation that need to be treated is the calculations done in \Stepref{coin:epsCalculation} of $\Coin$, \Stepref{addnoise:epsCalculation} in $\AddNoise$, and \Stepref{htdefenseprotocol:epsCalculation} in $\HTDefenseProtocol$. To be concrete, we focus on the calculation of $\eps = \sBias{\ms{1}}{\delta}$ for some $\delta\in[0,1]$ done in $\Coin$. Via sampling, for any $p\in \poly$, one can efficiently estimate $\eps$ by a value $\widetilde{\eps}$ such that $\size{\eps - \widetilde{\eps}}  \leq \frac{1}{p(\rnd)}$ with save but negligible probability in $m$. Since $\eps$ is merely used for sampling $q(\rnd) \in \poly$  $\eps$-bias $\oo$ coins, it follows that  statistical distance between the parties' views in random execution of $\ProtocolF{t}_\rnd$ and the efficient variant of   $\ProtocolF{t}_\rnd$ that uses the above estimation, is at most $\frac{q(\rnd)}{p(\rnd)} + \negl(\rnd)$ which is in  $O(1/m)$ for large enough $p$. It follows that  \cref{thm:MainFailStop} also holds \wrt the above efficient implementation  of $\Coin$.


	Next modification is to make all the oracle calls made by the parties to be  sequential (\ie one after the other). To do that, we merely replace the parallel calls to $\Defense$ done in \Stepref{protocol:inner:t:defense:1} and \Stepref{protocol:inner:t:defense:2} of \cref{protocol:inner:t}, with a single call per step. This is done by modifying $\Defense$ to get as input the inputs provided by the parties for \emph{all} parallel calls, compute the answer of each of this calls, and return the answers in an aggregated manner to the parties. Since our hybrid model dictates  that a single abort in one of the parallel calls to $\Defense$ aborts all calls, it is clear that this change does not effect the correctness and security of protocol $\ProtocolF{}$.
	
	Last modification it to make the protocol secure against arbitrary adversaries (not only fail-stop ones). Using information-theoretic one-time message authentication codes (\cf \cite{MoranNS09}), the functionalities $\Coin$, $\Defense$ and protocol $\ProtocolF{}$ can be compiled into functionalities and protocol that maintain the same correctness,  essentially the same efficiency, and the resulting protocol is $(\gamma+ \negl(\rnd'))$-fair against \emph{arbitrary} adversaries, assuming  the  protocol  $\gamma$-fair against fail-stop adversaries.

	Then next step is to define an hybrid-model protocol  whose characteristic are functions of the security parameter $\kappa$.  Let $\trnd = \trnd(\secParam) = \ceil{\rnd(\secParam)/c\cdot t(\secParam)^3} -a$, for $c>0$ to be determined by the analysis, and $a \in \set{0,\ldots,11}$ is the value such that $\trnd(\secParam) -a \equiv 1 \bmod 12$.\footnote{Note that the total number of coins,  $\frac{\trnd(\trnd+1)(2\trnd+1)}{6}$, is odd for $\trnd \equiv 1 \bmod 12$.} Consider  the  $O(\partNum\cdot \trnd)$-round, $\partNum$-party, polynomial-time protocol $\ProFH$ in the $(\Coin,\Defense)$-hybrid-model, that on input $\secParam$, the parties act as in  $\ProtocolF{\partNum}_{\trnd}(1^\partNum)$. \cref{thm:MainFailStop} and the above observations yields the  $\ProFH$ is a $\gamma(\kappa)\eqdef\left(O\left(\frac{t \cdot 2^t \cdot \sqrt{\log \trnd}}{\trnd^{1/2+1/\left(2^{t-1}-2\right)}}\right) =  O\left(\frac{t^4 \cdot 2^t \cdot \sqrt{\log \rnd}}{\rnd^{1/2+1/\left(2^{t-1}-2\right)}}\right) \right)$-fair  in the $(\Defense,\Coin)$-hybrid model.

	 Note that  $\ProFH$ makes sequential calls to the oracles, and  that since $\partNum(\kappa) \leq \frac12\loglog m$,   protocol $\ProFH$ runs in polynomial time. 
	 
	 We are finally able to present the real model protocol. Assuming protocols for securely computing \OT exist, there exists (see \cref{fact:fHybridMToReal}) an $O(\partNum^3 \trnd + \partNum\cdot \trnd)$-round, $\partNum$-party, polynomial-time protocol $\ProFR$ correct coin-flipping protocol, that is $(\gamma(\kappa) + \negl(\secParam))$-fair in the \emph{standard model}. By choosing $c$ in the definition of  $\trnd$ large enough, we have that the protocol has (at most) $\rnd$ rounds, yielding that $\ProFR$ is $O\left(\frac{t^4 \cdot 2^t \cdot \sqrt{\log m}}{m^{1/2+1/\left(2^{t-1}-2\right)}}\right)$-fair.
\end{proof}

\section{Leakage from Two-Step Boolean Processes}\label{sec:ExpChangeDueLeakage}
In this section we give bounds on  the  advantage one gains in predicting the outcome of certain types of Boolean  random variables, when some information  has ``leaked''. These bounds play a critical  role in the analysis of the  coin-flipping protocol presented  in \cref{sec:protocol}. Specifically,  they are used to prove \cref{protocol:lemma:inner:many:defense1:interface} that bounds the gain from aborting in the  first round of \cref{protocol:inner:t}, and to prove \cref{binomial:game:vec,binomial:game:hyp} that bounds the bias  of the online binomial games (which, in turn, captures the bias obtained by aborting in the main loop of  \cref{protocol:inner:t,protocol:inner:2}).

The types of random processes and leakage functions considered in this section are  given in \cref{subsec:SingleStepProcessWithLeak}, where the bounds on the  prediction gain for different types of random variables  and  leakage functions are given  in \cref{subsec:BoundsOnExpChange}.

\subsection{Two-step Processes and  Leakage Functions}\label{subsec:SingleStepProcessWithLeak}
Two-step Boolean processes are defined in \cref{subsubsec:SingleStepProcesses} and  the leakage functions we care about are defined in \cref{subsubsec:LeakageFunctions}.\footnote{Some of the definitions given  below  were already given in \cref{sec:protocol}, and they  recalled below  for the reader convenience.}

\subsubsection{Two-step Boolean Process} \label{subsubsec:SingleStepProcesses}
A two-step Boolean process is a pair of jointly-distributed random variables $(\ElementVar,\Value)$, where $\ElementVar$ is over an arbitrary domain $\ElementSet$ and $\Value$ is Boolean (\ie over $\zo$). It is instructive to think that  the process' first step is  choosing $\ElementVar$, and its second step is to choose $\Value$ as a random function of $\ElementVar$. Jumping ahead, the leakage functions we considered  are limited to  be functions of $\ElementVar$ (\ie of the process'  ``state'' after its  first step).
We focus on several types of such Boolean two-step processes.

 
\paragraph{Binomial process.}
Recall that $\Beroo{\eps}$ is the Bernoulli probability distribution over $\oo$ taking the value $1$ with probability $\frac{1}{2}\cdot (1+\eps)$, and that $\Beroo{n, \eps}$ is the probability distribution defined by $\Beroo{n, \eps}(k) = \ppr{(x_1,\ldots,x_n) \la (\Beroo{\eps})^n}{\sum_{i=1}^{n} x_i = k}$. 

\begin{definition}[Binomial process]\label{def:BinomialProcess}
	Let $\rnd \in \N$, $\curRound \in [\rnd]$, $\coinsFuncSym \colon \N \mapsto \N$, $\prevCoins \in \Z$ and $\eps \in [-1,1]$. An {\sf $(\rnd, \curRound, \coinsFuncSym, \prevCoins, \eps)$-binomial process} is the  two-step Boolean process $(\ElementVar,\Value)$ defined by 
	\begin{enumerate}	
		\item $\ElementVar = \CurCoinsVar$.
		
		\item $\Value = \sign(\prevCoins + \ElementVar + \NextNextCoinsVar)$,
	\end{enumerate}
	where $\CurCoinsVarSign_j$, for $j \in \set{\curRound, \ldots, \rnd}$, is an independent  random variable sampled  according to $\Beroo{\coinsFunc{j}, \eps}$. 
\end{definition}
Namely, in the first step $C_i$ is sampled, and the second step returns one if the value of $C_i$  plus a predetermined value $b$ and the sum  $C_{i+1},\ldots,C_m$ is non-negative. With the proper choice of parameters,  the two-step  binomial process captures the random process that happens in the execution  of \cref{protocol:inner:t,protocol:inner:2}.

\paragraph{Hypergeometric process.}
Recall that $\Berzo{\delta}$ is the Bernoulli probability distribution over $\zo$ taking the value $1$ with probability $\delta$ and $0$ otherwise,  that $\Hyp{n,p,\ell}$ is the hyper-geometric probability distribution defined by $\Hyp{n,p,\ell}(k) = \ppr{\cI \subset [n],  \size{\cI} = \ell}{\w(\vct_\cI) = k}$, where $\vct \in \oo^n$ is an arbitrary vector with $w(\vct)= \sum_{i=1}^n v_i = p$, and that  $\vHyp{n,p,\ell}(k) = \ppr{x\la \Hyp{n,p,\ell}}{x \geq  k} = \sum_{t=k}^{\ell}\Hyp{n,p,\ell}(t)$. Finally, recall that $\vBeroo{n, \eps}(k) \eqdef \ppr{x\la \Beroo{n,\eps}}{x \geq  k}$ and that $\sBias{n}{\delta}$ is the value $\eps \in [-1,1]$ with $\vBeroo{n,\eps}(0) = \delta$.

\begin{definition}[Hypergeometric process -- Restatement of \cref{def:HypProcess:pre}]\label{def:HypProcess}
    \HypProcess
\end{definition}
Namely,  $\ElementVar$  is set to the probability that  a random $s$-size subset of this vector contains more ones than zeros, and  $\Value$ is one with probability $\ElementVar$. This two-step process captures the random process that happens in \Stepref{protocol:inner:t:defense:1} of \cref{protocol:inner:t}. 

\subsubsection{Leakage Functions}\label{subsubsec:LeakageFunctions}
A leakage function $\HintFunc$ for a two-step process $(\ElementVar,\Value)$ is simply a randomized function over $\Supp(\ElementVar)$. We will later consider the  advantage in predicting the outcome of $B$ gained from knowing $\HintFunc(\ElementVar)$. That is,  we will measure the difference between $\eex{\Value}$ and  $\eex{\Value\mid \HintFunc(\ElementVar)=\hint}$, for a given ``hint'' (leakage) $\hint\in \Supp(\HintFunc(\ElementVar))$. In the following we define several such leakage   functions. The choice of the second and third leakage   functions considered below might seems somewhat arbitrary, but these are the functions one need to consider when analyzing \cref{protocol:inner:t,protocol:inner:2}.

\paragraph{All-information leakage.}
The all-information leakage function simply leaks the whole state of the process.
\begin{definition}[all-information leakage function]\label{def:AllInfHint}
	A  function $\HintFunc$ is {\sf an all-information leakage function} for a two-step Boolean process $(\ElementVar,\Value)$, if $\HintFunc(\element) = \element$ for every $\element \in \Supp(\ElementVar)$.
\end{definition}

\paragraph{Vector leakage.}
\begin{definition}[vector leakage function -- Restatement of \cref{def:VectorLeakageFunction:pre}]\label{def:VectorLeakageFunction}
   \VectorLeakageFunction
\end{definition}

Namely, the probability that the sum of $s$ bits  taken from the output of $\HintFunc(\element)$ is positive, is exactly $\pr{\Value =1 \mid \ElementVar = \element}$.
\paragraph{Hypergeometric leakage.}
\begin{definition}[hypergeometric leakage function]\label{def:HypHint}
	Let $\rnd \in \N$, $\curRound \in [\rnd]$, $\coinsFuncSym \colon \N \mapsto \N$, $\prevCoins \in \Z$ and $\hypBankWeight \in [-2\cdot \coinsSum{1}, 2\cdot \coinsSum{1}]$, for $\coinsSum{\coinsSumIndex} \eqdef \sum_{j=\coinsSumIndex}^{\rnd} \coinsFunc{j}$. A randomized function $\HintFunc$ is a {\sf $(\rnd,\curRound,\coinsFuncSym,\prevCoins,\hypBankWeight)$-hypergeometric leakage function} for the two-step process $(\ElementVar,\Value)$ with $\Supp(\ElementVar) \subseteq \Z$, if on input $\element\in \Supp(\ElementVar)$, $\HintFunc(\element) =\prevCoins + \element + \hypSample$, for $\hypSample \la \Hyp{2\cdot \coinsSum{1}, \hypBankWeight, \coinsSum{i+1}}$.
\end{definition}
Namely, a hypergeometric leakage function  \emph{masks} the state of the process with  an hypergeometric noise.

\subsubsection{Prediction  Advantage}\label{subsubsec:PredictionAdvantage}
We will be interested in bounding the difference in  the expected outcome of $\Value$  when $\HintVar$ leaks. This change is captured via the notion of prediction advantage.
\begin{definition}[prediction advantage -- Restatement of \cref{def:PredictionAdvantage:pre}]\label{def:PredictionAdvantage}
    \PredictionAdvantage
\end{definition}
The goal of the following section is to bound the prediction advantage $\PredictAdv_{\Pc,\HintFunc}$ in several processes with leakage functions. The bounds given in this section are used for proving the security of \cref{protocol:inner:t,protocol:inner:2}.

\subsection{Bounding Prediction  Advantage}\label{subsec:BoundsOnExpChange}

We give bounds on the prediction advantage in several combinations of two-step Boolean processes and leakage functions. The bounds are stated in \cref{TSP:subse:Bounds}. In \cref{subsubsec:ElementaryBound,subsubsec:ToolsForBinomialProcess,subsec:ToolsForVectorLeakage} we develop  tools for proving such bounds, and the proofs of the stated bounds are given in \cref{subsubsec:BoundingBinomialProcessAllInfHint,subsubsec:BoundingBinomialProcessHypHint,subsubsec:BoundingBinomialProcessVectorHint,subsubsec:HyperOneRoundGameVectorHint}. The choice of parameters we considered below are somewhat arbitrary, but these are the parameters   needed when analyzing  the security of \cref{protocol:inner:t,protocol:inner:2}.    

In the following recall that $\ml{i} = (\rnd - i + 1)^2$ and that $\ms{i} = \sum_{j=i}^{\rnd}\ml{i}$. 

\subsubsection{The Bounds}\label{TSP:subse:Bounds}

\paragraph{Bound on  binomial process with all-Information leakage.}
\begin{lemma}\label{lemma:BinomialProcessAllInfHint}
	Assume $\rnd \in \N$, $\curRound \in [\rnd]$, $\prevCoins \in \Z$ and $\eps \in [-1,1]$,  satisfy
	\begin{enumerate}		
		\item $\size{\eps} \leq 4 \cdot \sqrt{\frac{\log \rnd}{\ms{1}}}$,\label{lemma:BinomialProcessAllInfHint:asmp:epsBound}
		
		\item $\curRound \in [\rnd - \floor{\rnd^{\frac18}}]$, \label{lemma:BinomialProcessAllInfHint:asmp:iBound}
		
		\item $\size{\prevCoins + \eps \cdot \ms{\curRound}} \leq 4\cdot \sqrt{\log \rnd \cdot \ms{\curRound}}$, and\label{lemma:BinomialProcessAllInfHint:asmp:yBound}
						
		\item $-(\prevCoins + 1) \in \Supp(\Beroo{\ms{\curRound}, \eps})$.\label{lemma:BinomialProcessAllInfHint:asmp:minus1Supp}\footnote{We see $b$ as a valid bias of the first $i-1$ rounds of our coin flipping protocol, \ie satisfy the condition that $b+1$ and $\ms{i}$ has the same parity (recall that the first $i-1$ rounds has $\ms{1} - \ms{i}$ coins and that $\ms{1}$ is odd). By assuming that $b$ is not too large (condition \ref{lemma:BinomialProcessAllInfHint:asmp:yBound}), the above is equivalent to condition \ref{lemma:BinomialProcessAllInfHint:asmp:minus1Supp}.}
	\end{enumerate}
	Let $\Pc = (\ElementVar = \CurCoinsVar, \Value)$ be an $\bigl(\rnd, \curRound, \ells_{\rnd}, \prevCoins, \eps\bigr)$-binomial process according to \cref{def:BinomialProcess}, let $\HintFunc$ be an all-information leakage function for $\Pc$ according to \cref{def:AllInfHint}, and let $\PredictAdv_{\Pc,\HintFunc}$ be according to \cref{def:PredictionAdvantage}. Then, there exists a set $\GoodHintsFinal \subseteq \Supp(\HintVarBin)$ such that
	\begin{enumerate}
		\item $\pr{\HintVarBin \notin \GoodHintsFinal} \leq \frac1{\rnd^2}$, and 
		
		\item $\PredictAdv_{\Pc,\HintFunc}(\hint) \leq \const \cdot \sqrt{\ml{\curRound}} \cdot \sqrt{\log \rnd} \cdot \pr{\NextCoinsVar = -(\prevCoins+1)}$,  for every $\hint \in \GoodHintsFinal$ and a universal constant $\const > 0$.
	\end{enumerate}
\end{lemma}

In words, the above lemma (and also the following Lemmas \ref{lemma:BinomialProcessHyperHint} and \ref{lemma:BinomialProcessVectorHint}) states that we can bound the prediction advantage $\PredictAdv_{\Pc,\HintFunc}$ for ``typicall" leakages, using the value of $ \pr{\NextCoinsVar = -(\prevCoins+1)}$. Jumping ahead, such a bound on binomial process, together with \cref{lemma:MainInequalityForLP} which is the main result of \cref{sec:BinomialViaLP}, is used for analyzing our coin-flipping protocol. See the proofs of \cref{Binomial:lemma:VectorHint} and \cref{Binomial:lemma:HyperHint} for more details (which are restatements of \cref{binomial:game:vec} and \cref{binomial:game:hyp}, respectively).

\paragraph{Bound on binomial process with hypergeometric leakage.}

\begin{lemma}\label{lemma:BinomialProcessHyperHint}
	Assume $\rnd \in \N$, $\curRound \in [\rnd]$, $\prevCoins \in \Z$, $\eps \in [-1,1]$, $\const > 0$ and $\hypBankWeight \in [-2\cdot \ms{1}, 2\cdot \ms{1}]$, satisfy
	\begin{enumerate}		
		\item $\size{\hypBankWeight} \leq \const \cdot \sqrt{\log \rnd \cdot \ms{1}}$,\label{lemma:BinomialProcessHyperHint:asmp:pBound}
		
		\item $\size{\eps} \leq 4\cdot \sqrt{\frac{\log \rnd}{\ms{1}}}$, \label{lemma:BinomialProcessHyperHint:asmp:epsBound}
		
		\item $\curRound \in [\rnd - \floor{\rnd^{\frac18}}]$, \label{lemma:BinomialProcessHyperHint:asmp:iBound}
		
		\item $\size{\prevCoins + \eps \cdot \ms{\curRound}} \leq 4\cdot \sqrt{\log \rnd \cdot \ms{\curRound}}$, and\label{lemma:BinomialProcessHyperHint:asmp:yBound}
								
		\item $-(\prevCoins + 1) \in \Supp(\Beroo{\ms{\curRound}, \eps})$.\label{lemma:BinomialProcessHyperHint:asmp:minus1Supp}
	\end{enumerate}
	Let $\Pc = (\ElementVar = \CurCoinsVar, \Value)$ be a $\bigl(\rnd, \curRound, \ells_{\rnd}, \prevCoins, \eps\bigr)$-binomial process according to \cref{def:BinomialProcess}, let $\HintFunc$ be an $\bigl(\rnd, \curRound, \ells_{\rnd}, \prevCoins, \hypBankWeight\bigr)$-hypergeometric leakage function for $\Pc$ according to \cref{def:HypHint}, and let $\PredictAdv_{\Pc,\HintFunc}$ be according to \cref{def:PredictionAdvantage}.
	Then, there exists a set $\GoodHintsFinal \subseteq \Supp(\HintVarBin)$ such that
	\begin{enumerate}
		\item $\pr{\HintVarBin \notin \GoodHintsFinal} \leq \frac1{\rnd^2}$, and 
		
		\item for every $\hint \in \GoodHintsFinal$:
		\begin{enumerate}
			\item $\pr{\size{\CurCoinsVar} > 7\sqrt{\log \rnd \cdot \ml{\curRound}} \mid \HintVarBin = \hint} \leq \frac{\gamma}{\rnd^{12}}$, for a universal constant $\gamma > 0$. 
			
			\item $\PredictAdv_{\Pc,\HintFunc}(\hint) \leq \varphi(\const) \cdot \sqrt{\log \rnd} \cdot \sqrt{\frac{\ml{\curRound}}{\rnd-\curRound+1}} \cdot \pr{\NextCoinsVar = -(\prevCoins+1)}$, for a universal function $\varphi \colon \R^+ \rightarrow \R^+$.
		\end{enumerate}
	\end{enumerate}
\end{lemma}

\paragraph{Bound on binomial process with vector leakage.}

\begin{lemma}\label{lemma:BinomialProcessVectorHint}
	Assume $\vctBaseLen, \vctLenFact \in \N$, $\rnd \in \N$, $\curRound \in [\rnd]$, $\prevCoins \in \Z$ and $\eps \in [-1,1]$  satisfy
	
	\begin{enumerate}		
		\item $\size{\eps} \leq 4\sqrt{\frac{\log \rnd}{\ms{1}}}$,\label{lemma:BinomialProcessVectorHint:asmp:epsBound}
		
		\item $\curRound \in [\rnd - \floor{\rnd^{\frac18}}]$,\label{lemma:BinomialProcessVectorHint:asmp:iBound}
		
		\item $\size{\prevCoins + \eps \cdot \ms{\curRound}} \leq 4\cdot \sqrt{\log \rnd \cdot \ms{\curRound}}$,\label{lemma:BinomialProcessVectorHint:asmp:yBound}
		
		\item $-(\prevCoins + 1) \in \Supp(\Beroo{\ms{\curRound}, \eps})$,\label{lemma:BinomialProcessVectorHint:asmp:minus1Supp}
		
		\item $\vctBaseLen \geq \ms{1}$, and \label{lemma:BinomialProcessVectorHint:asmp:nBound}
		
		\item $\sqrt{\frac{\vctLenFact}{\rnd - \curRound}} \cdot \log \rnd \leq \frac1{100}$,\label{lemma:BinomialProcessVectorHint:asmp:kBound}
	\end{enumerate}
	Let $\Pc = (\ElementVar = \CurCoinsVar, \Value)$ be a $\bigl(\rnd, \curRound, \ells_{\rnd}, \prevCoins, \eps\bigr)$-binomial process according to \cref{def:BinomialProcess}, let $\HintFunc$ be an $\bigl(\vctBaseLen, \vctLenFact\bigr)$-vector leakage function for $\Pc$ according to \cref{def:VectorLeakageFunction}, and let $\PredictAdv_{\Pc,\HintFunc}$ be according to \cref{def:PredictionAdvantage}.
	Then, there exists a set $\GoodHintsFinal \subseteq \Supp(\HintVarBin)$ such that
	\begin{enumerate}
		\item $\pr{\HintVarBin \notin \GoodHintsFinal} \leq \frac1{\rnd^2}$, and
		
		\item for every $\hint \in \GoodHintsFinal$,
		\begin{enumerate}
			\item $\pr{\size{\CurCoinsVar} > 7\sqrt{\log \rnd \cdot \ml{\curRound}} \mid \HintVarBin = \hint} \leq \frac{\gamma}{\rnd^{12}}$, for a universal constant $\gamma > 0$. 
			
			\item $\PredictAdv_{\Pc,\HintFunc}(\hint) \leq \const \cdot \sqrt{\log \rnd \cdot \vctLenFact} \cdot \sqrt{\frac{\ml{\curRound}}{\rnd-\curRound+1}} \cdot \pr{\NextCoinsVar = -(\prevCoins+1)}$,  for a universal constant $\const > 0$.
		\end{enumerate}
	\end{enumerate}
\end{lemma}

\paragraph{Bound on  hypergeometric process with vector leakage.}

\def\HyperProcessLabel{def:HypProcess}
\def\VectorLeakageFunctionLabel{def:VectorLeakageFunction}
\def\PredictionAdvantageLabel{def:PredictionAdvantage}
\begin{lemma}[Restatement of \cref{lemma:HyperlProcessVectorHint:pre}]\label{lemma:HyperlProcessVectorHint}
    \HyperProcessVectorHint
\end{lemma}

\subsubsection{Data Processing on the Leakage}\label{subsubsec:TwoStepBoundBasicTools2} 
The following proposition shows that given access to a (randomize) function of the leakage cannot improve the prediction quality.
\begin{proposition}\label{prop:RedundentInfoInTheLeakage}
	Let $\Pc = (\ElementVar, \Value)$ be a two-step process and let $\HintFunc$ and $\HintFunc'$ be two leakage functions for $P$. Assume  there exists randomize function $g$ over the range of $\HintFunc$ such that $\HintFunc'(a) = \HintFunc(a) \circ g(\HintFunc(a))$ for every $a \in \Supp(\ElementVar)$, where the randomness of $g$ is independent of $f$ and $\Pc$.
	Then, for every $\gamma \in [0,1]$, it holds that
	\begin{align*}
	\ppr{\hint  \la \HintFunc(A)}  {\PredictAdv_{\Pc,\HintFunc}(\hint) > \gamma} =
	\ppr{\hint' \la \HintFunc'(A)} {\PredictAdv_{\Pc,\HintFunc'}(\hint') > \gamma}.
	\end{align*}
\end{proposition}

\begin{proof}
	Let $\gamma \in [0,1]$. Compute
	\begin{align}
		\ppr{\hint'\la \HintFunc'(\ElementVar)}{\PredictAdv_{\Pc,\HintFunc'}(\hint') > \gamma}
		&= 	\ppr{\hint \la \HintFunc(\ElementVar),\hint'' \la \SecHintFunc(\hint)}{\PredictAdv_{\Pc,\HintFunc'}(\hint \circ \hint'') > \gamma}\\
		&= \ppr{\hint \la \HintFunc(\ElementVar),\hint'' \la \SecHintFunc(\hint)}{\size{\pr{\Value = 1} - \pr{\Value = 1 \mid \HintFunc'(\ElementVar) = \hint \circ \hint''}} > \gamma}\nonumber\\
		&= \ppr{\hint \la \HintFunc(\ElementVar),\hint'' \la \SecHintFunc(\hint)}{\size{\pr{\Value = 1} - \pr{\Value = 1 \mid \HintFunc(\ElementVar) = \hint, \SecHintFunc(\hint) = \hint''}} > \gamma}\nonumber\\
		&= \ppr{\hint \la \HintFunc(\ElementVar)}{\size{\pr{\Value = 1} - \pr{\Value = 1 \mid \HintFunc(\ElementVar) = \hint}} > \gamma}\nonumber\\
		&= 	\ppr{\hint \la \HintFunc(\ElementVar)}{\PredictAdv_{\Pc,\HintFunc}(\hint) > \gamma}.\nonumber
	\end{align}
	The penultimate equation holds since $\pr{\Value = 1 \mid \HintFunc(\ElementVar) = \hint, \SecHintFunc(\hint) = \hint''} = \pr{\Value = 1 \mid \HintFunc(\ElementVar) = \hint}$.
\end{proof}

\remove{

\subsubsection{Basic Tools}\label{subsubsec:TwoStepBoundBasicTools222}
We will make use of the following simple observation
\remove{
\begin{proposition}
	Let $\Pc = (\ElementVar, \Value)$ be a two-step process, let $\HintFunc$ be a leakage function for $P$ and let $g$ be a randomized function over the range of $\HintFunc$. Then, for every $\hint \in \Supp(\HintVar)$, it holds that
	\begin{align*}
		\PredictAdv_{\Pc,g \circ \HintFunc}(g(\hint)) \leq \PredictAdv_{\Pc,\HintFunc}(\hint),
	\end{align*}
	where $\PredictAdv$ is according to \cref{def:PredictionAdvantage}. 
\end{proposition}
\begin{proof}
	Let $\hint \in \Supp(\HintVar)$. Compute
	\begin{align*}
		\PredictAdv_{\Pc,g \circ \HintFunc}(g(\hint))
		&= \size{\pr{\Value = 1} - \pr{\Value = 1 \mid g(\HintVar) = g(\hint)}}\\	
		&= \size{\ex{\hint' \la \HintVar}{\pr{\Value = 1} - \pr{\Value = 1 \mid \HintVar = h', g(\hint') = \hint}}}
	\end{align*}
\end{proof}}

\Enote{The above statement is not true. It's true that $\ex{\hint\la \HintFunc(A)}{\PredictAdv_{\Pc,\HintFunc}(\hint)} \geq \ex{\hint\la \CompHintFunc(A)}{\PredictAdv_{\Pc,\CompHintFunc}(\hint)}$, but I don't see how it helps us. I proved \cref{lemma:HyperlProcessVectorHint} without it. \Nnote{You're right Eliad. The old proposition wasn't true, and we actually haven't used it.}\\}
\Nnote{We use this one, when proving round 0 security. We need it since the adversary get more information than simply the banks, but since it's a function of the banks we don't care.}

}

\remove{
\begin{proof}
	Let $\alpha \in [0,1]$, let $I_{\HintFunc, \hint}$ be an indicator random variable of the event $\PredictAdv_{\Pc,\HintFunc}(\hint) > \alpha$. Compute
	\begin{align}
		\ppr{\hint\la \CompHintFunc(\ElementVar)}{\PredictAdv_{\Pc,\CompHintFunc}(\hint) > \alpha}
		&= \ppr{\hint\la \CompHintFunc(\ElementVar)}{\size{\pr{\Value = 1} - \pr{\Value = 1 \mid \CompHintFunc(\ElementVar) = \hint}} > \alpha}\nonumber\\
		&= \ppr{\hint\la \CompHintFunc(\ElementVar)}{\size{\ex{\hint' \la \HintFunc(\ElementVar) \mid \CompHintFunc(\ElementVar) = \hint}{\pr{\Value = 1} - \pr{\Value = 1 \mid \HintFunc(\ElementVar) = \hint'}}} > \alpha}\nonumber\\
		&\leq \ppr{\hint\la \CompHintFunc(\ElementVar)}{\ex{\hint' \la \HintFunc(\ElementVar) \mid \CompHintFunc(\ElementVar) = \hint}{\size{\pr{\Value = 1} - \pr{\Value = 1 \mid \HintFunc(\ElementVar) = \hint'}}} > \alpha}\nonumber\\ 
		&\leq \ppr{\hint\la \CompHintFunc(\ElementVar)}{\ex{\hint' \la \HintFunc(\ElementVar) \mid \CompHintFunc(\ElementVar) = \hint}{\PredictAdv_{\Pc,\HintFunc}(\hint')} > \alpha}\nonumber\\ 
	\end{align}
\end{proof}
}

\subsubsection{Expressing Prediction Advantage using Ratio }\label{subsubsec:ElementaryBound}
In this section we develop a general tool for bounding the prediction advantage $\PredictAdv_{\Pc,\HintFunc}$ of a process $\Pc = (\ElementVar, \Value)$ with leakage function $\HintFunc$. Informally, we reduce the task of bounding the prediction advantage into evaluating the ``ratio" of $\Pc$ with $\HintFunc$, where $\ratioo$ (defined below) is a useful measurement on how much the distribution of $\ElementVar$ changes when $\HintVar$  is given.

\begin{definition}\label{def:ratio}
	Let $\Pc = (\ElementVar, \Value)$ be a two-step process and let $\HintFunc$ be a leakage function for $\Pc$.
	For $\hint \in \Supp(\HintVar)$, $\GoodElements \subseteq \Supp(\ElementVar)$ and $\element \in \GoodElements$, define
	\begin{align*}
		\ratioo_{\hint,\GoodElements}(\element) = \frac{\pr{\ElementVar = \element \mid \HintVar = \hint, \ElementVar \in \GoodElements}}{\pr{\ElementVar = \element \mid \ElementVar \in \GoodElements}}
	\end{align*}
\end{definition}

Namely, $\ratioo_{\hint,\GoodElements}(\element)$ measures the change (in multiplicative term) of the probability that $\ElementVar =\element$, due to the knowledge of  $\hint$, assuming that  $\ElementVar$ is in some ``typical" set (\ie  $\ElementVar\in \GoodElements$).

An alternative and equivalent definition of $\ratioo$ is stated below.
\begin{definition}\label{def:ratio:eq}
	Let $\Pc = (\ElementVar, \Value)$ be a two-step process and let $\HintFunc$ be a leakage function for $\Pc$.
	For $\hint \in \Supp(\HintVar)$, $\GoodElements \subseteq \Supp(\ElementVar)$ and $\element \in \GoodElements$, define
	\begin{align*}
	\ratioo_{\hint,\GoodElements}(\element) = \frac{\pr{\HintVar = \hint \mid \ElementVar = \element}}{\pr{\HintVar = \hint \mid \ElementVar \in \GoodElements}}
	\end{align*}
\end{definition}

As the next claim states, the above two definitions of $\ratioo$ are indeed equivalent.

\begin{claim}\label{claim:ratioEqDefs}
	\cref{def:ratio,def:ratio:eq} are equivalent.
\end{claim}
\begin{proof}
	Let $\ratioo_{\hint,\GoodElements}(\element)$ be according to \cref{def:ratio:eq}.
	A simple calculation yields that
	\begin{align}\label{claim:ratioDifferentVal:1}
		\ratioo_{\hint,\GoodElements}(\element)
		&= \frac{\pr{\HintVar=\hint \mid  \ElementVar = \element}}{\pr{\HintVar=\hint \mid  \ElementVar \in \GoodElements}}\\
		&= \frac{\pr{\ElementVar = \element \mid  \HintVar=\hint}}{\pr{\ElementVar = \element}} \cdot \frac{\pr{\ElementVar \in \GoodElements}}{\pr{\ElementVar \in \GoodElements \mid  \HintVar = \hint}}\nonumber
	\end{align}
	Since $\element \in \GoodElements$, it follows that
	\begin{align}\label{claim:ratioDifferentVal:2}
		\pr{\ElementVar = \element \mid  \ElementVar \in \GoodElements} = \frac{\pr{\ElementVar = \element}}{\pr{\ElementVar \in \GoodElements}}
	\end{align}
	and
	\begin{align}\label{claim:ratioDifferentVal:3}
		\pr{\ElementVar = \element \mid  \ElementVar \in \GoodElements, \HintVar = \hint} = \frac{\pr{\ElementVar = \element \mid  \HintVar = \hint}}{\pr{\ElementVar \in \GoodElements \mid  \HintVar = \hint}}
	\end{align}
	We conclude that
	\begin{align*}
		\ratioo_{\hint,\GoodElements}(\element)
		= \frac{\pr{\ElementVar = \element \mid  \ElementVar \in \GoodElements}}{\pr{\ElementVar = \element \mid  \ElementVar \in \GoodElements, \HintVar = \hint}},
	\end{align*}
	as required.
\end{proof}

The following   lemma  allows us to bound the prediction advantage $\PredictAdv_{\Pc,\HintFunc}$ of a process $\Pc$ with a leakage function $\HintFunc$, using its \emph{ratio} and a ``small" additive term. 

\begin{lemma}\label{lemma:ElementaryBound:GenericDiffBound}
	Let $\Pc = (\ElementVar, \Value)$ be a two-step process and let $\HintFunc$ be a leakage function for $\Pc$.
	Then, for every $\hint \in \Supp(\HintVar)$ and $\GoodElements \subseteq \Supp(\ElementVar)$, it holds that
	\begin{align*}
		\PredictAdv_{\Pc, \HintFunc}(\hint) \leq \ex{\element \la \ElementVar \mid \element \in \GoodElements}{\size{\pr{\Value = 1} -\pr{\Value = 1 \mid \ElementVar = \element}}\cdot \size{1 - \ratioo_{\hint,\GoodElements}(\element)}} + \tail_{\hint,\GoodElements},
	\end{align*}
	for $\tail_{\hint,\GoodElements} = 2\cdot (\pr{\element \notin \GoodElements} + \pr{\element \notin \GoodElements \mid \HintVar = \hint})$.
\end{lemma}
\begin{proof}
	Let $p = \Pr[\ElementVar \in \GoodElements]$, let $q = 1-p$, let $p_\hint = \Pr[\ElementVar \in \GoodElements \mid \HintVar = \hint]$, let $q_\hint = 1 - p_\hint$, let $p'=\Pr[\Value = 1 \mid \ElementVar \notin \GoodElements]$ and let  $p''=\Pr[\Value = 1 \mid \HintVar = \hint,\element \notin \GoodElements]$. Note that
	\begin{align}\label{OneRoundGame:lemma:GenericDiffBound:1}
	\lefteqn{\Pr[\Value = 1]}\\
	&= p \cdot \Pr[\Value = 1 \mid \ElementVar \in \GoodElements] +	q \cdot \Pr[\Value = 1 \mid \ElementVar \notin \GoodElements]\nonumber\\
	&= p \cdot \ex{\element \la \ElementVar \mid \element \in \GoodElements}{\Pr[\Value = 1 \mid \ElementVar = \element]} + q \cdot p'\nonumber\\
	&= p \cdot \ex{\element \la \ElementVar \mid \element \in \GoodElements}{\Pr[\Value = 1 \mid \ElementVar = \element]} + q \cdot p'\nonumber\\
	&= p_\hint \cdot \ex{\element \la \ElementVar \mid \element \in \GoodElements}{\Pr[\Value = 1 \mid \ElementVar = \element]} + (p - p_\hint)\cdot \ex{\element \la \ElementVar \mid \element \in \GoodElements}{\Pr[\Value = 1 \mid \ElementVar = \element]} + q \cdot p'.\nonumber
	\end{align}
	In addition, note that
	\begin{align}\label{OneRoundGame:lemma:GenericDiffBound:2}
	\lefteqn{\Pr[\Value = 1 \mid \HintVar = \hint]}\\
	&= p_\hint \cdot \Pr[\Value = 1 \mid \HintVar = \hint,\element \in \GoodElements]
	+ q_\hint \cdot \Pr[\Value = 1 \mid \HintVar = \hint,\element \notin \GoodElements]\nonumber\\
	&= p_\hint \cdot \frac{\Pr[\Value = 1 \land \HintVar = \hint \mid \element \in \GoodElements]}{\Pr[\HintVar = \hint \mid \element \in \GoodElements]}
	+ q_\hint \cdot p''\nonumber\\
	&= p_\hint \cdot \frac{\ex{\element \la \ElementVar \mid \element \in \GoodElements}{\Pr[\Value = 1 \land \HintVar = \hint \mid \ElementVar = \element]}}{\Pr[\HintVar = \hint \mid \element \in \GoodElements]}
	+ q_\hint \cdot p''\nonumber\\
	&= p_\hint \cdot \frac{\ex{\element \la \ElementVar \mid \element \in \GoodElements}{\Pr[\Value = 1 \mid \ElementVar = \element] \cdot \Pr[\HintVar = \hint \mid \ElementVar = \element]}}{\Pr[\HintVar = \hint \mid \element \in \GoodElements]} + q_\hint \cdot p''\nonumber\\
	&= p_\hint \cdot \ex{\element \la \ElementVar \mid \element \in \GoodElements}{\Pr[\Value = 1 \mid \ElementVar = \element] \cdot \frac{\Pr[\HintVar = \hint \mid \ElementVar = \element]}{\Pr[\HintVar = \hint \mid \element \in \GoodElements]}} + q_\hint \cdot p''\nonumber\\
	&= p_\hint \cdot \ex{\element \la \ElementVar \mid \element \in \GoodElements}{\Pr[\Value = 1 \mid \ElementVar = \element] \cdot \ratioo_{\hint,\GoodElements}(\element)} + q_\hint \cdot p''.\nonumber
	\end{align}
	Combing \cref{OneRoundGame:lemma:GenericDiffBound:1,OneRoundGame:lemma:GenericDiffBound:2} yields that
\begin{align*}	
\PredictAdv_{\Pc,\HintFunc}(\hint)
&= \size{\Pr[\Value = 1] - \Pr[\Value = 1 \mid \HintVar = \hint]}\\
&\leq p_\hint \cdot \size{\ex{\element \la \ElementVar \mid \element \in \GoodElements}{\Pr[\Value = 1 \mid \ElementVar = \element] \cdot (1 - \ratioo_{\hint,\GoodElements}(\element))}} + \abs{p - p_\hint} + q + q_\hint\\
&= p_\hint \cdot \size{\ex{\element \la \ElementVar \mid \element \in \GoodElements}{(\Pr[\Value = 1 \mid \ElementVar = \element] - \Pr[\Value = 1]) \cdot (1 - \ratioo_{\hint,\GoodElements}(\element))}} + \abs{p - p_\hint} + q + q_\hint\\
&\leq \size{\ex{\element \la \ElementVar \mid \element \in \GoodElements}{(\Pr[\Value = 1 \mid \ElementVar = \element] - \Pr[\Value = 1]) \cdot (1 - \ratioo_{\hint,\GoodElements}(\element))}} + \abs{q - q_\hint} + q + q_\hint\\
&\leq \ex{\element \la \ElementVar \mid \element \in \GoodElements}{\size{\Pr[\Value = 1 \mid \ElementVar = \element] - \Pr[\Value = 1]} \cdot \size{1 - \ratioo_{\hint,\GoodElements}(\element)}}
+ 2\cdot(q + q_\hint).
\end{align*}
	 The second equality holds by the following calculation
	\begin{align*}
	\lefteqn{\ex{\element \la \ElementVar \mid \element \in \GoodElements}{\Pr[\Value = 1 \mid \ElementVar = \element] \cdot (1 -\ratioo_{\hint,\GoodElements}(\element))}}\\
	&= \size{\ex{\element \la \ElementVar \mid \element \in \GoodElements}{(\Pr[\Value = 1] + \Pr[\Value = 1 \mid \ElementVar = \element] - \Pr[\Value = 1]) \cdot (1 - \ratioo_{\hint,\GoodElements}(\element))}}\\
	&= \ex{\element \la \ElementVar \mid \element \in \GoodElements}{\Pr[\Value = 1] \cdot (1 - \ratioo_{\hint,\GoodElements}(\element))} +
	\ex{\element \la \ElementVar \mid \element \in \GoodElements}{(\Pr[\Value = 1 \mid \ElementVar = \element] - \Pr[\Value = 1]) \cdot (1 - \ratioo_{\hint,\GoodElements}(\element))}\\
	&= \Pr[\Value = 1] \cdot (1 - \ex{\element \la \ElementVar \mid \element \in \GoodElements}{\ratioo_{\hint,\GoodElements}(\element)}) +
	\ex{\element \la \ElementVar \mid \element \in \GoodElements}{(\Pr[\Value = 1 \mid \ElementVar = \element] - \Pr[\Value = 1]) \cdot (1 - \ratioo_{\hint,\GoodElements}(\element))}\\
	&= \ex{\element \la \ElementVar \mid \element \in \GoodElements}{\Pr[\Value = 1] \cdot (1 - 1)} +
	\ex{\element \la \ElementVar \mid \element \in \GoodElements}{(\Pr[\Value = 1 \mid \ElementVar = \element] - \Pr[\Value = 1]) \cdot (1 - \ratioo_{\hint,\GoodElements}(\element))}\\
	&=  \ex{\element \la \ElementVar \mid \element \in \GoodElements}{(\Pr[\Value = 1 \mid \ElementVar = \element] - \Pr[\Value = 1]) \cdot (1 - \ratioo_{\hint,\GoodElements}(\element))}.
	\end{align*}
\end{proof}

\subsubsection{Bounding Prediction Advantage for  Binomial Processes}\label{subsubsec:ToolsForBinomialProcess}

In this section we develop tools for bounding the prediction advantage $\PredictAdv_{\Pc, \HintFunc}$ of a binomial process $\Pc$ \wrt  an arbitrary leakage $\HintFunc$.  In \cref{subsubsec:BoundingBinomialProcessAllInfHint,subsubsec:BoundingBinomialProcessHypHint,subsubsec:BoundingBinomialProcessVectorHint}, we use these tools to bound the prediction advantage of binomial process \wrt specific leakage functions.  

The following lemma, proven in \cref{subsubsection:BoundingBinomialProcess:ConnectionToLP}, is our first tool for bounding the prediction advantage of a binomial process $(\ElementVar = \CurCoinsVar, \Value)$ with arbitrary leakage. The lemma uses $\ratioo$, defined in \cref{def:ratio:eq}, and states that an appropriate upper-bound on $\size{1-\ratioo}$ yields an upper-bound on the prediction advantage. This tool is used in \cref{subsubsec:BoundingBinomialProcessHypHint,subsubsec:BoundingBinomialProcessVectorHint} for bounding the prediction advantage with hypergeometric and vector leakage, respectively.

\begin{lemma}\label{lemma:BoundingBinomialProcess:ConnectionToLP}
	Let $\rnd \in \N$, $\curRound \in [\rnd]$, $\prevCoins \in \Z$ and $\eps \in [-1,1]$ and assume that $\curRound \in [\rnd - \floor{\rnd^{\frac18}}]$, that $\size{\eps} \leq 4\cdot \sqrt{\frac{\log \rnd}{\ms{1}}}$, that $-(\prevCoins + 1) \in \Supp(\NextCoinsVar)$ and that $\size{\prevCoins + \eps \cdot \ms{\curRound}} \leq 4\sqrt{\log \rnd \cdot \ms{\curRound}}$.
	Let $\Pc = (\CurCoinsVar, \Value)$ be a $(\rnd,\curRound,\ells_{\rnd},\prevCoins,\eps)$-binomial process according to \cref{def:BinomialProcess}, let $\HintFunc$ be a leakage function for $\Pc$, let $\PredictAdv_{\Pc, \HintFunc}$ be according to \cref{def:PredictionAdvantage} and let $\GoodCoins \eqdef \set{\curCoins \in \Supp(\CurCoinsVar) \mid \size{\sigma(\curCoins)} \leq 6\cdot \sqrt{\log \rnd \cdot \ml{\curRound}}}$ for $\sigma(\curCoins) \eqdef \curCoins - \ex{\curCoins' \la \CurCoinsVar}{\curCoins'} = \curCoins - \eps \cdot \ml{\curRound}$.
	Let $\hint \in \Supp(\HintVarBin)$ be such that
	\begin{enumerate}
		\item $\pr{\CurCoinsVar \notin \GoodCoins \mid \HintVarBin = \hint} \leq \frac1{\rnd^{12}}$, and \label{lemma:BoundingBinomialProcess:ConnectionToLP:asmp1}
		
		\item $\size{1 - \ratioo_{\hint}(\curCoins)} \leq \ratioBoundFactor\cdot \frac{\size{\sigma(\curCoins)} + \sqrt{\ml{\curRound}}}{\sqrt{\ms{\curRound}}}$ for every $\curCoins \in \GoodCoins$\label{lemma:BoundingBinomialProcess:ConnectionToLP:asmp2}
	\end{enumerate}
	for $\ratioo_{\hint} = \ratioo_{\hint,\GoodCoins}$ being according to \cref{def:ratio:eq}.
	Then
	\begin{align*}
		\PredictAdv_{\Pc, \HintFunc}(\hint) \leq \const\cdot (\ratioBoundFactor + 1) \cdot \sqrt{\frac{\ml{\curRound}}{\rnd - \curRound + 1}} \cdot \pr{\NextCoinsVar = -(\prevCoins+1)}
	\end{align*}
	for a universal constant $\const > 0$.
\end{lemma}
Namely, in order to bound the prediction advantage, it is enough to bound the value of $\size{1 - \ratioo}$ for the set of  ``typical" coins $\GoodCoins$.

The next lemma, proven in  \cref{subsubsection:ProvingLPLemma:GameValueDiffOneRound}, is our  second tool for bounding the prediction advantage of a binomial process $(\ElementVar = \CurCoinsVar, \Value)$. This tool is used directly in \cref{subsubsec:BoundingBinomialProcessAllInfHint} for bounding the prediction advantage with all-information leakage and is one of the main building blocks for proving \cref{lemma:BoundingBinomialProcess:ConnectionToLP}.
\begin{lemma}\label{ProvingLPLemma:lemma:GameValueDiffOneRound}
	Let $\rnd \in \N$, $\curRound \in [\rnd]$, $\prevCoins \in \Z$ and $\eps \in [-1,1]$ and assume that $\curRound \in [\rnd - \floor{\rnd^{\frac18}}]$, that $\size{\eps} \leq 4\cdot \sqrt{\frac{\log \rnd}{\ms{1}}}$, that $-(\prevCoins + 1) \in \Supp(\NextCoinsVar)$ and that $\size{\prevCoins + \eps \cdot \ms{\curRound}} \leq 4\cdot \sqrt{\log \rnd \cdot \ms{\curRound}}$.
	Let $\Pc = (\CurCoinsVar, \Value)$ be a $(\rnd,\curRound,\ells_{\rnd},\prevCoins,\eps)$-binomial process according to \cref{def:BinomialProcess} and let $\GoodCoins$ and $\sigma$ be as defined in \cref{lemma:BoundingBinomialProcess:ConnectionToLP}. Then, for every $\curCoins \in \GoodCoins$ it holds that
	\begin{align*}
	\size{\pr{\Value = 1} - \pr{\Value = 1 \mid \CurCoinsVar = \curCoins}} \leq \const \cdot \left( \size{\sigma(\curCoins)} + \sqrt{\ml{\curRound}}\right) \cdot \pr{\NextCoinsVar = -(\prevCoins+1)}
	\end{align*}
	for a universal constant $\const > 0$.
\end{lemma}
Namely, the above lemma bounds the expectation change of $\Value$, given a ``typical" value for $\CurCoinsVar$.

\mparagraph{Proving \texorpdfstring{\cref{lemma:BoundingBinomialProcess:ConnectionToLP}}{First Lemma}}\label{subsubsection:BoundingBinomialProcess:ConnectionToLP}

\begin{proof}[Proof of \cref{lemma:BoundingBinomialProcess:ConnectionToLP}]
	In the following we assume \wlg that $\rnd$ is larger than a universal constant determined by the proof  (otherwise, the proof is trivially holds by choosing large enough $\const$).
	Assume $\hint \in \Supp(\HintVarBin)$ satisfies assumptions \ref{lemma:BoundingBinomialProcess:ConnectionToLP:asmp1} and \ref{lemma:BoundingBinomialProcess:ConnectionToLP:asmp2} of \cref{lemma:BoundingBinomialProcess:ConnectionToLP}. 
	Since $\size{\prevCoins + \eps \cdot \ms{\curRound}} \leq 4 \sqrt{\log \rnd \cdot \ms{\curRound}}$,
	\cref{prop:binomProbEstimation} yields that
	\begin{align}\label{ProvingLPLemma:proposition:GameValueDiffOneRound:eq:prMinus1LowerBound}
		\pr{\NextCoinsVar = -(\prevCoins+1)} 
		&\geq \frac{1}{\sqrt{\ms{\curRound}}} \cdot e^{-\frac{(\prevCoins+1+\eps \cdot \ms{\curRound})^2}{2\cdot \ms{\curRound}}}\\
		&\geq \frac1{\rnd^{10}}\nonumber
	\end{align}
	Therefore, by Hoeffding's inequality (\cref{claim:Hoeffding}) and assumption \ref{lemma:BoundingBinomialProcess:ConnectionToLP:asmp1} on $\hint$, it holds that
	\begin{align}\label{ProvingLPLemma:proposition:GameValueDiffOneRound:eq:TailEst}
	\frac{2(\pr{\CurCoinsVar \notin \GoodCoins} + \pr{\CurCoinsVar \notin \GoodCoins \mid \HintVar = \hint})}{\pr{\NextCoinsVar = -(\prevCoins+1)}}
	\leq \frac{\frac4{\rnd^{12}}}{\frac1{\rnd^{10}}}
	= \frac4{\rnd^2}
	\end{align}
	It follows that
	\begin{align*}
	\lefteqn{\frac{\PredictAdv_{\Pc,\HintFunc}(\hint)}{\pr{\NextCoinsVar = -(\prevCoins+1)}}}\\
	&\leq \ex{\curCoins\la \CurCoinsVar\mid \curCoins\in \GoodCoins}{\frac{\size{\pr{\Value = 1} - \pr{\Value = 1 \mid \CurCoinsVar = \curCoins}}}{\pr{\NextCoinsVar = -(\prevCoins+1)}} \cdot \size{1 - \ratioo_{\hint}(\curCoins)}} + \frac{4}{\rnd^2}\\
	&\leq \ex{\curCoins\la \CurCoinsVar\mid \curCoins\in \GoodCoins}{\const' \bigl(\size{\sigma(\curCoins)} + \sqrt{\ml{\curRound}}\bigr) \cdot \ratioBoundFactor\cdot \frac{\size{\sigma(\curCoins)} + \sqrt{\ml{\curRound}}}{\sqrt{\ms{\curRound}}}} + \frac{4}{\rnd^2}\\
	&\leq \const' \cdot \ratioBoundFactor \cdot \ex{\curCoins\la \CurCoinsVar\mid \curCoins\in \GoodCoins}{\frac{\size{\sigma(\curCoins)}^2 + 2\size{\sigma(\curCoins)}\sqrt{\ml{\curRound}} + \ml{\curRound}}{\sqrt{\ms{\curRound}}}} + \frac{4}{\rnd^2}\\
	&\leq 4\const' \cdot \ratioBoundFactor \cdot \frac{\ml{\curRound}}{\sqrt{\ms{\curRound}}} + \frac{4}{\rnd^2}\\
	&\leq 4\const' \cdot \ratioBoundFactor \cdot \sqrt{\frac{\ml{\curRound}}{\rnd - \curRound + 1}} + \frac{4}{\rnd^2}\\
	&\leq 4(\const' \cdot \ratioBoundFactor + 1) \cdot \sqrt{\frac{\ml{\curRound}}{\rnd - \curRound + 1}}
	\end{align*}
	for  $\const'$ being the constant guaranteed in \cref{ProvingLPLemma:lemma:GameValueDiffOneRound}. The first inequality holds by \cref{lemma:ElementaryBound:GenericDiffBound} and by \cref{ProvingLPLemma:proposition:GameValueDiffOneRound:eq:TailEst}, the seconds one  by \cref{ProvingLPLemma:lemma:GameValueDiffOneRound} and assumption \ref{lemma:BoundingBinomialProcess:ConnectionToLP:asmp2} on $\hint$, the fourth one by \cref{prop:binomTailExpectation} and the fifth one holds since $\ms{\curRound} \leq (\rnd - i + 1) \cdot \ml{\curRound}$.
\end{proof}

\mparagraph{Proving \texorpdfstring{\cref{ProvingLPLemma:lemma:GameValueDiffOneRound}}{Second Lemma}}\label{subsubsection:ProvingLPLemma:GameValueDiffOneRound}

\begin{proof}[Proof of \cref{ProvingLPLemma:lemma:GameValueDiffOneRound}]		
	In the following we assume \wlg that $\rnd$ is larger than a universal constant to be determined by the proof (otherwise, the proof is trivially holds by choosing large enough $\const$). Since $\size{\prevCoins + \eps \cdot \ms{\curRound}} \leq 4 \sqrt{\log \rnd \cdot \ms{\curRound}}$ and since $\rnd$ is large,
	\cref{prop:binomProbEstimation} yields that
	\begin{align*}
		\pr{\NextCoinsVar = -(\prevCoins+1)} \geq \frac1{\rnd^{10}}
	\end{align*}
	Therefore, by Hoeffding's inequality (\cref{claim:Hoeffding}), it holds that
	\begin{align}\label{ProvingLPLemma:lemma:GameValueRealDiffOneRound:eq:TailEst}
		\frac{\pr{\CurCoinsVar \notin \GoodCoins}}{\pr{\NextCoinsVar = -(\prevCoins+1)}}
		\leq \frac{\frac1{\rnd^{12}}}{\frac1{\rnd^{10}}}
		= \frac1{\rnd^2}
	\end{align}
	
	\remove{
	 We bound $\size{\pr{\Value = 1} - \pr{\Value = 1 \mid \CurCoinsVar = \curCoins}}$ using an upper-bound on $\size{\pr{\Value = 1 \mid \CurCoinsVar = 0} - \pr{\Value = 1 \mid \CurCoinsVar = \curCoins}}$, for every $\curCoins \in \GoodCoins$. This separation \Inote{which separation? }is done since the first difference, which is quite close to the second one, is easier to bound. Specifically, we use the following claim, proven below.}
	We use the following claim (proven in the next section).
	
	\begin{claim}\label{ProvingLPLemma:claim:GameValueDiffOneRound}
		For every $\curCoins, \curCoins' \in \GoodCoins$, it holds that
		\begin{align*}
			\size{\pr{\Value = 1 \mid \CurCoinsVar = \curCoins} - \pr{\Value = 1 \mid \CurCoinsVar = \curCoins'}} \leq \const' \cdot (\size{\sigma(\curCoins)} + \size{\sigma(\curCoins')}) \cdot \pr{\NextCoinsVar = -(\prevCoins+1)}
		\end{align*}
		for a universal constant $\const' > 0$.
	\end{claim}
	Fix $\curCoins \in \GoodCoins$ and compute
	\begin{align*}
		\lefteqn{\frac{\size{\pr{\Value = 1} - \pr{\Value = 1 \mid \CurCoinsVar = \curCoins}}}{\pr{\NextCoinsVar = -(\prevCoins+1)}}}\\
		&\leq \ex{\curCoins' \la \CurCoinsVar}{\frac{\size{\pr{\Value = 1 \mid \CurCoinsVar = \curCoins'} - \pr{\Value = 1 \mid \CurCoinsVar = \curCoins}}}{\pr{\NextCoinsVar = -(\prevCoins+1)}}}\\
		&\leq \ex{\curCoins' \la \CurCoinsVar \mid \curCoins' \in \GoodCoins}{\frac{\size{\pr{\Value = 1 \mid \CurCoinsVar = \curCoins'} - \pr{\Value = 1 \mid \CurCoinsVar = \curCoins}} + \pr{\CurCoinsVar \notin \GoodCoins}}{\pr{\NextCoinsVar = -(\prevCoins+1)}}}\\
		&\leq \ex{\curCoins' \la \CurCoinsVar \mid \curCoins' \in \GoodCoins}{\const'\cdot (\size{\sigma(\curCoins)}+ \size{\sigma(\curCoins')})} + \frac1{\rnd^2}\\
		&\leq \const' \cdot \sqrt{\ml{\curRound}} + \const'\cdot \size{\sigma(\curCoins)} + \frac1{\rnd^2}\\
		&\leq (\const' + 1) \cdot \bigl(\sqrt{\ml{\curRound}} + \const'\cdot \size{\sigma(\curCoins)}\bigr).
	\end{align*}
	The third inequality holds by \cref{ProvingLPLemma:claim:GameValueDiffOneRound,ProvingLPLemma:lemma:GameValueRealDiffOneRound:eq:TailEst}, and the fourth one  by \cref{prop:binomTailExpectation}.
\end{proof}

\mparagraph{Proving \cref{ProvingLPLemma:claim:GameValueDiffOneRound}}\label{ProvingLPLemma:paragraph:GameValueDiffOneRound}


\begin{proof}[Proof of \cref{ProvingLPLemma:claim:GameValueDiffOneRound}]
	We consider two cases.
	
	\mparagraph{The case $\size{\prevCoins + \eps \cdot \ms{\curRound+1}} \leq \sqrt{\ms{\curRound+1}}$}
	
	In this case, it holds that
	\begin{align}\label{ProvingLPLemma:proposition:GameValueDiffOneRound:case1:eq1}
		\pr{\NextCoinsVar = -(\prevCoins+1)}
		&= \vBeroo{\ms{\curRound+1}, \eps}(-(\prevCoins+1))\\
		&\geq \frac12 \cdot \frac1{\sqrt{\ms{\curRound+1}}} \cdot e^{-\frac{(-\prevCoins-1-\eps \cdot \ms{\curRound+1})^2}{2 \cdot \ms{\curRound + 1}}}\nonumber\\
		&\geq \frac12 \cdot \frac1{\sqrt{\ms{\curRound+1}}} \cdot e^{-1},\nonumber
	\end{align}
	where the first inequality holds by \cref{prop:binomProbEstimation}.
	In addition, it holds that
	\begin{align}\label{ProvingLPLemma:proposition:GameValueDiffOneRound:case1:eq2}
		\size{\pr{\Value = 1 \mid \CurCoinsVar = \curCoins} - \pr{\Value = 1 \mid \CurCoinsVar = \curCoins'}}&= \size{\vBeroo{\ms{\curRound+1}, \eps}\bigl(-(\prevCoins+\curCoins)\bigr) - \vBeroo{\ms{\curRound+1}, \eps}\bigl(-(\prevCoins + \curCoins')\bigr)}\\
		&\leq \frac{\size{\sigma(\curCoins) - \sigma(\curCoins')}}{\sqrt{\ms{\curRound+1}}}\nonumber\\
		&\leq \frac{\size{\sigma(\curCoins)} + \size{\sigma(\curCoins')}}{\sqrt{\ms{\curRound+1}}},\nonumber
	\end{align}
	where the first inequality also holds by \cref{prop:binomProbEstimation}.
	Combining \cref{ProvingLPLemma:proposition:GameValueDiffOneRound:case1:eq1,ProvingLPLemma:proposition:GameValueDiffOneRound:case1:eq2} yields that
	\begin{align}
		\frac{\size{\pr{\Value = 1 \mid \CurCoinsVar = \curCoins} - \pr{\Value = 1 \mid \CurCoinsVar = \curCoins'}}}{\pr{\NextCoinsVar = -(\prevCoins+1)}} \leq 2e \cdot (\size{\sigma(\curCoins)} + \size{\sigma(\curCoins')})
	\end{align}
	\mparagraph{The case $\size{\prevCoins + \eps \cdot \ms{\curRound+1}} > \sqrt{\ms{i+1}}$}
	
	Assume for simplicity that $\prevCoins + \eps \cdot \ms{\curRound+1} > \sqrt{\ms{i+1}}$ (the case $\prevCoins + \eps \cdot \ms{\curRound+1} \leq -\sqrt{\ms{i+1}}$ follows by an analogues arguments). In addition, we assume \wlg that $\size{\curCoins} \geq \size{\curCoins'}$.
	Note that for every $\curCoins''$ with $\size{\curCoins''} \leq \size{\curCoins}$, it holds that
	\begin{align}\label{ProvingLPLemma:proposition:GameValueDiffOneRound:case2:eq1}
		\pr{\NextNextCoinsVar = -(\prevCoins + \curCoins'')}
		&= \vBeroo{\ms{\curRound+1}, \eps}(-(\prevCoins + \curCoins''))\\
		&\leq \frac1{\sqrt{\ms{\curRound+1}}} \cdot e^{-\frac{(-\prevCoins-\curCoins''-\eps \cdot \ms{\curRound+1})^2}{2 \cdot \ms{\curRound + 1}}}\nonumber\\
		&\leq \frac1{\sqrt{\ms{\curRound+1}}} \cdot e^{-\frac{(-\prevCoins+\size{\curCoins}-\eps \cdot \ms{\curRound+1})^2}{2 \cdot \ms{\curRound + 1}}}\nonumber\\
		&\leq 2\cdot \vBeroo{\ms{\curRound+1}, \eps}(-(\prevCoins - \size{\curCoins}))\nonumber\\
		&= 2\cdot \pr{\NextNextCoinsVar = -(\prevCoins - \size{\curCoins})}.\nonumber
	\end{align}
	The first and third inequalities hold by \cref{prop:binomProbEstimation} and the second inequality holds since $\size{\curCoins''} \leq \size{\curCoins} < \sqrt{\ms{\curRound+1}} < \prevCoins + \eps \cdot \ms{\curRound+1}$.
	
	Therefore,
	\begin{align*}
	\lefteqn{\frac{\pr{\NextCoinsVar = -(\prevCoins+1)}}{\size{\pr{\Value = 1 \mid \CurCoinsVar = \curCoins} - \pr{\Value = 1 \mid \CurCoinsVar = \curCoins'}}}}\\
	&\geq \frac{\pr{\NextCoinsVar = -(\prevCoins+1) \mid \CurCoinsVar \in \GoodCoins} \cdot \pr{\CurCoinsVar \in \GoodCoins}}{\size{\pr{\Value = 1 \mid \CurCoinsVar = \curCoins} - \pr{\Value = 1 \mid \CurCoinsVar = \curCoins'}}}\\
	&\geq \frac12 \cdot \ex{\curCoins'' \la \CurCoinsVar \mid \curCoins'' \in \GoodCoins}{\frac{\pr{\NextCoinsVar = -(\prevCoins+1) \mid \CurCoinsVar = \curCoins''}}{\size{\pr{\Value = 1 \mid \CurCoinsVar = \curCoins} - \pr{\Value = 1 \mid \CurCoinsVar = \curCoins'}}}}\\
	&\geq \frac12 \cdot \ex{\curCoins'' \la \CurCoinsVar \mid \curCoins'' \in \GoodCoins}{\frac{\pr{\NextNextCoinsVar = -(\prevCoins+\curCoins''+1)}}{\size{\curCoins-\curCoins'}\cdot \pr{\NextNextCoinsVar = -(\prevCoins-\size{\curCoins})}}}\\
	&\geq \frac1{2\bigl(\size{\sigma(\curCoins)}+\size{\sigma(\curCoins')}\bigr)} \cdot \ex{\curCoins'' \la \CurCoinsVar \mid \curCoins'' \in \GoodCoins}{\frac{\pr{\NextNextCoinsVar = -(\prevCoins+\curCoins''+1)}}{\pr{\NextNextCoinsVar = -(\prevCoins-\size{\curCoins})}}}\\
	&\geq \frac1{4(\size{\sigma(\curCoins)}+\size{\sigma(\curCoins')})} \cdot \ex{\curCoins'' \la \CurCoinsVar \mid \curCoins'' \in \GoodCoins}{\frac{\exp\left(-\frac{(-\prevCoins-\curCoins''-1-\eps\cdot\ms{\curRound+1})^2}{2\cdot\ms{\curRound+1}}\right)}{\exp\left(-\frac{(-\prevCoins+\size{\curCoins}-\eps\cdot\ms{\curRound+1})^2}{2\cdot\ms{\curRound+1}}\right)}}\\
	&= \frac{\ex{\curCoins'' \la \CurCoinsVar \mid \curCoins'' \in \GoodCoins}{\exp\left(-\frac{(\curCoins''+1)^2+2\prevCoins(\curCoins''+1)+2(\curCoins''+1)\cdot\eps\cdot\ms{\curRound+1}-\curCoins^2+2\size{\curCoins}\prevCoins+2\size{\curCoins}\cdot\eps\cdot\ms{\curRound+1}}{2\cdot\ms{\curRound+1}}\right)}}{4(\size{\sigma(\curCoins)}+\size{\sigma(\curCoins')})}\\
	&\geq \frac{1}{4e\cdot (\size{\sigma(\curCoins)}+\size{\sigma(\curCoins')})}.
	\end{align*}
	The third inequality holds by \cref{ProvingLPLemma:proposition:GameValueDiffOneRound:case2:eq1}, the fifth one  by \cref{prop:binomProbEstimation}, and the last one  since the expression  in the exponent is smaller than one (note that $\size{\curCoins}, \size{\curCoins''} \leq 6\cdot \sqrt{\log \rnd \cdot \ml{i+1}} + \eps \cdot \ml{i} \leq 7\cdot \sqrt{\log \rnd \cdot \ml{i+1}}$). We conclude that
	\begin{align*}
	\frac{\size{\pr{\Value = 1 \mid \CurCoinsVar = \curCoins} - \pr{\Value = 1 \mid \CurCoinsVar = \curCoins'}}}{\pr{\NextCoinsVar = -(\prevCoins+1)}}
	&\leq 4e\cdot (\size{\sigma(\curCoins)} + \size{\sigma(\curCoins')}),
	\end{align*}
	as required.
\end{proof}

\subsubsection{A Bound on Binomial Process with All-Information Leakage}\label{subsubsec:BoundingBinomialProcessAllInfHint}
In this section we prove \cref{lemma:BinomialProcessAllInfHint}.
Let $\rnd \in \N$, $\curRound \in [\rnd]$, $\prevCoins \in \Z$ and $\eps \in [-1,1]$ that satisfy the assumptions of \cref{lemma:BinomialProcessAllInfHint}. We assume \wlg that $\rnd$ is larger than some universal constant and we focus on the $\bigl(\rnd, \curRound, \ells_{\rnd}, \prevCoins, \eps\bigr)$-binomial process $\Pc = (\ElementVar = \CurCoinsVar, \Value)$ (according to \cref{def:BinomialProcess}) with all-information leakage function $\HintFunc$.
Let $\GoodCoins \eqdef \set{\curCoins \in \Supp(\CurCoinsVar) \mid \size{\sigma(\curCoins)} \leq 6\cdot \sqrt{\log \rnd \cdot \ml{\curRound}}}$ for $\sigma(\curCoins) \eqdef \curCoins - \ex{\curCoins' \la \CurCoinsVar}{\curCoins'} = \curCoins - \eps \cdot \ml{\curRound}$.
\begin{proof}[Proof of \cref{lemma:BinomialProcessAllInfHint}]
	Let $\GoodHintsFinal = \GoodCoins$.
	By \cref{claim:Hoeffding} (Hoeffding's inequality), it holds that
	\begin{align}\label{OneRoundGame:lemma:BinomialCoinsHintGame:con1}
		\pr{\HintVarBin \notin \GoodHintsFinal} = \pr{\CurCoinsVar \notin \GoodCoins} \leq \frac1{\rnd^2}.
	\end{align}
	Fix $\curCoins \in \GoodCoins$.
	Since $\HintVarBin = \CurCoinsVar$, \cref{ProvingLPLemma:lemma:GameValueDiffOneRound} yields that
	there exists some universal constant $\const' > 0$ such that
	\begin{align}\label{OneRoundGame:lemma:BinomialCoinsHintGame:con2}
		\frac{\size{\pr{\Value = 1} - \pr{\Value = 1 \mid \HintVarBin = \curCoins}}}{\pr{\NextCoinsVar = -(\prevCoins+1)}}
		&\leq \const' \cdot \bigl(\size{\sigma(\curCoins)} + \sqrt{\ml{\curRound}}\bigr)\\
		&\leq 7\const' \cdot \sqrt{\ml{\curRound}} \cdot \sqrt{\log \rnd},\nonumber
	\end{align}
	where the second inequality holds since $\size{\sigma(\curCoins)} \leq 6\cdot \sqrt{\log \rnd \cdot \ml{\curRound}}$.
	The proof follows by \cref{OneRoundGame:lemma:BinomialCoinsHintGame:con1,OneRoundGame:lemma:BinomialCoinsHintGame:con2}.
\end{proof}

\subsubsection{Bound on Binomial Process with Hypergeometric Leakage}\label{subsubsec:BoundingBinomialProcessHypHint}
In this section we prove \cref{lemma:BinomialProcessHyperHint}.
Let $\rnd \in \N$, $\curRound \in [\rnd]$, $\prevCoins \in \Z$, $\eps \in [-1,1]$, $\hypBankWeight \in [-2\cdot \ms{1}, 2\cdot \ms{1}]$ and $\const > 0$ that satisfy the assumptions of \cref{lemma:BinomialProcessHyperHint}.
In the following,  we assume \wlg that $\rnd$ is larger than some universal constant (we can choose $\gamma$ and $\varphi$ to be large enough on small values of $m$), and we focus on the $\bigl(\rnd, \curRound, \ells_{\rnd}, \prevCoins, \eps\bigr)$-binomial process $\Pc = (\ElementVar = \CurCoinsVar, \Value)$ with $\bigl(\rnd, \curRound, \ells_{\rnd}, \prevCoins, \hypBankWeight\bigr)$-hypergeometric leakage function $\HintFunc$.
We let $\GoodCoins \eqdef \set{\curCoins \in \Supp(\CurCoinsVar) \mid \size{\sigma(\curCoins)} \leq 6\cdot \sqrt{\log \rnd \cdot \ml{\curRound}}}$ for $\sigma(\curCoins) \eqdef \curCoins - \ex{\curCoins' \la \CurCoinsVar}{\curCoins'} = \curCoins - \eps \cdot \ml{\curRound}$, and $\GoodHints \eqdef \set{\hint \in \Supp(\HintVarBin) \mid \size{\hint - \prevCoins} \leq (\const+4)\cdot \sqrt{\log \rnd \cdot \ms{\curRound}}}$.

The following proposition, which wraps the main analysis of this section, bounds how much $\ratioo_{\hint}(\element)$ can be far from $1$.

\begin{proposition}\label{prop:BinHypHint:RatioEstimation}
	For every $\hint \in \GoodHints$ and $\curCoins \in \GoodCoins$, it holds that
	\begin{align*}
		\size{1-\ratioo_{\hint}(\curCoins)} \leq \varphi(\const)\cdot \sqrt{\log\rnd} \cdot \frac{\size{\sigma(\curCoins)} + \sqrt{\ml{\curRound}}}{\sqrt{\ms{\curRound}}},
	\end{align*}
	for some universal function $\varphi \colon \R^+ \rightarrow \R^+$.
\end{proposition}
\begin{proof}
	Fix $\hint = \prevCoins + \hypSample \in \GoodHints$ and $\curCoins \in \GoodCoins$. Compute
	
	\begin{align}\label{OneRoundGame:lemma:BinHypHint:oneOverRatio}
	\frac{1}{\ratioo_{\hint}(\curCoins)}
	&= \frac{\pr{\HintVar=\prevCoins + \hypSample \mid \CurCoinsVar \in \GoodCoins}}{\pr{\HintVar=\prevCoins + \hypSample \mid \CurCoinsVar = \curCoins}}\\
	&= \ex{\curCoins' \la \CurCoinsVar \mid \curCoins'\in \GoodCoins}{\frac{\Hyp{2\cdot \ms{1},\hypBankWeight,\ms{i+1}}(\hypSample - \curCoins')}{\Hyp{2\cdot \ms{1},\hypBankWeight,\ms{i+1}}(\hypSample - \curCoins)}}\nonumber\\
	&\in \ex{\curCoins' \la \CurCoinsVar \mid \curCoins'\in \GoodCoins}{e^{\frac{(\hypSample - \curCoins - \frac{p \cdot \ms{\curRound+1}}{2\cdot \ms{1}})^2 - (\hypSample - \curCoins' - \frac{\hypBankWeight \cdot \ms{\curRound+1}}{2\cdot \ms{1}})^2}{2\cdot \ms{\curRound+1}\cdot (1 - \frac{\ms{\curRound + 1}}{2\cdot \ms{1}})}}} \cdot \left(1 \pm 4\varphi'(\const)\cdot \frac{\log^{1.5} \rnd}{\sqrt{\ms{\curRound+1}}}\right)\nonumber\\
	&= \ex{\curCoins' \la \CurCoinsVar \mid \curCoins'\in \GoodCoins}{e^{\frac{2(\curCoins' - \curCoins)\cdot (\hypSample - \frac{\hypBankWeight \cdot \ms{\curRound+1}}{2\cdot \ms{1}}) + \curCoins^2 - {\curCoins'}^2}{2\cdot \ms{\curRound+1}\cdot (1 - \frac{\ms{\curRound + 1}}{2\cdot \ms{1}})}}} \cdot \left(1 \pm 4\varphi'(\const)\cdot \frac{\log^{1.5} \rnd}{\sqrt{\ms{\curRound+1}}}\right)\nonumber\\
	&= \ex{\curCoins' \la \CurCoinsVar \mid \curCoins'\in \GoodCoins}{e^{\frac{2\bigl(\sigma(\curCoins') - \sigma(\curCoins)\bigr)\cdot (\hypSample - \frac{\hypBankWeight \cdot \ms{\curRound+1}}{2\cdot \ms{1}}) + \sigma(\curCoins)^2 - \sigma(\curCoins')^2 + 2\cdot\bigl(\sigma(\curCoins) - \sigma(\curCoins')\bigr)\cdot \eps \cdot \ml{\curRound}}{2\cdot \ms{\curRound+1}\cdot (1 - \frac{\ms{\curRound + 1}}{2\cdot \ms{1}})}}} \cdot \left(1 \pm 4\varphi'(\const)\cdot \frac{\log^{1.5} \rnd}{\sqrt{\ms{\curRound+1}}}\right)\nonumber
	\end{align}
	where the third transition follows by \cref{prop:hyperProbTightEstimation} where $\varphi'$ is the function from it.
	Since $\size{\sigma(\curCoins)},\size{\sigma(\curCoins')} \leq 6\sqrt{\log \rnd \cdot \ml{\curRound}}$, $\size{\hypSample} \leq (\const + 4)\sqrt{\log \rnd \cdot \ms{\curRound}}$, $\size{\hypBankWeight} \leq \const\sqrt{\log \rnd \cdot \ms{1}}$ and $i \in [\rnd - \floor{\rnd^{\frac18}}]$, it holds that $\size{\frac{2\bigl(\sigma(\curCoins') - \sigma(\curCoins)\bigr)\cdot (\hypSample - \frac{\hypBankWeight \cdot \ms{\curRound+1}}{2\cdot \ms{1}}) + \sigma(\curCoins)^2 - \sigma(\curCoins')^2 + 2\cdot\bigl(\sigma(\curCoins) - \sigma(\curCoins')\bigr)\cdot \eps \cdot \ml{\curRound}}{2\cdot \ms{\curRound+1}\cdot (1 - \frac{\ms{\curRound + 1}}{2\cdot \ms{1}})}} < 1$. Therefore,
	\begin{align}\label{OneRoundGame:lemma:BinHypHint:expEst}
	\lefteqn{\ex{\curCoins' \la \CurCoinsVar \mid \curCoins'\in \GoodCoins}{e^{\frac{2\bigl(\sigma(\curCoins') - \sigma(\curCoins)\bigr)\cdot (\hypSample - \frac{\hypBankWeight \cdot \ms{\curRound+1}}{2\cdot \ms{1}}) + \sigma(\curCoins)^2 - \sigma(\curCoins')^2 + 2\cdot\bigl(\sigma(\curCoins) - \sigma(\curCoins')\bigr)\cdot \eps \cdot \ml{\curRound}}{2\cdot \ms{\curRound+1}\cdot (1 - \frac{\ms{\curRound + 1}}{2\cdot \ms{1}})}}}}\\
	&\in \left(1 \pm \ex{\curCoins' \la \CurCoinsVar \mid \curCoins'\in \GoodCoins}{\frac{4\cdot \size{\sigma(\curCoins') - \sigma(\curCoins)}\cdot \size{\hypSample - \frac{\hypBankWeight \cdot \ms{\curRound+1}}{2\cdot \ms{1}}} + 2\size{\sigma(\curCoins)^2 - \sigma(\curCoins')^2} +4\cdot\size{\sigma(\curCoins) - \sigma(\curCoins')}\cdot \eps \cdot \ml{\curRound}}{\ms{\curRound+1}}}\right)\nonumber\\
	&\in \left(1 \pm \frac{4 \cdot \bigl(\size{\sigma(\curCoins)} + \sqrt{\ml{\curRound}}\bigr) \cdot (3\const + 4) \sqrt{\log \rnd \cdot \ms{\curRound+1}} + 2\cdot \sigma(\curCoins)^2 + 2\cdot \ml{\curRound} + \bigl(\size{\sigma(\curCoins)} + \sqrt{\ml{\curRound}}\bigr) \cdot 1}{\ms{\curRound+1}}\right)\nonumber\\
	&\in \left(1 \pm (12\const + 17) \cdot \sqrt{\log \rnd} \cdot \frac{\size{\sigma(\curCoins)} + \sqrt{\ml{\curRound}}}{\sqrt{\ms{\curRound+1}}}\right)\nonumber
	\end{align}
	where the first transition holds since $e^{a} \in 1 \pm 2\size{a}$ for every $a \in [-1,1]$ and the second one holds by \cref{prop:binomTailExpectation} and by the bound on $\size{p}$, $\size{\hypSample}$ and $\size{\eps}$.
	Combining \cref{OneRoundGame:lemma:BinHypHint:oneOverRatio,OneRoundGame:lemma:BinHypHint:expEst} yields that
	\begin{align*}
	\frac{1}{\ratioo_{\hint}(\curCoins)}
	&\in \left(1 \pm (12\const + 17) \cdot \sqrt{\log \rnd} \cdot \frac{\size{\sigma(\curCoins)} + \sqrt{\ml{\curRound}}}{\sqrt{\ms{\curRound+1}}}\right) \cdot \left(1 \pm 4\varphi'(\const)\cdot \frac{\log^{1.5} \rnd}{\sqrt{\ms{\curRound+1}}}\right)\\
	&\in \left(1 \pm (12\const + 18) \cdot \sqrt{\log \rnd} \cdot \frac{\size{\sigma(\curCoins)} + \sqrt{\ml{\curRound}}}{\sqrt{\ms{\curRound+1}}}\right)
	\end{align*}
	Since $(12\const + 18) \cdot \sqrt{\log \rnd} \cdot \frac{\size{\sigma(\curCoins)} + \sqrt{\ml{\curRound}}}{\sqrt{\ms{\curRound+1}}} < 0.5$ and since $\frac{1}{1 \pm a} \in 1 \pm 2a$ for every $a \in (-0.5,0.5)$, we conclude that
	\begin{align*}
	\size{1 - \ratioo_{\hint}(\curCoins)} \leq (24\const + 36) \cdot \sqrt{\log \rnd} \cdot \frac{\size{\sigma(\curCoins)} + \sqrt{\ml{\curRound}}}{\sqrt{\ms{\curRound+1}}},
	\end{align*}
	as required.
\end{proof}

The following proposition combines the analysis done in \cref{prop:BinHypHint:RatioEstimation} with the main tool of \cref{subsubsec:ToolsForBinomialProcess} in order to bound the expectation change of $\Value$.

\begin{proposition}\label{prop:BinHypHint:BiasEstimation}
	For every $\hint \in \GoodHints$ such that $\pr{\CurCoinsVar \notin \GoodCoins \mid \HintVarBin = \hint} \leq \frac1{\rnd^{12}}$, it holds that
	\begin{align*}
		\frac{\size{\pr{\Value = 1} - \pr{\Value = 1 \mid \HintVarBin = \hint}}}{\pr{\NextCoinsVar = -(\prevCoins + 1)}} \leq \varphi(\const)\sqrt{\log \rnd} \cdot \sqrt{\frac{\ml{\curRound}}{\rnd - \curRound + 1}},
	\end{align*}
	for some universal function $\varphi \colon \R^+ \rightarrow \R^+$.
\end{proposition}
\begin{proof}
	The proof immediately follows by \cref{prop:BinHypHint:RatioEstimation,lemma:BoundingBinomialProcess:ConnectionToLP}.
\end{proof}

We are finally ready for proving \cref{lemma:BinomialProcessHyperHint}.

\mparagraph{Proof of \texorpdfstring{\cref{lemma:BinomialProcessHyperHint}}{Main Lemma}}\label{subsubsec:ProvingBinomialHypHintGameLemma}
\begin{proof}
	Let $\GoodHintsFinal \eqdef \set{\hint \in \GoodHints \mid \pr{\CurCoinsVar \notin \GoodCoins \mid \HintVarBin = \hint} \leq \frac1{\rnd^{12}}}$.
	Assume by contradiction that $\ppr{\hint \la \HintVarBin}{\pr{\CurCoinsVar \notin \GoodCoins \mid \HintVarBin = \hint} > \frac1{\rnd^{12}}} > \frac1{2\rnd^2}$. Then
	\begin{align*}
		\pr{\CurCoinsVar \notin \GoodCoins}\\
		&\geq \ppr{\hint \la \HintVarBin}{\CurCoinsVar \notin \GoodCoins \mid \pr{\CurCoinsVar \notin \GoodCoins \mid \HintVarBin = \hint} > \frac1{\rnd^{12}}} \cdot \ppr{\hint \la \HintVarBin}{\pr{\CurCoinsVar \notin \GoodCoins \mid \HintVarBin = \hint} > \frac1{\rnd^{12}}}\\
		&\geq \frac1{\rnd^{12}} \cdot \frac1{2\rnd^2}\\
		&= \frac1{2\rnd^{14}},
	\end{align*}
	In contradiction to Hoeffding's inequality (\cref{claim:Hoeffding}).
	Hence,
	\begin{align}\label{OneRoundGame:lemma:BinomialHypHintGame:eq1}
		\ppr{\hint \la \HintVarBin}{\pr{\CurCoinsVar \notin \GoodCoins \mid \HintVarBin = \hint} > \frac1{\rnd^{12}}} \leq \frac1{2\rnd^2},
	\end{align}
	In addition, it holds that
	\begin{align}\label{OneRoundGame:lemma:BinomialHypHintGame:eq2}
		\pr{\HintVarBin \notin \GoodHints}
		&= \ppr{\hint \la \HintVarBin}{\size{\hint-\prevCoins} > (\const+4)\sqrt{\log \rnd \cdot \ms{\curRound}}}\\
		&= \ppr{\hypSample \la \Hyp{2\ms{1},\hypBankWeight,\ms{\curRound+1}}}{\size{\CurCoinsVar+\hypSample} > (\const+4)\sqrt{\log \rnd \cdot \ms{\curRound}}}\nonumber\\
		&\leq \ppr{\hypSample \la \Hyp{2\ms{1},\hypBankWeight,\ms{\curRound+1}}}{\size{\CurCoinsVar+\hypSample} > (\const+4)\sqrt{\log \rnd \cdot \ms{\curRound}} \mid \CurCoinsVar\in \GoodCoins} + \pr{\CurCoinsVar \notin \GoodCoins}\nonumber\\
		&\leq \ppr{\hypSample \la \Hyp{2\ms{1},\hypBankWeight,\ms{\curRound+1}}}{\size{\hypSample} > (\const+3)\sqrt{\log \rnd \cdot \ms{\curRound}}} + \frac1{4\rnd^{2}}\nonumber\\
		&\leq \ppr{\hypSample \la \Hyp{2\ms{1},\hypBankWeight,\ms{\curRound+1}}}{\size{\hypSample - \frac{\ms{\curRound+1}\cdot \hypBankWeight}{\ms{1}}} > 3\sqrt{\log \rnd \cdot \ms{\curRound}}} + \frac1{4\rnd^{2}}\nonumber\\
		&\leq \frac1{2\rnd^2}.\nonumber
	\end{align}	
	The second inequality holds since $\size{\CurCoinsVar} \leq 7\sqrt{\log \rnd \cdot \ml{\curRound}} \leq \sqrt{\log \rnd \cdot \ms{\curRound+1}}$ and by Hoeffding's inequality (\cref{claim:Hoeffding}), the third one holds since $\size{\frac{\ms{\curRound+1}\cdot \hypBankWeight}{\ms{1}}} \leq \const \sqrt{\log \rnd \cdot \ms{\curRound+1}}$ and the last one holds by \cref{fact:hyperHoeffding} (Hoeffding's inequality for hypergeometric distribution).

	Combining \cref{OneRoundGame:lemma:BinomialHypHintGame:eq1,OneRoundGame:lemma:BinomialHypHintGame:eq2} yields that
	\begin{align}\label{OneRoundGame:lemma:BinomialHyperHintGame:con1}
		\pr{\HintVarBin \notin \GoodHintsFinal} \leq \frac1{\rnd^2}.
	\end{align}

	In addition, note that for every $\hint \in \GoodHintsFinal$ it holds that
	\begin{align}\label{OneRoundGame:lemma:BinomialHyperHintGame:con2}
		\pr{\size{\CurCoinsVar} > 7 \sqrt{\log \rnd \cdot \ml{i}} \mid \HintVarBin = \hint}
		&= \pr{\size{\CurCoinsVar} > 7 \sqrt{\log \rnd \cdot \ml{i}} \mid \HintVarBin = \hint}\\
		&\leq \pr{\size{\sigma(\CurCoinsVar)} > 6 \sqrt{\log \rnd \cdot \ml{i}} \mid \HintVarBin = \hint}\nonumber\\
		&= \pr{\CurCoinsVar \notin \GoodCoins \mid \HintVarBin = \hint}\nonumber\\
		&\leq \frac1{\rnd^{12}},\nonumber
	\end{align}
	where the first inequality holds since $\size{\CurCoinsVar - \sigma(\CurCoinsVar)} = \eps \cdot \ml{\curRound} < \sqrt{\log \rnd \cdot \ml{i}}$ and the last inequality holds by the definition of $\GoodHintsFinal$.
	The rest of proof immediately follows by \cref{OneRoundGame:lemma:BinomialHyperHintGame:con1}, \cref{OneRoundGame:lemma:BinomialHyperHintGame:con2} and by \cref{prop:BinHypHint:BiasEstimation}.
\end{proof}

\subsubsection{Bounding the Ratio for Processes with Vector Leakage}\label{subsec:ToolsForVectorLeakage}
In this section, we prove the following lemma which states a general property about the $\ratioo$ function, defined in \cref{subsubsec:ElementaryBound}, for any process $\Pc = (\ElementVar, \Value)$ with a vector leakage function $\HintFunc$. This property, together with \cref{lemma:ElementaryBound:GenericDiffBound}, will be used for proving \cref{lemma:BinomialProcessVectorHint,lemma:HyperlProcessVectorHint}.

\remove{
The following lemma bounds the expectation change of $\Value$ for any two-step process $(\ElementVar, \Value)$ with vector leakage. \Inote{Explain what is goind on, where used, building on which previously proved bounds, etc}}

\begin{lemma}\label{lemma:VctHint:RatioEstimation}
	Let $\vctBaseLen, \vctLenFact \in \N$, let $(\ElementVar, \Value)$ be a two-step process, let $\HintFunc$ be an $(\vctBaseLen,\vctLenFact)$-vector leakage function for $(\ElementVar, \Value)$ according to \cref{def:vectorHint}, and let $\ratioo$ be according to \cref{def:ratio:eq}.
	Then, for every $\hint \in \Supp(\HintVar)$, $\GoodElements \subseteq \Supp(\ElementVar)$ and $\element \in \GoodElements$, it holds that
	\begin{align*}
		\frac{1}{\ratioo_{\hint,\GoodElements}(\element)} \in \ex{\element' \la \ElementVar \mid \element' \in \GoodElements}{e^{\left(\eps_{\element'} - \eps_{\element}\right)\cdot \left(\w(\hint) - \frac{\eps_{\element'} + \eps_\element}{2}\cdot \vctTotalLen\right)}} \cdot (1 \pm \error),
	\end{align*}
	for $\eps_\element \eqdef \sBias{\vctBaseLen}{\pr{\Value = 1 \mid \ElementVar = \element}}$ and $\error \eqdef \max_{\element',\element'' \in \GoodElements, \z \in \pm\size{\eps^4_{\element'} - \eps^4_{\element''}}}\size{e^{\frac{\eps^3_{\element'} - \eps^3_{\element''}}{3} \cdot \w(\hint) + \z \cdot \vctTotalLen} - 1}$.
\end{lemma}

\begin{proof}
	Note that for every $\element \in \Supp(\ElementVar)$ and $\hint \in \Supp(\HintVar)$, it holds that
	\begin{align}\label{OneRoundGame:lemma:RatioEstimation:probEstimation}
		\pr{\HintVar = \hint \mid \ElementVar = \element}
		&= \pr{f(\element) = \hint}\\
		&= 2^{-\vctTotalLen}\cdot (1 + \eps_\element)^{\frac12(\vctTotalLen + \w(\hint))} \cdot (1 - \eps_\element)^{\frac12(\vctTotalLen - \w(\hint))}\nonumber\\
		&\in 2^{-\vctTotalLen}\cdot e^{\left(\eps_\element - \frac{\eps^2_\element}{2} + \frac{\eps^3_\element}{3} \pm \eps^4_\element\right) \cdot \frac12(\vctTotalLen + \w(\hint))} \cdot e^{\left(-\eps_\element - \frac{\eps^2_\element}{2} - \frac{\eps^3_\element}{3} \pm \eps^4_\element\right) \cdot \frac12(\vctTotalLen - \w(\hint))}\nonumber\\
		&= 2^{-\vctTotalLen}\cdot e^{\eps_\element\cdot \w(\hint) - \frac{\eps^2_\element}{2}\cdot \vctTotalLen} \cdot e^{\frac{\eps^3_\element}{3} \cdot \w(\hint) \pm \eps^4_\element \cdot \vctTotalLen},\nonumber
	\end{align}
	where the third transition holds by the Taylor series $\ln(1+x) = x - \frac{x^2}{2} + \frac{x^3}{3} - \frac{x^4}{4} + \ldots$.
	
	Hence
	\begin{align}\label{OneRoundGame:lemma:RatioEstimation:ratioEstimation}
		\frac{1}{\ratioo_\hint(\element)}
		&= \frac{\pr{\HintVar = \hint \mid \ElementVar \in \GoodElements}}{\pr{\HintVar = \hint \mid \ElementVar = \element}}\\
		&= \ex{\element' \la \ElementVar \mid \element' \in \GoodElements}{\frac{\pr{\HintVar = \hint' \mid \ElementVar = \element'}}{\pr{\HintVar = \hint \mid \ElementVar = \element}}}\nonumber\\
		&\in \ex{\element' \la \ElementVar \mid \element' \in \GoodElements}{e^{(\eps_{\element'} - \eps_\element)\cdot \w(\hint) - \frac{\eps^2_{\element'} - \eps^2_{\element}}{2}\cdot \vctTotalLen} \cdot e^{\frac{\eps^3_{\element'} - \eps^3_\element}{3} \cdot \w(\hint) \pm \size{\eps^4_{\element'} - \eps^4_{\element}} \cdot \vctTotalLen}}\nonumber\\
		&= \ex{\element' \la \ElementVar \mid \element' \in \GoodElements}{e^{\left(\eps_{\element'} - \eps_{\element}\right)\cdot \left(\w(\hint) - \frac{\eps_{\element'} + \eps_\element}{2}\cdot \vctTotalLen\right)} \cdot \left(1 + \left(e^{\frac{\eps^3_{\element'} - \eps^3_\element}{3} \cdot \w(\hint) \pm \size{\eps^4_{\element'} - \eps^4_{\element}} \cdot \vctTotalLen} - 1\right)\right)}\nonumber\\
		&\in \ex{\element' \la \ElementVar \mid \element' \in \GoodElements}{e^{\left(\eps_{\element'} - \eps_{\element}\right)\cdot \left(\w(\hint) - \frac{\eps_{\element'} + \eps_\element}{2}\cdot \vctTotalLen\right)}} \cdot \left(1 \pm \error\right),\nonumber
	\end{align}
	where the third transition holds by \cref{OneRoundGame:lemma:RatioEstimation:probEstimation}.
\end{proof}

\subsubsection{Bound on Binomial Process with Vector Leakage}\label{subsubsec:BoundingBinomialProcessVectorHint}
In this section we prove \cref{lemma:BinomialProcessVectorHint}.
Let $\vctBaseLen, \vctLenFact \in \N$, $\rnd \in \N$, $\curRound \in [\rnd]$, $\prevCoins \in \Z$ and $\eps \in [-1,1]$ that satisfy the assumptions of \cref{lemma:BinomialProcessVectorHint}.
In the following, we assume that $\rnd$ is larger than some universal constant (we can choose $\gamma$ and $\varphi$ to be large enough on small values of $m$), and we focus on the $(\rnd, \curRound, \ells_{\rnd}, \prevCoins, \eps)$-binomial process $(\ElementVar = \CurCoinsVar, \Value)$ with $(\vctBaseLen,\vctLenFact)$-vector leakage function $\HintFunc$.
We let $\GoodCoins \eqdef \set{\curCoins \in \Supp(\CurCoinsVar) \mid \size{\sigma(\curCoins)} \leq 6\cdot \sqrt{\log \rnd \cdot \ml{\curRound}}}$ for $\sigma(\curCoins) \eqdef \curCoins - \ex{\curCoins' \la \CurCoinsVar}{\curCoins'} = \curCoins - \eps \cdot \ml{\curRound}$, and $\GoodHints \eqdef \set{\hint \in \Supp(\HintVarBin) \mid \size{\sigma(\hint)} \leq 4\cdot \sqrt{\log \rnd \cdot \vctTotalLen}}$ for $\sigma(\hint) \eqdef \w(\hint) - \ex{\curCoins \la \CurCoinsVar,\hint' \la \HintFunc(\curCoins)}{\w(\hint')} = \w(\hint) - \ex{\curCoins \la \CurCoinsVar}{\eps_{\curCoins}}\cdot \vctTotalLen$, where $\eps_{\curCoins} = \sBias{\vctBaseLen}{\pr{\Value = 1 \mid \CurCoinsVar = \curCoins}}$.

The following proposition, which wraps the main analysis of this section, bounds how much $\ratioo_{\hint}(\element)$ can be far from one.

\begin{proposition}\label{OneRoundGame:prop:BinVctHint:RatioEstimation}
	For every $\hint \in \GoodHints$ and $\curCoins \in \GoodCoins$, it holds that
	\begin{align*}
	\size{1-\ratioo_{\hint}(\curCoins)} \leq \const \sqrt{\log\rnd \cdot \vctLenFact} \cdot \frac{\size{\sigma(\curCoins)} + \sqrt{\ml{\curRound}}}{\sqrt{\ms{\curRound}}},
	\end{align*}
	for some universal constant $\const > 0$.
\end{proposition}
\begin{proof}
	Fix $\hint \in \GoodHints$ and $\curCoins \in \GoodCoins$.
	By \cref{lemma:VctHint:RatioEstimation}, it holds that
	\begin{align*}
	\frac{1}{\ratioo_\hint(\curCoins)} \in \ex{\curCoins' \la \CurCoinsVar \mid \curCoins' \in \GoodCoins}{e^{\left(\eps_{\curCoins'} - \eps_{\curCoins}\right)\cdot \left(\w(\hint) - \frac{\eps_{\curCoins'} + \eps_\curCoins}{2}\cdot \vctTotalLen\right)}} \cdot (1 \pm \error),
	\end{align*}
	where $\error = \max_{\curCoins',\curCoins'' \in \GoodCoins, \z \in \pm\size{\eps^4_{\curCoins'} - \eps^4_{\curCoins''}}}\size{e^{\frac{\eps^3_{\curCoins'} - \eps^3_{\curCoins''}}{3} \cdot \w(\hint) + \z \cdot \vctTotalLen} - 1}$.
	
	Since $\eps_{\curCoins'} = \sBias{\vctBaseLen}{\pr{\Value = 1 \mid \CurCoinsVar = \curCoins'}}$ for every $\curCoins' \in \GoodCoins$,
	\cref{prop:epsDiff} yields that
	\begin{align}\label{OneRoundGame:lemma:BinVctHint:RatioEstimation:eq1}
	\eps_{\curCoins'} \in \frac{\eps \cdot \ms{\curRound+1} + \prevCoins + \curCoins'}{\sqrt{\vctBaseLen \cdot \ms{\curRound+1}}} \pm \frac{\log^2 \rnd}{2\cdot \sqrt{\vctBaseLen \cdot \ms{\curRound+1}}},
	\end{align}
	for every $\curCoins' \in \GoodCoins$, which yields that
	\begin{align}
	\size{\eps_{\curCoins'}}
	&\leq \frac{\size{\eps} \cdot \ms{\curRound+1} + \size{\prevCoins} + \size{\curCoins'}}{\sqrt{\vctBaseLen \cdot \ms{\curRound+1}}} + \frac{\log^2 \rnd}{2\cdot \sqrt{\vctBaseLen \cdot \ms{\curRound+1}}}\\
	&\leq 10\cdot \sqrt{\frac{\log \rnd}{\vctBaseLen}},\nonumber
	\end{align}
	where the second inequality holds by the bound on $\size{\eps}, \size{\prevCoins}$ (assumptions \ref{lemma:BinomialProcessVectorHint:asmp:epsBound} and \ref{lemma:BinomialProcessVectorHint:asmp:yBound} of \cref{lemma:BinomialProcessVectorHint}) and by the bound on $\size{\curCoins'}$.
	Therefore, for every $\curCoins',\curCoins'' \in \GoodCoins$ and $\z \in \pm \size{\eps^4_{\curCoins'} - \eps^4_{\curCoins''}}$, it holds that
	\begin{align}\label{OneRoundGame:lemma:BinVctHint:RatioEstimation:eq2}
	\lefteqn{\size{\frac{\eps^3_{\curCoins'} - \eps^3_{\curCoins''}}{3} \cdot \w(\hint) + \z \cdot \vctTotalLen}}\\
	&\leq \frac{\size{\eps_{\curCoins'}}^3 + \size{\eps_{\curCoins''}}^3}{3} \cdot \w(\hint) + \size{\z} \cdot \vctTotalLen\nonumber\\
	&\leq \frac{2000}3 \cdot \frac{\log^{1.5} \rnd}{\vctBaseLen^{1.5}} \cdot \left(10\cdot \sqrt{\frac{\log \rnd}{\vctBaseLen}}\cdot \vctTotalLen + 4\sqrt{\log \rnd \cdot \vctTotalLen}\right) + 20000 \cdot \frac{\log^2 \rnd}{\vctBaseLen^2} \cdot \vctTotalLen\nonumber\\
	&\leq 30000 \cdot \log^2 \rnd \cdot \frac{\vctLenFact}{\vctBaseLen}\nonumber\\
	&\leq 1.\nonumber
	\end{align}
	The second inequality holds by the bounds on $\size{\eps_{\curCoins'}}$, $\size{\eps_{\curCoins''}}$ and $\size{\w(\hint)}$,
	and the last one by assumptions \ref{lemma:BinomialProcessVectorHint:asmp:nBound} and \ref{lemma:BinomialProcessVectorHint:asmp:kBound} of \cref{lemma:BinomialProcessVectorHint} and by assuming that $\rnd$ is larger than some universal constant.
	This yields that
	\begin{align}\label{OneRoundGame:lemma:BinVctHint:RatioEstimation:eq3}
	\error
	&= \max_{\curCoins',\curCoins'' \in \GoodCoins, \z \in \pm\size{\eps^4_{\curCoins'} - \eps^4_{\curCoins''}}}\size{e^{\frac{\eps^3_{\curCoins'} - \eps^3_{\curCoins''}}{3} \cdot \w(\hint) + \z \cdot \vctTotalLen} - 1}\\
	&\leq \max_{\curCoins',\curCoins'' \in \GoodCoins, \z \in \pm\size{\eps^4_{\curCoins'} - \eps^4_{\curCoins''}}} 2\cdot\size{\frac{\eps^3_{\curCoins'} - \eps^3_{\curCoins''}}{3} \cdot \w(\hint) + \z \cdot \vctTotalLen}\nonumber\\
	&\leq 60000 \cdot \log^2 \rnd \cdot \frac{\vctLenFact + \sqrt{\vctBaseLen}}{\vctBaseLen},\nonumber
	\end{align}
	where the first inequality holds since $e^{a} \in 1 \pm 2\size{a}$ for every $a \in [-1,1]$
	and the second one holds by \cref{OneRoundGame:lemma:BinVctHint:RatioEstimation:eq2}.
	
	In addition, note that
	\begin{align}\label{OneRoundGame:lemma:BinVctHint:RatioEstimation:eq4}
	\w(\hint)
	&= \ex{\curCoins' \la \CurCoinsVar}{\eps_{\curCoins'}} \cdot \vctTotalLen + \sigma(\hint)\\
	&\in \ex{\curCoins' \la \CurCoinsVar}{\eps_{\curCoins'}} \cdot \vctTotalLen \pm 4\cdot \sqrt{\log \rnd \cdot \vctTotalLen}\nonumber\\
	&= \frac{\eps \cdot \ms{\curRound+1} + \prevCoins + \eps \cdot \ml{\curRound}}{\sqrt{\vctBaseLen \cdot \ms{\curRound+1}}} \cdot \vctTotalLen \pm \frac{\log^2 \rnd}{2\cdot \sqrt{\vctBaseLen \cdot \ms{\curRound+1}}} \pm 4\cdot \sqrt{\log \rnd \cdot \vctTotalLen}\nonumber\\
	&\in \sqrt{\frac{\vctBaseLen}{\ms{\curRound+1}}} \cdot \left(\frac{\eps\cdot\ms{\curRound+1} + \prevCoins + \eps \cdot \ml{\curRound}}{\vctBaseLen}\cdot \vctTotalLen \pm 5\cdot \sqrt{\log \rnd \cdot \ms{\curRound+1}\cdot \vctLenFact}\right),\nonumber
	\end{align}
	where the second equality holds by \cref{OneRoundGame:lemma:BinVctHint:RatioEstimation:eq1}.
	Therefore, for every $\curCoins' \in \GoodCoins$, it holds that
	\begin{align}\label{OneRoundGame:lemma:BinVctHint:RatioEstimation:eq5}
	\lefteqn{\left(\eps_{\curCoins'} - \eps_{\curCoins}\right)\cdot \left(\w(\hint) - \frac{\eps_{\curCoins'} + \eps_\curCoins}{2}\cdot \vctTotalLen\right)}\\
	&\in \frac{\curCoins' - \curCoins \pm \log^2 \rnd}{\sqrt{\vctBaseLen \cdot \ms{\curRound+1}}} \cdot \left( \w(\hint) - \frac{2\eps \cdot \ms{\curRound+1} + 2\prevCoins + \curCoins + \curCoins' \pm \log^2 \rnd}{2\sqrt{\vctBaseLen \cdot \ms{\curRound+1}}} \cdot \vctTotalLen \right)\nonumber\\
	&= \frac{\sigma(\curCoins') - \sigma(\curCoins) \pm \log^2 \rnd}{\ms{\curRound+1}} \cdot \left( \sqrt{\frac{\ms{\curRound+1}}{\vctBaseLen}}\cdot \w(\hint) - \frac{2\eps \cdot \ms{\curRound+1} + 2\prevCoins + \curCoins + \curCoins' \pm \log^2 \rnd}{2\vctBaseLen} \cdot \vctTotalLen \right)\nonumber\\
	&\subseteq \frac{\sigma(\curCoins') - \sigma(\curCoins) \pm \log^2 \rnd}{\ms{\curRound+1}} \cdot \left( \pm 5\cdot \sqrt{\log \rnd \cdot \ms{\curRound+1}\cdot \vctLenFact} - \frac12 \left(\sigma(\curCoins) + \sigma(\curCoins') \pm \log^2 \rnd\right) \cdot \vctLenFact \right)\nonumber\\
	&\subseteq \frac{\sigma(\curCoins') - \sigma(\curCoins) \pm \log^2 \rnd}{\ms{\curRound+1}} \cdot \left( \pm 5\cdot \sqrt{\log \rnd \cdot \ms{\curRound+1}\cdot \vctLenFact} \pm 4\cdot \sqrt{\log \rnd \cdot \ml{\curRound}} \cdot \vctLenFact \right)\nonumber\\
	&\subseteq \frac{\sigma(\curCoins') - \sigma(\curCoins) \pm \log^2 \rnd}{\sqrt{\ms{\curRound+1}}} \cdot \left( \pm 5\cdot \sqrt{\log \rnd \cdot \vctLenFact} \pm 4\cdot \frac{\vctLenFact \cdot \sqrt{\log \rnd} }{\sqrt{\rnd - \curRound}} \right)\nonumber\\
	&\subseteq \frac{\sigma(\curCoins') - \sigma(\curCoins) \pm \log^2 \rnd}{\sqrt{\ms{\curRound+1}}} \cdot \left( \pm 9\cdot \sqrt{\log \rnd \cdot \vctLenFact} \right),\nonumber
	\end{align}
	where the first transition holds by \cref{OneRoundGame:lemma:BinVctHint:RatioEstimation:eq1}, the third one holds by \cref{OneRoundGame:lemma:BinVctHint:RatioEstimation:eq4}, the fourth one holds since $\curCoins,\curCoins' \in \GoodCoins$ and the fifth one holds since $\ms{\curRound + 1} \leq (\rnd - \curRound) \cdot \ml{\curRound}$.
	
	By the bounds on $\size{\sigma(\curCoins)}, \size{\sigma(\curCoins')}, \size{\eps_{\curCoins}}, \size{\eps_{\curCoins'}}$ and by assumption \ref{lemma:BinomialProcessVectorHint:asmp:kBound} of \cref{lemma:BinomialProcessVectorHint}, it holds that\remove{\Enote{verify it and add more explanation}}
	\begin{align*}
	\size{\left(\eps_{\curCoins'} - \eps_{\curCoins}\right)\cdot \left(\w(\hint) - \frac{\eps_{\curCoins'} + \eps_\curCoins}{2}\cdot \vctTotalLen\right)}
	\leq 1,
	\end{align*}
	for every $\curCoins' \in \GoodCoins$.
	Hence,
	\begin{align}\label{OneRoundGame:lemma:BinVctHint:RatioEstimation:eq6}
	\frac{1}{\ratioo_\hint(\curCoins)}
	&\in \ex{\curCoins' \la \CurCoinsVar \mid \curCoins' \in \GoodCoins}{e^{\left(\eps_{\curCoins'} - \eps_{\curCoins}\right)\cdot \left(\w(\hint) - \frac{\eps_{\curCoins'} + \eps_\curCoins}{2}\cdot \vctTotalLen\right)}} \cdot (1 \pm \error)\\
	&\subseteq \ex{\curCoins' \la \CurCoinsVar \mid \curCoins' \in \GoodCoins}{1 \pm 18\cdot \frac{\size{\sigma(\curCoins)} + \size{\sigma(\curCoins')} + \log^2 \rnd}{\sqrt{\ms{\curRound+1}}} \cdot \sqrt{\log \rnd \cdot \vctLenFact}}\cdot (1 \pm \error)\nonumber\\
	&\subseteq \left(1 \pm 18\cdot \frac{\size{\sigma(\curCoins)} + \sqrt{\ml{\curRound}} + \log^2 \rnd}{\sqrt{\ms{\curRound+1}}} \cdot \sqrt{\log \rnd \cdot \vctLenFact}\right)\cdot \left(1 \pm 60000 \cdot \log^2 \rnd \cdot \frac{\vctLenFact }{\vctBaseLen}\right)\nonumber\\
	&\subseteq \left(1 \pm 19\cdot \frac{\size{\sigma(\curCoins)} + \sqrt{\ml{\curRound}}}{\sqrt{\ms{\curRound+1}}} \cdot \sqrt{\log \rnd \cdot \vctLenFact}\right),\nonumber
	\end{align}
	where the second transition holds by \cref{OneRoundGame:lemma:BinVctHint:RatioEstimation:eq5} and since $e^{a} \in 1 \pm 2\size{a}$ for every $a \in [-1,1]$, the third one holds by \cref{prop:binomTailExpectation} and the last one holds by assumptions \ref{lemma:BinomialProcessVectorHint:asmp:iBound}, \ref{lemma:BinomialProcessVectorHint:asmp:nBound}, \ref{lemma:BinomialProcessVectorHint:asmp:kBound} of \cref{lemma:BinomialProcessVectorHint} and by assuming that $\rnd$ is larger than some universal constant, which yields that $\sqrt{\frac{\ml{\curRound}}{\ms{\curRound}}} \geq \sqrt{\frac{1}{\rnd-\curRound}} = \omega(\log^2 \rnd \cdot \frac{\vctLenFact + \sqrt{\vctBaseLen}}{\vctBaseLen})$
	
	By assumption \ref{lemma:BinomialProcessVectorHint:asmp:kBound} of \cref{lemma:BinomialProcessVectorHint} and since $\curCoins \in \GoodCoins$ and $\frac1{1\pm a} \subseteq 1 \pm 2a$ for every $a \in (-0.5,0.5)$, we deduce from \cref{OneRoundGame:lemma:BinVctHint:RatioEstimation:eq6} that
	\begin{align}
	\ratioo_\hint(\curCoins) \in \left(1 \pm 38\sqrt{\log \rnd \cdot \vctLenFact} \cdot \frac{\size{\sigma(\curCoins)} + \sqrt{\ml{\curRound}}}{\sqrt{\ms{\curRound+1}}}\right).
	\end{align}
	Thus
	\begin{align*}
	\size{1-\ratioo_\hint(\curCoins)} \leq 38\sqrt{\log \rnd \cdot \vctLenFact} \cdot \frac{\size{\sigma(\curCoins)} + \sqrt{\ml{\curRound}}}{\sqrt{\ms{\curRound+1}}}
	\end{align*}
\end{proof}

The following proposition combines the analysis done in \cref{OneRoundGame:prop:BinVctHint:RatioEstimation} with the main tool of \cref{subsubsec:ToolsForBinomialProcess} in order to bound the expectation change of $\Value$.

\begin{proposition}\label{OneRoundGame:prop:BinVctHint:BiasEstimation}
	For every $\hint \in \GoodHints$ such that $\pr{\CurCoinsVar \notin \GoodCoins \mid \HintVarBin = \hint} \leq \frac1{\rnd^{12}}$, it holds that
	\begin{align*}
		\frac{\size{\pr{\Value = 1} - \pr{\Value = 1 \mid \HintVarBin = \hint}}}{\pr{\NextCoinsVar = -(\prevCoins + 1)}} \leq \const\sqrt{\log \rnd \cdot \vctLenFact} \cdot \sqrt{\frac{\ml{\curRound}}{\rnd - \curRound + 1}}
	\end{align*}
\end{proposition}
\begin{proof}
	The proof immediately follows by \cref{OneRoundGame:prop:BinVctHint:RatioEstimation,lemma:BoundingBinomialProcess:ConnectionToLP}.
\end{proof}

We are finally ready for proving \cref{lemma:BinomialProcessVectorHint}.

\mparagraph{Proving \texorpdfstring{\cref{lemma:BinomialProcessVectorHint}}{Main Lemma}}\label{subsubsec:ProvingBinomialVectorHintGameLemma}
\begin{proof}
	Let $\GoodHintsFinal \eqdef \set{\hint \in \GoodHints \mid \pr{\CurCoinsVar \notin \GoodCoins \mid \HintVarBin = \hint} \leq \frac1{\rnd^{12}}}$.
	Using similar arguments as in the proof of \cref{lemma:BinomialProcessHyperHint}, it holds that
	\begin{align}\label{OneRoundGame:lemma:BinomialVectorHintGame:eq1}
		\ppr{\hint \la \HintVarBin}{\pr{\CurCoinsVar \notin \GoodCoins \mid \HintVarBin = \hint} > \frac1{\rnd^{12}}} \leq \frac1{2\rnd^2},
	\end{align}
	In addition, Hoeffding's inequality (\cref{claim:Hoeffding}) yields that
	\begin{align}\label{OneRoundGame:lemma:BinomialVectorHintGame:eq2}
		\pr{\HintVarBin \notin \GoodHints}
		\leq \frac1{2\rnd^2}
	\end{align}
	Therefore, we conclude from \cref{OneRoundGame:lemma:BinomialVectorHintGame:eq1,OneRoundGame:lemma:BinomialVectorHintGame:eq2} that
	\begin{align}\label{OneRoundGame:lemma:BinomialVectorHintGame:con1}
		\pr{\HintVarBin \notin \GoodHintsFinal}
		\leq \frac1{\rnd^2}
	\end{align}
	
	In addition, as proven in \cref{lemma:BinomialProcessHyperHint}, for every $\hint \in \GoodHintsFinal$ it holds that
	\begin{align}\label{OneRoundGame:lemma:BinomialVectorHintGame:con2}
		\pr{\size{\CurCoinsVar} > 7 \sqrt{\log \rnd \cdot \rnd} \mid \HintVarBin = \hint} \leq \frac1{\rnd^{12}},
	\end{align}
	
	The  proof now follows by \cref{OneRoundGame:lemma:BinomialVectorHintGame:con1,OneRoundGame:lemma:BinomialVectorHintGame:con2,OneRoundGame:prop:BinVctHint:BiasEstimation}.
\end{proof}

\subsubsection{Bound on Hypergeometric Process with Vector Leakage}\label{subsubsec:HyperOneRoundGameVectorHint}
In this section we prove \cref{lemma:HyperlProcessVectorHint}.
Let $\vctBaseLen, \vctLenFact, \hypVctLenFact \in \N$ and $\delta \in [0,1]$ that satisfy the assumptions of \cref{lemma:HyperlProcessVectorHint},
assume that $\delta \in [\frac1{\vctBaseLen^4}, 1-\frac1{\vctBaseLen^4}]$ and let $\eps \eqdef \sBias{\vctBaseLen}{\delta}$ (note that by \cref{claim:Hoeffding}, $\size{\eps} \leq 4\sqrt{\frac{\log \vctBaseLen}{\vctBaseLen}}$). We  assume \wlg that $\vctBaseLen$ is larger than some universal constant (otherwise, the proof trivially holds by taking large enough $\const$). 

Let $(\ElementVar, \Value)$ be a $\bigl(\vctBaseLen,\hypVctLenFact, \delta\bigr)$-hypergeometric process with $(\vctBaseLen,\vctLenFact)$-vector leakage function $\HintFunc$, as defined in \cref{def:VectorLeakageFunction}.

Let $\GoodVectors = \set{\vct \in \oo^{\hypTotalVctLen} \mid \size{\sigma(\vct)} \leq 4\sqrt{\log \vctBaseLen \cdot \hypTotalVctLen}}$, for $\sigma(\vct) \eqdef \w(\vct) - \ex{\vct' \la (\Beroo{\eps})^{\hypTotalVctLen}}{\w(\vct')} = \w(\vct) - \eps \cdot \hypTotalVctLen$, let $\GoodElements \eqdef \bigcup\limits_{\vct \in \GoodVectors} \set{\vHyp{\hypTotalVctLen, \w(\vct), \vctBaseLen}(0)}$ and let $\GoodHints \eqdef \set{\hint \in \Supp(\HintVar) \mid \size{\sigma(\hint)} \leq 4\cdot \sqrt{\log \vctBaseLen \cdot \vctTotalLen}}$, for $\sigma(\hint) \eqdef \w(\hint) - \ex{\element \la \ElementVar, \hint \la \HintVar \mid \ElementVar = \element}{\w(\hint)} = \w(\hint) - \ex{\element \la \ElementVar}{\eps_{\element}}\cdot \vctTotalLen$ for $\eps_{\element} \eqdef \sBias{\vctBaseLen}{\pr{\Value = 1 \mid \ElementVar = \element}}$.
In addition, for $\element \in \GoodElements$, let $\w(\element)$ be the value $\weight \in \Z$ with $\element = \vHyp{\hypTotalVctLen, \weight, \vctBaseLen}(0)$ and we let $\sigma(\element) = \w(\element)- \eps \cdot \hypTotalVctLen$ (note that by definition, $\size{\sigma(\element)} \leq 4\sqrt{\log \vctBaseLen \cdot \hypTotalVctLen}$ for every $\element \in \GoodElements$).

\remove{\Inote{Explain the main steps towards proving the lemma. The first step,   \cref{HyperOneRoundGameVectorHint:GameValueDiffOneRound}, is ... }}
Proving \cref{lemma:HyperlProcessVectorHint} is done by bounding $\PredictAdv_{\Pc,\HintFunc}(\hint)$ for ``typical'' values of $\hint$. The first step (\cref{HyperOneRoundGameVectorHint:GameValueDiffOneRound}) is to bound $\size{\pr{\Value = 1} - \pr{\Value = 1 \mid \ElementVar = \element}}$ for ``typical'' values of $\hint$. The second step (\cref{OneRoundGame:lemma:HypVctHint:RatioEstimation}) is to bound $\size{1-\ratioo_{\hint}(\element)}$ for ``typical'' values of $\element$ and $\hint$. Then, \cref{OneRoundGame:prop:HypVctHint:BiasEstimation} combines the two step using \cref{lemma:ElementaryBound:GenericDiffBound} in order to achieve the desired bound on $\PredictAdv_{\Pc,\HintFunc}(\hint)$.

\remove{The following proposition bounds the expectation change of $\Value$ given the first step $\ElementVar$.}

\begin{proposition}\label{HyperOneRoundGameVectorHint:GameValueDiffOneRound}
	For every $\element \in \GoodElements$, it holds that
	\begin{align*}
		\size{\pr{\Value = 1} - \pr{\Value = 1 \mid \ElementVar = \element}} \leq \frac{\sigma(\element) + 2\sqrt{\hypVctLenFact \cdot \vctBaseLen}}{\hypVctLenFact \cdot \sqrt{\vctBaseLen}}.
	\end{align*}
\end{proposition}
\begin{proof}
	Note that for every $\element' \in \GoodElements$, it holds that $\pr{\Value = 1 \mid \ElementVar = \element'} = \element' = \vHyp{\hypTotalVctLen, \w(\element'), \vctBaseLen}(0)$. Therefore, by \cref{prop:hyperToNormal} it holds that
	\begin{align}\label{HyperOneRoundGameVectorHint:GameValueDiffOneRound:eq1}
	\pr{\Value = 1 \mid \ElementVar = \element'}
	&\in \Phi\left(\frac{-\frac{\w(\element')\cdot \vctBaseLen}{\hypTotalVctLen}}{\sqrt{\vctBaseLen(1-\frac{\vctBaseLen}{\hypTotalVctLen})}}\right) \pm \varphi(4) \cdot \frac{\log^{1.5} \vctBaseLen}{\sqrt{\vctBaseLen}}\\
	&= \Phi\left(-\frac{\w(\element')}{\hypVctLenFact\cdot \sqrt{\vctBaseLen\cdot (1-\frac1{\hypVctLenFact})}}\right) \pm \varphi(4) \cdot \frac{\log^{1.5} \vctBaseLen}{\sqrt{\vctBaseLen}}\nonumber
	\end{align}
	for every $\element' \in \GoodElements$. This yields that
	\begin{align*}
	\lefteqn{\size{\pr{\Value = 1} - \pr{\Value = 1 \mid \ElementVar = \element}}}\\
	&= \ex{\element' \la \ElementVar}{\pr{\Value = 1 \mid \ElementVar = \element} - \pr{\Value = 1 \mid \ElementVar = \element}}\\
	&\leq \ex{\element' \la \ElementVar}{\size{\Phi\left(-\frac{\w(\element')}{\hypVctLenFact\cdot \sqrt{\vctBaseLen}}\right) - \Phi\left(-\frac{\w(\element)}{\hypVctLenFact\cdot \sqrt{\vctBaseLen}}\right)}} + 2\varphi(4)\cdot \frac{\log^{1.5} \vctBaseLen}{\sqrt{\vctBaseLen}}\\
	&\leq \ex{\element' \la \ElementVar}{\size{\int_{\frac{\w(\element)}{\hypVctLenFact\cdot \sqrt{\vctBaseLen}}}^{\frac{\w(\element')}{\hypVctLenFact\cdot \sqrt{\vctBaseLen}}}e^{-\frac{t^2}{2}}dt}} + 2\varphi(4)\cdot \frac{\log^{1.5} \vctBaseLen}{\sqrt{\vctBaseLen}}\\
	&\leq \ex{\element' \la \ElementVar}{\size{\frac{\w(\element')}{\hypVctLenFact\cdot \sqrt{\vctBaseLen}} - \frac{\w(\element)}{\hypVctLenFact\cdot \sqrt{\vctBaseLen}}}} + 2\varphi(4)\cdot \frac{\log^{1.5} \vctBaseLen}{\sqrt{\vctBaseLen}}\\
	&= \ex{\element' \la \ElementVar}{\size{\frac{\sigma(\element')}{\hypVctLenFact\cdot \sqrt{\vctBaseLen}} - \frac{\sigma(\element)}{\hypVctLenFact\cdot \sqrt{\vctBaseLen}}}} + 2\varphi(4)\cdot \frac{\log^{1.5} \vctBaseLen}{\sqrt{\vctBaseLen}}\\
	&\leq \ex{\element' \la \ElementVar}{\frac{\size{\sigma(\element)} + \size{\sigma(\element')}}{\hypVctLenFact\cdot \sqrt{\vctBaseLen}}} + 2\varphi(4)\cdot \frac{\log^{1.5} \vctBaseLen}{\sqrt{\vctBaseLen}}\\
	&= \frac{\size{\sigma(\element)} + \ex{\vct \la (\Beroo{\eps})^{\hypTotalVctLen}}{\size{\sigma(\vct)}}}{\hypVctLenFact\cdot \sqrt{\vctBaseLen}} + 2\varphi(4)\cdot \frac{\log^{1.5} \vctBaseLen}{\sqrt{\vctBaseLen}}\\ 
	&\leq \frac{\size{\sigma(\element)} + 2\sqrt{\hypTotalVctLen}}{\hypVctLenFact\cdot \sqrt{\vctBaseLen}}.
	\end{align*}
	The second transition holds by \cref{HyperOneRoundGameVectorHint:GameValueDiffOneRound:eq1}, the  penultimate one holds since $\ex{\element' \la \ElementVar}{\size{\sigma(\element')}} = \ex{\vct \la (\Beroo{\eps})^{\hypTotalVctLen}}{\size{\sigma(\vct)}}$ and the last one by \cref{prop:binomTailExpectation}.
\end{proof}

The following proposition, which wraps the main analysis of this section, bounds how much $\ratioo_{\hint}(\element)$ can be far from $1$.

\begin{proposition}\label{OneRoundGame:lemma:HypVctHint:RatioEstimation}
	For every $\hint \in \GoodHints$ and $\element \in \GoodElements$, it holds that
	\begin{align*}
		\size{1-\ratioo_{\hint}(\element)} \leq \const \sqrt{\log \vctBaseLen \cdot \frac{\vctLenFact}{\hypVctLenFact}} \cdot \frac{\size{\sigma(\element)} + \sqrt{\hypTotalVctLen}}{\sqrt{\hypTotalVctLen}},
	\end{align*}
	for some universal constant $\const > 0$.
\end{proposition}
\begin{proof}
	By \cref{lemma:VctHint:RatioEstimation} it holds that
	\begin{align*}
	\frac{1}{\ratioo_\hint(\element)} \in \ex{\element' \la \ElementVar \mid \element' \in \GoodElements}{e^{\left(\eps_{\element'} - \eps_{\element}\right)\cdot \left(\w(\hint) - \frac{\eps_{\element'} + \eps_\element}{2}\cdot \vctTotalLen\right)}} \cdot (1 \pm \error),
	\end{align*}
	for  $\error = \max_{\element',\element'' \in \GoodElements, \z \in \pm\size{\eps^4_{\element'} - \eps^4_{\element''}}}\size{e^{\frac{\eps^3_{\element'} - \eps^3_{\element''}}{3} \cdot \w(\hint) + \z \cdot \vctTotalLen} - 1}$. Recall that $\eps_{\element'} = \sBias{\vctBaseLen}{\pr{\Value = 1 \mid \ElementVar = \element'}}$, for every $\element' \in \GoodElements$, where $\pr{\Value = 1 \mid \ElementVar = \element'} = \element' = \vHyp{\hypTotalVctLen, \w(\element'),\vctBaseLen}(0)$.
	Therefore, \cref{prop:hyperEpsEstimation} yields that
	\begin{align}\label{OneRoundGame:lemma:HypVctHint:RatioEstimation:eq1}
	\eps_{\element'}
	&\in \frac{\frac{\w(\element')\cdot \vctBaseLen}{\hypTotalVctLen}}{\sqrt{\vctBaseLen \cdot \vctBaseLen \cdot (1-\frac{\vctBaseLen}{\hypTotalVctLen})}} \pm \frac{\log^2 \vctBaseLen}{2\vctBaseLen}\\
	&= \frac{\w(\element')}{\hypTotalVctLen\cdot \sqrt{1-\frac1{\hypVctLenFact}}}\pm \frac{\log^2 \vctBaseLen}{2\vctBaseLen},\nonumber
	\end{align}
	for every $\element' \in \GoodElements$, which yields that
	\begin{align}
	\size{\eps_{\element'}}
	&\leq \frac{\size{\w(\element')}}{\hypTotalVctLen} + \frac{\log^2 \vctBaseLen}{\vctBaseLen}\\
	&= \frac{\size{\eps \cdot \hypTotalVctLen + \sigma(\element')}}{\hypTotalVctLen} + \frac{\log^2 \vctBaseLen}{\vctBaseLen}\nonumber\\
	&\leq 10 \cdot \sqrt{\frac{\log \vctBaseLen}{\vctBaseLen}},\nonumber
	\end{align}
	where the second inequality holds by the bound on $\size{\eps}$ and $\size{\sigma(\element')}$.
	Therefore, for every $\element',\element'' \in \GoodElements$ and $\z \in \pm \size{\eps^4_{\element'} - \eps^4_{\element''}}$, it holds that
	\begin{align}\label{OneRoundGame:lemma:HypVctHint:RatioEstimation:eq2}
	\size{\frac{\eps^3_{\element'} - \eps^3_{\element''}}{3} \cdot \w(\hint) + \z \cdot \vctTotalLen}
	&\leq \frac{\size{\eps_{\element'}}^3 + \size{\eps_{\element''}}^3}{3} \cdot \w(\hint) + \size{\z} \cdot \vctTotalLen\\
	&\leq \frac{2000}3 \cdot \frac{\log^{1.5} \vctBaseLen}{\vctBaseLen^{1.5}} \cdot \left(10\cdot \sqrt{\frac{\log \vctBaseLen}{\vctBaseLen}}\cdot \vctTotalLen + 4\sqrt{\log \vctBaseLen \cdot \vctTotalLen}\right) + 20000 \cdot \frac{\log^2 \vctBaseLen}{\vctBaseLen^2} \cdot \vctTotalLen\nonumber\\
	&\leq 2700\cdot \frac{\log^2 \vctBaseLen}{\sqrt{\vctBaseLen}} + 27000 \cdot \log^2 \vctBaseLen \cdot \frac{\vctLenFact}{\vctBaseLen}\nonumber\\
	&\leq 27000 \cdot \log^2 \vctBaseLen \cdot \frac{\vctLenFact + \sqrt{\vctBaseLen}}{\vctBaseLen}\nonumber\\
	&\leq 1,\nonumber
	\end{align}
	where the second inequality holds by the bounds on $\size{\eps_{\element'}}$, $\size{\eps_{\element''}}$ and $\size{\w(\hint)}$,
	and the last one holds by assumption \ref{lemma:HyperlProcessVectorHint:asmp3}.
	This yields that
	\begin{align}\label{OneRoundGame:lemma:HypVctHint:RatioEstimation:eq3}
	\error
	&= \max_{\element',\element'' \in \GoodElements, \z \in \pm\size{\eps^4_{\element'} - \eps^4_{\element''}}}\size{e^{\frac{\eps^3_{\element'} - \eps^3_{\element''}}{3} \cdot \w(\hint) + \z \cdot \vctTotalLen} - 1}\\
	&\leq \max_{\element',\element'' \in \GoodElements, \z \in \pm\size{\eps^4_{\element'} - \eps^4_{\element''}}} 2\cdot\size{\frac{\eps^3_{\element'} - \eps^3_{\element''}}{3} \cdot \w(\hint) + \z \cdot \vctTotalLen}\nonumber\\
	&\leq 54000 \cdot \log^2 \vctBaseLen \cdot \frac{\vctLenFact + \sqrt{\vctBaseLen}}{\vctBaseLen},\nonumber
	\end{align}
	where the first inequality holds since $e^{a} \in 1 \pm 2\size{a}$ for every $a \in [-1,1]$,
	and the second one holds by \cref{OneRoundGame:lemma:HypVctHint:RatioEstimation:eq2}.
	
	In addition, note that
	\begin{align}\label{OneRoundGame:lemma:HypVctHint:RatioEstimation:eq4}
	\w(\hint)
	&= \ex{\element' \la \ElementVar}{\eps_{\element'}}\cdot \vctTotalLen + \sigma(\hint)\\
	&\in \ex{\element' \la \ElementVar}{\eps_{\element'}} \cdot \vctTotalLen \pm 4\cdot \sqrt{\log \vctBaseLen \cdot \vctTotalLen}\nonumber\\
	&\in \left(\frac{\ex{\element' \la \ElementVar}{\w(\element')}}{\hypTotalVctLen \cdot \sqrt{1-\frac1{\hypVctLenFact}}} \pm \frac{\log^2 \vctBaseLen}{\vctBaseLen}\right) \cdot \vctTotalLen \pm 4\cdot \sqrt{\log \vctBaseLen \cdot \vctTotalLen}\nonumber\\
	&\in \left(\frac{\eps \cdot \hypTotalVctLen \pm 4\sqrt{\log \vctBaseLen \cdot \hypTotalVctLen}}{\hypTotalVctLen \cdot \sqrt{1-\frac1{\hypVctLenFact}}} \pm \frac{\log^2 \vctBaseLen}{\vctBaseLen}\right) \cdot \vctTotalLen \pm 4\cdot \sqrt{\log \vctBaseLen \cdot \vctTotalLen}\nonumber
	\end{align}
	where the third equality holds by \cref{OneRoundGame:lemma:HypVctHint:RatioEstimation:eq1}.
	Therefore, for every $\element' \in \GoodElements$, it holds that
	\begin{align}\label{OneRoundGame:lemma:HypVctHint:RatioEstimation:eq5}
	\lefteqn{\left(\eps_{\element'} - \eps_{\element}\right)\cdot \left(\w(\hint) - \frac{\eps_{\element'} + \eps_\element}{2}\cdot \vctTotalLen\right)}\\
	&\in \frac{\w(\element)' - \w(\element) \pm \hypVctLenFact \cdot \log^2 \vctBaseLen}{\hypTotalVctLen \cdot \sqrt{1-\frac1{\hypVctLenFact}}} \cdot \left(\w(\hint) - \frac{\w(\element) + \w(\element') \pm \hypVctLenFact \cdot \log^2 \vctBaseLen}{2\cdot \hypTotalVctLen \cdot \sqrt{1-\frac1{\hypVctLenFact}}} \cdot \vctTotalLen \right)\nonumber\\
	&\in \frac{\sigma(\element)' - \sigma(\element) \pm \hypVctLenFact \cdot \log^2 \vctBaseLen}{\hypTotalVctLen \cdot \sqrt{1-\frac1{\hypVctLenFact}}} \cdot \left(\frac{\sigma(\element) + \sigma(\element') \pm \hypVctLenFact \cdot \log^2 \vctBaseLen}{2\cdot \hypTotalVctLen \cdot \sqrt{1-\frac1{\hypVctLenFact}}} \cdot \vctTotalLen  \pm 4\cdot \sqrt{\log \vctBaseLen \cdot \vctTotalLen}\right) \nonumber\\
	&\in \frac{\sigma(\element)' - \sigma(\element) \pm \hypVctLenFact \cdot \log^2 \vctBaseLen}{\hypTotalVctLen \cdot \sqrt{1-\frac1{\hypVctLenFact}}} \cdot \left(\frac{\pm 9\cdot \sqrt{\log \vctBaseLen \cdot \hypTotalVctLen}}{2\cdot \hypVctLenFact \cdot \sqrt{1-\frac1{\hypVctLenFact}}} \cdot \vctLenFact \pm 4\cdot \sqrt{\log \vctBaseLen \cdot \vctTotalLen}\right)\nonumber\\
	&\in \frac{\sigma(\element)' - \sigma(\element) \pm \hypVctLenFact \cdot \log^2 \vctBaseLen}{\hypTotalVctLen \cdot \sqrt{1-\frac1{\hypVctLenFact}}} \cdot \left(\pm \frac{\vctLenFact}{\hypVctLenFact} \cdot 5 \cdot \sqrt{\log \vctBaseLen \cdot \hypTotalVctLen} \pm 4\cdot \sqrt{\log \vctBaseLen \cdot \vctTotalLen}\right)\nonumber\\
	&\in \frac{\sigma(\element)' - \sigma(\element) \pm \hypVctLenFact \cdot \log^2 \vctBaseLen}{\hypTotalVctLen} \cdot \left(\pm 8\cdot \sqrt{\log \vctBaseLen \cdot \vctTotalLen}\right)\nonumber\\
	&= \frac{\sigma(\element)' - \sigma(\element) \pm \hypVctLenFact \cdot \log^2 \vctBaseLen}{\sqrt{\hypTotalVctLen}} \cdot \left(\pm 8\cdot \sqrt{\log \vctBaseLen \cdot \frac{\vctLenFact}{\hypVctLenFact}}\right),\nonumber
	\end{align}
	where the first transition holds by \cref{OneRoundGame:lemma:HypVctHint:RatioEstimation:eq1}, the second one holds by \cref{OneRoundGame:lemma:HypVctHint:RatioEstimation:eq4} and the third one holds since $\element,\element' \in \GoodElements$.
	
	By the bounds on $\size{\sigma(\element)}, \size{\sigma(\element')}$ and by assumption \ref{lemma:HyperlProcessVectorHint:asmp4}, it holds that
	\begin{align*}
	\size{\left(\eps_{\element'} - \eps_{\element}\right)\cdot \left(\w(\hint) - \frac{\eps_{\element'} + \eps_\element}{2}\cdot \vctTotalLen\right)}
	\leq 1,
	\end{align*}
	for every $\element' \in \GoodElements$.
	Hence,
	\begin{align}\label{OneRoundGame:lemma:HypVctHint:RatioEstimation:eq6}
	\frac{1}{\ratioo_\hint(\element)}
	&\in \ex{\element' \la \ElementVar \mid \element' \in \GoodElements}{e^{\left(\eps_{\element'} - \eps_{\element}\right)\cdot \left(\w(\hint) - \frac{\eps_{\element'} + \eps_\element}{2}\cdot \vctTotalLen\right)}} \cdot (1 \pm \error)\\
	&\in \ex{\element' \la \ElementVar \mid \element' \in \GoodElements}{1 \pm 16\cdot \frac{\size{\sigma(\element)} + \size{\sigma(\element')} + \hypVctLenFact \cdot \log^2 \vctBaseLen}{\sqrt{\hypTotalVctLen}} \cdot \sqrt{\log \vctBaseLen \cdot \frac{\vctLenFact}{\hypVctLenFact}}}\cdot (1 \pm \error)\nonumber\\
	&\in \left(1 \pm 16\cdot \frac{\size{\sigma(\element)} + \sqrt{\hypTotalVctLen} + \hypVctLenFact \cdot \log^2 \vctBaseLen}{\sqrt{\hypTotalVctLen}} \cdot \sqrt{\log \vctBaseLen \cdot \frac{\vctLenFact}{\hypVctLenFact}}\right)\cdot \left(1 \pm 54000 \cdot \log^2 \vctBaseLen \cdot \frac{\vctLenFact + \sqrt{\vctBaseLen}}{\vctBaseLen}\right)\nonumber\\
	&\in \left(1 \pm 18\cdot \frac{\size{\sigma(\element)} + \sqrt{\hypTotalVctLen}}{\sqrt{\hypTotalVctLen}} \cdot \sqrt{\log \vctBaseLen \cdot \frac{\vctLenFact}{\hypVctLenFact}}\right),\nonumber
	\end{align}
	where the second transition holds by \cref{OneRoundGame:lemma:HypVctHint:RatioEstimation:eq5} and since $e^{a} \in 1 \pm 2\size{a}$ for every $a \in [-1,1]$, the third one holds by \cref{prop:binomTailExpectation} and the last one holds by assumptions \ref{lemma:HyperlProcessVectorHint:asmp3} and \ref{lemma:HyperlProcessVectorHint:asmp4}.
	
	By assumption \ref{lemma:HyperlProcessVectorHint:asmp4} and since $\element \in \GoodElements$ and $\frac1{1\pm a} \subseteq 1 \pm 2a$ for every $a \in (-0.5,0.5)$, we deduce from \cref{OneRoundGame:lemma:HypVctHint:RatioEstimation:eq6} that
	\begin{align}
	\ratioo_\hint(\element) \in \left(1 \pm 36\cdot \frac{\size{\sigma(\element)} + \sqrt{\hypTotalVctLen}}{\sqrt{\hypTotalVctLen}} \cdot \sqrt{\log \vctBaseLen \cdot \frac{\vctLenFact}{\hypVctLenFact}}\right).
	\end{align}
	Thus
	\begin{align*}
	\size{1-\ratioo_\hint(\element)} \leq 36\sqrt{\log \vctBaseLen \cdot \frac{\vctLenFact}{\hypVctLenFact}}\cdot \frac{\size{\sigma(\element)} + \sqrt{\hypTotalVctLen}}{\sqrt{\hypTotalVctLen}}
	\end{align*}
\end{proof}

The following proposition combines \cref{HyperOneRoundGameVectorHint:GameValueDiffOneRound} and \cref{OneRoundGame:lemma:HypVctHint:RatioEstimation} in order to achieve a bound on the prediction advantage $\PredictAdv_{\Pc,\HintFunc}(\hint)$ for ``typical'' values of $\hint$.

\begin{proposition}\label{OneRoundGame:prop:HypVctHint:BiasEstimation}
	For every $\hint \in \GoodHints$ such that $\pr{\ElementVar \notin \GoodElements \mid \HintVar = \hint} \leq \frac1{\vctBaseLen^2}$, it holds that
	\begin{align*}
		\PredictAdv_{\Pc,\HintFunc}(\hint) \leq \const \sqrt{\log \vctBaseLen} \cdot \frac{\sqrt{\vctLenFact}}{\hypVctLenFact},
	\end{align*}
	for a universal constant $\const > 0$.
\end{proposition}
\begin{proof}
	Compute
	\begin{align*}
	\PredictAdv_{\Pc,\HintFunc}(\hint)
	&\leq \ex{\element \la \GoodElements}{\size{\pr{\Value=1} - \pr{\Value=1 \mid \ElementVar = \element}} \cdot \size{1-\ratioo_{\hint}(\element)}} + \frac2{\vctBaseLen^2}\\
	&\leq \ex{\element \la \GoodElements}{\left(\frac{\size{\sigma(\element)} + 2\sqrt{\hypTotalVctLen}}{\hypVctLenFact\cdot \sqrt{\vctBaseLen}}\right) \cdot \left(\const' \sqrt{\log \vctBaseLen \cdot \frac{\vctLenFact}{\hypVctLenFact}} \cdot \frac{\size{\sigma(\element)} + \sqrt{\hypTotalVctLen}}{\sqrt{\hypTotalVctLen}}\right)} + \frac2{\vctBaseLen^2}\\
	&= \const' \sqrt{\log \vctBaseLen} \cdot \frac{\sqrt{\vctLenFact}}{\hypVctLenFact} \cdot \ex{\element \la \GoodElements}{\frac{\size{\sigma(\element)}^2 + 3\size{\sigma(\element)}\cdot\sqrt{\hypTotalVctLen} + 2\cdot \hypTotalVctLen}{\hypTotalVctLen}} + \frac2{\vctBaseLen^2}\\
	&= \const' \sqrt{\log \vctBaseLen} \cdot \frac{\sqrt{\vctLenFact}}{\hypVctLenFact} \cdot \ex{\vct \la (\Beroo{\eps})^{\hypTotalVctLen}}{\frac{\size{\sigma(\vct)}^2 + 3\size{\sigma(\vct)}\cdot\sqrt{\hypTotalVctLen} + 2\cdot \hypTotalVctLen}{\hypTotalVctLen}} + \frac2{\vctBaseLen^2}\\
	&\leq \const' \sqrt{\log \vctBaseLen} \cdot \frac{\sqrt{\vctLenFact}}{\hypVctLenFact} \cdot 6 + \frac2{\vctBaseLen^2}\\
	&\leq 7\const' \sqrt{\log \vctBaseLen} \cdot \frac{\sqrt{\vctLenFact}}{\hypVctLenFact}.
	\end{align*}
	The first inequality holds by \cref{lemma:ElementaryBound:GenericDiffBound} and since $\pr{\ElementVar \notin \GoodElements \mid \HintVar = \hint} \leq \frac1{\vctBaseLen^2}$ and $\pr{\ElementVar \notin \GoodElements} \leq \frac1{\vctBaseLen^2}$ by \cref{claim:Hoeffding} (Hoeffding's inequality), the second one holds by \cref{OneRoundGame:lemma:HypVctHint:RatioEstimation,HyperOneRoundGameVectorHint:GameValueDiffOneRound}, and the third one holds by \cref{prop:binomTailExpectation}.
\end{proof}

We are finally ready to prove \cref{lemma:HyperlProcessVectorHint}.

\mparagraph{Proving \texorpdfstring{\cref{lemma:HyperlProcessVectorHint}}{Main Lemma}}\label{subsubsec:ProvingHypVctHintLemma}
\begin{proof}[Proof of \cref{lemma:HyperlProcessVectorHint}]
	
	We divide the proof into two cases:
	\mparagraph{Case $\delta \notin [\frac1{\vctBaseLen^4}, 1-\frac1{\vctBaseLen^4}]$}
	Assume that $\delta \in [0,\frac1{\vctBaseLen^4}]$, where the proof of the case $\delta \in [1-\frac1{\vctBaseLen^4},1]$ is analogous. Assume by contradiction that
	\begin{align}\label{OneRoundGame:lemma:HyperVectorHintGame:eq0}
		\ppr{\hint \la \HintVar}{\PredictAdv_{\Pc,\HintFunc}(\hint) > \frac1{\vctBaseLen}} > \frac1{\vctBaseLen^2}
	\end{align}
	Therefore,
	\begin{align*}
		2\delta
		&= \pr{\Value=1} + \ex{\hint \la \HintVar}{\pr{\Value=1 \mid \HintVar = \hint}}\\
		&\geq \ex{\hint \la \HintVar}{\size{\pr{\Value=1} - \pr{\Value=1 \mid \HintVar = \hint}}}\\
		&\geq \frac1{\vctBaseLen^3},
	\end{align*}
	in contradiction to the assumption that $\delta \in [0,\frac1{\vctBaseLen^4}]$.
	The proof immediately follows by \cref{OneRoundGame:lemma:HyperVectorHintGame:eq0} since $\frac1{\vctBaseLen} < \frac{\sqrt{\vctLenFact}}{\hypVctLenFact}$ by assumption \ref{lemma:HyperlProcessVectorHint:asmp2} of \cref{lemma:HyperlProcessVectorHint}.
	\mparagraph{Case $\delta \in [\frac1{\vctBaseLen^4}, 1-\frac1{\vctBaseLen^4}]$}
	
	Let $\GoodHintsFinal \eqdef \set{\hint \in \GoodHints \mid \pr{\ElementVar \notin \GoodElements \mid \HintVar = \hint} \leq \frac1{\vctBaseLen^2}}$.
	Assume by contradiction that $\ppr{\hint \la \HintVar}{\pr{\ElementVar \notin \GoodElements \mid \HintVar = \hint} > \frac1{\vctBaseLen^2}} > \frac1{2\vctBaseLen^2}$. Then
	\begin{align*}
	\lefteqn{\pr{\ElementVar \notin \GoodElements}}\\
		&\geq \ppr{\hint \la \HintVar}{\ElementVar \notin \GoodElements \mid \pr{\ElementVar \notin \GoodElements \mid \HintVar = \hint} > \frac1{\vctBaseLen^2}} \cdot \ppr{\hint \la \HintVar}{\pr{\ElementVar \notin \GoodElements \mid \HintVar = \hint} > \frac1{\vctBaseLen^2}}\\
		&\geq \frac1{\vctBaseLen^2} \cdot \frac1{2\vctBaseLen^2} = \frac1{2\vctBaseLen^4},
	\end{align*}
	In contradiction to Hoeffding's inequality (\cref{claim:Hoeffding}).
	Hence,
	\begin{align}\label{OneRoundGame:lemma:HyperVectorHintGame:eq1}
		\ppr{\hint \la \HintVar}{\pr{\ElementVar \notin \GoodElements \mid \HintVar = \hint} > \frac1{\vctBaseLen^2}} \leq \frac1{2\vctBaseLen^2}
	\end{align}
	It follows that 
	\begin{align*}
		\pr{\HintVar \notin \GoodHintsFinal}
		&\leq \pr{\HintVar \notin \GoodHints} + \ppr{\hint \la \HintVar}{\pr{\ElementVar \notin \GoodElements \mid \HintVar = \hint} > \frac1{\vctBaseLen^2}} \leq \frac1{\vctBaseLen^2},
	\end{align*}
	where the last inequality holds by \cref{OneRoundGame:lemma:HyperVectorHintGame:eq1} and \cref{claim:Hoeffding} (Hoeffding's inequality).
	The  proof now  follows by \cref{OneRoundGame:prop:HypVctHint:BiasEstimation}.
\end{proof}

\newcommand{\rem}[1]{\mathsf{rem}(#1)}
\newcommand{\reml}{\mathsf{rem}(\ell)}
\newcommand{\minusOne}{{f_{\minus 1}}}
\newcommand{\vectCoin}[1]{\vect{{#1}}}
\newcommand{\sumvect}[1]{\mathsf{sum}({#1})}
\newcommand{\opt}{\mathsf{opt}}
\newcommand{\extRnd}[1]{\mathsf{\mathcal{A}ugState}_{#1}}
\newcommand{\smpRnd}[1]{\mathsf{\mathcal{S}tate}_{#1}}
\newcommand{\Greedy}{\mathsf{Greedy}}
\newcommand{\Strgy}{\mathsf{T}}
\newcommand{\offset}{{\em offset} }
\newcommand{\noHintState}{{\em no-hint state} }
\newcommand{\noHintStates}{{\em no-hint states} }
\newcommand{\finalState}{{\em final state} }
\newcommand{\finalStates}{{\em final states} }
\newcommand{\withHintState}{{\em with-hint state} }
\newcommand{\withHintStates}{{\em with-hint states} }
\newcommand{\pos}{{\mathsf{pos}}}

\section{Bounding Online-Binomial Games via Linear Programs}\label{sec:BinomialViaLP}

In this section we show how to bound online binomial games via a linear programming. In \cref{sec:additional_notions} we give additional notations and facts related to an online Binomial games (hereafter, a Binomial game). In \cref{sec:lp} we present a linear program whose feasible solution set characterizes all valid strategies for an adversary. In \cref{sec:dual-sol} we construct a feasible dual solution that bounds the binomial game that is relevant for our work. To be consistent with the common naming in the literature, in the following we sometimes refer to a player in  an online binomial game as a \emph{strategy}.

\subsection{Notation} \label{sec:additional_notions}
In this section we present the notation used in  \cref{sec:BinomialViaLP}. To make the reader life easier, we start with recalling the basic definitions from \cref{subsec:formal_binomial}.

\begin{definition}[Online  binomial games -- Restatement of \cref{def:game}]\label{bounds:def:game}
    \onlineBinomialGame
\end{definition}

\begin{definition}[Game bias -- Restatement of \cref{def:gameBias}]\label{bounds:def:gameBias}
    \biasOnlineBinomialGame
\end{definition}
Let $\game_{m,\eps,f} = \set{C_1,\ldots,C_m,f}$ be a Binomial game. In the following it be will convenient to identify a round of the game by the number of rounds left until the game ends.
Thus, referring  the  $i$-th round of $\game_m$, as \emph{level} $m-i+1$. For any level $\ell \in [m]$, let $D_\ell=C_{m-\ell+1}$, and let $\reml = (\ell-1)^2 + \ldots + 1^2 = O(\ell^3)$ be the remaining coins when at level $\ell$.

We define two types of events/states. A \noHintState $\seq{\ell,b}$ corresponds to the event that $S_{m-\ell} = b$. A \withHintState  $\seq{\ell,b, h}$ corresponds to the event that $\seq{\ell,b}$ happens and $H_{\rnd-\ell+1}= h$. In some cases, we abuse notation and refer to state $u=\seq{\ell,b,h}$ as the tuple $(\ell,b,h)$. For a set of states $S$, let $\pr{S}$ be $\pr{\bigcup_{u \in S} u}$.
For a \withHintState $u = \seq{\ell,b,h}$ or \noHintState $u = \seq{ \ell,b}$, let $\ell$ be the {\em level} of $u$, and $b$ be the \offset of $u$. For two states $u,v$, we write $u < v$ to indicate that $u$ occurred in an earlier round.
For a \withHintState $u = \seq{\ell,b,hint}$, let $u^- = \seq{\ell,b}$ be the corresponding \noHintState. For \noHintState $u = \seq{\ell,b}$, $u^-$ is the same as $u$. The final \noHintState $\seq{0,b}$ is referred to as a $f_b$. In particular, $f_{\minus 1}$ be the final state with \offset\ $\minus 1$. Let $F^\pos$ be the set of all \finalStates with positive \offset. Let $\widehat{V}$ be the union of all \withHintState and final states. 
Given some state $u$ (\withHintState or \noHintState), let $c_u \eqdef \Pr[F^\pos \mid u^-]$, and $v_u \eqdef \Pr[F^\pos\mid u]$.

We next define the \emph{final state} in which a strategy $T$ stops.

\begin{definition}[Abort state]\label{def:additional_to_strategy}
For  a strategy $\Strgy$, let $U_{\Strgy}$ be the \withHintState in which the strategy $\Strgy$ aborts, or the \finalState that the game reached if no abort occurs.
\begin{align*}
  U_{\Strgy}=
  \begin{cases}
    \seq{ \ell,b,h}  &\mbox{if \ \ $I_{\Strgy}=m-\ell+1,\: S_{m-\ell}=b,\: H_{m-\ell+1}=h$ }  \\
    \seq{ 0,b} &\mbox{if \ \ $I_{\Strgy}=m+1, \: S_m=b$ }
  \end{cases}
\end{align*}
\end{definition}

Using the above  notation,  it holds that

 \begin{align}\label{clm:bias_of_strategy}
    \bias_\Strgy = \sum_{u \in \widehat{V}} (c_u - v_u) \cdot \pr{U_\Strgy=u}
  \end{align}
for any  strategy $\Strgy$.

\subsection{The Linear Program}\label{sec:lp}
In this section we present a linear program which we show characterizes \emph{all} possible strategies $S$ of the adversary in the binomial game. The linear program captures not only deterministic strategies, but any strategy, including probabilistic strategies. Thus, finding the best strategy for the adversary is equivalent to finding the optimal solution to the linear program. The linear program and its dual appear in Figure \ref{fig:basicLP}. The primal LP has variables $a_v$ for every \withHintState $v$ that represent the probability that the strategy aborts at state $v$. The LP is, of course, specific for each family of binomial games under consideration (with its specific states, number of rounds etc.).
The following lemma shows that every strategy for the adversary induces a feasible solution for the linear program with the same value.

\begin{figure}[t!]
\begin{center}
\begin{tabular}{cc|cc}
$(P)$ \ $\max \sum_{v\in \widehat{V}}\ a_v \cdot (c_v-v_v)$& & $(D)$ \ $\min \sum_{u\in \widehat{V}} y_u \cdot \Pr[u]$ &\\
 $\st$ &&  $\st$  \\
 $a_v + \sum_{u | u<v} a_u \cdot \Pr[v|u] \leq \Pr[v]$ & $\forall v \in \widehat{V}$ &  $y_u + \sum_{v | u<v} y_v \cdot \Pr[v|u] \geq c_u-v_u$ & $\forall u \in \widehat{V}$  \\
 $a_v \geq 0$ & $\forall v \in \widehat{V}$ & $y_u \geq 0$ & $\forall u \in \widehat{V}$ \\
\end{tabular}
\end{center}
\caption{Linear program and its dual for the Binomial game $\game_{m,\eps,f}$}
\label{fig:basicLP}
\end{figure}

\begin{lemma}[Strategy to LP solution]\label{lem:feasible}
  Let $\Strgy$ be an adversarial strategy for the $m$-round binomial game $\game_{m,\eps,f}$. For any $v \in \widehat{V}$ let $a^\Strgy_v$ be the probability that the strategy aborts at state $v$, where probability is taken over the randomness of both the game and possibly the strategy (formally, $a^\Strgy_v = \pr{U_\Strgy=v}$) . Then, $a^\Strgy_v$ is a feasible solution to the linear program.  Moreover, the objective value $\sum_{v\in \widehat{V}}\ a^\Strgy_v \cdot (c_v-v_v)$ is the bias obtained by strategy $\Strgy$.
\end{lemma}

\begin{proof}
Let $\Strgy$ be an $m$-round strategy. Obviously, $a^\Strgy_v \geq 0$. Using \cref{clm:bias_of_strategy}, we have:
$$\bias_\Strgy = \sum_{v \in \widehat{V}}\ a^\Strgy_v \cdot (c_v-v_v)$$


For every states $u<v$, since the visited nodes (induced by the coins) form a Markov chain, and since the event $U_T=u$ is a random fucntion of the node $u$ and its ancestors, it holds that

\begin{align}
  \pr{v \mid U_\Strgy=u,u} = \pr{v \mid u} \label{lem:feasible:EQ1}
\end{align}
Thus, we have:
\begin{align}
  \pr{v} &\geq \sum_{u \leq v} \pr{v \mid U_\Strgy=u} \cdot  \pr{U_\Strgy=u} \label{lem:feasible:EQ2} \\
         & = \sum_{u \leq v} \pr{v \mid U_\Strgy=u,u} \cdot  \pr{U_\Strgy=u} \label{lem:feasible:EQ3} \\
         & = \sum_{u \leq v} \pr{v \mid u} \cdot  a^\Strgy_u \label{lem:feasible:EQ4}\\
         & = a^\Strgy_v + \sum_{u < v} a^\Strgy_u \cdot  \pr{v|u}. \nonumber
\end{align}

Inequality \eqref{lem:feasible:EQ2} follows by total probability on disjoint events (without the probability that $\Strgy$ does not abort until $v$'s round).
Equality \eqref{lem:feasible:EQ3} is due that the event $U_\Strgy=u$ is contained in $u$. Equality \eqref{lem:feasible:EQ4} is due \cref{lem:feasible:EQ1}. Thus, the variables satisfy the main constraint.
\end{proof}

The next lemma is a direct implication of \cref{lem:feasible} along with weak duality.

\begin{lemma}[Upper bound  on game value]\label{lem:bound_game_value}
  Let $\game=\game_{m,\eps,f}$ be a Binomial game  and let $\set{y_u\ |\ u \in \widehat{V}}$ be a feasible solution to the dual LP $(D)$ induced by $\game$. Then,
  $$ \bias_{m,\eps,f} \leq \sum_{u \in \widehat{V}} \Pr[u] \cdot y_u\enspace .$$
\end{lemma}
\begin{proof}
Consider the primal-dual LPs defined in \cref{fig:basicLP}. By Weak duality theorem the value of any feasible solution to the $(D)$ is an upper bound on the value of any feasible solution to $(P)$. By \cref{lem:feasible} for any positively aimed strategy $\Strgy$ and any feasible solution $\set{y_u\ |\ u \in \widehat{V}}$ for $(D)$,
\begin{align}
   \bias_\Strgy = \sum_{v \in \widehat{V}}\ a^\Strgy_v \cdot (c_v-v_v) \leq \sum_{u\in \widehat{V}} y_u \cdot \pr{u} \label{lem:bound_game_value:eq1}
\end{align}
Thus, $\bias_{m,\eps,f} \leq \sum_{u \in \widehat{V}} \Pr[u] \cdot y_u $.
\end{proof}

For completeness, we also show that a solution to the linear program implies a strategy for the adversary with the same value.

\begin{lemma}[LP solution to strategy]\label{lem:other_direction}
Let $a_v$ for $v\in \widehat{V}$ be a feasible solution to $(P)$. Let $\Strgy$ be a strategy that aborts at state $v$ with probability $\frac{a_v}{\Pr[v]-\sum_{u<v} a_u \cdot \Pr[v|u]}$ whenever the execution gets to state $v$ and $\Strgy$ did not abort in any previous state. Then, $\Strgy$ is a valid strategy that achieves bias of $\sum_{v\in \widehat{V}}\ a_v \cdot (c_v-v_v)$.
\end{lemma}

\begin{proof}
Let $\Strgy(v) = \frac{a_v}{\Pr[v]-\sum_{u<v} a_u \cdot \Pr[v|u]}$. First, by the constraints of $(P)$ $0\leq \Strgy(v)\leq 1$ and so the strategy define a valid conditional probability of stopping at state $v$. We prove by induction on the rounds that the strategy aborts at every state $v$ with probability $a_v$. This immediately implies (from \cref{clm:bias_of_strategy}) that the strategy has bias $\sum_{v\in \widehat{V}}\ a_v \cdot (c_v-v_v)$.
For the first round the probability that the game visits $v$ is $\Pr[v]$. Hence, the strategy aborts with probability $\Pr[v] \cdot \frac{a_v}{\Pr[v]-\sum_{u<v} a_u \cdot \Pr[v|u]}=\Pr[v] \cdot \frac{a_v}{\Pr[v]} = a_v$.
For an arbitrary state $v$ at round $k$ we have,

\begin{align}
\lefteqn{\Pr[\text{$\Strgy$ aborts at $v$}] = \Strgy(v) \cdot \Pr[\text{game visits state $v$ and did not abort at any $u<v$}]}\nonumber \\
&= \Strgy(v)\cdot \left(\Pr[\text{game visits state $v$}] - \sum_{u<v}\Pr[\text{game visits $v$}| \text{$S$ aborts at $u$}]\Pr[\text{$S$ aborts at $u$}]\right) \label{inee1}\\
&= \Strgy(v)\cdot \left(\Pr[\text{game visits state $v$}] - \sum_{u<v}\Pr[\text{game visits $v$}| \text{game visits $u$}]\Pr[\text{$S$ aborts at $u$}]\right) \label{inee2}\\
&= \Strgy(v) \cdot \left(\Pr[v]-\sum_{u<v} a_u \cdot \Pr[v|u]\right) = a_v \label{inee3}
\end{align}

Equality \eqref{inee1} follows by total probability on disjoint events. Inequality \eqref{inee2} follows since given that the game visits $u$, the probability that the strategy aborts on state $u$ is independent of the event that the game visits $v$ (that depends on coins that are tossed at later rounds. Finally, Equality \eqref{inee3} follows by the induction hypothesis.
\end{proof}

\subsection{Useful Tools}\label{sec:usefull_tools}
In this section we develop several useful tools, that are later used to analyze the dual-LP. We start with the intuitive claim that states that the best possible hint is the result of current coins.

\begin{claim}[best possible hint]\label{best_possible_hint}
Let {\sf $\game_{m,\eps,f} = \set{C_1,\ldots,C_m,f}$ } be an $m$-round online Binomial game, where $f \colon [m] \times \Z \times \Z \ra \hintSet$.
Let {\sf $\game'=\game'_{m,\eps,f'}$ } be the $m$-round online Binomial game, that uses the function $f'$, where  $f'\ : [m] \times \Z \times \Z \ra \hintSet \cup \set{-\ell^2,\ldots,\ell^2}$ is defined as follows:

\begin{align*}
	f'(i,b,z)= \begin{cases} z & z\in Z' \subset Z \\ f(i,b,z) & z \in Z \setminus Z'\end{cases}
\end{align*}
(note that in the first case, $f'$ outputs the current round coins).
Let $\{y_u\}_{u \in \widehat{V'}}$ be a feasible solution for the dual LP, induced by $\game'$. Then there exist a feasible solution $\{x_u\}_{u \in \widehat{V}}$ for the dual LP induced by $\game$, such that,
\begin{align}
  \sum_{u \in \widehat{V'}} \Pr[u] \cdot y_u = \sum_{u \in \widehat{V}} \Pr[u] \cdot x_u  \label{best_possible_hint:EQ1}
\end{align}
\end{claim}

\begin{proof}[Proof of \cref{best_possible_hint}]
  The following proves the claim for a hint function $f'$, that agrees with $f$ on all \noHintStates except of one. That is, $f'(i,b,z) = f(i,b,z)$ for all $Z$ except one coordinate.
  The validity for any $f'$ will follow by easy induction. So assume that $f$ agrees with $f'$, on all \noHintStates, except from $\seq{\ell',b'}$, and for every $z \in \Z$, $f'(\ell',b',z)=z$. Denote by $\seq{\ell,b}^+ \eqdef \set{\seq{\ell,b,h} \mid h \in \hintSet}$ that is the set of all \withHintStates with corresponding \noHintState: $\seq{\ell,b}$. We define the the solution for the dual-LP induced by $\game$, to be:
\begin{align*}
  x_u=
  \begin{cases}
    \sum_{i=-\ell^2}^{\ell^2}\ \  y_{\seq{\ell',b',i}} \cdot \pr{D_{\ell'}=i \mid u} &\mbox{if \ \ $u=\seq{\ell',b',h}, h \in \hintSet$ } \\
    y_u   &\mbox{otherwise}
  \end{cases}
\end{align*}
We start by proving that the target function has the same value in both LPs. Indeed for every $u=\seq{\ell,b,h}$, where $\seq{\ell,b} \neq \seq{\ell',b'}$, we have $\pr{u} \cdot y_u =\pr{u} \cdot x_u$, hence those states contribute the same to the sums in \cref{best_possible_hint:EQ1}. We calculate:
\begin{align*}
  \sum_{h \in \hintSet} \pr{\seq{\ell',b',h}} \cdot x_u &= \sum_{h \in \hintSet} \pr{\seq{\ell',b',h}} \cdot \sum_{i=-\ell^2}^{\ell^2}\ \  y_{\seq{\ell',b',i}} \cdot \pr{D_{\ell'}=i \mid \seq{\ell',b',h}} \\
                                                       &= \sum_{i=-\ell^2}^{\ell^2}\ \  y_{\seq{\ell',b',i}} \cdot \sum_{h \in \hintSet} \pr{\seq{\ell',b',h}} \cdot \pr{D_{\ell'}=i \mid \seq{\ell',b',h}} \\
                                                       &= \sum_{i=-\ell^2}^{\ell^2}\ \  y_{\seq{\ell',b',i}} \cdot \pr{\seq{\ell',b'}, D_{\ell'}=i}
\end{align*}
So we conclude that states of the type $\seq{\ell',b',\cdot}$ contribute the same to the sums in \cref{best_possible_hint:EQ1}, hence \cref{best_possible_hint:EQ1} follows.

Next, we prove the that $\{x_u\}_{u \in \widehat{V}}$ is a feasible solution for the dual-LP induced by $\game$. Constraints relevant to states $\seq{\ell,b,h}$, with $\ell<\ell'$, or $\ell=\ell'$, and $b \neq b'$, are trivially satisfied because they look the same in the LP induced by $\game'$. Consider now states of the form $\seq{\ell',b',h}$:
\begin{align}
  \lefteqn{x_{\seq{\ell',b',h}} + \sum_{v > \seq{\ell',b',h}}\ x_v \cdot \pr{v \mid \seq{\ell',b',h}}} \nonumber \\
  &= \sum_{i=-\ell^2}^{\ell^2}\ \  y_{\seq{\ell',b',i}} \cdot \pr{D_{\ell'}=i \mid \seq{\ell',b',h}} + \sum_{v > \seq{\ell',b',h}}\ y_v \cdot \pr{v \mid \seq{\ell',b',h}} \nonumber\\
  & =\sum_{i=-\ell^2}^{\ell^2}\ \  y_{\seq{\ell',b',i}} \cdot \pr{D_{\ell'}=i \mid \seq{\ell',b',h}} + \nonumber \\
  &  \ \ \ \ \ \ \ \ \ \ \ \ \
  + \sum_{v > \seq{\ell',b',h}}\ y_v \cdot \sum_{i=-\ell^2}^{\ell^2}\ \pr{v \mid D_{\ell'}=i,\seq{\ell',b',h}} \cdot \pr{D_{\ell'}=i \mid \seq{\ell',b',h}} \nonumber\\
  &=\sum_{i=-\ell^2}^{\ell^2}\ \  (y_{\seq{\ell',b',i}} \cdot \pr{D_{\ell'}=i \mid \seq{\ell',b',h}} + \nonumber \\
  & \ \ \ \ \ \ \ \ \ \ \ \ \ +
  \pr{D_{\ell'}=i \mid \seq{\ell',b',h}} \cdot \sum_{v > \seq{\ell',b',h}}\ y_v \cdot \pr{v \mid D_{\ell'}=i, \seq{\ell',b',h}}) \label{best_possible_hint:long:1}
\end{align}
Continuing from \cref{best_possible_hint:long:1} we get
\begin{align}
  &=\sum_{i=-\ell^2}^{\ell^2}\ \  \pr{D_{\ell'}=i \mid \seq{\ell',b',h}}
  \cdot \left(y_{\seq{\ell',b',i}} + \sum_{v > \seq{\ell',b',h}}\ y_v \cdot \pr{v \mid D_{\ell'}=i,\seq{\ell',b',h}}\right) \nonumber \\
  &=\sum_{i=-\ell^2}^{\ell^2}\ \  \pr{D_{\ell'}=i \mid \seq{\ell',b',h}}
  \cdot \left(y_{\seq{\ell',b',i}} + \sum_{v > \seq{\ell',b',h}}\ y_v \cdot \pr{v \mid D_{\ell'}=i,\seq{\ell',b'}}\right) \label{best_possible_hint:long:2}\\
  &=\sum_{i=-\ell^2}^{\ell^2}\ \  \pr{D_{\ell'}=i \mid \seq{\ell',b',h}}
  \cdot \left(y_{\seq{\ell',b',i}} + \sum_{v > \seq{\ell',b',i}}\ y_v \cdot \pr{v \mid D_{\ell'}=i,\seq{\ell',b'}}\right) \nonumber\\
  &\geq \sum_{i=-\ell^2}^{\ell^2}\ \  \pr{D_{\ell'}=i \mid \seq{\ell',b',h}}
  \cdot \left(\pr{F^\pos \mid \seq{\ell',b'}} - \pr{F^\pos \mid \seq{\ell',b'}, D_{\ell'}=i}\right) \label{best_possible_hint:long:3} \\
  &= \pr{F^\pos \mid \seq{\ell',b'}} \cdot \sum_{i=-\ell^2}^{\ell^2}\ \ \pr{D_{\ell'}=i \mid \seq{\ell',b',h}} - \nonumber  \\
  & \ \ \ \ \ \ \ \ \ \ \ \ \ - \sum_{i=-\ell^2}^{\ell^2}\ \  \pr{F^\pos \mid \seq{\ell',b',h}, D_{\ell'}=i} \cdot \pr{D_{\ell'}=i \mid \seq{\ell',b',h}} \label{best_possible_hint:long:4} \\
  &=\pr{F^\pos \mid \seq{\ell',b'}} - \pr{F^\pos \mid \seq{\ell',b',h}}\nonumber
\end{align}
Where Equality \eqref{best_possible_hint:long:2}, and Equality \eqref{best_possible_hint:long:4} we used the fact that $\pr{v \mid D_{\ell'}=i,\seq{\ell',b'}} = \pr{v \mid D_{\ell'}=i,\seq{\ell',b',h}}$ (Intuitively, once we know the value of $D_{\ell'}$, the hint $h$ gives us no more information), and in Inequality \eqref{best_possible_hint:long:3} we use the feasibility of the solution $\{y_u\}_{u \in \widehat{V'}}$.
The feasibility for states $\seq{\ell,b',h}$ for $\ell>\ell'$ involves same kind of computation, and we omit it.
\end{proof}

Recall that $S$ is a set of states. In the following we abuse notation and write $\pr{S}$ instead of $\pr{\bigcup_{u \in S} u}$.

\begin{claim}[low profit states]\label{low_profit}
  Let $\delta>0$ be a positive constant. Let $S$ be a set of \withHintStates such that $S \subset \{u=\langle \ell,b,h \rangle\ |\ c_u - v_u \leq \delta \ ,\ \ell \neq 0 \}$. Then, there are values $y_u$ ($u \in S$) such that $\sum_{u \in S} y_u \cdot \Pr[u] \leq \delta \cdot \pr{S}$ and for every state $u \in S$: $y_u + \sum_{v \in S\ :\ v>u } y_v \cdot \Pr[v|u]  \geq c_u-v_u$.
\end{claim}

\begin{proof}
  Fix some $\delta$, and $S$. Denote by $S^i$, all the states from $S$, that belong to level $i$. Define $y_u$ to be:
  \begin{align*}
    y_u = \begin{cases} \delta  &\mbox{if } u \in S^1 \\
      \delta \cdot \Pr[\overline{S^{i-1}} , \ldots , \overline{S^1}\ |\ u] & \mbox{if } u \in S^i \text{ for } i>1 \end{cases}
  \end{align*}

Where $\overline{S^i}$ are all states that are not in $S^i$.
Take some state $u \in S^i$. We have:
  \begin{align*}
    y_u + \sum_{v \in S\ :\ v>u } y_v \cdot \Pr[v|u]
    &= \delta \cdot \Pr[\overline{S^{i-1}} , \ldots , \overline{S^1}\ |\ u] + \sum_{j<i}\ \sum_{v \in S^j} \delta \cdot \Pr[\overline{S^{j-1}} , \ldots , \overline{S^1}\ |\ v] \cdot \Pr[v|u]\\
    &= \delta \cdot (\Pr[\overline{S^{i-1}} , \ldots , \overline{S^1}\ |\ u] + \sum_{j<i}\ \Pr[S^{j},\overline{S^{j-1}} , \ldots , \overline{S^1}\ |\ u] )\\
    &= \delta \geq c_u - v_u
  \end{align*}
  Also we have:
  \begin{align*}
    \sum_{u \in S} \Pr[u] \cdot y_u &= \sum_{i=1}^m \ \sum_{u \in S^i} \delta \cdot \Pr[\overline{S^{i-1}} , \ldots , \overline{S^1}\ |\ u] \cdot \Pr[u]\\
                                    &= \sum_{i=1}^m \ \sum_{u \in S^i} \delta \cdot \Pr[u , \overline{S^{i-1}} , \ldots , \overline{S^1}]\\
                                    &= \sum_{i=1}^m \ \delta \cdot \Pr[S^i , \overline{S^{i-1}} , \ldots , \overline{S^1}]\\
                                    &= \delta \cdot \Pr[S^m \cup \ldots \cup S^1]\\
                                    &=\delta \cdot \pr{S}
  \end{align*}
\end{proof}


\begin{claim}[$\frac1m$-profit states]\label{frac_1_m_profit_states}
  Let {\sf $\game_{m,\eps,f} = \set{C_1,\ldots,C_m,f}$ } be an $m$-round online Binomial game with $\abs{\eps} \leq \frac{4 \cdot \sqrt{\log m}}{m\sqrt {m}}$, and for every $\seq{i,b} \in [\rnd,\Z]$, such that $\abs{b + \eps \cdot \reml} \geq 4 \sqrt{\log{\rnd} \cdot \reml}$, and every $z \in \Z$, $f(i,b,z)=z$ \footnote{For such $\seq{i,b}$, $f$ output current round coins.}.
  Let $S \eqdef S_\pos \cup S_\negl$ the set of states such that,
\begin{align*}
  S_\pos &\eqdef \set{\seq{\ell,b,h}\ \colon\  b + \eps \cdot \reml \geq 4 \sqrt{\log{\rnd} \cdot \reml}\ ,\ - \sqrt{\log{m} \cdot \reml} \leq h} \\
  S_\negl &\eqdef \set{\seq{\ell,b,h}\ \colon\  b + \eps \cdot \reml \leq -4 \sqrt{\log{\rnd} \cdot \reml}\ ,\ h \leq \sqrt{\log{m} \cdot \reml} }
\end{align*}
 Then for every $u \in S$ we have:
  \begin{align*}
    c_u - v_u = O(\frac1m)
  \end{align*}
\end{claim}

\newcommand{\hbound}{{\mathsf{HBound}}}

\begin{proof}
    We prove that claim for $u \in S_\pos$. The case of $u \in S_\negl$ can be done similarly. Let $u=\seq{\ell,b,h} \in S_\pos$ be such a state. For $\ell \in [m]$, let $X_\ell = D_{\ell-1} + \ldots + D_1$ (informally, $X_\ell$ is the sum of the remaining coins to be toss after level $\ell$). Finally, let $\hbound \eqdef \sqrt{\log{m} \cdot \reml}$. We have:

  \begin{align*}
    c_u - v_u &= \Pr[X_{\ell}+D_{\ell}>-b] - \Pr[X_{\ell}+h>-b]\\
              &= \sum_i \Pr[X_{\ell}+i>-b] \cdot \Pr[D_{\ell}=i] - \sum_i \Pr[X_{\ell}+h>-b] \cdot \Pr[D_{\ell}=i]\\
              &= \sum_i (\Pr[X_{\ell}+i>-b]-\Pr[X_{\ell}+h>-b]) \cdot \Pr[D_{\ell}=i]\\
              &\leq \sum_i (1 - \Pr[X_{\ell}>-(b+h)]) \cdot \Pr[D_{\ell}=i]\\
              &= \sum_i (\Pr[X_{\ell} \leq -(b+h)]) \cdot \Pr[D_{\ell}=i]\\
              &\leq \Pr[X_{\ell} \leq -(b - \hbound)] \\
              &= \Pr[X_{\ell} - \eps \cdot \reml \leq -(b-\hbound+\eps \cdot \reml)]\\
              &\leq 2 \cdot e^{-\frac{((b+\eps \cdot \reml)-\hbound)^2}{2 \cdot \reml}}
  \end{align*}
  Where the final inequality follows by \cref{claim:Hoeffding}. We will show that,
\begin{align*}
  e^{-\frac{((b+\eps \cdot \reml)-\hbound)^2}{2 \cdot \reml}}  \leq \frac{1}{\rnd}
\end{align*}

Simplifying, we get that we should show that,

\begin{align}
  2 \cdot \reml \cdot \log{\rnd} + 2 (b+\eps \cdot \reml)\cdot \hbound \leq (b+\eps \cdot \reml)^2 + \hbound^2 \label{frac_1_m_profit_states:EQ1}
\end{align}
To conclude we prove that:
\begin{align*}
  2 \cdot \reml \cdot \log{\rnd} + 2 (b+\eps \cdot \reml)\cdot \hbound \leq (b+\eps \cdot \reml)^2
\end{align*}
The above holds since:  $\hbound = \sqrt{\log{m} \cdot \reml} \leq \frac14 \cdot (b+\eps \cdot \reml)$, hence $2 (b+\eps \cdot \reml)\cdot \hbound \leq \frac12 \cdot (b+\eps \cdot \reml)^2$. Also since $2\sqrt{\log{m} \cdot \reml} \leq b+\eps \cdot \reml$, we get that $2 \log{m} \cdot \reml \leq \frac12 \cdot (b+\eps \cdot \reml)^2$, so Inequality  \eqref{frac_1_m_profit_states:EQ1} holds.
\end{proof}

\begin{claim}[Trivial Satisfaction]\label{trivially_satisfied_states}
  Let $u = \seq{\ell,b,h} $ be a non \finalState. Let $\set{y_v}_{v \in \widehat{V}}$ be any assignment to the dual variables, where for each $v \in F^\pos$, $y_v \geq 0$, and $y_u \geq c_u-v_u$. Then the dual constraint of state $u$, is satisfied. That is:
  \begin{align*}
    y_u + \sum_{v:\ v>u} y_v \cdot \Pr[v|u] \geq c_u-v_u
  \end{align*}
\end{claim}

\begin{proof}
    Immediately from $y_u \geq c_u-v_u$.
\end{proof}

\newcommand{\LB}{\sqrt{\log{\rnd} \cdot \reml }}
\newcommand{\epsB}{ \frac{4\sqrt{\log{\rnd}}}{m\sqrt{m}} }

\begin{claim}[marginal states]\label{marginal_states}
  Let {\sf $\game_{m,\eps,f} = \set{C_1,\ldots,C_m,f}$ } be an $m$-round online Binomial game such that $\abs{\eps} \leq \epsB$. Let $S \eqdef \set{u=\seq{\ell,b,h}\ \colon \ \abs{b + \eps \cdot \reml} \geq 4 \sqrt{\log{\rnd} \cdot \reml}}$.
  Then, there exists an assignment of values $y_u$ for $u \in S$, that satisfies:
  \begin{align}
    \sum_{u \in S} y_u \cdot \Pr[u] \leq O(\frac1{\rnd}) \label{marginal_states:EQ1}
  \end{align}
  \begin{align}
    y_u + \sum_{v \in S\ :\ v>u } y_v \cdot \Pr[v|u]  \geq c_u-v_u & \qquad \forall u\in S\label{marginal_states:EQ2}
  \end{align}
\end{claim}

\begin{proof}
  By \cref{best_possible_hint} it enough to prove the claim for the case that the hint function $f$ simply output the coins of current state. Define the following sets:
  \begin{align*}
    S_\pos &\eqdef \set{u=\seq{\ell,b,h}\ \colon \ b + \eps \cdot \reml \geq 4 \sqrt{\log{\rnd} \cdot \reml},\ h \geq -\sqrt{\log{\rnd} \cdot \reml}}\\
    A_\pos &\eqdef \set{u=\seq{\ell,b,h}\ \colon \ b + \eps \cdot \reml \geq 4 \sqrt{\log{\rnd} \cdot \reml},\ h < -\sqrt{\log{\rnd} \cdot \reml}}\\
    S_\negl &\eqdef \set{u=\seq{\ell,b,h}\ \colon \ b + \eps \cdot \reml \leq-4 \sqrt{\log{\rnd} \cdot \reml},\ h \leq \sqrt{\log{\rnd} \cdot \reml}}\\
    A_\negl &\eqdef \set{u=\seq{\ell,b,h}\ \colon \ b + \eps \cdot \reml \leq-4 \sqrt{\log{\rnd} \cdot \reml},\ h > \sqrt{\log{\rnd} \cdot \reml}}
  \end{align*}
  Obviously $S=S_\pos \cup S_\negl \cup A_\pos \cup A_\negl$.  Also note that $S_\pos$, and $S_\negl$, are the same as in \cref{frac_1_m_profit_states}. We prove the claim for every $u \in A_\pos \cup S_\pos$. The proof for $S_\negl \cup A_\negl$ can be done similarly. Start with $S_\pos$. We define $y_u$ for $u \in S_\pos$, according to \cref{low_profit}, with $\delta=O(\frac1m)$. By \cref{frac_1_m_profit_states} we know that for each $u \in S_\pos$, $c_u-v_u \leq O(\frac1m)$. For $u \in A_\pos$, define $y_u \eqdef c_u-v_u$. By \cref{trivially_satisfied_states}, \cref{marginal_states:EQ2} holds. It is left to prove \cref{marginal_states:EQ1} where summation is over $A_\pos$.

We start with lower bounding the following expression:

\begin{align}
  \LB + \eps \cdot \ell^2 &\geq \LB - \size{\eps} \cdot \ell^2    \nonumber  \\
                          &\geq \LB - \epsB \cdot \ell^2   \nonumber  \\
                          &\geq \sqrt{\log m} \cdot (\sqrt{\reml} - 4\sqrt{\ell}) \nonumber   \\
                          &\geq \sqrt{\log m} \cdot \ell^{1.25}   \label{marginal_states:EQ3}
\end{align}

Before we continue, recall that for a set of \withHintStates $W$, $W^- \eqdef \set{u^-\ |\ u \in W}$, and $W^\ell$ is the set of $W$ that are in level $\ell$.
For every $\seq{\ell,b} \in A_\pos^-$, the following holds:
\begin{align}
  \pr{\seq{\ell,b,h} \in A_\pos^\ell\ |\ \seq{\ell,b}} &= \pr{D_{\ell} < -\LB} \nonumber \\
                                                        &= \pr{D_{\ell} - \eps \cdot \ell^2 < -(\LB + \eps \cdot \ell^2)} \nonumber \\
                                                        &\leq 2 \cdot e^{-\frac{(\LB + \eps\ell^2)^2}{2\ell^2}} \label{marginal_states:proof:1}  \\
                                                        &\leq 2 \cdot e^{-\frac{(\sqrt{\log m} \cdot \ell^{1.25})^2}{2\ell^2}} \label{marginal_states:proof:2}   \\
                                                        &= 2 \cdot e^{-\frac12 \cdot \log (m) \cdot \sqrt{\ell}} \nonumber   \\
                                                        &\leq \frac2{\rnd} \cdot e^{-\frac{\sqrt{\ell}}2}  \label{marginal_states:EQ4}
\end{align}
Where Inequality \eqref{marginal_states:proof:1} follows by \cref{claim:Hoeffding}, and Inequality \eqref{marginal_states:proof:2} follows by \cref{marginal_states:EQ3}.

Now we perform our final calculation:
{\allowdisplaybreaks
\begin{align}
  \sum_{u \in A_\pos} \Pr[u] \cdot y_u &= \sum_{u \in A_\pos} \Pr[u] \cdot (c_u-v_u)    \nonumber \\
                                  &\leq \sum_{u \in A_\pos} \Pr[u]    \nonumber \\
                                  &= \sum_{\seq{\ell,b} \in A_\pos^-}\ \ \   \sum_{h<-\LB}  \pr{\seq{\ell,b,h}}  \label{marginal_states:proof:C:1}   \\
                                  &= \sum_{\seq{\ell,b} \in A_\pos^-}\ \ \   \sum_{h<-\LB}  \Pr[\seq{\ell,b,h} | \seq{\ell,b}] \cdot \pr{\seq{\ell,b}}   \nonumber \\
                                  &= \sum_{\seq{\ell,b} \in A_\pos^-}  \Pr[\seq{\ell,b,h} \in A_\pos | \seq{\ell,b}] \cdot \pr{\seq{\ell,b}}    \nonumber   \\
                                  &= \sum_\ell \sum_{\seq{\ell,b} \in (A_\pos^\ell)^-}  \Pr[\seq{\ell,b,h} \in A_\pos^\ell | \seq{\ell,b} )] \cdot \pr{\seq{\ell,b}} \nonumber   \\
                                  &\leq \sum_\ell \sum_{\seq{\ell,b} \in (A_\pos^\ell)^-} \frac1{\rnd} \cdot e^{-\frac{\sqrt{\ell}}2} \cdot \pr{\seq{\ell,b}} \label{marginal_states:proof:C:2} \\
                                  &\leq \frac1{\rnd} \cdot \sum_\ell e^{-\frac{\sqrt{\ell}}2} \cdot  \sum_{\seq{\ell,b} \in (A_\pos^\ell)^-} \pr{\seq{\ell,b}}  \nonumber   \\
                                  &\leq \frac1{\rnd} \cdot \sum_\ell e^{-\frac{\sqrt{\ell}}2} \cdot  1  \nonumber  \\
                                  &= O(\frac1m)  \nonumber
\end{align}
}
Where Equality \eqref{marginal_states:proof:C:1} is simply by the definition of $A_\pos$, and Inequality \eqref{marginal_states:proof:C:2} is due \cref{marginal_states:EQ4}
\end{proof}


\begin{claim}[final rounds 8]\label{final_rounds_8}
  Let {\sf $\game_{m,\eps,f} = \set{C_1,\ldots,C_m,f}$ } be an $m$-round online Binomial game with $\abs{\eps} \leq \epsB$.
  Let $S \eqdef \set{u=\seq{\ell,b,h}\ \colon\ \ell \leq \sqrt[8]{m}}$.
  Then there exists an assignment of values $y_u$  for $u \in S$, that satisfies:
  \begin{align}
    \sum_{u \in S} y_u \cdot \Pr[u] \leq O(\frac1{\rnd}) \label{final_rounds_8:EQ1}
  \end{align}
  \begin{align}
    y_u + \sum_{v \in S\ :\ v>u } y_v \cdot \Pr[v|u]  \geq c_u-v_u & \qquad \forall u\in S\label{final_rounds_8:EQ2}
  \end{align}
\end{claim}

\begin{proof}
  Let $u=\seq{\ell,b,h} \in S$. By \cref{marginal_states} we may assume that
  $$- 4 \sqrt{\log{\rnd} \cdot \reml}  \leq  b + \eps \cdot \reml \leq 4 \sqrt{\log{\rnd} \cdot \reml} $$
  For $u \in S$ we set $y_u=c_u-v_u$. By \cref{trivially_satisfied_states} we know that \cref{final_rounds_8:EQ2} holds. To prove \cref{final_rounds_8:EQ1}, we calculate:
  \begin{align*}
    \sum_{u \in S'} y_u \cdot \Pr[u] &= \sum_{\ell \leq \sqrt[8]{m}} \sum_{b} (c_u-v_u) \cdot \pr{\seq{\ell,b}} \\
                                     &\leq \sum_{\ell \leq \sqrt[8]{m}} \sum_{b} \pr{\seq{\ell,b}} \\
                                     &=O(\frac{1}{m\sqrt{m}} \cdot \sqrt{\log{\rnd} \cdot \rem{\sqrt[8]{m}}} \cdot \sqrt[8]{m}) \leq O(\frac{1}{m})
  \end{align*}
\end{proof}

\begin{lemma}\label{lemma:boundingGaveValueJumpLastRounds}
  Let {\sf $\game_{m,\eps,f} = \set{C_1,\ldots,C_m,f}$ } be an $m$-round online Binomial game with $\abs{\eps} \leq 4\sqrt{\frac{\log \rnd}{\ms{1}}}$ and a hint function $f$ that simply output currents round coins.
  Let $S \eqdef \set{ \seq{\ell,b,h}\ \colon \ \abs{b + \eps \cdot \reml} \leq 4 \sqrt{\log{\rnd} \cdot \reml},\ \rnd^{\frac18} \leq \ell }$.
  Then, for every $\seq{\ell,b} \in S^-$, there exists a set $\aset_{\ell,b}$ such that the following two conditions hold:
  \begin{enumerate}
  \item
    \begin{equation*}
      \sum_{h \notin \aset_{\ell,b}} \pr{\seq{\ell,b,h}  \mid \seq{\ell,b} } \leq \frac1{\rnd^2}
    \end{equation*}
  \item For every $u=\seq{\ell,b,h}$ where $h \in \aset_{\ell,b}$, the following holds:
    $$ c_u - v_u \leq \const \cdot \ell \cdot \sqrt{\log(m)}\cdot \pr{f_{-1} \mid u} $$
  \end{enumerate}
  for some universal constant $\const > 0$.
\end{lemma}
\begin{proof}
	Let $\seq{\ell,\prevCoins} \in S^-$, let $\curRound = \rnd - \ell + 1$, let $(\ElementVar, \Value)$ be the variables of a $\bigl(\rnd, \curRound, \prevCoins, \eps\bigr)$-binomial two-step process (as defined in \cref{def:BinomialProcess}) and let $g$ be an all-information hint function for $(\ElementVar, \Value)$ (as defined in \cref{def:AllInfHint}).
	Since $\rnd$, $\curRound$, $\prevCoins$ and $\eps$ satisfy all the conditions of \cref{lemma:BinomialProcessAllInfHint}, it holds that there exists a set $\aset_{\ell,\prevCoins} \subseteq \Supp(\HintVar)$ such that
	\begin{enumerate}
		\item $\pr{g(\ElementVar) \notin \aset_{\ell,\prevCoins}} \leq \frac1{\rnd^2}$
		
		\item For every $\hint \in \aset_{\ell,\prevCoins}$, $$\frac{\size{\pr{\Value = 1} - \pr{\Value = 1 \mid g(\ElementVar) = \hint}}}{\pr{\NextCoinsVar = -(\prevCoins+1)}} \leq \const \cdot \sqrt{\ml{\curRound}} \cdot \sqrt{\log \rnd},$$ for some universal constant $\const > 0$.
	\end{enumerate}
	The proof follows since $c_u - v_u = \size{\pr{\Value = 1} - \pr{\Value = 1 \mid g(\ElementVar) = \hint}}$, $\pr{f_{-1} \mid u} = \pr{\NextCoinsVar = -(\prevCoins+1)}$ and since $\sqrt{\ml{\curRound}} = \ell$.
\end{proof}

\begin{claim}[final rounds]\label{final_rounds}
  Let {\sf $\game_{m,\eps,f} = \set{C_1,\ldots,C_m,f}$ } be an $m$-round online Binomial game with $\abs{\eps} \leq \epsB$. Let $S \eqdef \set{u=\seq{\ell,b,h}\ \colon\ \ell \leq \gamma}$ where $\gamma \in [\rnd]$.
  Then, there exists an assignment to the variables $y_u$ for $u \in S \cup \set{f_{-1}}$ that satisfies:
  \begin{align}
    \sum_{u \in S} y_u \cdot \Pr[u] \leq O(\frac1{\rnd}) \label{final_rounds:optimality}
  \end{align}
  \begin{align}
    y_{f_{-1}} \leq O(\gamma \cdot \sqrt{\log{\rnd}}) \label{final_rounds:minus_1}
  \end{align}
  \begin{align}
    y_u + \sum_{v \in S \cup \set{f_{-1}} \ :\ v>u } y_v \cdot \Pr[v|u]  \geq c_u-v_u & \qquad \forall u\in S\label{final_rounds:feasibility}
  \end{align}
\end{claim}
\begin{proof}
  By \cref{best_possible_hint} it enough to prove the claim for the case that the hint function $f$ simply outputs the coins of current state.
 By \cref{marginal_states}, and \cref{final_rounds_8} we can assume that the set $S$ is actually:
  $S' \eqdef \set{u=\seq{\ell,b,h}\ \colon\ \rnd^{\frac18} \leq \ell \leq \gamma \ ,\ \abs{b + \eps \cdot \reml} \leq 4 \sqrt{\log{\rnd} \cdot \reml}}$.

  Next, by \cref{lemma:boundingGaveValueJumpLastRounds}, we get that for every $\seq{\ell,b} \in S^-$, there exists a set $\aset_{\ell,b}$ such that the following two conditions hold:
  \begin{enumerate}
  \item
    \begin{equation}
      \sum_{h \notin \aset_{\ell,b}} \pr{\seq{\ell,b,h}  \mid \seq{\ell,b} } \leq \frac1{\rnd^2} \label{final_rounds:EQ3}
    \end{equation}
  \item For every $u=\seq{\ell,b,h}$ where $h \in \aset_{\ell,b}$, the following holds:
    \begin{align}
      c_u - v_u \leq \const \cdot \ell \cdot \sqrt{\log(m)}\cdot \pr{f_{-1} \mid u}  \label{final_rounds:EQ4}
    \end{align}
    where $\const$ is some universal constant.
  \end{enumerate}

Let,
  \begin{align*}
    L \eqdef & \set{u=\seq{\ell,b,h} \in S' \mid  h \notin \aset_{\ell,b},\ c_u-v_u > 0\ }\\
  \bar{L} \eqdef & \set{u=\seq{\ell,b,h} \in S' \mid  h \in \aset_{\ell,b}}
  \end{align*}
  The assignment of values $y_u$ for $u \in S' \cup \set{\minusOne}$ is as follows.
  For state $\minusOne$ define $y_\minusOne = \const \cdot \gamma \cdot \sqrt{\log{\rnd}}$.
  For $u \in L$, define $y_u=c_u-v_u$. For all other states $u$, define $y_u=0$. \cref{final_rounds:minus_1} is satisfied trivially. To prove that \cref{final_rounds:optimality} holds we recall that $L^\ell$ is the set of all the states in $L$ from level $\ell$. We have:
  \begin{align}
    \sum_{u \in L} \Pr[u] \cdot y_u &= \sum_{u \in L} \Pr[u] \cdot (c_u-v_u) \nonumber\\
                                    &\leq \sum_{u \in L} \Pr[u] \nonumber\\
                                    &= \sum_{\seq{\ell,b} \in L^-}\ \  \sum_{h \notin \aset_{\ell,b}} \pr{\seq{\ell,b,h}} \nonumber\\
                                    &= \sum_{\seq{\ell,b} \in L^-}\ \  \sum_{h \notin \aset_{\ell,b}} \pr{\seq{\ell,b}} \cdot \pr{\seq{\ell,b,h} \mid \seq{\ell,b}} \nonumber\\
                                    &= \sum_{\seq{\ell,b} \in L^-} \pr{\seq{\ell,b}} \sum_{h \notin \aset_{\ell,b}} \pr{\seq{\ell,b,h} \mid \seq{\ell,b}} \nonumber\\
                                    &\leq \sum_{\seq{\ell,b} \in L^-} \pr{\seq{\ell,b}} \cdot \frac1{m^2} \label{ineq-cl5131}\\
                                    &= \frac1{m^2} \sum_{\ell} \sum_{\seq{\ell,b} \in (L^\ell)^-} \pr{\seq{\ell,b}} \nonumber\\
                                    &\leq \frac1{m^2} \sum_{\ell} 1 = \frac1{m^2} \cdot m  = O(\frac1m) \nonumber
  \end{align}
  Where Inequality \eqref{ineq-cl5131} follows by Inequality \eqref{final_rounds:EQ3}. Thus, we conclude that \cref{final_rounds:optimality} holds.

  We next prove the feasibility of this solution, for states in $S'$ (\cref{final_rounds:feasibility}). For states $u \in L$ it's immediate from \cref{trivially_satisfied_states}. For states $u \in \bar{L}$ we calculate:
  \begin{align*}
    c_u-v_u &\leq  \const \cdot \ell \cdot \sqrt{\log(m)} \cdot \pr{f_{-1} \mid u}     \\
            &\leq  \const \cdot \gamma \cdot \sqrt{\log{m}} \cdot \pr{f_{-1} \mid u}   \\
            &\leq  y_{f_{\minus 1}} \cdot \pr{f_{\minus 1} \mid u} \\
            &\leq  y_u + \sum_{v:\ v>u} y_v \cdot \pr{v \mid u}
  \end{align*}
  Where the first Inequality follows by Inequality \eqref{final_rounds:EQ4}.
\end{proof}

\begin{claim}[big $\eps$]\label{big_eps}
  Let {\sf $\game_{m,\eps,f} = \set{C_1,\ldots,C_m,f}$ } be an $m$-round online Binomial game with $\abs{\eps} \geq \epsB$. Then, there are values $y_u$ (for $u \in \hat{V}$) such that,
  $\sum_{u \in \hat{V}} y_u \cdot \Pr[u] \leq O(\frac1{\rnd})$, and for every state $u \in S$: $y_u + \sum_{v \in \hat{V}\ :\ v>u } y_v \cdot \Pr[v|u]  \geq c_u-v_u$.
\end{claim}

\begin{proof}
  We prove for $\eps \geq \epsB$. The proof for $\eps \leq -\epsB$ is equivalent. First, we have that,
  \begin{align}
    \eps^2 \cdot \rem{m} \geq (\epsB)^2 \cdot \rem{m} \geq 16 \cdot \log{m} \cdot \frac{\rem{m}}{m^3} \geq 4 \cdot \log{m} \label{big_eps:EQ1}
  \end{align}

  Let $F^\negl$ be the union of all \finalStates with negative \offset. We have:
  \begin{align}
    \pr{F^\negl} &= \pr{S_m \leq 0} = \pr{S_m -\eps \cdot \rem{m} \leq -\eps \cdot \rem{m} }  \nonumber \\
                 &\leq \pr{\abs{S_m -\eps \cdot \rem{m}} \geq \eps \cdot \rem{m} } \label{big_eps:EQ3} \\
                 &\leq 2 \cdot e^{-\frac{(\eps \cdot \rem{m})^2}{2 \cdot \rem{m}}} \nonumber \\
                 &\leq 2 \cdot e^{-\frac12 \cdot \eps^2 \cdot \rem{m}} \nonumber \\
                 &\leq 2 \cdot e^{-2 \log{m}} = \frac2{m^2}  \label{big_eps:EQ2}
  \end{align}
  Where Inequality \eqref{big_eps:EQ3} is by \cref{claim:Hoeffding}, and Inequality \eqref{big_eps:EQ2} is by \cref{big_eps:EQ1}. We conclude that $\pr{F^\negl} \leq \frac2{m^2}$.

  Next we prove that $\pr{\exists u \colon \pr{F^\negl \mid u} \geq \frac1m} \leq \frac2m$,
  where by "$\exists u \colon \pr{F^\negl \mid u} \geq \frac1m$" we mean the event that the game reaches a state $u$, such that $\pr{F^\negl \mid u} \geq \frac1m$.
  Assume to the contrary that
  $\pr{\exists u \colon \pr{F^\negl \mid u} \geq \frac1m} > \frac2m$. We get:
  \begin{align*}
    \pr{F^\negl} &\geq \pr{F^\negl \ \Bigl\vert\  \exists u \colon \pr{F^\negl \mid u} \geq \frac1m} \cdot \pr{\exists u \colon \pr{F^\negl \mid u} \geq \frac1m} \\
                 &> \pr{F^\negl \ \Bigl\vert\  \exists u \colon \pr{F^\negl \mid u} \geq \frac1m} \cdot \frac2m \\
                 &\geq \frac1m \cdot \frac2m = \frac2{m^2}
  \end{align*}
  Contradicting Inequality \eqref{big_eps:EQ2}.

  Denote $S \eqdef \set{u \mid v_u < 1 - \frac1m}$. The above calculation shows that $\pr{S} \leq \frac2m$. Obviously for every $u \in S$, we have $c_u-v_u \leq \frac1m$. Define a solution for the dual LP as follow:
  \begin{itemize}
  \item For $u \notin S$, define $y_u$ by \cref{low_profit} with $\delta=\frac1m$.
  \item For $u \in S$, define $y_u$ by \cref{low_profit} with $\delta=1$.
  \end{itemize}
  By \cref{low_profit}, the above solution is feasible. Also by the same claim, and the fact that $\pr{S} \leq \frac2m$ we get that
    $\sum_{u} y_u \cdot \Pr[u] \leq O(\frac1m)$ and thus the claim follows.
\end{proof}

\subsection{Solving the Dual LP }\label{sec:dual-sol}

By the preceding discussion in \cref{sec:lp}, any feasible solution to the dual linear program in \cref{fig:basicLP} upper bounds the profit of any adversary.
In this section we construct a feasible dual solution with the desired properties.

\begin{lemma}\label{lemma:MainInequalityForLP-bigEps}[Solving the Dual LP---large $\eps$]
    Let {\sf $\game_{\rnd,\eps,f} = \set{C_1,\ldots,C_{\rnd},f}$ } be some $m$-round online Binomial game, and assume $\abs{\eps} \geq 4\sqrt{\frac{\log \rnd}{\ms{1}}}$
    then $\bias(\game) = O(\frac{1}{\rnd})$.
\end{lemma}

\begin{proof}[Proof of \cref{lemma:MainInequalityForLP-bigEps}]
    By \cref{lem:bound_game_value}, it is enough to show a feasible solution $\{y_u\}$, of the dual-LP such that: $\sum_{u \in \widehat{V}} \Pr[u] \cdot y_u = O(\frac{1}{\rnd})$.
    Since $\abs{\eps} > \epsB$ it follows immediately from \cref{big_eps}.
\end{proof}

\begin{lemma}\label{lemma:MainInequalityForLP}[Solving the Dual LP]
    Let {\sf $\game_{\rnd,\eps,f} = \set{C_1,\ldots,C_{\rnd},f}$ } be an $m$-round online Binomial game, and assume $\abs{\eps} < 4\sqrt{\frac{\log \rnd}{\ms{1}}}$.
    Let $\tau \in [\rnd]$ be such that $\frac{\tau}{\sqrt{m}} \cdot \log^3(\rnd) < 1$ and let
    \[S \eqdef \left\{\seq{\ell,b}\ \colon \ \abs{b + \eps \cdot \reml} \leq 4 \sqrt{\log{\rnd} \cdot \reml},\
    \ell \geq \max \left( \floor{\rnd^{\frac18}},\tau^2 \log^3(\rnd) \right) , b+1 \equiv \reml\ (mod\ 2) \right\}\]
    Assume that for every $\seq{\ell,b} \in S$, there exists a set $\aset_{\ell,b}$ (of hints) such that the following conditions hold:
    \begin{enumerate}
        \item
        \begin{equation}
            \sum_{h \notin \aset_{\ell,b}} \pr{\seq{\ell,b,h}  \mid \seq{\ell,b} } \leq \frac1{\rnd^2} \label{lemma:MainInequalityForLP:1}
        \end{equation}
        \item For every $u=\seq{\ell,b,h}$, where $\seq{\ell,b} \in S$ and $h \in \aset_{\ell,b}$, the following two condition holds:
        \begin{enumerate}
            \item
          \begin{equation}
              c_u-v_u \leq \const' \cdot \tau \cdot \sqrt{\rnd \log \rnd} \cdot \pr{f_{-1} \mid u^-} \label{lemma:MainInequalityForLP:2}
           \end{equation}
          where $\const'$ is some universal constant.
        \item
        \begin{equation}
            \pr{D_\ell > 9 \cdot \sqrt{\ell \cdot \log \rnd}\ \mid\ \seq{\ell,b,h}} \leq \frac{\const''}{\rnd^{12}} \label{lemma:MainInequalityForLP:3}
        \end{equation}
        where $\const''$ is some universal constant.
       \end{enumerate}
    \end{enumerate}

    Then $\bias(\game) = O(\frac{\tau \cdot \sqrt{\log \rnd}}{\rnd})$.
\end{lemma}

\newcommand{\threshold}{\tau \cdot \sqrt{\rnd}}
\begin{proof}[Proof of \cref{lemma:MainInequalityForLP}]
By \cref{lem:bound_game_value}, it is enough to show a feasible solution $\{y_u\}$, of the dual-LP such that: $\sum_{u \in \widehat{V}} \Pr[u] \cdot y_u = O(\frac{\tau \sqrt{\log \rnd}}{\rnd})$. We define the set $S'$ as follows:
\begin{align*}
    S' \eqdef \left\{ \seq{\ell,b}\ \mid \ \abs{b + \eps \cdot \reml} \leq 4 \sqrt{\log{\rnd} \cdot \reml},\
    \ell \geq \threshold ,\ b+1 \equiv \reml\ (mod\ 2) \right\}
\end{align*}

Since $\frac{\tau}{\sqrt{m}} \cdot \log^3(\rnd) < 1$, it follows that $\tau \cdot \log^3(\rnd) < \sqrt{m}$. Hence if $\ell \geq \threshold$, it implies that $\ell \geq \threshold \geq \tau^2 \cdot \log^3(\rnd)$, and so $\ell \geq \max \left( \floor{\rnd^{\frac18}},\tau^2 \log^3(\rnd) \right)$. Hence we conclude that $S' \subseteq S$, and for the rest of the proof we use the properties guaranteed for $S$ only for the states in $S'$.

We define the solution for the dual LP as follow:
\begin{enumerate}
\item For non final states $u=\langle \ell,b,h \rangle$, with: $\ell \leq \threshold$, define $y_u$ according to \cref{final_rounds} where $\gamma = \threshold$.

\item For $u=f_{\minus 1}$ (the \finalState with $b=\minus 1$) define $y_{f_{-1}}$ according to \cref{final_rounds} where $\gamma = \threshold$, with the following enhancement. By (\cref{final_rounds:minus_1}) we knows that
\begin{align}
  y_{f_{-1}} \leq O(\threshold \cdot \sqrt{\log{\rnd}})
\end{align}

Let $\const$ be a constant \st $y_{f_{-1}} \leq \const \cdot \threshold \cdot \sqrt{\log{\rnd}})$. Define $\const^{\max} \eqdef \max(\const,\const', \const'')$. We define $y_{f_{-1}}$ to be
$y_{f_{-1}} = \const^{\max} \cdot \threshold \cdot \sqrt{\log{\rnd}}$.\footnote{Since we only enlarged the value of $y_{f_{-1}}$ guaranteed to exist by \cref{final_rounds}, we know that all levels up to $\gamma=\threshold$ are covered.}

\item For non final states $u=\langle \ell,b,h \rangle$, with: $\ell \geq \threshold$, $\seq{\ell,b} \in S'$, and $h \notin \aset_{\ell,b}$, and $c_u-v_u > 0$, take $y_u=c_u-v_u$.
\item For non final states $u=\langle \ell,b,h \rangle$, with: $\ell \geq \threshold$, $\seq{\ell,b} \notin S'$, define $y_u$, according to \cref{marginal_states}.
\item for all other states $u$ in $\widehat{V}$, take $y_u=0$.
\end{enumerate}
We start by proving that:
\begin{align}
  \sum_{u \in \widehat{V}} \Pr[u] \cdot y_u = O(\frac{\tau \cdot \sqrt{\log \rnd}}{\rnd}) \label{lemma:MainInequalityForLP:optimally}
\end{align}
By \cref{prop:binomProbEstimation}:
\begin{align}
  \Pr[f_{\minus 1}] \cdot y_{f_{\minus 1}} = O(\Pr[f_{\minus 1}] \cdot \tau \cdot \sqrt{m} \cdot \sqrt{\log{m}}) = O(\frac{\tau \cdot \sqrt{\log{m}}}{m}) \label{optimally:final:minus_1}
\end{align}
By \cref{final_rounds}, and by \cref{marginal_states}, states $u$ defined in case $2$ or $4$, contribute to the sum $O(\frac1m)$. So, it remains to deal with states of the case $3$. Define
\begin{align*}
  L \eqdef  \set{u=\seq{\ell,b,h} \mid  \seq{\ell,b} \in S',  h \notin \aset_{\ell,b},\ c_u-v_u > 0\ }
\end{align*}
Recall that $L^\ell$ is the set of all the states in $L$ from level $\ell$. We have:
\begin{align}
  \sum_{u \in L} \Pr[u] \cdot y_u &= \sum_{u \in L} \Pr[u] \cdot (c_u-v_u)   \nonumber \\
                                  &\leq \sum_{u \in L} \Pr[u]   \nonumber  \\
                                  &= \sum_{\seq{\ell,b} \in L^-}\ \  \sum_{h \notin \aset_{\ell,b}} \pr{\seq{\ell,b,h}}   \nonumber  \\
                                  &= \sum_{\seq{\ell,b} \in L^-}\ \  \sum_{h \notin \aset_{\ell,b}} \pr{\seq{\ell,b}} \cdot \pr{\seq{\ell,b,h} \mid \seq{\ell,b}}   \nonumber  \\
                                  &= \sum_{\seq{\ell,b} \in L^-} \pr{\seq{\ell,b}} \sum_{h \notin \aset_{\ell,b}} \pr{\seq{\ell,b,h} \mid \seq{\ell,b}}   \nonumber  \\
                                  &\leq \sum_{\seq{\ell,b} \in L^-} \pr{\seq{\ell,b}} \cdot \frac3{m^2} \label{lemma:MainInequalityForLP:proof:1}    \\
                                  &= \frac3{m^2} \sum_{\ell} \sum_{\seq{\ell,b} \in (L^\ell)^-} \pr{\seq{\ell,b}}    \nonumber \\
                                  &\leq \frac3{m^2} \sum_{\ell} 1 = \frac3{m^2} \cdot m  = O(\frac1m)    \label{lemma:MainInequalityForLP:proof:2}
\end{align}
Where Inequality \eqref{lemma:MainInequalityForLP:proof:1} is due to \cref{lemma:MainInequalityForLP:1}. Combining \cref{lemma:MainInequalityForLP:proof:2}, and \cref{optimally:final:minus_1}, we conclude that \cref{lemma:MainInequalityForLP:optimally} holds.

We move now to prove the feasibility of our solution. For that, we need to show that for every state $u$, the following holds:
\begin{align*}
  y_u + \sum_{v:\ v>u} y_v \cdot \Pr[v|u] \geq c_u-v_u
\end{align*}
We divide the proof into 5 types of states $u$:
\begin{enumerate}
\item Final states: For final states $u$, with positive or negative \offset we have $c_u-v_u=0$, so the constraint holds.
\item For non final states $u=\langle \ell,b,h \rangle$, with $\ell \leq \threshold$: the feasibility follows immediately from \cref{final_rounds}.
\remove{Since $\ell \leq 10000 \cdot \tau \cdot \log^2 \rnd < \tau \cdot \sqrt{\rnd}$, and since we defined $y_{f_{-1}} = \tau \cdot \sqrt{m} \cdot \sqrt{\log{m}}$, }
\item For non final states $u=\langle \ell,b,h \rangle$, with: $\ell \geq \threshold$, and $\seq{\ell,b} \notin S'$, feasibility follows from \cref{marginal_states}.
\item For non final states $u=\langle \ell,b,h \rangle$, with: $\ell \geq \threshold$,\ $\seq{\ell,b} \in S'$,\  $h \notin \aset_{i,b}$, and $c_u-v_u > 0$: follows immediately from \cref{trivially_satisfied_states}.
\item For non final states $u=\langle \ell,b,h \rangle$, with: $\ell \geq \threshold$,\ $\seq{\ell,b} \in S'$,\ $h \in \aset_{i,b}$, and $c_u-v_u > 0$, we prove below.
\end{enumerate}

Consider some state $u = \langle \ell,b,h \rangle$ as defined in the case $5$. We first prove that there exist a constant $\nu$, such that $\pr{f_{\minus 1} \mid u^-} \leq \nu \cdot \pr{f_{\minus 1} \mid u}$.\\

\newcommand{\goodcond}{\size{i} \leq 9 \cdot \sqrt{\ell \cdot \log(\rnd)}}
\newcommand{\badcond}{\size{i} > 9 \cdot \sqrt{\ell \cdot \log(\rnd)}}
\begin{align}
  \pr{f_{\minus 1} \mid u} &= \sum_{i} \pr{f_{\minus 1} \mid \seq{\ell,b},D=i} \cdot \pr{D=i \mid u}  \nonumber    \\
                           &\geq \sum_{\goodcond} \pr{f_{\minus 1} \mid \seq{\ell,b},D=i} \cdot \pr{D=i \mid u}  \nonumber    \\
                           &= \sum_{\goodcond} \pr{X=-(b+i)}  \cdot \pr{D=i \mid u}  \nonumber    \\
                           &\geq \sum_{\goodcond} \nu \cdot \pr{X+D=-b} \cdot \pr{D=i \mid u}  \label{lemma:MainInequalityForLP:proof:3}    \\
                           &= \sum_{\goodcond} \nu \cdot \pr{f_{\minus 1} \mid \seq{\ell,b}} \cdot \pr{D=i \mid u}  \nonumber    \\
                           &= \nu \cdot \pr{f_{\minus 1} \mid \seq{\ell,b}} \cdot  \sum_{\goodcond} \pr{D=i \mid u}  \nonumber    \\
                           &\geq \frac12 \cdot \nu \cdot \tilde{\lambda} \cdot \pr{f_{\minus 1} \mid u^-}  \label{lemma:MainInequalityForLP:proof:4}
\end{align}
Where Inequality \eqref{lemma:MainInequalityForLP:proof:3} is because $\nu \cdot \pr{X+D=-b} \leq \pr{X=-(b+i)}$ for some constant $\nu$, and every $i$ such that $\goodcond$.  Inequality \eqref{lemma:MainInequalityForLP:proof:4} is due to \cref{lemma:MainInequalityForLP:3}.

We have:
  \begin{align}
    y_u + \sum_{v:\ v>u} y_v \cdot \pr{v \mid u} &\geq y_{f_{\minus 1}} \cdot \pr{f_{\minus 1} \mid u}  \nonumber   \\
                                                 &= \nu \cdot y_{f_{\minus 1}} \cdot \pr{f_{\minus 1} \mid u^-}   \nonumber   \\
                                                 &= \const' \cdot \tau \cdot \sqrt{m} \cdot \sqrt{\log{m}} \cdot \pr{f_{\minus 1} \mid u^-}  \nonumber \\
                                                 &\geq c_u-v_u  \label{lemma:MainInequalityForLP:proof:5}
  \end{align}
Where Inequality \eqref{lemma:MainInequalityForLP:proof:5} is due to \cref{lemma:MainInequalityForLP:2}.
\end{proof}

\subsection{Bounding Vector and Hypergeometric Games}

A main tool for this section is \cref{lemma:MainInequalityForLP}, proved in previous section. We use \cref{lemma:MainInequalityForLP} together with the tools of \cref{sec:ExpChangeDueLeakage} to prove \cref{binomial:game:vec} and \cref{binomial:game:hyp}.

\def\GameLabel{bounds:def:game}
\begin{lemma}\label{Binomial:lemma:VectorHint}[Restatement of \cref{binomial:game:vec}]
	\BinomialVectorHintInterfaceLemma
\end{lemma}
\begin{proof}
	If $\size{\eps} > 4\sqrt{\frac{\log \rnd}{\ms{1}}}$, the proof immediately follows by \cref{lemma:MainInequalityForLP-bigEps}.
	Therefore, we assume that $\size{\eps} \leq 4\sqrt{\frac{\log \rnd}{\ms{1}}}$.
	Let $S$ be as defined in \cref{lemma:MainInequalityForLP}, with respect to $\tau = \sqrt{k}$ and $\game = \game_{\rnd,\eps,f}$ for $f = \fvec{\rnd,\eps,k\cdot \ms{1}}$, as defined in \cref{def:vectorHint}. Namely, $f$ on input $(\curRound,\prevCoins,\curCoins)$ calculates $\delta = \vBeroo{\ms{\curRound+1},\eps}(-\prevCoins-\curCoins)$ and outputs a random sample from $(\Beroo{\eps})^{k\cdot \ms{1}}$, for $\eps \eqdef \sBias{\ms{1}}{\delta}$.
	
	In the following, let $\seq{\ell,\prevCoins} \in S^-$, let $\curRound = \rnd - \ell + 1$, let $\vctBaseLen = \ms{1}$, let $(\ElementVar = \CurCoinsVar, \Value)$ be the variables of a $\bigl(\rnd, \curRound, \prevCoins, \eps\bigr)$-binomial two-step process (as defined in \cref{def:BinomialProcess}) and let $g$ be a $(\vctBaseLen, k)$-vector leakage function for $(\ElementVar, \Value)$ (as defined in \cref{def:VectorLeakageFunction}).
	Since $\rnd$,$\curRound$,$\prevCoins$,$\vctBaseLen,\eps$ and $\alpha = k$ satisfy all the conditions of \cref{lemma:BinomialProcessVectorHint}, the lemma yields that there exists $\cH_{\ell, \prevCoins} \subseteq \oo^{k\cdot \ms{1}}$ such that
	\begin{enumerate}
		\item $\pr{g(\ElementVar) \notin \cH_{\ell, \prevCoins}} \leq \frac1{\rnd^2}$\label{Binomial:lemma:VectorHint:prop1}
		
		\item For every $\hint \in \cH_{\ell, \prevCoins}$,
		\begin{enumerate}
			\item $\pr{\size{\CurCoinsVar} > 9\sqrt{\log \rnd \cdot \ml{\curRound}} \mid g(\ElementVar) = \hint} \leq \frac{\const}{\rnd^{12}}$, for some universal constant $\const > 0$.\label{Binomial:lemma:VectorHint:prop2a}
			
			\item $\size{\pr{\Value = 1} - \pr{\Value = 1 \mid g(\ElementVar) = \hint}} \leq \const' \sqrt{\log \rnd \cdot k} \cdot \sqrt{\frac{\ml{\curRound}}{\rnd-\curRound+1}}\cdot \pr{\NextCoinsVar = -(\prevCoins+1)},$ for some universal constant $\const' > 0$.\label{Binomial:lemma:VectorHint:prop2b}
		\end{enumerate}
	\end{enumerate}
	By doing the translations from the notations of \cref{sec:BinomialViaLP} to the notations of \cref{sec:ExpChangeDueLeakage}, we get that $\sum_{h \notin \aset_{\ell,b}} \pr{\seq{\ell,b,h}  \mid \seq{\ell,b}} = \pr{g(\ElementVar) \notin \cH_{\ell, \prevCoins}} \leq \frac1{\rnd^2}$ and for every $u=\seq{\ell,\prevCoins,\hint}$,
	\begin{enumerate}[$\ast$]
		\item $c_u-v_u = \pr{\Value = 1} - \pr{\Value = 1 \mid g(\ElementVar) = \hint}$,
		\item $\pr{f_{-1} \mid u^-} = \pr{\NextCoinsVar = -(\prevCoins+1)}$,
		\item $\pr{D_\ell > 9 \cdot \sqrt{\ell \cdot \log \rnd}\ \mid\ \seq{\ell,b,h}} = \pr{\CurCoinsVar > 9\cdot \sqrt{\log \rnd \cdot \ml{\curRound}} \mid g(\ElementVar) = \hint}$.
	\end{enumerate}
	Therefore, combining these equalities with properties \ref{Binomial:lemma:VectorHint:prop2a} and \ref{Binomial:lemma:VectorHint:prop2b} of $\cH_{\ell, \prevCoins}$, together with the fact that $\sqrt{\frac{\ml{\curRound}}{\rnd-\curRound+1}} \leq \sqrt{\rnd}$, yields that
	\begin{enumerate}[(a)]
	\item $\pr{D_\ell > 9 \cdot \sqrt{\ell \cdot \log \rnd}\ \mid\ \seq{\ell,\prevCoins,\hint}} \leq \frac{\const}{\rnd^2}$, and
	\item $\size{c_u-v_u} \leq \const' \sqrt{k} \cdot \sqrt{\rnd \log \rnd} \cdot \pr{f_{-1} \mid u^-},$
	\end{enumerate}
	for every $u=\seq{\ell,\prevCoins,\hint}$ with $\hint \in \cH_{\ell, \prevCoins}$.
	In summary, we proved that for every $\seq{\ell,\prevCoins} \in S^-$ there exists a set $\cH_{\ell, \prevCoins}$ that satisfy the three conditions of \cref{lemma:MainInequalityForLP} with $\tau = \sqrt{k}$. Therefore, applying \cref{lemma:MainInequalityForLP} yields that $\bias_{\game} \in O(\frac{\sqrt{k}}{\rnd} \cdot \sqrt{\log \rnd})$, as required.
\end{proof}

\def\GameLabel{bounds:def:game}
\begin{lemma}\label{Binomial:lemma:HyperHint}[Restatement of \cref{binomial:game:hyp}]
	\BinomialHyperHintInterfaceLemma
\end{lemma}
\begin{proof}
	If $\size{\eps} > 4\sqrt{\frac{\log \rnd}{\ms{1}}}$, the proof immediately follows by \cref{lemma:MainInequalityForLP-bigEps}.
	Therefore, we assume that $\size{\eps} \leq 4\sqrt{\frac{\log \rnd}{\ms{1}}}$.
	
	Let $\game'$ be the binomial game  $\game_{\rnd,\eps,\HintFunc'}$ according to \cref{bounds:def:game}, where $\HintFunc'$ is a random function that on input $(\curRound,\prevCoins,\curCoins)$, samples $\hypSample$ according to $\Hyp{2\cdot \ms{1},\hypBankWeight,\ms{\curRound+1}}$ and outputs $\prevCoins + \curCoins + \hypSample$.
	Recall that $\fhyp{\rnd,\hypBankWeight}$, defined in \cref{def:hypHint}, is a random function that on input $(\curRound,\prevCoins,\curCoins)$, outputs $1$ with probability $\vHyp{2\cdot \ms{1},p,\ms{\curRound+1}}(-\prevCoins-\curCoins)$ and $-1$ otherwise.
	Note that $\fhyp{\rnd,\hypBankWeight} = f'' \circ f'$ for $f''$ that on input $z \in \Z$ output $1$ if $z \geq 0$ and $-1$ otherwise.
	Therefore, since $\fhyp{\rnd,\hypBankWeight}$ is just a function on the output of $f'$, it is enough to bound $\bias_{\game'}$ (Lemma 4.3 of \cite{HaitnerT17}).
	
	In the following, let $S$ be as defined in \cref{lemma:MainInequalityForLP}, with respect to $\tau = 1$ and $\game'$, let $\seq{\ell,\prevCoins} \in S^-$, let $\curRound = \rnd - \ell + 1$, let $(\ElementVar = \CurCoinsVar\, \Value)$ be the variables of a $\bigl(\rnd, \curRound, \prevCoins, \eps\bigr)$-binomial two-step process (as defined in \cref{def:BinomialProcess}) and let $g$ be a $\bigl(\rnd, \curRound, \prevCoins, \hypBankWeight\bigr)$-hypergeometric leakage function for $(\ElementVar, \Value)$ (as defined in \cref{def:HypHint}).
	Since $\rnd$, $\curRound$, $\prevCoins$, $\eps$, $\hypBankWeight$ and $\const$ satisfy all the conditions of \cref{lemma:BinomialProcessHyperHint}, the lemma yields that there exists a set $\cH_{\ell, \prevCoins}$ such that
	\begin{enumerate}
		\item $\pr{g(\ElementVar) \notin \cH_{\ell, \prevCoins}} \leq \frac1{\rnd^2}$\label{Binomial:lemma:HyperHint:prop1}
		
		\item For every $\hint \in \cH_{\ell, \prevCoins}$,
		\begin{enumerate}
			\item $\pr{\size{\CurCoinsVar} > 9\sqrt{\log \rnd \cdot \ml{\curRound}} \mid g(\ElementVar) = \hint} \leq \frac{\const'}{\rnd^{12}}$, for some universal constant $\const' > 0$.\label{Binomial:lemma:HyperHint:prop2a}
			
			\item $\size{\pr{\Value = 1} - \pr{\Value = 1 \mid g(\ElementVar) = \hint}} \leq \varphi(\const) \sqrt{\log \rnd} \cdot \sqrt{\frac{\ml{\curRound}}{\rnd-\curRound+1}}\cdot \pr{\NextCoinsVar = -(\prevCoins+1)},$ for some universal function $\varphi \colon \R^+ \rightarrow \R^+$.\label{Binomial:lemma:HyperHint:prop2b}
		\end{enumerate}
	\end{enumerate}
	By doing the translations from the notations of \cref{sec:BinomialViaLP} to the notations of \cref{sec:ExpChangeDueLeakage} (as done in \cref{Binomial:lemma:VectorHint}) and by combining the above properties of $\cH_{\ell, \prevCoins}$ together with the fact that $\sqrt{\frac{\ml{\curRound}}{\rnd-\curRound+1}} \leq \sqrt{\rnd}$, we get that
	\begin{enumerate}[(a)]
		\item $\pr{D_\ell > 9 \cdot \sqrt{\ell \cdot \log \rnd}\ \mid\ \seq{\ell,\prevCoins,\hint}} \leq \frac{\const'}{\rnd^2}$, and
		\item $\size{c_u-v_u} \leq \varphi(\const) \cdot \sqrt{\rnd \log \rnd} \cdot \pr{f_{-1} \mid u^-},$
	\end{enumerate}
	for every $u=\seq{\ell,\prevCoins,\hint}$ with $\hint \in \cH_{\ell, \prevCoins}$.
	In summary, we proved that for every $\seq{\ell,\prevCoins} \in S^-$ there exists a set $\cH_{\ell, \prevCoins}$ that satisfy the three conditions of \cref{lemma:MainInequalityForLP} with $\tau = 1$. Therefore, applying \cref{lemma:MainInequalityForLP} yields that $\bias_{\game'} \in O(\frac{\sqrt{\log \rnd}}{\rnd})$, as required.
\end{proof}



\bibliographystyle{abbrvnat}
\bibliography{../crypto}
\appendix
\fi

\section{Missing Proofs}\label{sec:missinProofs}
This section contains missing proofs for statement given in \cref{sec:prelim:Binomial,sec:prelim:HypGeo}.

\subsection{Properties of Bell-Like Distributions}\label{app:missinProofs:DistProp}

This section proves useful properties of "bell-like" distributions, which in particular gives useful properties on the binomial and hypergeometric distributions.

Recall that for $a\in \R$ and $b\geq 0$,  $a\pm b$ denotes for the interval $[a-b,a+b]$, and that given sets $\cs_1,\ldots,\cs_k$ and $k$-input function $f$,  $f(\cs_1,\ldots,\cs_k) = \set{f(x_1,\ldots,x_j) \colon x_i\in \cs_i}$, \eg $f(1\pm 0.1) = \set{f(x) \colon x\in [.9,1.1]}$.


\begin{definition}[bell-like distributions]\label{app:missinProofs:Dn:Properties}
	For $\rng \in \N$, $\vr \in [1,\rng]$, $\const > 0$ and $\xi > 0$, we say that a distribution $\Dst$ is a $(\rng,\vr,\const,\xi)$-bell-like distribution if
	\begin{enumerate}
		
		\item $\size{\mu} \leq \const\cdot \sqrt{\vr \log \vr}$ where $\mu \eqdef \ex{t\la \Dst}{t}$.\label{app:missinProofs:Dn:expBoundProp}
		
		\item $\ppr{t\la \Dst}{\size{t-\mu} \geq a} \leq 2\cdot e^{-\frac{a^2}{2\vr}}$ [Hoeffding's Inequality].\label{app:missinProofs:Dn:hoeffdingProp}
		
		\item $\Dst(t) = 0$ for every $t \in \Z$ with $\frac{\rng+t}{2} \notin (\rng)$.\label{app:missinProofs:Dn:zeroProp}
		
		\item $\Dst(t) \in (1 \pm \xi \cdot \frac{\log^{1.5}\vr}{\sqrt{\vr}})\cdot \sqrt{\frac2{\pi}}\cdot \frac1{\sqrt{\vr}}\cdot e^{-\frac{(t-\mu)^2}{2\vr}}$ for every $t \in \Z$ with $\size{t}\leq \const\cdot \sqrt{\vr\log \vr}$ and $\frac{\rng+t}{2} \in (\rng)$.\label{app:missinProofs:Dn:estimationProp}
	\end{enumerate}
\end{definition}

In the following, let $\rng \in \N$, $\vr \in [1,\rng]$, $\const \geq 1$ and $\xi > 0$ and let $\Dst$ be a $(\rng,\vr,\const,\xi)$-bell-like distribution with (according to \cref{app:missinProofs:Dn:Properties}) and let $\mu \eqdef \ex{t\la \Dst}{t}$.
In the following, we make some observations regards $\Dst$.

Recall that the function $\Phi\colon \R \mapsto (0,1)$ defined as $\Phi(x) \eqdef \frac{1}{\sqrt{2\pi}}\int_{x}^{\infty}e^{-\frac{t^2}{2}}dt$ is the cumulative distribution function of the standard normal distribution.

\begin{fact}[\cite{AbramowitzS64}]\label{app:missinProofs:normalBound}
	For $x \geq 0$ it holds that
	\begin{align*}
	\sqrt{\frac{2}{\pi}} \cdot \frac{e^{-\frac{x^2}{2}}}{x + \sqrt{x^2 + 4}} \leq
	\Phi(x) \leq \sqrt{\frac{2}{\pi}}\cdot \frac{e^{-\frac{x^2}{2}}}{x + \sqrt{x^2 + \frac{8}{\pi}}}.
	\end{align*}
\end{fact}

\begin{proposition}\label{app:missinProofs:estimateSumWithIntegral}
	Let $\vr \in \N$, $\mu \in \Z$ and $k,\ell \in \Z$ be such that $\ell \geq k \geq \frac{\mu}{2}$. Then
	\begin{align*}
	\abs{\sum_{t=k}^{\ell}e^{-\frac{(2t-\mu)^2}{2\vr}} - \int_{k}^{\ell} e^{-\frac{(2t-\mu)^2}{2\vr}}dt}
	\leq e^{-\frac{(2k - \mu)^2}{2\vr}}.
	\end{align*}
\end{proposition}
\begin{proof}
	See \cite{HaitnerT17}.
\end{proof}

The following proposition states the connection between a bell-like distribution and the normal distribution.

\begin{proposition}\label{app:missinProofs:generalGameValueEstimation}
	For every $k\in \Z$ with $\size{k} < \const\cdot \sqrt{\vr\log{\vr}}$, it holds that
	\begin{align*}
	\vDst(k) \in \Phi(\frac{k - \mu}{\sqrt{\vr}}) \pm \error,
	\end{align*}
	where $\error = \varphi(\xi) \cdot \frac{\log^{1.5}\vr}{\sqrt{\vr}}\cdot e^{-\frac{(k-\mu)^2}{2\vr}}$ for $\varphi(\xi) = 4\xi+5$.
\end{proposition}
\begin{proof}
	Assume for simplicity that $\rng$ and $k$ are both even, where the proofs of the other cases are analogous. Let $\ell = \ell(\const, \vr) \eqdef 4\cdot \ceil{\const\sqrt{\vr\log \vr}}< 5\const\cdot \sqrt{\vr\log \vr}$. We start by handling the case $k \geq \mu$. It holds that
	\begin{align}\label{app:missinProofs:generalGameValueEstimation:1}
	\sum_{t=k}^{\ell} \Dst(t)
	&= \sum_{t=\frac{k}{2}}^{\frac{\ell}{2}} \Dst(2t)\\
	&\in \sum_{t=\frac{k}{2}}^{\frac{\ell}{2}} \sqrt{\frac{2}{\pi}}(1 \pm \xi \cdot \frac{\log^{1.5}\vr}{\sqrt{\vr}}) \cdot \frac{1}{\sqrt{\vr}} \cdot e^{-\frac{(2t-\mu)^2}{2\vr}}\nonumber\\
	&\subseteq (1 \pm \xi \cdot \frac{\log^{1.5}\vr}{\sqrt{\vr}}) \cdot A(\vr,k,\const), \nonumber
	\end{align}
	letting $A(\vr,k,\const) \eqdef \sum_{t=\frac{k}{2}}^{\frac{\ell}{2}}  \sqrt{\frac{2}{\pi}} \cdot\frac{1}{\sqrt{\vr}} \cdot e^{-\frac{(2t-\mu)^2}{2\vr}}$. The first transition holds by property \ref{app:missinProofs:Dn:zeroProp} of $\Dst$ and the second one by property \ref{app:missinProofs:Dn:estimationProp} of $\Dst$.
	
	Compute
	\begin{align}\label{app:missinProofs:generalGameValueEstimation:2}
	A(\vr,k,\const)
	&= \sum_{t=\frac{k}{2}}^{\frac{\ell}{2}}  \sqrt{\frac{2}{\pi}} \cdot\frac{1}{\sqrt{\vr}} \cdot e^{-\frac{(2t-\mu)^2}{2\vr}}\\
	&\in \int_{\frac{k}{2}}^{\frac{\ell}{2}} \sqrt{\frac{2}{\pi}} \cdot\frac{1}{\sqrt{\vr}} \cdot e^{-\frac{(2t-\mu)^2}{2\vr}}dt \pm \frac{1}{\sqrt{\vr}}\cdot e^{-\frac{(k-\mu)^2}{2\vr}}\nonumber\\
	&= \int_{\frac{k-\mu}{\sqrt{\vr}}}^{\frac{\ell - \mu}{\sqrt{\vr}}} \frac{1}{\sqrt{2\pi}} \cdot e^{-\frac{x^2}{2}}dx \pm \frac{1}{\sqrt{\vr}}\cdot e^{-\frac{(k-\mu)^2}{2\vr}}\nonumber\\
	&= \Phi(\frac{k-\mu}{\sqrt{\vr}}) - \Phi(\frac{\ell - \mu}{\sqrt{\vr}}) \pm \frac{1}{\sqrt{\vr}}\cdot e^{-\frac{(k-\mu)^2}{2\vr}}\nonumber\\
	&\subseteq \Phi(\frac{k-\mu}{\sqrt{\vr}}) \pm \frac{1}{\vr^{4\const^2}} \pm \frac{1}{\sqrt{\vr}}\cdot e^{-\frac{(k-\mu)^2}{2\vr}}\nonumber\\
	&\subseteq \Phi(\frac{k-\mu}{\sqrt{\vr}}) \pm \frac{2}{\sqrt{\vr}}\cdot e^{-\frac{(k-\mu)^2}{2\vr}},\nonumber
	\end{align}
	where the second transition holds by \cref{app:missinProofs:estimateSumWithIntegral} (and since $k \geq \mu$), the third one holds by letting $x = \frac{2t-\mu}{\sqrt{\vr}}$, the fifth one holds by \cref{app:missinProofs:normalBound} together with property \ref{app:missinProofs:Dn:expBoundProp} of $\Dst$ which yields that $\Phi(\frac{\ell -\mu}{\sqrt{\vr}}) \leq \Phi(3\const\sqrt{\log \vr}) \leq \frac{1}{\vr^{4\const^2}}$, and the last one holds since $\frac{1}{\sqrt{\vr}}\cdot e^{-\frac{(k-\mu)^2}{2\vr}} \geq \frac{1}{\vr^{2\const^2 + \frac{1}{2}}} \geq  \frac{1}{\vr^{4\const^2}}$.
	
	Applying \cref{app:missinProofs:generalGameValueEstimation:2} on \cref{app:missinProofs:generalGameValueEstimation:1} yields that
	\begin{align}
	\sum_{t=k}^{\ell} \Dst(t)\label{app:missinProofs:gameValueEst:1}
	&\in (1 \pm \xi \cdot \frac{\log^{1.5}\vr}{\sqrt{\vr}}) \cdot (\Phi(\frac{k-\mu}{\sqrt{\vr}}) \pm \frac{2}{\sqrt{\vr}}\cdot e^{-\frac{(k-\mu)^2}{2\vr}}) \\
	&= \Phi(\frac{k-\mu}{\sqrt{\vr}}) \pm \xi \cdot \frac{\log^{1.5}\vr}{\sqrt{\vr}} \cdot \Phi(\frac{k-\mu}{\sqrt{\vr}}) \pm
	2\cdot \xi \cdot \frac{\log^{1.5}\vr}{\vr}\cdot e^{-\frac{(k-\mu)^2}{2\vr}} \pm \frac{2}{\sqrt{\vr}}\cdot e^{-\frac{(k-\mu)^2}{2\vr}}\nonumber\\
	&\subseteq \Phi(\frac{k-\mu}{\sqrt{\vr}}) \pm (3\xi + 2)\cdot \frac{\log^{1.5}\vr}{\sqrt{\vr}}\cdot e^{-\frac{(k-\mu)^2}{2\vr}},\nonumber
	\end{align}
	We conclude that
	\begin{align}
	\vDst(k) &= \sum_{t=k}^{n} \Dst(t)\label{app:missinProofs:gameValueEst:2}\\
	&= \sum_{t=k}^{\ell} \Dst(t) + \ppr{x\la \Dst}{x > \ell}\nonumber\\
	&\in \sum_{t=k}^{\ell} \Dst(t)\pm \frac{2}{\vr^{4\const^2}}\nonumber\\
	&\subseteq \left(\Phi(\frac{k-\mu}{\sqrt{\vr}}) \pm (3\xi + 2)\cdot \frac{\log^{1.5}\vr}{\sqrt{\vr}}\cdot e^{-\frac{(k-\mu)^2}{2\vr}}\right)
	\pm \frac{2}{\vr^{4\const^2}}\nonumber\\
	&\subseteq \Phi(\frac{k-\mu}{\sqrt{\vr}}) \pm (3\xi + 4)\cdot \frac{\log^{1.5}\vr}{\sqrt{\vr}}\cdot e^{-\frac{(k-\mu)^2}{2\vr}},\nonumber
	\end{align}
	where the third transition holds by property \ref{app:missinProofs:Dn:hoeffdingProp} of $\Dst$ and the fourth one holds by \cref{app:missinProofs:gameValueEst:1}.
	It is left to handle the case $k < \mu$. For such $k$, it holds that
	\begin{align}
	\vDst(k) &= 1 - (\vNegDst)(-k) + \Dst(k)\\
	&\in 1 - (\vNegDst)(-k) + (1 \pm \xi \cdot \frac{\log^{1.5}\vr}{\sqrt{\vr}})\cdot \sqrt{\frac2{\pi}}\cdot \frac1{\sqrt{\vr}}\cdot e^{-\frac{(k-\mu)^2}{2\vr}}\nonumber\\
	&\subseteq \left(1 - \Phi(\frac{-k + \mu}{\sqrt{\vr}}) \pm (3\xi + 4)\cdot \frac{\log^{1.5}\vr}{\sqrt{\vr}}\cdot e^{-\frac{(k-\mu)^2}{2\vr}}\right)
	+ (1 \pm \xi \cdot \frac{\log^{1.5}\vr}{\sqrt{\vr}})\cdot \sqrt{\frac2{\pi}}\cdot \frac1{\sqrt{\vr}}\cdot e^{-\frac{(k-\mu)^2}{2\vr}}\nonumber\\
	&\subseteq \Phi(\frac{k - \mu}{\sqrt{\vr}}) \pm (4\xi + 5)\cdot \frac{\log^{1.5}\vr}{\sqrt{\vr}}\cdot e^{-\frac{(k-\mu)^2}{2\vr}},\nonumber
	\end{align}
	where the second transition holds by property \ref{app:missinProofs:Dn:estimationProp} of $\Dst$
	and the third one holds by \cref{app:missinProofs:gameValueEst:2} applied to $-\Dst$ and $-k$ (The distribution $-\Dst$, which defined as $-\Dst(t) = \Dst(-t)$, is also a $(\rng,\vr,c,\xi)$-bell-like distribution).
\end{proof}

\begin{proposition}\label{app:missinProofs:epsDiffImproved}
	Let $n \in \N$, $\delta \in [0,1]$ and $\const > 0$ be such that $\delta \in (\frac{1}{n^\const}, 1-\frac{1}{n^\const})$. Then,
	$$\sBias{n}{\delta} \in -\frac{\Phi^{-1}(\delta)}{\sqrt{n}} \pm \error$$
	for $\error = \varphi(\const) \cdot \frac{\log^{1.5} n}{n}$ and a universal function $\varphi$.
\end{proposition}
\begin{proof}
	See \cite{HaitnerT17}.
\end{proof}

\begin{proposition}\label{app:missinProofs:generalPhiMinusOne}
	Let $\delta = \vDst(k)$ for some $k\in \Z$ with $\size{k} < \const\cdot \sqrt{\vr\log{\vr}}$. Assuming $e^{-4\xi'(\xi'+\const)\cdot\frac{\log^3 \vr}{\sqrt{\vr}}} \geq \frac12$, it holds that
	\begin{align*}
	\Phi^{-1}(\delta) \in \frac{k - \mu}{\sqrt{\vr}} \pm \error,
	\end{align*}
	for $\error = (8\xi+10) \cdot \frac{\log^{1.5} \vr}{\sqrt{\vr}}$.
\end{proposition}
\begin{proof}
	Let $\xi' = 4\xi + 5$, let $\Delta \eqdef 2\xi'\cdot \log^{1.5}\vr$ and let $k_0 \eqdef k-\mu$.
	
	We prove that $\Phi(\frac{k_0 + \Delta}{\sqrt{\vr}}) \leq \delta \leq \Phi(\frac{k_0 - \Delta}{\sqrt{\vr}})$, which yields the required bound since $\Phi$ is monotonic decreasing. We focus on the upper bound, whereas the lower bound can be proven analogously.
	Since
	\begin{align}\label{app:missinProofs:generalPhiMinusOne:1}
	\frac{\Delta}{\sqrt{\vr}} \cdot e^{-\frac{k_0^2}{2\vr}} \geq \xi' \cdot \frac{\log^{1.5}\vr}{\sqrt{\vr}}\cdot e^{-\frac{k_0^2}{2\vr}}
	\end{align}
	and
	\begin{align}\label{app:missinProofs:generalPhiMinusOne:2}
	\frac{\Delta}{\sqrt{\vr}} \cdot e^{-\frac{(k_0-\Delta)^2}{2\vr}}
	&= \frac{\Delta}{\sqrt{\vr}} \cdot e^{-\frac{k_0^2}{2\vr}} \cdot e^{\frac{2k_0\Delta - \Delta^2}{2\vr}}\\
	&\geq \frac{\Delta}{\sqrt{\vr}} \cdot e^{-\frac{k_0^2}{2\vr}} \cdot e^{-4\xi'(\xi'+\const)\cdot\frac{\log^3 \vr}{\sqrt{\vr}}}\nonumber\\
	&\geq \frac{\Delta}{\sqrt{\vr}} \cdot e^{-\frac{k_0^2}{2\vr}} \cdot \frac12\nonumber\\
	&= \xi' \cdot \frac{\log^{1.5}\vr}{\sqrt{\vr}}\cdot e^{-\frac{k_0^2}{2\vr}},\nonumber
	\end{align}
	it follows that
	\begin{align}
	\delta
	&\leq \Phi(\frac{k_0}{\sqrt{\vr}}) + \xi' \cdot \frac{\log^{1.5}\vr}{\sqrt{\vr}}\cdot e^{-\frac{k_0^2}{2\vr}}\\
	&\leq \Phi(\frac{k_0}{\sqrt{\vr}}) + \frac{\Delta}{\sqrt{\vr}}\cdot \min(e^{-\frac{k_0^2}{2\vr}}, e^{-\frac{(k_0-\Delta)^2}{2\vr}})\nonumber\\
	&\leq \Phi(\frac{k_0}{\sqrt{\vr}}) + \int_{\frac{k_0 - \Delta}{\sqrt{\vr}}}^{\frac{k_0}{\sqrt{\vr}}}e^{-\frac{t^2}{2}}dt\nonumber\\
	&= \Phi(\frac{k_0}{\sqrt{\vr}} - \frac{\Delta}{\sqrt{\vr}}),\nonumber
	\end{align}
	where the first inequality holds by \cref{app:missinProofs:generalGameValueEstimation} and the second one by \cref{app:missinProofs:generalPhiMinusOne:1} and \cref{app:missinProofs:generalPhiMinusOne:2}.
\end{proof}

\begin{proposition}\label{app:missinProofs:generalEpsEstimation}
	Let $\delta = \vDst(k)$ for some $k\in \Z$ with $\size{k} < \const\cdot \sqrt{\vr\log{\vr}}$. Assume
	\begin{enumerate}
		
		\item $\vr \geq 16$ \label{app:missinProofs:generalEpsEstimation:bound1}
		
		\item $\max(\const,\xi') \cdot \frac{\log^2 \vr}{\sqrt{\vr}} < \frac18$ \label{app:missinProofs:generalEpsEstimation:bound2}
		
		\item $e^{-4\xi'(\xi'+\const)\cdot\frac{\log^3 \vr}{\sqrt{\vr}}} \geq \frac12$, \label{app:missinProofs:generalEpsEstimation:bound3}
		
	\end{enumerate}
	where $\xi' = 4\xi + 5$ and $\varphi'$ is the function from \cref{app:missinProofs:epsDiffImproved}. Then
	\begin{align*}
	\sBias{n}{\delta} \in \frac{\mu - k}{\sqrt{n\cdot \vr}} \pm \error,
	\end{align*}
	for $\error = \bigl(\varphi'(2\const^2+1)+2\xi'\bigr) \cdot \frac{\log^{1.5} \vr}{\sqrt{n\cdot \vr}}$.
\end{proposition}
\begin{proof}
	In order to use \cref{app:missinProofs:epsDiffImproved}, we first prove that $\delta \in (\frac{1}{\vr^{2\const^2+1}}, 1-\frac{1}{\vr^{2\const^2+1}})\subseteq (\frac{1}{n^{2\const^2+1}}, 1-\frac{1}{n^{2\const^2+1}})$.
	Let $k_0 \eqdef k-\mu$. For simplicity, we assume $k_0 \geq 0$, whereas the case $k_0 < 0$ holds by symmetry.
	Compute
	\begin{align}
	\delta
	&\in \Phi(\frac{k_0}{\sqrt{\vr}}) \pm \xi'\cdot \frac{\log^{1.5}\vr}{\sqrt{\vr}}\cdot e^{-\frac{k_0^2}{2\vr}}\\
	&\subseteq \left(\frac1{\frac{k_0}{\sqrt{\vr}} + \sqrt{\frac{k_0^2}{\vr} + 4 \pm 2}} \pm \xi'\cdot \frac{\log^{1.5}\vr}{\sqrt{\vr}}\right)\cdot e^{-\frac{k_0^2}{2\vr}},\nonumber\\
	&\subseteq \frac{1\pm\frac12}{\frac{k_0}{\sqrt{\vr}} + \sqrt{\frac{k_0^2}{\vr} + 4 \pm 2}} \cdot e^{-\frac{k_0^2}{2\vr}}\nonumber\\
	&\subseteq (\frac{1}{8\const\cdot \sqrt{\log \vr}\cdot \vr^{2c^2}}, \frac34)\nonumber\\
	&\subseteq (\frac{1}{\vr^{2\const^2+1}}, 1-\frac{1}{\vr^{2\const^2+1}})\nonumber
	\end{align}
	where the first transition holds by \cref{app:missinProofs:generalGameValueEstimation}, the second one holds by \cref{app:missinProofs:normalBound}, the third one holds by condition \ref{app:missinProofs:generalEpsEstimation:bound2} and since $k_0 \leq 2\const\cdot \sqrt{\vr\log \vr}$, the fourth one also holds since $k_0 \leq 2\const\cdot \sqrt{\vr\log \vr}$ and the last one holds by conditions \ref{app:missinProofs:generalEpsEstimation:bound1} and \ref{app:missinProofs:generalEpsEstimation:bound2}.
	
	Finally, it holds that
	\begin{align}
	\sBias{n}{\delta}
	&\in -\frac{\Phi^{-1}(\delta)}{\sqrt{n}} \pm \varphi'(2\const^2+1) \cdot \frac{\log^{1.5}n}{n}\\
	&\subseteq -\frac{\left(\frac{k - \mu}{\sqrt{\vr}} \pm 2\xi' \cdot \frac{\log^{1.5} \vr}{\sqrt{\vr}}\right)}{\sqrt{n}} \pm \varphi'(2\const^2+1) \cdot \frac{\log^{1.5}n}{n}\nonumber\\
	&\subseteq \frac{\mu - k}{\sqrt{n\cdot \vr}} \pm \bigl(\varphi'(2\const^2+1)+2\xi'\bigr) \cdot \frac{\log^{1.5} \vr}{\sqrt{n\cdot \vr}},\nonumber
	\end{align}
	where the first transition holds by \cref{app:missinProofs:epsDiffImproved}, the second one by \cref{app:missinProofs:generalPhiMinusOne}
	and the last one holds since $n \geq \vr$.
\end{proof}


\subsection{Facts about binomial distribution}\label{app:missinProofs:Binomial}

Recall that for $n\in \N$ and $\eps \in [-1,1]$, we let $\Beroo{n,\eps}$ be the binomial distribution induced by the sum of $n$ independent random variables over $\oo$, each takes the value $1$ with probability $\frac{1}{2}(1+\eps)$ and $-1$ otherwise.

\begin{proposition}\label{app:missinProofs:binomTailExpectation}[Restatement of \cref{prop:binomTailExpectation}]
	\propBinomTailExpectation
\end{proposition}
\begin{proof}
	The right inequality in \cref{prop:binomTailExpectation:Item1} holds since
	\begin{align*}
		\ex{x \la \Beroo{n,\eps}}{(x-\mu)^2}
		&= \ex{x \la \Beroo{n,\eps}}{x^2} - 2 \mu \cdot \ex{x \la \Beroo{n,\eps}}{x} + \mu^2\\
		&= \Var_{x \la \Beroo{n,\eps}}[x]\\
		&= n\cdot (1-\eps^2)\\
		&\leq n,
	\end{align*}
	where the right inequality in \cref{prop:binomTailExpectation:Item2} holds since $\ex{x \la \Beroo{n,\eps}}{\size{x-\mu}} \leq \sqrt{\ex{x \la \Beroo{n,\eps}}{(x-\mu)^2}}$.
	
	The left inequality in \cref{prop:binomTailExpectation:Item2} holds since
	\begin{align*}
		\lefteqn{\ex{x \la \Beroo{n,\eps}}{\size{x-\mu}}}\\
		&= 	\ppr{x \la \Beroo{n,\eps}}{\size{x-\mu} \leq k}\cdot \ex{x \la \Beroo{n,\eps} \mid \size{x-\mu} \leq k}{\size{x-\mu}} + \ppr{x \la \Beroo{n,\eps}}{\size{x-\mu} > k} \cdot \ex{x \la \Beroo{n,\eps} \mid \size{x-\mu} > k}{\size{x-\mu}}\\
		&\geq \ppr{x \la \Beroo{n,\eps}}{\size{x-\mu} \leq k}\cdot \ex{x \la \Beroo{n,\eps} \mid \size{x-\mu} \leq k}{\size{x-\mu}} + \ppr{x \la \Beroo{n,\eps}}{\size{x-\mu} > k} \cdot \ex{x \la \Beroo{n,\eps} \mid \size{x-\mu} \leq k}{\size{x-\mu}}\\
		&= \ex{x \la \Beroo{n,\eps} \mid \size{x-\mu} \leq k}{\size{x-\mu}},
	\end{align*}
	where the left inequality in \cref{prop:binomTailExpectation:Item1} holds analogously to the above calculation.
\end{proof}

\begin{fact}\label{app:missinProofs:Hoeffding}[Restatement of \cref{claim:Hoeffding} (Hoeffding's inequality)]
	Let $n,t \in \N$ and $\eps \in [-1,1]$. Then
	\begin{align*}
		\ppr{x\la \Beroo{n,\eps}}{\abs{x-\eps n} \geq t} \leq 2e^{-\frac{t^2}{2n}}.
	\end{align*}
\end{fact}

\begin{proposition}\label{app:missinProofs:binomProbEstimation}[Restatement of \cref{prop:binomProbEstimation}]
	\propBinomProbEstimation
\end{proposition}

\begin{proposition}\label{app:missinProofs:epsDiff}[Restatement of \cref{prop:epsDiff}]
	\propBinomEps
\end{proposition}
\begin{proof}
	Let $\varphi'$ be the function from \cref{app:missinProofs:binomProbEstimation}, and let $\varphi''$ be the function from \cref{app:missinProofs:epsDiffImproved}. By \cref{app:missinProofs:Hoeffding,app:missinProofs:binomProbEstimation} and using the proposition's bounds, it follows that $\Beroo{n,\eps}$ is a $(n,n,\const,\varphi'(\const))$-bell-like distribution according to \cref{app:missinProofs:Dn:Properties}. Note that there exists a function $\vartheta \colon \R^+ \mapsto \N$ such that conditions \ref{app:missinProofs:generalEpsEstimation:bound1}, \ref{app:missinProofs:generalEpsEstimation:bound2} and \ref{app:missinProofs:generalEpsEstimation:bound3} of \cref{app:missinProofs:epsDiffImproved} holds for every $n \geq \vartheta(\const)$. In the following we focus on $n \geq \vartheta(\const)$, where smaller $n$'s are handled by setting the value of $\varphi(\const)$ to be large enough on these values. Now we can apply \cref{app:missinProofs:generalEpsEstimation} to get that
	\begin{align*}
		\sBias{n'}{\delta} \in \frac{\eps n - k}{\sqrt{n\cdot n'}} \pm \varphi''(\const) \cdot \frac{\log^{1.5} n}{\sqrt{n\cdot n'}},
	\end{align*}
	as required.
\end{proof}

\subsection{Facts About the Hypergeometric Distribution}\label{app:missinProofs:HypGeo}

Recall that for a vector $\vct \in \oo^\ast$ we let $\w(\vct) \eqdef \sum_{i \in [\size{\cI}]}\vct_i$, and given a set of indexes $\cI \subseteq [\size{\vct}]$, we let $\vct_{\cI} = (\vct_{i_1},\ldots,\vct_{i_{\size{\cI}}})$ where $i_1,\ldots,i_{\size{\cI}}$ are the ordered elements of $\cI$. In addition, recall that for $n\in \N$, $\ell \in [n]$, and an integer $p\in [-n,n]$, we define the hypergeometric probability distribution $\Hyp{n,p,\ell}$ by $\Hyp{n,p,\ell}(k) \eqdef \ppr{\cI}{\w(\vct_\cI) = k}$, where $\cI$ is an $\ell$-size set uniformly chosen from $[n]$ and $\vct \in \oo^n$ with $w(\vct)= p$. 

\begin{fact}[Hoeffding's inequality for hypergeometric distribution]\label{app:missinProofs:hyperHoeffding}
	Let $\ell \leq n \in \N$, and $p \in \Z$ with $\size{p}\ \leq n$. Then
	$$\ppr{x\la \Hyp{n,p,\ell}}{{\abs{x-\mu}} \geq t} \leq e^{-\frac{t^2}{2\ell}},$$
	for $\mu = \ex{x\la \Hyp{n,p,\ell}}{x} = \frac{\ell  \cdot p}{n}$.
\end{fact}
\begin{proof}
	Immediately follows by  \cite[Equations (10),(14)]{scala2009hypergeometric}.
\end{proof}

We use the following estimation of an almost-central binomial coefficients.
\begin{proposition}\label{app:missinProofs:binomCoeffEstimation}
	Let $n \in \N$ and $t \in \Z$ be such that $\abs{t} \leq n^{\frac{3}{5}}$ and $\frac{n+t}{2} \in (n)$. Then
	\begin{align*}
	\binom{n}{\frac{n+t}{2}} \cdot 2^{-n} \in (1 \pm \error) \cdot  \sqrt{\frac{2}{\pi}} \cdot \frac{1}{\sqrt{n}} \cdot e^{-\frac{t^2}{2n}},
	\end{align*}
	for  $\error = \xi \cdot (\frac{\abs{t}^3}{n^2} + \frac{1}{n})$ and a universal constant $\xi$.
\end{proposition}
\begin{proof}
	See \cite{HaitnerT17}.
\end{proof}

The following claim calculates $\Hyp{n,p,\ell}(t)$ using an almost-central binomial coefficients.

\begin{claim}\label{app:missinProofs:hyperProbEquation}
Let $n \in \N$, $\ell\in[n]$, $p, t\in \Z$ be such that $\size{p}\leq n^{\frac35}$, $\size{t} \leq \ell^{\frac35}$ and $t \in \Supp(\Hyp{n, p, \ell})$. Then
\begin{align*}
\Hyp{n, p, \ell}(t) = \frac{\binom{\ell}{\frac{\ell+t}{2}}\cdot \binom{n-\ell}{\frac{(n-\ell)+(p-t)}{2}}}{\binom{n}{\frac{n+p}2}}
\end{align*}
\end{claim}
\begin{proof}
By definition it holds that
\begin{align}
\Hyp{n, p, \ell}(t) = \frac{\binom{\frac{n+p}{2}}{\frac{\ell+t}{2}} \cdot \binom{\frac{n-p}{2}}{\frac{\ell-t}{2}}}{\binom{n}{\ell}}
\end{align}
Compute
\begin{align*}
\Hyp{n, p, \ell}(t) &= \frac{(\frac{n+p}{2})!}{(\frac{\ell+t}{2})!(\frac{(n+p)-(\ell+t)}{2})!} \cdot \frac{(\frac{n-p}{2})!}{(\frac{\ell-t}{2})!(\frac{(n-p)-(\ell-t)}{2})!}\cdot \frac{\ell!(n-\ell)!}{n!} \\
&= \frac{\binom{\ell}{\frac{\ell+t}{2}}\cdot \binom{n-\ell}{\frac{(n-\ell)+(p-t)}{2}}}{\binom{n}{\frac{n+p}2}},
\end{align*}
as required.
\end{proof}

The following propositions gives an estimation for the hypergeometric probability $\Hyp{n, p, \ell}(t)$ using the almost central binomial coefficients' estimation done in \cref{app:missinProofs:binomCoeffEstimation}.

\begin{proposition}\label{app:missinProofs:hyperProbEstimation}
Let $n \in \N$, $\ell\in[\floor{\frac{n}2}]$, $p, t\in \Z$ be such that $\size{p}\leq \frac14 n^{\frac35}$, $\size{t} \leq \frac14 \ell^{\frac35}$ and $t \in \Supp(\Hyp{n, p, \ell})$. Then
\begin{align*}
\Hyp{n, p, \ell}(t) = (1 \pm \error) \cdot \sqrt{\frac{2}{\pi}} \cdot \frac{1}{\sqrt{\ell(1-\frac{\ell}{n})}}\cdot e^{-\frac{(t-\frac{p\ell}{n})^2}{2\ell(1-\frac{\ell}{n})}},
\end{align*}
for $\error= \xi \cdot (\frac1{\ell} + \frac{\size{t}^3}{\ell^2} + \frac{\size{p}^3}{n^2})$ and a universal constant $\xi$.
\end{proposition}
\begin{proof}
Let $\xi'$ be the constant from \cref{app:missinProofs:binomCoeffEstimation}. In the following we focus on $n \geq 1000(1 + {\xi'}^2)$, smaller $n$'s are handled by setting the value of $\xi$ to be large enough on these values. Compute
\begin{align}
\Hyp{n, p, \ell}(t)
&= \frac{\binom{\ell}{\frac{\ell+t}{2}}\cdot \binom{n-\ell}{\frac{(n-\ell)+(p-t)}{2}}}{\binom{n}{\frac{n+p}2}}\\
&= \frac{\left(\bigl(1 \pm \xi'\cdot(\frac1{\ell} + \frac{\size{t}^3}{\ell^2})\bigr)\cdot \sqrt{\frac{2}{\pi}}\cdot \frac1{\sqrt{\ell}}\cdot e^{-\frac{t^2}{2\ell}}\right)\cdot
\left(\bigl(1 \pm \xi'\cdot(\frac1{n-\ell} + \frac{\size{p-t}^3}{(n-\ell)^2})\bigr)\cdot \sqrt{\frac{2}{\pi}}\cdot \frac1{\sqrt{n-\ell}}\cdot e^{-\frac{(p-t)^2}{2(n-\ell)}}\right)}
{\bigl(1 \pm \xi'\cdot(\frac1{n} + \frac{\size{p}^3}{n^2})\bigr)\cdot \sqrt{\frac{2}{\pi}}\cdot \frac1{\sqrt{n}}\cdot e^{-\frac{p^2}{2n}}}\nonumber\\
&= (1 \pm \error')\cdot \frac{\left(\sqrt{\frac{2}{\pi}}\cdot \frac1{\sqrt{\ell}}\cdot e^{-\frac{t^2}{2\ell}}\right)\cdot
\left(\sqrt{\frac{2}{\pi}}\cdot \frac1{\sqrt{n-\ell}}\cdot e^{-\frac{(p-t)^2}{2(n-\ell)}}\right)}
{\sqrt{\frac{2}{\pi}}\cdot \frac1{\sqrt{n}}\cdot e^{-\frac{p^2}{2n}}}\nonumber\\
&= (1 \pm \error')\cdot \sqrt{\frac2{\pi}}\cdot \frac1{\sqrt{\ell(1-\frac{\ell}{n})}} \cdot e^{-\frac{t^2}{2\ell}-\frac{(p-t)^2}{2(n-\ell)}+\frac{p^2}{n}}\nonumber\\
&= (1 \pm \error')\cdot \sqrt{\frac2{\pi}}\cdot \frac1{\sqrt{\ell(1-\frac{\ell}{n})}} \cdot e^{\frac{-t^2(1-\frac{\ell}{n})-(p-t)^2\cdot\frac{\ell}{n}+p^2\cdot\frac{\ell}{n}(1-\frac{\ell}{n})}{2\ell(1-\frac{\ell}{n})}}\nonumber\\
&= (1 \pm \error')\cdot \sqrt{\frac2{\pi}}\cdot \frac1{\sqrt{\ell(1-\frac{\ell}{n})}} \cdot e^{-\frac{(t-\frac{p\ell}{n})^2}{2\ell(1-\frac{\ell}{n})}},\nonumber
\end{align}
for $\error' = 8(\xi'+{\xi'}^2)\cdot(\frac1{\ell} + \frac{\size{t}^3}{\ell^2} + \frac1{n-\ell} + \frac{\size{p-t}^3}{(n-\ell)^2} + \frac1{n} + \frac{\size{p}^3}{n^2})$. In the second transition, the evaluation of $\binom{n-\ell}{\frac{(n-\ell)+(p-t)}{2}}$ using \cref{app:missinProofs:binomCoeffEstimation} holds since $\abs{p-t} \leq \frac12 n^{\frac35} \leq (\frac12 n)^{\frac35} \leq (n-\ell)^{\frac35}$. By letting $\error = \xi\cdot(\frac1{\ell} + \frac{\size{t}^3}{\ell^2} + \frac{\size{p}^3}{n^2})$ for $\xi = 40(\xi'+{\xi'}^2)$, we conclude that
\begin{align}
\Hyp{n, p, \ell}(t) = (1 \pm \error)\cdot \sqrt{\frac2{\pi}}\cdot \frac1{\sqrt{\ell(1-\frac{\ell}{n})}} \cdot e^{-\frac{(t-\frac{p\ell}{n})^2}{2\ell(1-\frac{\ell}{n})}},
\end{align}
as required.
\end{proof}

I case we have tighter bound on $\size{n}$ and $\size{t}$, we get the following estimation.

\begin{proposition}\label{app:missinProofs:hyperProbTightEstimation}
\propHyperProbTightEstimation
\end{proposition}
\begin{proof}
There exists a function $\vartheta \colon \R^+ \mapsto \N$ such that $\frac14 \ell^{\frac35} > \const\cdot \sqrt{\ell\log{\ell}}$ for every $\ell \geq \vartheta(\const)$. In the following we focus on $\ell \geq \max(\vartheta(\const),10)$, where smaller $\ell$'s are handled by setting the value of $\varphi(\const)$ to be large enough on these values. Let $\xi$ be the constant from \cref{app:missinProofs:hyperProbEstimation}.
Note that
\begin{align}
\xi\cdot(\frac1{\ell} + \frac{\size{t}^3}{\ell^2} + \frac{\size{p}^3}{n^2}) \leq
\xi\cdot(2\const^3 + 1) \cdot \frac{\log^{1.5} \ell}{\sqrt{\ell}}
\end{align}
Thus, the proposition holds by \cref{app:missinProofs:hyperProbEstimation} and by setting $\varphi(\const) \eqdef \xi\cdot(2\const^3 + 1)$.
\end{proof}


\begin{proposition}\label{app:missinProofs:hyperToNormal}[Restatement of \cref{prop:hyperToNormal}]
\propHyperToNormal
\end{proposition}
\begin{proof}
	Let $\varphi'$ be the function from \cref{app:missinProofs:hyperProbTightEstimation}. By \cref{app:missinProofs:hyperHoeffding,app:missinProofs:hyperProbTightEstimation} and using the proposition's bounds, it follows that $\Hyp{n, p, \ell}$ is a $(\ell,\ell(1-\frac{\ell}{n}),\const,\varphi'(\const))$-bell-like distribution according to \cref{app:missinProofs:Dn:Properties}. Therefore, by \cref{app:missinProofs:generalGameValueEstimation} it follows that
	\begin{align*}
		\vHyp{n, p, \ell}(k) \in \Phi\left(\frac{k-\frac{p\cdot \ell}{n}}{\sqrt{\ell(1-\frac{\ell}{n})}}\right) \pm (4\varphi'(\const) + 5) \cdot \frac{\log^{1.5} \ell}{\sqrt{\ell}},
	\end{align*}
	as required.
\end{proof}

\begin{proposition}\label{app:missinProofs:hyperEpsEstimation}[Restatement of \cref{prop:hyperEpsEstimation}]
\propHyperEpsEstimation
\end{proposition}
\begin{proof}
Let $\varphi'$ be the function from \cref{app:missinProofs:hyperProbTightEstimation}, and let $\varphi''$ be the function from \cref{app:missinProofs:epsDiffImproved}. By \cref{app:missinProofs:hyperHoeffding,app:missinProofs:hyperProbTightEstimation} and using the proposition's bounds, it follows that $\Hyp{n, p, \ell}$ is a $(\ell,\ell(1-\frac{\ell}{n}),\const,\varphi'(\const))$-bell-like distribution according to \cref{app:missinProofs:Dn:Properties}. Note that there exists a function $\vartheta \colon \R^+ \mapsto \N$ such that conditions \ref{app:missinProofs:generalEpsEstimation:bound1}, \ref{app:missinProofs:generalEpsEstimation:bound2} and \ref{app:missinProofs:generalEpsEstimation:bound3} of \cref{app:missinProofs:epsDiffImproved} hold for every $\ell \geq \vartheta(\const)$ (with respect to $\vr \eqdef \ell(1-\frac{\ell}{n})$ and $\xi \eqdef \varphi'(\const)$). In the following we focus on $\ell \geq \vartheta(\const)$, where smaller $\ell$'s are handled by setting the value of $\varphi(\const)$ to be large enough on these values. Now we can apply \cref{app:missinProofs:generalEpsEstimation} to get that
\begin{align*}
\sBias{m}{\delta} \in \frac{\frac{p\cdot\ell}{n} - k}{\sqrt{m\cdot \vr}} \pm \left(\varphi''(2\const^2+1)+2\cdot\bigl(4\varphi'(\const)+5\bigr)\right)\cdot \frac{\log^{1.5} \ell}{\sqrt{m\cdot \ell}},
\end{align*}
as required.
\end{proof}

\end{document}